%% file: main.tex
\pdfoutput=0

\documentclass[preprintnumbers,11pt,onecolumn]{article}
\usepackage{fullpage}

\usepackage{xcolor}
\usepackage{amsmath}
\usepackage{amssymb}
\usepackage{url} 
\usepackage{epsfig}
\usepackage{diagrams}
\newarrow{A} ---->
\newarrow{Aa} ----{>>}
\newarrow{Ad} ....>
\newarrow{Aad} ....{>>}

\usepackage{subfigure}	

\def\puncture {\psdots[dotstyle=+,dotangle=45,dotscale=1](0,0)}
\usepackage{amsthm}
\usepackage{enumerate}

\definecolor{DarkGray}{rgb}{0.1,0.1,0.5}
\usepackage[colorlinks=true,breaklinks, linkcolor= DarkGray,citecolor= DarkGray,urlcolor= DarkGray]{hyperref}	
\newtheorem{theorem}{Theorem}[section]

\newtheorem{lemma}[theorem]{Lemma}

\usepackage[usenames,dvipsnames]{pstricks}
\usepackage{pst-node}

\newcommand{\eqnref}[1]{\hyperref[#1]{{(\ref*{#1})}}}
\newcommand{\thmref}[1]{\hyperref[#1]{{Theorem~\ref*{#1}}}}
\newcommand{\lemref}[1]{\hyperref[#1]{{Lemma~\ref*{#1}}}}
\newcommand{\corref}[1]{\hyperref[#1]{{Corollary~\ref*{#1}}}}
\newcommand{\defref}[1]{\hyperref[#1]{{Definition~\ref*{#1}}}}
\newcommand{\secref}[1]{\hyperref[#1]{{Section~\ref*{#1}}}}
\newcommand{\figref}[1]{\hyperref[#1]{{Figure~\ref*{#1}}}}
\newcommand{\tabref}[1]{\hyperref[#1]{{Table~\ref*{#1}}}}
\newcommand{\remref}[1]{\hyperref[#1]{{Remark~\ref*{#1}}}}
\newcommand{\appref}[1]{\hyperref[#1]{{Appendix~\ref*{#1}}}}
\newcommand{\claimref}[1]{\hyperref[#1]{{Claim~\ref*{#1}}}}
\newcommand{\propref}[1]{\hyperref[#1]{{Proposition~\ref*{#1}}}}
\newcommand{\exampleref}[1]{\hyperref[#1]{{Example~\ref*{#1}}}}
\newcommand{\conjref}[1]{\hyperref[#1]{{Conjecture~\ref*{#1}}}}

\newcommand*{\id}{\mathbf{1}}

\newcommand*{\cB}{\mathcal{B}}
\newcommand*{\cC}{\mathcal{C}}

\newcommand*{\cD}{\mathcal{D}}

\newcommand*{\cH}{\mathcal{H}}

\newcommand*{\cM}{\mathcal{M}}

\newcommand*{\cP}{\mathcal{P}}

\newcommand*{\cT}{\mathcal{T}}

\newcommand{\C}{\mathbb{C}}

\newcommand*{\tr}{\mathsf{tr}}
\newcommand*{\ket}[1]{|#1\rangle}
\newcommand*{\bra}[1]{\langle #1|}
\newcommand*{\proj}[1]{\ket{#1}\bra{#1}}

\newcommand{\ketbra}[2]{| #1 \rangle\!\langle #2 |}

\newcommand{\braket}[2]{\langle #1|#2\rangle}       
\newcommand*{\spr}[2]{\langle #1|#2\rangle}

\newcommand{\norm}[1]{\| #1 \|}

\newcommand{\be}{\begin{equation}}
\newcommand{\ee}{\end{equation}}
\newcommand{\bea}{\begin{eqnarray}}
\newcommand{\eea}{\end{eqnarray}}
\newcommand{\bestar}{\begin{equation*}}
\newcommand{\eestar}{\end{equation*}}
\newcommand{\beastar}{\begin{eqnarray*}}
\newcommand{\eeastar}{\end{eqnarray*}}

\newcommand{\abs}[1]{{\lvert #1\rvert}}	

\newcommand{\DFib}{\mathrm{DFib}}
\newcommand{\Fib}{\mathrm{Fib}}

\newcommand*{\taubar}{\bar{\tau}}
\newcommand*{\tautau}{\boldsymbol\tau\boldsymbol\tau}
\newcommand*{\onetau}{\boldsymbol 1\boldsymbol\tau}
\newcommand*{\tauone}{\boldsymbol \tau\boldsymbol 1}
\newcommand*{\oneone}{\boldsymbol 1\boldsymbol 1} 

\newcommand*{\oneanyon}{\boldsymbol 1}
\newcommand*{\tauanyon}{\boldsymbol\tau}

\renewcommand*{\v}{{0}}
\renewcommand*{\t}{{1}}

\newcommand*{\coeff}[1]{{_{#1}}}

\input{newdefs}
\input{Qcircuit}
\def\puncture {\psdots[dotstyle=+,dotangle=45,dotscale=1](0,0)}

\def\vac{1}  
 
\begin{document}

\title{Quantum computation with Turaev-Viro codes}

\author{ 
Robert K{\"o}nig%
	\thanks{Institute for Quantum Information, California Institute of Technology}
\and
Greg Kuperberg%
	\thanks{Department of Mathematics, University of California, Davis}
\and
Ben W.~Reichardt%
	\thanks{School of Computer Science and Institute for Quantum Computing, University of Waterloo}
} 
\date{}

\maketitle

\begin{abstract}
The Turaev-Viro invariant for a closed $3$-manifold is defined as the contraction of a certain tensor network.  The tensors correspond to tetrahedra in a triangulation of the manifold, with values determined by a fixed spherical category.  For a manifold with boundary, the tensor network has free indices that can be associated to qudits, and its contraction gives the coefficients of a quantum error-correcting code.  The code has local stabilizers determined by Levin and Wen.  For example, applied to the genus-one handlebody using the ${\bf Z}_2$ category, this construction yields the well-known toric code.  

For other categories, such as the Fibonacci category, the construction realizes a non-abelian anyon model over a discrete lattice.  By studying braid group representations acting on equivalence classes of colored ribbon graphs embedded in a punctured sphere, we identify the anyons, and give a simple recipe for mapping fusion basis states of the doubled category to ribbon graphs.  We explain how suitable initial states can be prepared efficiently, how to implement braids, by successively changing the triangulation using a fixed five-qudit local unitary gate, and how to measure the topological charge.  Combined with known universality results for anyonic systems, this provides a large family of schemes for quantum computation based on local deformations of stabilizer codes.  These schemes may serve as a starting point for developing fault-tolerance schemes using continuous stabilizer measurements and active error-correction.  
\end{abstract}

\clearpage

\tableofcontents
\clearpage

\def\anyontimes{\times}
\def\anyonplus{+}
\def\anyoncong{=}

\input{intro}

\input{stringnethilbertspace}
\input{anyonsdefinition}
\input{levinwenmodel}
\input{anyoniccomputation}
\input{generalization}
\input{TVdefinition}

\subsection*{Acknowledgments}

We would like to thank Gorjan Alagic, Michael Freedman, Stephen Jordan, Alexei Kitaev, Liang Kong and John Preskill.  

R.K.\ acknowledges support by NSF grant PHY-0803371 and SNF~PA00P2-126220.  
G.K.\ acknowledges support by NSF grant DMS-0606795.  
B.R.\ acknowledges support from NSERC and ARO.  
Some of this research was conducted while two of the authors were visiting the Kavli Institute for Theoretical Physics, supported by NSF grant PHY05-51164.

\appendix
\input{anyongeneral}

\input{normalization_ip}
\input{completeness}
\input{fibonacciexamples}
\input{highgenus}

\input{levinwenappendix}
\input{partitionfunction}

\bibliographystyle{alpha-eprint}
\bibliography{q}

\end{document}

%% file: newdefs.tex
\def\h {\widehat}

\newcommand*{\horizontalt}{
\raisebox{-1em}{\includegraphics{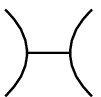}}}

\newcommand*{\horizontalv}{
\raisebox{-1em}{\includegraphics{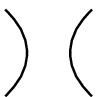}}}

\newcommand*{\verticalv}{ 
\raisebox{-1em}{\includegraphics{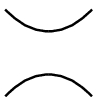}}}

\newcommand*{\verticalt}{ 
\raisebox{-1em}{\includegraphics{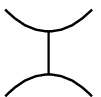}}}

\newcommand*{\labeledcylinder}[1]{
\raisebox{0ex}{\begin{picture}(25,37)(0,0)
\put(7,17){$#1$}
\put(0,0){\includegraphics[scale=.3]{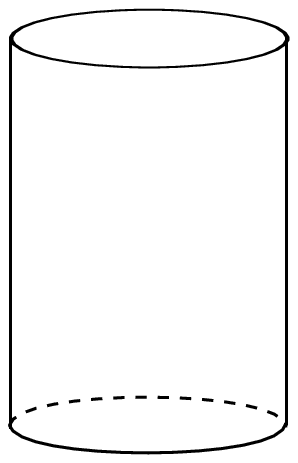}}
\end{picture}}}

\newcommand*{\labeledhcylinder}[1]{
\raisebox{-1.7ex}{\begin{picture}(40,25)(0,9)
\put(14,17){$#1$}
\put(0,0){\includegraphics[scale=.3]{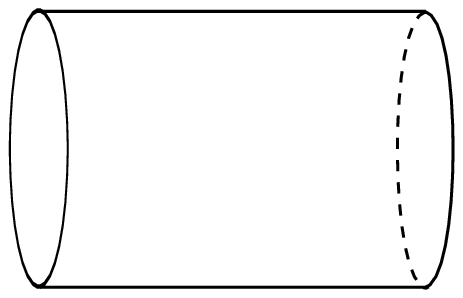}}
\end{picture}}}

\newcommand*{\labeledhpants}[1]{
\raisebox{-1.7ex}{\begin{picture}(40,38)(0,9)
\put(15,18){$#1$}
\put(0,0){\includegraphics[scale=.3]{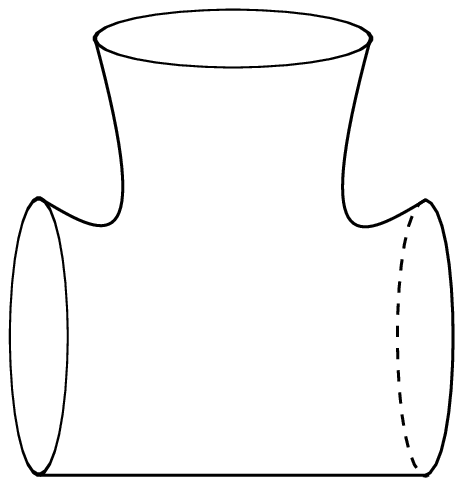}}
\end{picture}}}

\newcommand*{\cylinder}{\raisebox{-4.4ex}{\includegraphics[scale=.3]{SVG/cylinder}}}
\newcommand*{\cylinderdehntwistcurve}{\raisebox{-4.4ex}{\includegraphics[scale=.3]{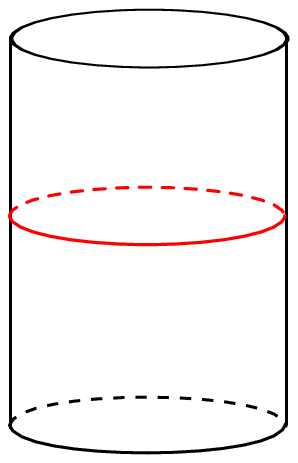}}}

\newcommand*{\cylinderttI}{\raisebox{-4.4ex}{\includegraphics[scale=.3]{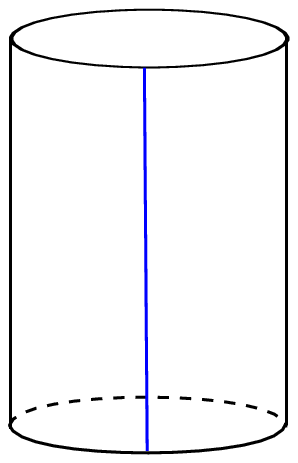}}}
\newcommand*{\cylinderS}{\raisebox{-4.4ex}{\includegraphics[scale=.3]{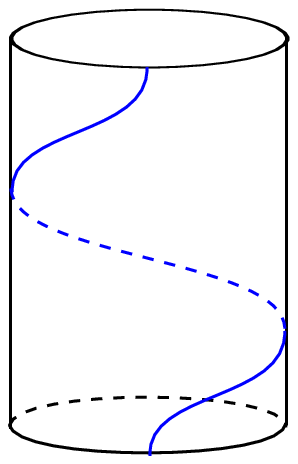}}}
\newcommand*{\cylinderSdag}{\raisebox{-4.4ex}{\includegraphics[scale=.3]{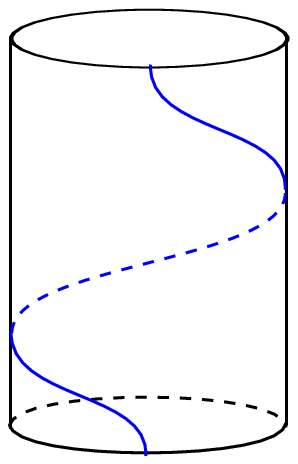}}}
\newcommand*{\cylinderO}{\raisebox{-4.4ex}{\includegraphics[scale=.3]{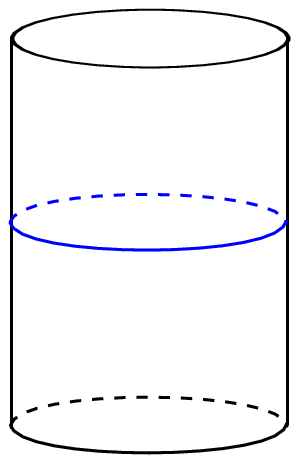}}}
\newcommand*{\cylindervtD}{\raisebox{-4.4ex}{\includegraphics[scale=.3]{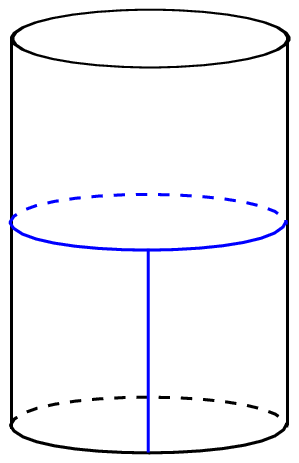}}}
\newcommand*{\cylindertvD}{\raisebox{-4.4ex}{\includegraphics[scale=.3]{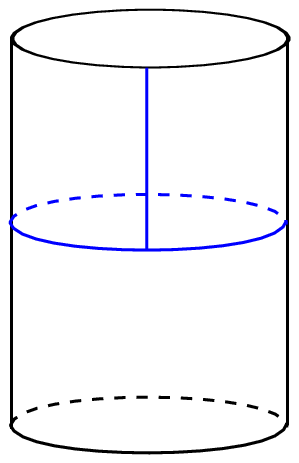}}}

\newcommand*{\pantsttt}{\raisebox{-4.4ex}{\includegraphics[scale=.3]{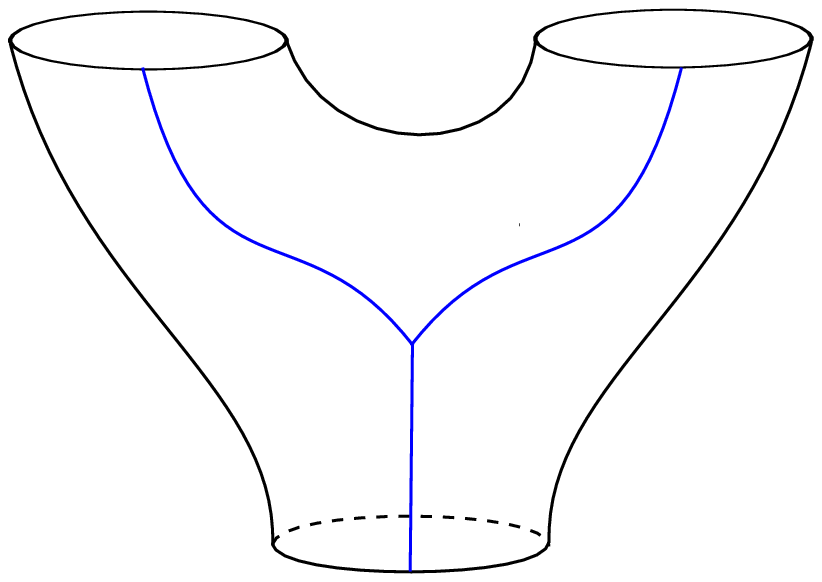}}}

\newcommand*{\pantsttv}{\raisebox{-4.4ex}{\includegraphics[scale=.3]{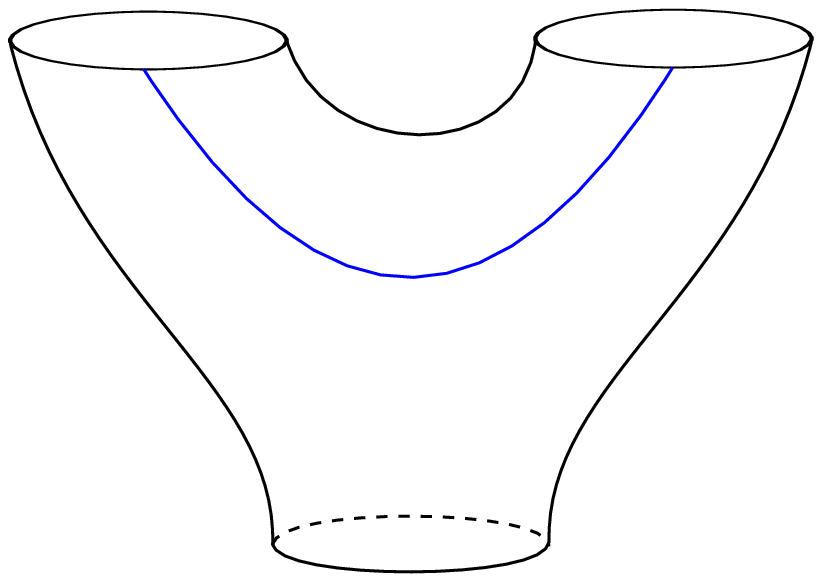}}}

\newcommand*{\pantstwisted}{\raisebox{-4.4ex}{\includegraphics[scale=.3]{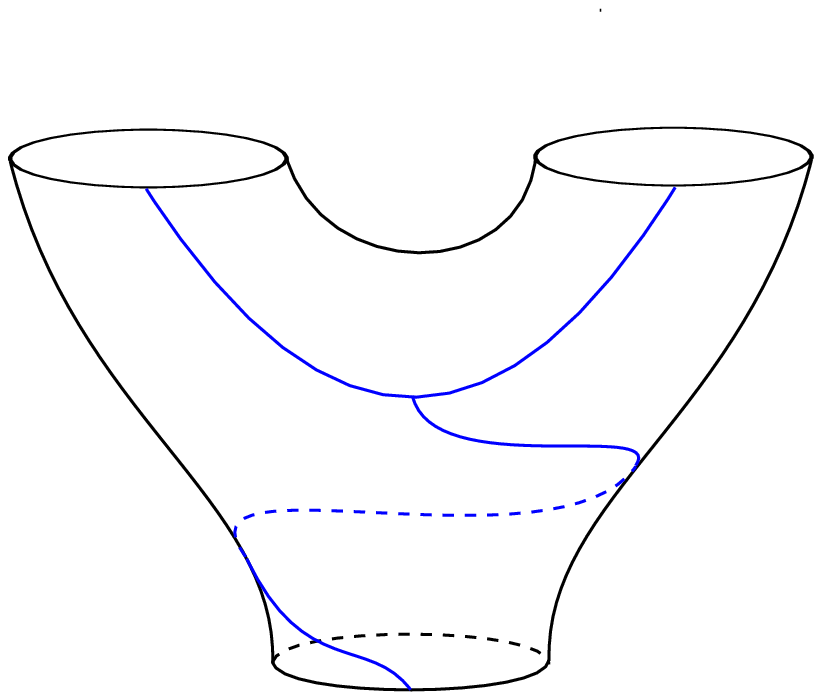}}}

\newcommand*{\pantsdeformed}{\raisebox{-4.4ex}{\includegraphics[scale=.3]{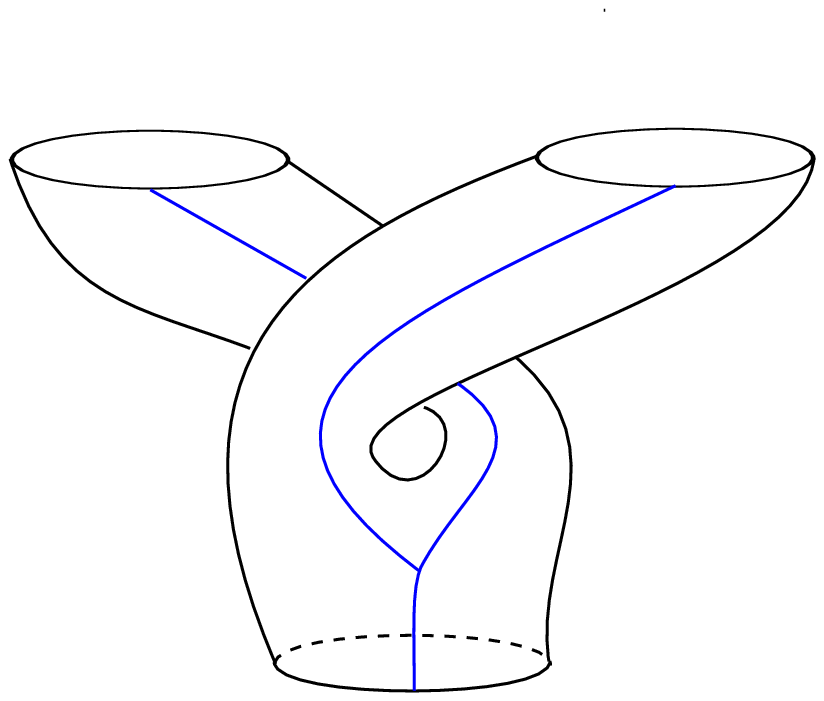}}}

\newcommand*{\pantsdehnind}{\raisebox{0.0ex}{
\begin{picture}(82,52)(-10,0)
\put(0,0){\includegraphics[scale=.3]{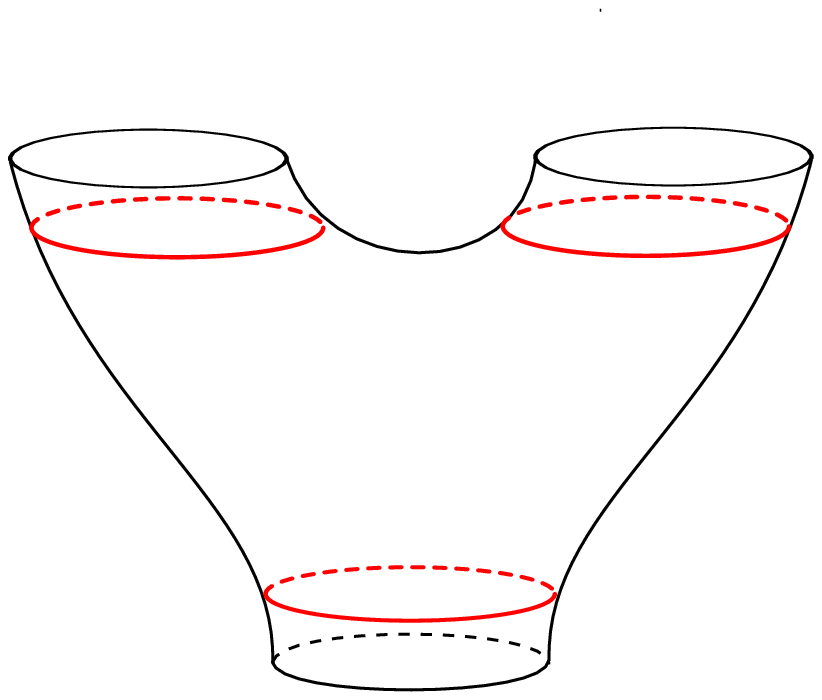}}
\put(8,8){$\gamma_C$}
\put(71,40){$\gamma_B$}
\put(-12,40){$\gamma_A$}
\end{picture}
}
}
\newcommand*{\pantsdehnfirst}{\raisebox{-4.4ex}{\includegraphics[scale=.3]{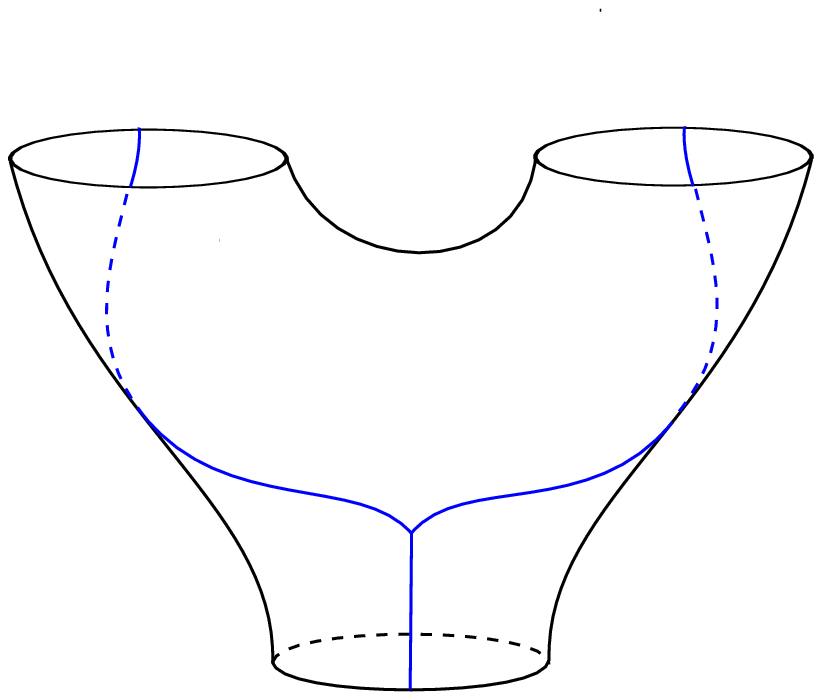}}}
\newcommand*{\pantsdehnsecond}{\raisebox{-4.4ex}{\includegraphics[scale=.3]{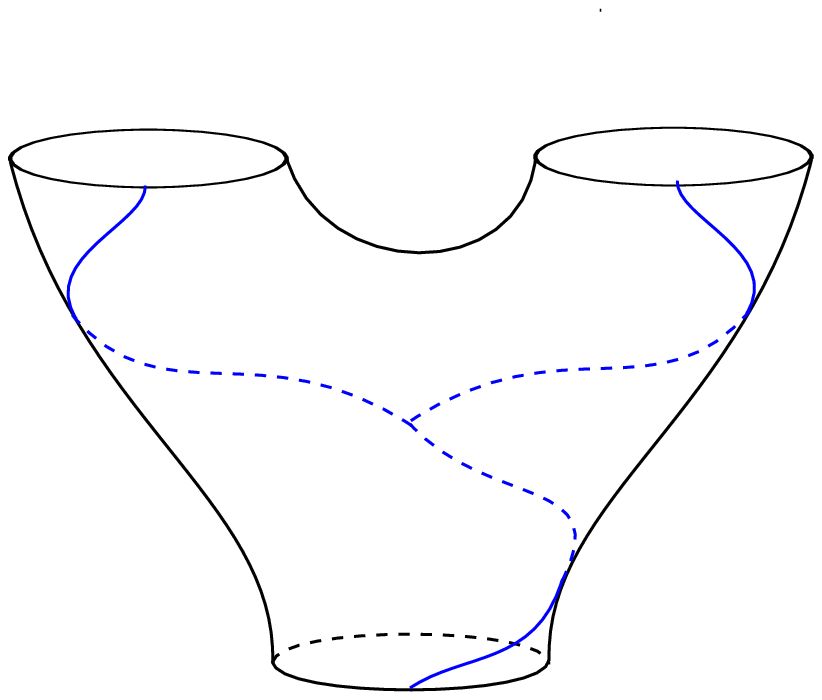}}}
\newcommand*{\pantsbraided}{\raisebox{-4.4ex}{\includegraphics[scale=.3]{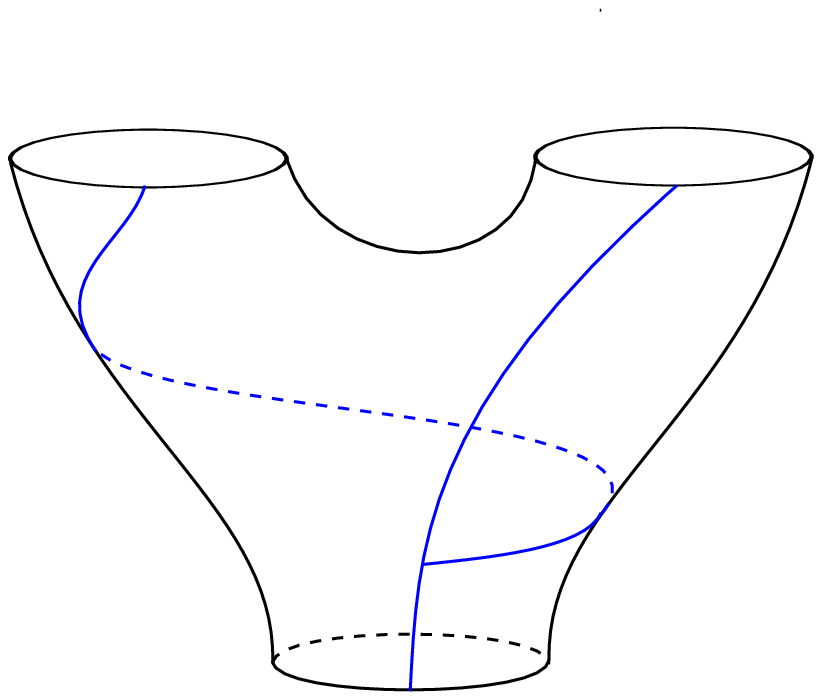}}}

\newcommand*{\fourpuncturedspherebraid}[2]{\raisebox{-5.0ex}{
\begin{picture}(70,52)(-2,0)
\put(12,10){$#1$}
\put(56,10){$#2$}
\put(0,0){\includegraphics[scale=.3]{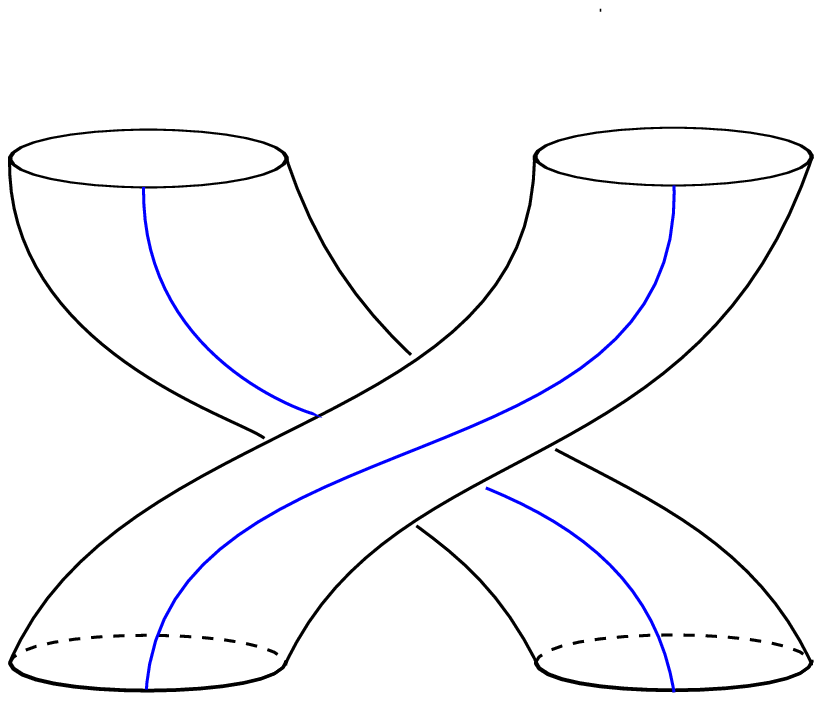}}
\end{picture}}}

\newcommand*{\wvvI}{\raisebox{-0.4ex}{\includegraphics[scale=.1]{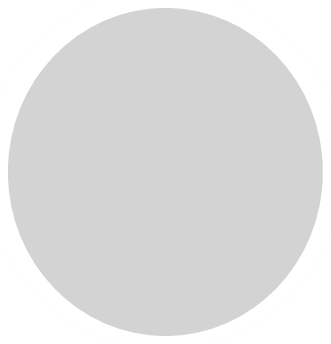}}} 

\newcommand*{\zvvI}{\raisebox{-0.4ex}{\includegraphics[scale=.1]{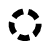}}} 
\newcommand*{\zttI}{\raisebox{-0.4ex}{\includegraphics[scale=.1]{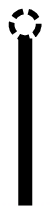}}}
\newcommand*{\zS}{\raisebox{-0.4ex}{\includegraphics[scale=.1]{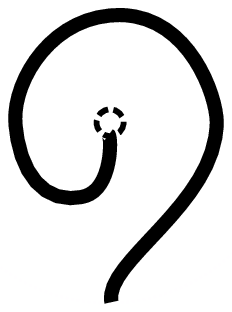}}}
\newcommand*{\zSdag}{\raisebox{-0.4ex}{\includegraphics[scale=.1]{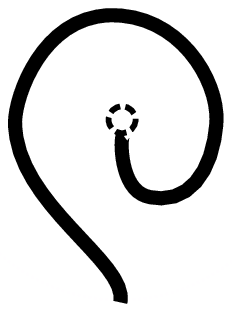}}}
\newcommand*{\zO}{\raisebox{-0.4ex}{\includegraphics[scale=.1]{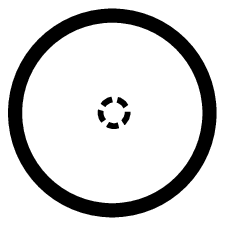}}}
\newcommand*{\zvtD}{\raisebox{-0.4ex}{\includegraphics[scale=.1]{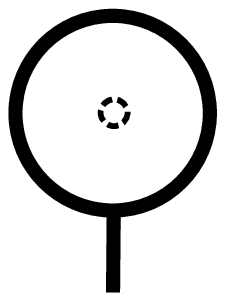}}}
\newcommand*{\ztvD}{\raisebox{-0.4ex}{\includegraphics[scale=.1]{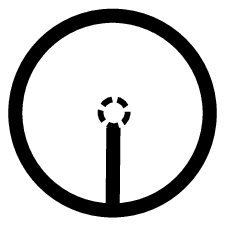}}}
\newcommand*{\annulusoutboundary}{\raisebox{-0.4ex}{\includegraphics[scale=.1]{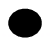}}}

\newcommand*{\outb}[1]{
\raisebox{0.0ex}{\protect\ \begin{picture}(2.8,2.8)(4,1.1)
\put(0,0){#1}  
\put(0,0){\annulusoutboundary}
\end{picture}\ }
} 

\newcommand*{\twocob}[3]{  
\raisebox{-1.10ex}{\begin{picture}(19,15)(2,0)
\put(0,0){#3}   
\put(0.0,7.8){#1}
\put(10.55,7.8){#2}
\put(0,0){\threepuncturedoutboundary}
\end{picture}}}

\newcommand*{\threecob}[4]{  
\raisebox{-1.15ex}{\begin{picture}(25,29)(-1,-3)
\put(0,0){#3}   
\put(0.0,7.8){#1}
\put(10.55,7.8){#2}
\put(-4.7,-8.30){#4}
\put(-4.7,-8.30){\threeoutboundary}
\end{picture}}}

\newcommand*{\vvI}{\protect\outb{\protect\zvvI}} 
\newcommand*{\ttI}{\protect\outb{\protect\zttI}}
\renewcommand*{\S}{\protect\outb{\protect\zS}}
\newcommand*{\Sdag}{\protect\outb{\protect\zSdag}}
\renewcommand*{\O}{\protect\outb{\protect\zO}}
\newcommand*{\vtD}{\protect\outb{\protect\zvtD}}
\newcommand*{\tvD}{\protect\outb{\protect\ztvD}}

\newcommand*{\ttt}{\raisebox{-0.3ex}{\includegraphics[scale=.1]{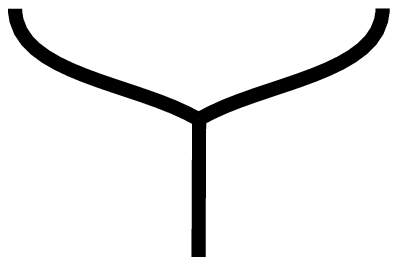}}} 
\newcommand*{\ttv}{\raisebox{-0.3ex}{\includegraphics[scale=.1]{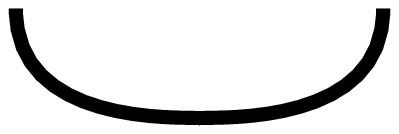}}} 
\newcommand*{\vtt}{\raisebox{-0.3ex}{\includegraphics[scale=.1]{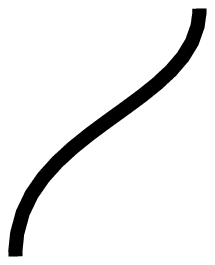}}} 
\newcommand*{\threepuncturedoutboundary}{\raisebox{-0.3ex}{\includegraphics[scale=.1]{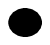}}}

\newcommand*{\bttv}{\twocob{\wvvI}{\wvvI}{\ttv}}

\newcommand*{\bttt}{\twocob{\wvvI}{\wvvI}{\ttt}}

\newcommand*{\threeoutboundary}{\raisebox{-0.1ex}{\includegraphics[scale=.1]{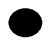}}}
\newcommand*{\threeSdag}{\raisebox{-0.1ex}{\includegraphics[scale=.1]{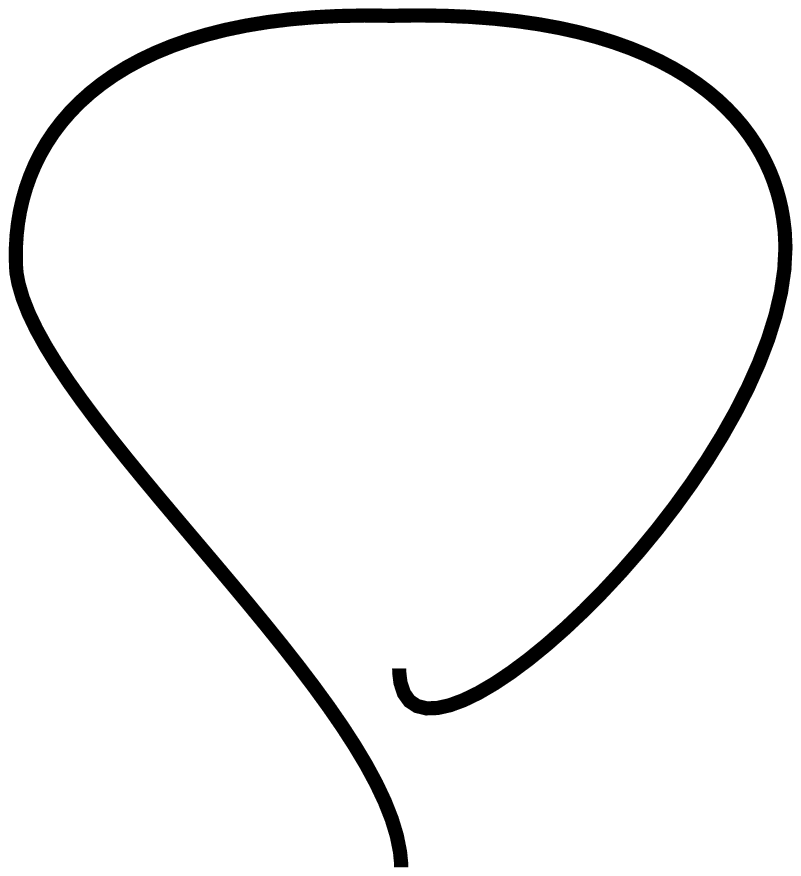}}}

\newcommand*{\fuse}[3]{\raisebox{-3.0ex}{
\begin{picture}(28,30)(-5,-8)
\put(0,0){\includegraphics[scale=.5]{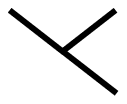}}
\put(-7,15){\small $#1$}
\put(10,15){\small $#2$}
\put(13,-7){\small $#3$}
\end{picture}}
} 

\newcommand*{\overcrossing}{\raisebox{-1.5ex}{
\begin{picture}(27,23)(0,0)
\put(0,0){\includegraphics[scale=.5]{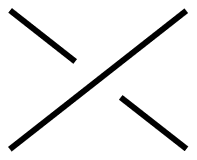}}
\end{picture}
}}

\newcommand*{\overcrossingresolvedv}{\raisebox{-1.5ex}{
\begin{picture}(27,23)(0,0)
\put(0,0){\includegraphics[scale=.5]{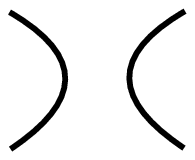}}
\end{picture}
}}

\newcommand*{\overcrossingresolvedupv}{\raisebox{-1.5ex}{
\begin{picture}(27,23)(0,0)
\put(0,0){\includegraphics[scale=.5]{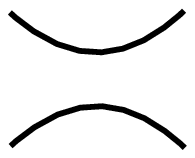}}
\end{picture}
}}

\newcommand*{\fusethree}[5]{\raisebox{-3.0ex}{
\begin{picture}(39,28)(-2,-4)
\put(0,0){\includegraphics[scale=.5]{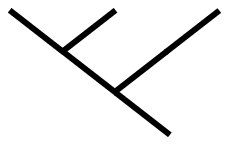}}
\put(-4,20){\small $#1$}
\put(12,20){\small $#2$}
\put(27,20){\small $#3$}
\put(0,4){\small $#4$}
\put(20,-4){\small $#5$}
\end{picture}}
} 
 
\newcommand*{\fusethrees}[5]{\raisebox{-3.0ex}{
\begin{picture}(39,28)(-2,-4)
\put(0,0){\includegraphics[scale=.5]{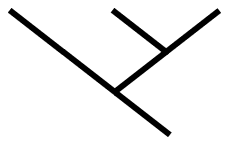}}
\put(-4,20){\small $#1$}
\put(10,20){\small $#2$}
\put(27,20){\small $#3$}
\put(20,5){\small $#4$}
\put(20,-4){\small $#5$}
\end{picture}}
}
 
 \newcommand*{\fusethreedirected}[5]{\raisebox{-3.0ex}{
\begin{picture}(39,28)(-2,-4)
\put(0,0){\includegraphics[scale=.5]{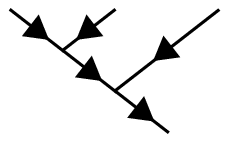}}
\put(-4,20){\small $#1$}
\put(12,20){\small $#2$}
\put(28,20){\small $#3$}
\put(1,4){\small $#4$}
\put(24,-4){\small $#5$}
\end{picture}}
} 
 
\newcommand*{\fusethreesdirected}[5]{\raisebox{-3.0ex}{
\begin{picture}(39,28)(-2,-4)
\put(0,0){\includegraphics[scale=.5]{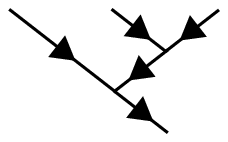}}
\put(-4,20){\small $#1$}
\put(12,20){\small $#2$}
\put(28,20){\small $#3$}
\put(22,5){\small $#4$}
\put(24,-4){\small $#5$}
\end{picture}}
}

 \newcommand*{\fusethreedirectedreverse}[5]{\raisebox{-3.0ex}{
\begin{picture}(39,28)(-2,-4)
\put(0,0){\includegraphics[scale=.5]{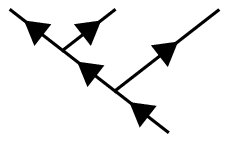}}
\put(-4,20){\small $#1$}
\put(12,20){\small $#2$}
\put(28,20){\small $#3$}
\put(1,4){\small $#4$}
\put(24,-4){\small $#5$}
\end{picture}}
} 

\newcommand*{\fusefour}[7]{\raisebox{-3.0ex}{
\begin{picture}(40,38)(-10,-4)
\put(-14,17){\includegraphics[scale=.5]{SVG/fusiondiagramdiagonal}}
\put(0,0){\includegraphics[scale=.5]{SVG/fusiondiagramthree_second}}
\put(-15,30){\small $#1$}
\put(-2,30){\small $#2$}
\put(12,20){\small $#3$}
\put(27,20){\small $#4$}
\put(-1,8){\small $#5$}
\put(21,5){\small $#6$}
\put(22,-7){\small $#7$}
\end{picture}}
} 
 
\newcommand*{\Smatrix}[2]{\raisebox{-1.7ex}{
\begin{picture}(39,28)(-2,-4)
\put(0,0){\includegraphics[scale=.5]{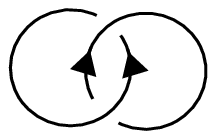}}
\put(-4,5){\small $#1$}
\put(30,5){\small $#2$}
\end{picture}}}

\newcommand*{\threeDcrossingunresolved}[2]{\raisebox{-2.3ex}{
\begin{picture}(35,35)(-5,-4)
\put(0,0){\includegraphics[scale=.5]{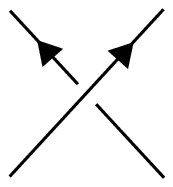}}
\put(-5,-4){\small $#1$}
\put(25,-4){\small $#2$}
\end{picture}}}

\newcommand*{\arrowfusiondiagram}[3]{\raisebox{-3.0ex}{
\begin{picture}(30,38)(0,0)
\put(0,0){\includegraphics[scale=.6]{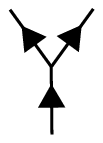}}
\put(3,32){\small $#1$}
\put(18,32){\small $#2$}
\put(10,0){\small $#3$}
\end{picture}}}

\newcommand*{\arrowbraiddiagram}[3]{\raisebox{-3.0ex}{
\begin{picture}(30,38)(0,0)
\put(0,0){\includegraphics[scale=.6]{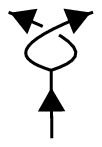}}
\put(3,32){\small $#1$}
\put(18,32){\small $#2$}
\put(10,0){\small $#3$}
\end{picture}}}

\newcommand*{\arrowbraidoppositediagram}[3]{\raisebox{-3.0ex}{
\begin{picture}(30,38)(0,0)
\put(0,0){\reflectbox{\includegraphics[scale=.6]{SVG/3Dcrossingdoubleddiagrambraid}}}
\put(3,32){\small $#1$}
\put(18,32){\small $#2$}
\put(10,0){\small $#3$}
\end{picture}}}

\newcommand*{\threeDcrossingresolved}[3]{\raisebox{-2.3ex}{
\begin{picture}(35,35)(-5,-4)
\put(0,0){\includegraphics[scale=.5]{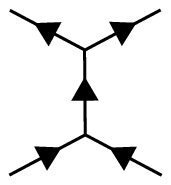}}
\put(-5,-4){\small $#1$}
\put(25,-4){\small $#2$}
\put(15,10){\small $#3$}
\put(-5,23){\small $#2$}
\put(25,23){\small $#1$}
\end{picture}}}

\newcommand*{\threeDlinewithringvacuum}[1]{\raisebox{-1.7ex}{
\begin{picture}(30,35)(-2,-4)
\put(0,0){\includegraphics[scale=.5]{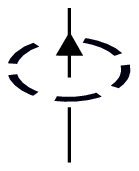}}
\put(10,0){\small $#1$}
\end{picture}}}

\newcommand*{\threeDlinewithring}[2]{\raisebox{-1.7ex}{
\begin{picture}(30,35)(-2,-4)
\put(0,0){\includegraphics[scale=.5]{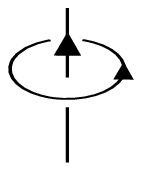}}
\put(10,0){\small $#1$}
\put(19,10){\small $#2$}
\end{picture}}}

\newcommand*{\threeDline}[1]{\raisebox{-2.1ex}{ 
\begin{picture}(7,35)(7,-4)
\put(0,0){\includegraphics[scale=.5]{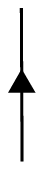}}
\put(12,8){\small $#1$}
\end{picture}}}

\newcommand*{\threeDlinedown}[1]{\raisebox{-2.1ex}{ 
\begin{picture}(7,35)(14,-4)
\put(0,0){\includegraphics[scale=.5]{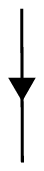}}
\put(12,8){\small $#1$}
\end{picture}}}

\newcommand*{\threeDlinevacuum}{\raisebox{-2.1ex}{ 
\begin{picture}(7,35)(14,-4)
\put(0,0){\includegraphics[scale=.5]{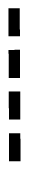}}
\end{picture}}}

\newcommand*{\threeDlinetwistright}[1]{ 
\raisebox{-2.1ex}{\begin{picture}(13,35)(12,-4)
\put(0,0){\reflectbox{\includegraphics[scale=.5]{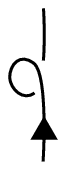}}}
\put(23,4){\small $#1$}
\end{picture}}}

\newcommand*{\threeDlinetopbottom}[5]{\raisebox{-3.4ex}{ 
\begin{picture}(40,40)(-2,-8)
\put(0,0){\includegraphics[scale=.5]{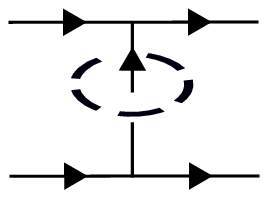}}
\put(0,27){\small $#1$}
\put(30,27){\small $#2$}
\put(20,5){\small $#3$}
\put(0,-8){\small $#4$}
\put(30,-8){\small $#5$}
\end{picture}}}

\newcommand*{\threeDlinetopbottomFmove}[5]{\raisebox{-3.4ex}{ 
\begin{picture}(40,40)(-2,-8)
\put(0,0){\includegraphics[scale=.5]{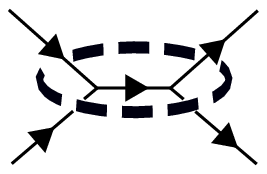}}
\put(0,27){\small $#1$}
\put(30,27){\small $#2$}
\put(15,15){\small $#3$}
\put(0,-8){\small $#4$}
\put(30,-8){\small $#5$}
\end{picture}}}

\newcommand*{\threeDlinetopbottomunconnected}[2]{\raisebox{-3.4ex}{ 
\begin{picture}(40,40)(-2,-8)
\put(0,0){\includegraphics[scale=.5]{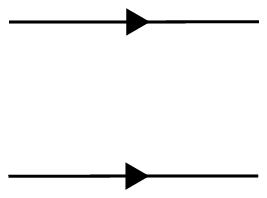}}
\put(17,27){\small $#1$}
\put(17,-8){\small $#2$}
\end{picture}}}

\newcommand*{\threeDdoubleringvv}[7]{\raisebox{-2.6ex}{ 
\begin{picture}(130,40)(-2,-8)
\put(0,0){\includegraphics[scale=.5]{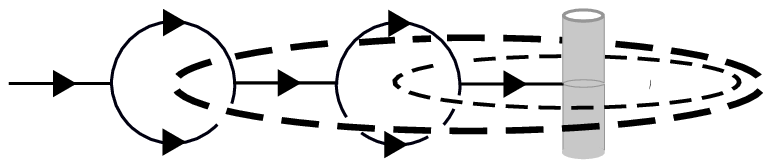}}
\put(5,15){\small $#1$}
\put(21,22){\small $#2$}
\put(21,-7){\small $#3$}
\put(38,15){\small $#4$}
\put(55,22){\small $#5$}
\put(55,-7){\small $#6$}
\put(70,15){\small $#7$}
\end{picture}}}

\newcommand*{\threeDdoubleringbvv}[7]{\raisebox{-2.6ex}{ 
\begin{picture}(130,40)(-2,-8)
\put(0,0){\includegraphics[scale=.5]{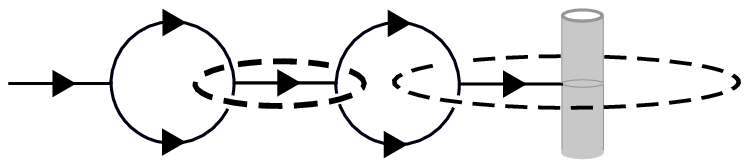}}
\put(5,15){\small $#1$}
\put(21,22){\small $#2$}
\put(21,-7){\small $#3$}
\put(38,15){\small $#4$}
\put(55,22){\small $#5$}
\put(55,-7){\small $#6$}
\put(70,15){\small $#7$}
\end{picture}}}

\newcommand*{\threeDdoubleringsingle}[4]{\raisebox{-2.6ex}{ 
\begin{picture}(80,40)(30,-8)
\put(0,0){\includegraphics[scale=.5]{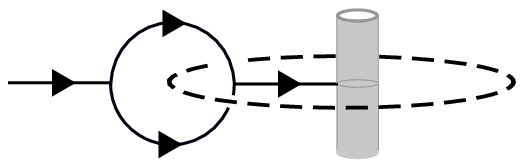}}
\put(37,15){\small $#1$}
\put(54,22){\small $#2$}
\put(54,-7){\small $#3$}
\put(72,15){\small $#4$}
\end{picture}}}

\newcommand*{\threeDdoubleringsinglenoboundary}[4]{\raisebox{-2.6ex}{ 
\begin{picture}(50,40)(30,-8)
\put(0,0){\includegraphics[scale=.5]{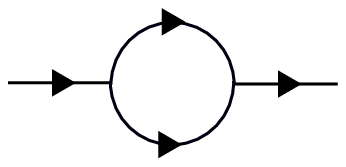}}
\put(37,15){\small $#1$}
\put(54,22){\small $#2$}
\put(54,-7){\small $#3$}
\put(72,15){\small $#4$}
\end{picture}}}

\newcommand*{\threeDdoubleringstraightline}[1]{\raisebox{-1.6ex}{ 
\begin{picture}(55,20)(25,0)
\put(0,0){\includegraphics[scale=.5]{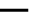}}
\put(54,17){\small $#1$}
\end{picture}}}

\newcommand*{\twoDleftline}[1]{\put(0,0){\includegraphics[scale=.6]{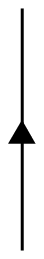}}\put(6,17){$#1$}}
\newcommand*{\twoDrightline}[1]{\put(0,0){\includegraphics[scale=.6]{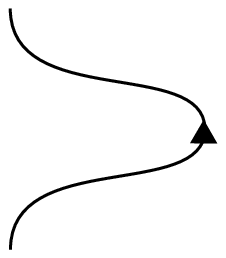}}\put(38,17){$#1$}}
\newcommand*{\twoDvacuum}{\put(0,0){\includegraphics[scale=.6]{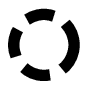}}}
\newcommand*{\twoDoutervacuum}{\put(0,0){\includegraphics[scale=.6]{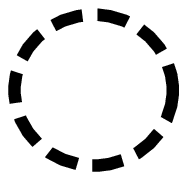}}}
\newcommand*{\twoDoutersolid}[1]{\put(0,0){\includegraphics[scale=.6]{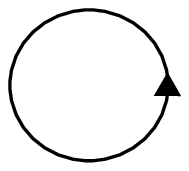}}\put(35,17){#1}}
\newcommand*{\twoDoutersolidshifted}[1]{\put(30,0){\includegraphics[scale=.6]{SVG/2Dlineoutersolid}}\put(65,17){#1}}
\newcommand*{\twoDsolid}[1]{\put(0,0){\includegraphics[scale=.6]{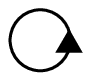}}\put(27,17){$#1$}}
\newcommand*{\twoDhole}{\put(0,0){\includegraphics[scale=.6]{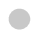}}}
\newcommand*{\twoDtopbottom}[4]{\put(0,0){\includegraphics[scale=.6]{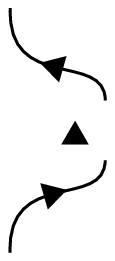}}\put(6,17){$#1$}\put(10,37){$#4$}\put(10,0){$#3$}\put(0,0){\includegraphics[scale=.6]{SVG/2Dlinesolid}}\put(27,17){$#2$}
}

\newcommand*{\Danyons}[4]{
\raisebox{-3.0ex}{
\begin{picture}(23,48)(4,0)
\put(0,0){\includegraphics[scale=.6]{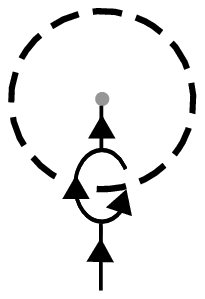}}
\put(18,2){$#1$}
\put(4,12){$#2$}
\put(22,12){$#3$}
\put(18,26){$#4$}
\end{picture}
}
}

\newcommand*{\Danyonsstandardbasisonex}[4]
{\put(0,0){\includegraphics[scale=.6]{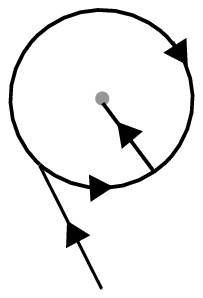}}
\put(0,6){$#1$}
\put(20,10){$#2$}
\put(33,42){$#3$}
\put(23,29){$#4$}
}
\newcommand*{\Danyonsstandardbasisone}[4]{
\raisebox{-3.0ex}{
\begin{picture}(32,52)(0,0)
\put(0,0){\Danyonsstandardbasisonex{#1}{#2}{#3}{#4}}
\end{picture}}
}
\newcommand*{\Danyonsstandardbasistwo}[4]{
\raisebox{-3.0ex}{
\begin{picture}(32,52)(2,0)
\put(0,0){\includegraphics[scale=.6]{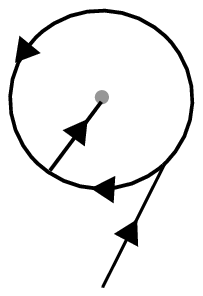}}
\put(26,6){$#1$}
\put(8,10){$#2$}
\put(-6,42){$#3$}
\put(5,29){$#4$}
\end{picture}}
}

\newcommand*{\doubledexample}[1]{
\raisebox{-6.3ex}{
\begin{picture}(70,78)(0,0)
\put(0,0){\includegraphics[scale=.6]{SVG/doubled_ex#1}}
\end{picture}}
}

\newcommand*{\doubledanyonsvt}{ 
\raisebox{-4.0ex}{
\begin{picture}(35,50)(0,0)
\put(0,0){\includegraphics[scale=.6]{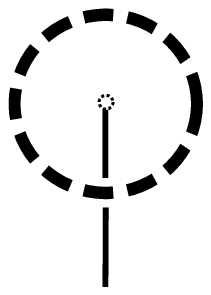}}
\end{picture}}
}

\newcommand*{\doubledanyonsttv}{
\raisebox{-4.0ex}{
\begin{picture}(35,50)(0,0)
\put(0,0){\includegraphics[scale=.6]{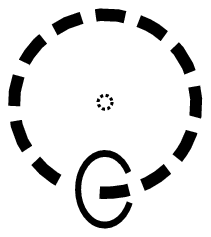}}
\end{picture}}
}

\newcommand*{\doubledanyonsttt}{ 
\raisebox{-4.0ex}{
\begin{picture}(35,50)(0,0)
\put(0,0){\includegraphics[scale=.6]{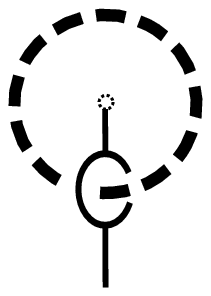}}
\end{picture}}
}

\newcommand*{\doubledanyonsdeltaonezero}{ 
\raisebox{-4.0ex}{
\begin{picture}(35,50)(0,0)
\put(0,0){\includegraphics[scale=.6]{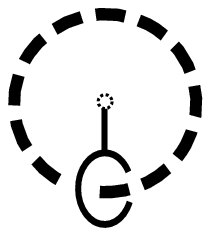}}
\end{picture}}
}

\newcommand*{\doubledanyonsdeltazeroone}{ 
\raisebox{-4.0ex}{
\begin{picture}(35,50)(0,0)
\put(0,0){\includegraphics[scale=.6]{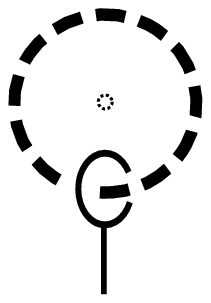}}
\end{picture}}
}

\newcommand*{\doubledanyonsdeltaonezeroZ}{ 
\raisebox{-4.0ex}{
\begin{picture}(35,50)(0,0)
\put(0,0){\includegraphics[scale=.6]{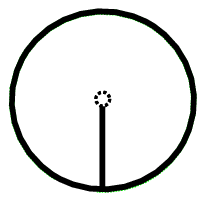}}
\end{picture}}
}

\newcommand*{\doubledanyonsdeltazerooneZ}{ 
\raisebox{-4.0ex}{
\begin{picture}(35,50)(0,0)
\put(0,0){\includegraphics[scale=.6]{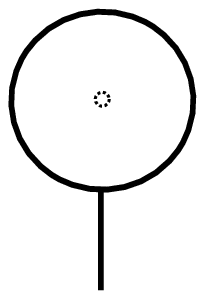}}
\end{picture}}
}

\newcommand*{\doubledanyonstv}{ 
\raisebox{-4.0ex}{
\begin{picture}(35,50)(0,0)
\put(0,0){\includegraphics[scale=.6]{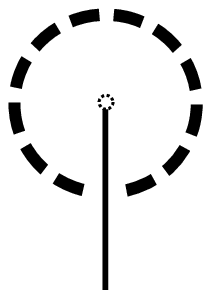}}
\end{picture}}
}

\newcommand*{\doubledanyonsvv}{ 
\raisebox{-4.0ex}{
\begin{picture}(35,50)(0,0)
\put(0,0){\includegraphics[scale=.6]{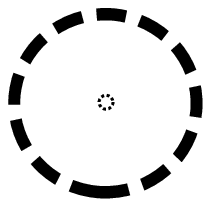}}
\end{picture}}
}

\newcommand*{\bSdag}{
\raisebox{-4.0ex}{
\begin{picture}(35,50)(0,0)
\put(0,0){\includegraphics[scale=.6]{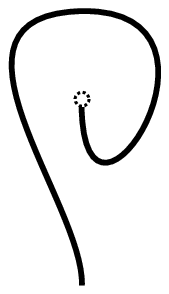}}
\end{picture}}}

\newcommand*{\bS}{
\raisebox{-4.0ex}{
\begin{picture}(35,50)(0,0)
\put(0,0){\includegraphics[scale=.6]{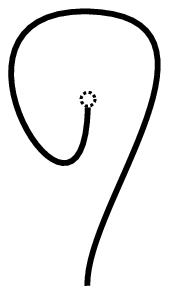}}
\end{picture}}}

\newcommand*{\bttinv}{
\raisebox{-4.0ex}{
\begin{picture}(35,50)(0,0)
\put(0,0){\includegraphics[scale=.6]{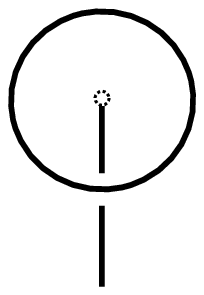}}
\end{picture}}}

\newcommand*{\bt}{
\raisebox{-4.0ex}{
\begin{picture}(35,50)(0,0)
\put(0,0){\includegraphics[scale=.6]{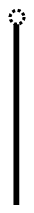}}
\end{picture}}}

\newcommand*{\Danyonsfmove}[1]{
\put(0,0){\includegraphics[scale=.6]{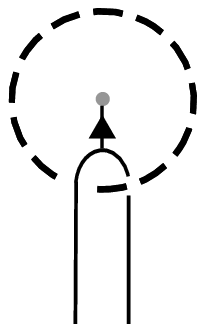}}
\put(19,30){$#1$}
}

\newcommand*{\Danyonsfmovebase}[1]{
\put(0,0){\includegraphics[scale=.6]{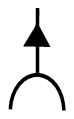}}
\put(19,60){$#1$}
}

\newcommand*{\Danyonsfmovebasetwo}[1]{
\put(0,0){\includegraphics[scale=.6]{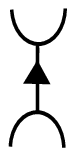}}
\put(19,60){$#1$}
}

\newcommand*{\Danyonsfourpunctured}[5]{
\put(0,0){\includegraphics[scale=.6]{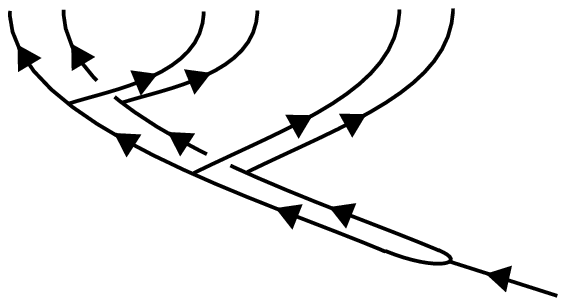}}
\put(0,46){${#1}_1$}
\put(21,46){${#2}_1$}
\put(33,46){${#1}_2$}
\put(55,46){${#2}_2$}
\put(67,46){${#1}_3$}
\put(89,46){${#2}_3$}
\put(19,20){${#4}_1$}
\put(41,29){${#5}_1$}
\put(51,8){${#4}_2$}
\put(71,18){${#5}_2$}
\put(95,8){${#3}$}
}

\newcommand*{\Danyonsfourpuncturedsingle}[2]{
\put(0,0){\includegraphics[scale=.6]{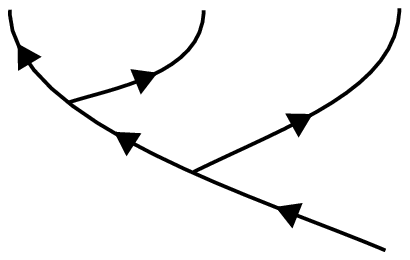}}
\put(0,71){$#1_1$}
\put(33,71){$#1_2$}
\put(67,71){$#1_3$}
\put(19,45){$#2_1$}
\put(51,33){$#2_2$}
}

\newcommand*{\Danyonsfourpuncturedflippedsingle}{
\put(0,0){\includegraphics[scale=.6]{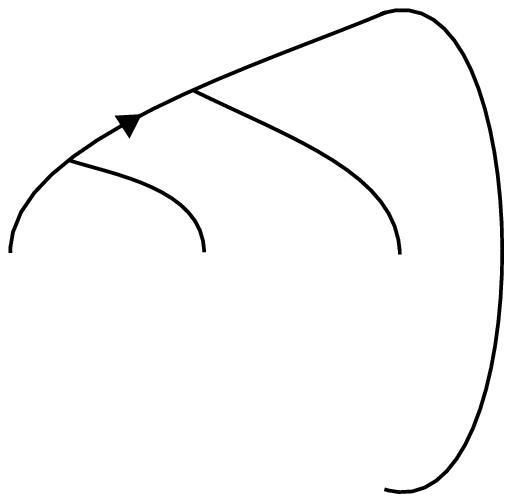}}
\put(19,104){$c'_1$}
}

\newcommand*{\Danyonsfourpuncturedflipped}{
\put(0,0){\includegraphics[scale=.6]{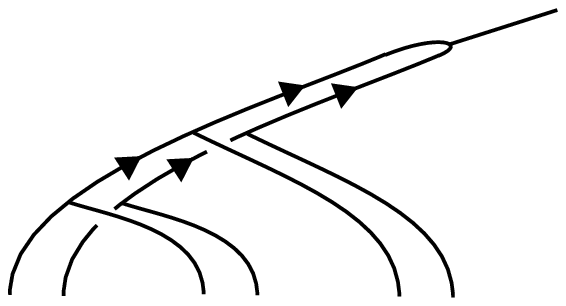}}
\put(19,79){$c'_1$}
\put(41,65){$d'_1$}
\put(51,93){$c'_2$}
\put(71,78){$d'_2$}
}

\newcommand*{\Danyonsfourpuncturedflippedtwo}{
\put(0,0){\includegraphics[scale=.6]{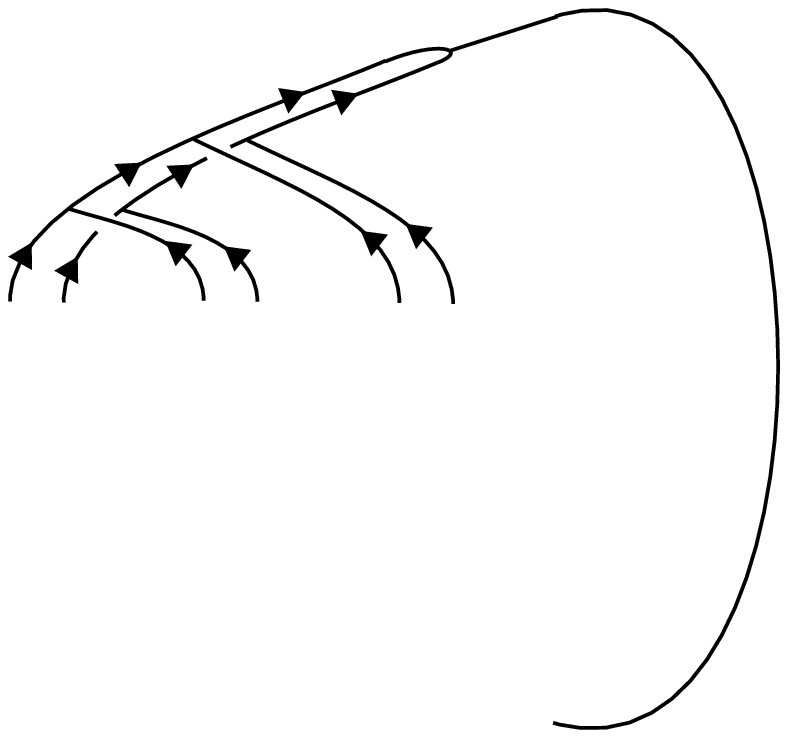}}
\put(19,104){$c'_1$}
\put(41,90){$d'_1$}
\put(51,118){$c'_2$}
\put(71,103){$d'_2$}
\put(0,71){$a_1$}
\put(22,71){$b_1$}
\put(33,71){$a_2$}
\put(55,71){$b_2$}
\put(67,71){$a_3$}
\put(89,71){$b_3$}
}

\newcommand*{\Danyonsfmoveb}[4]{
\put(0,0){\includegraphics[scale=.6]{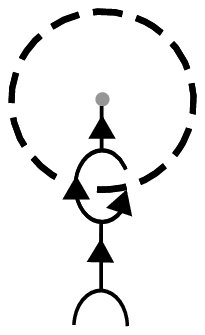}}
\put(4,17){$#2$}
\put(22,17){$#3$}
\put(18,9){$#1$}
\put(19,30){$#4$}
}

\newcommand*{\myvecstandard}[6]{
{_{#1}}\Psi_{\vec{#2}} (\vec{#3}, \vec{#4}, \vec{#5}, \vec{#6})}

\newcommand*{\myvecstand}[4]{
{_{#1}}\Psi_{#2} (#3,#4)
}

\newcommand*{\myvecstandardtwopic}[6]{\raisebox{-5.5ex}{
\begin{picture}(110,70)(-2,0)
\put(0,0){\Danyonsfourpunctured{#3}{#4}{#1}{#5}{#6}}
\put(0,0){\Danyonsfmovebase{{#2}_1}}
\put(33,0){\Danyonsfmovebase{{#2}_2}}
\put(67,0){\Danyonsfmovebase{{#2}_3}}
\end{picture}}
}

\newcommand*{\myvecstandardpic}[6]{
\raisebox{-9.0ex}{
\begin{picture}(110,103)(0,0)
\put(0,48){\Danyonsfmove{{#2}_1}}
\put(33,48){\Danyonsfmove{{#2}_2}}
\put(67,48){\Danyonsfmove{{#2}_3}}
\put(0,0){\Danyonsfourpunctured{#3}{#4}{#1}{#5}{#6}}
\end{picture}}
}


%% file: intro.tex
\section{Introduction}

Proposed topological quantum computers protect quantum information in a physical medium, while allowing logical operations to be applied robustly using adiabatic local processes~\cite{Kitaev03, preskill, Freedmanetal03, Nayaketal08}.  Building systems with the required properties is generally believed to be difficult.  However, some of the highest fault-tolerance threshold estimates are based on a code---Kitaev's toric code---and procedures that are motivated by topological quantum computation~\cite{Raussendorfetal05, Raussendorfharrington06, Dennisetal02, Fowleretal08}.  A limitation of the toric code is that it corresponds to an abelian anyon model, so does not allow for universal quantum computation.  By translating more general anyon models into the language of error-correcting codes, this limitation can be surpassed.  

A first step is to realize the Hilbert space of a topological quantum field theory (TQFT) on a spin lattice.  Turaev and Viro completed this for a class of ribbon categories~\cite{TuraevViro92topology}, and their construction was extended by Barrett and Westbury to spherical categories~\cite{BarrettWestbury93spherical, BarrettWestbury96invariants}.  For example, if~$\cC$ is the Fibonacci category, a modular tensor category with one non-trivial particle~$\tau$, the Turaev-Viro construction gives a TQFT for the ``doubled'' category $\cC \otimes \cC^*$, with particles $1\otimes 1,1\otimes \tau, \tau\otimes 1$ and $\tau\otimes\tau$.  Levin and Wen showed that the TQFT Hilbert space is the code space of a set of commuting local stabilizers~\cite{LevinWen}.  

Roughly, a TQFT associates to every surface $\Sigma$ a Hilbert space~$\cH_\Sigma$, and to every diffeomorphism $h: \Sigma \rightarrow \Sigma'$ a linear map $U(h): \cH_\Sigma \rightarrow \cH_{\Sigma'}$.  By restricting to self-diffeomorphisms of~$\Sigma$, the maps $\{U(h)\}_h$ form a representation of $\Sigma$'s mapping class group.  For example, if $\Sigma$ is the sphere with $n$ punctures, its mapping class group is the braid group on $n-1$ strands.  For certain TQFTs, including that based on the Fibonacci category, the resulting computational model has been shown to be equivalent to standard quantum computers~\cite{Freedmanuniversality00, Freedmanetal02}.  

We address the problem of implementing TQFT computations for modular tensor categories in the Turaev-Viro spin-lattice code space.  
This involves the following three steps:
\begin{enumerate}[(i)]
\item
We identify bases of $\cH_\Sigma$ corresponding to fusion diagrams for the doubled theory.  This essentially boils down to decomposing the braid group representation on~$\cH_\Sigma$, a Hilbert space of colored ribbon graphs modulo local equivalence relations.  
\item
For every braid group element $b$, we show how to realize the encoded unitary $U(b)$ as a sequence of protected, local gates.  Pachner moves on the surface's triangulation correspond to unitary $F$-moves on the spins.  Under such a transformation, the code space, i.e., the ground space~$\cH_t$ of a local stabilizer Hamiltonian~$H_t$, is transformed to the ground space of a local stabilizer Hamiltonian $H_{t+1}$ which can be obtained by adiabatically changing three terms in~$H_t$.
\item
Finally, we give procedures for preparing certain initial states in $\cH_\Sigma$, and for measuring the topological charge within a region.
\end{enumerate}
Altogether, this provides a large family of schemes for universal quantum computation using geometrically local operations on locally stabilized codewords.  These schemes may serve as a starting point for developing fault-tolerance schemes using continuous stabilizer measurements and active error-correction. 

The computational scheme presented here is intimately linked to fundamental concepts in topological quantum computation.  The existence of such a computational scheme based on a doubled category was previously conjectured in~\cite[Section~6.2]{RowStonZhen09}.
A similar scheme was proposed by Freedman using the $SU(2)$ Witten-Chern-Simons modular functor at level three~\cite{Freedman00}, but in contrast to our work the stability of his proposal relies on a conjectured energy gap. 

The purpose of this paper is partly introductory and we keep the presentation as self-contained as possible.  In particular, though helpful, we will not require previous knowledge of anyonic fusion spaces; see~\cite{preskill} and the appendix of~\cite{KitaevAnyons} for excellent introductory reviews.

The ideas in this paper originated in an email exchange between one of the authors (G.K.) and Alexei Kitaev.

\subsection*{Outline}

In \secref{sec:TVHilbertspace}, we introduce the Hilbert space $\cH_\Sigma$ associated with a surface~$\Sigma$ by the Turaev-Viro-TQFT defined by a category~$\cC$.  This is a Hilbert space of colored ribbon graphs embedded in~$\Sigma$ modulo local equivalence relations.  We explain how self-diffeomorphisms $h: \Sigma \rightarrow \Sigma$ act on this space by deforming these graphs.  In \secref{sec:anyonsdefinition}, we show how to decompose this mapping class group action into its anyonic content: the representation is described by the ``doubled'' category~$\cC \otimes \cC^*$.  In \secref{sec:realization}, we explain how the space~$\cH_\Sigma$ arises as a subspace of qudits on a lattice; this is the Turaev-Viro code.  We also introduce the corresponding stabilizer Hamiltonian defined by Levin and Wen~\cite{LevinWen}.  In \secref{sec:computation}, we show how to encode information into the code space, execute logical operations, and read out encoded information.  

Throughout Sections~\ref{sec:TVHilbertspace} through~\ref{sec:computation}, we restrict our attention to the Fibonacci category.  
In \secref{sec:generalanyonmodels}, we extend these results to arbitrary ribbon categories.  We conclude in \secref{sec:tvinvariantcode} with a discussion of the relation between the Turaev-Viro $3$-manifold invariant and the Turaev-Viro code.  

In the appendices, we give proofs and examples for various claims, and show how our results extend to mapping class group representations of surfaces with higher genus.  In particular, this leads to an alternative (known) expression for the Turaev-Viro invariant, as discussed in \appref{sec:crane}.

%% file: stringnethilbertspace.tex
\section{The Hilbert space \texorpdfstring{$\cH_\Sigma$}{H\_Sigma} of Fibonacci ribbon graphs on \texorpdfstring{$\Sigma$}{Sigma}} \label{sec:TVHilbertspace}

\subsection{Definition of Fibonacci ribbon graphs and \texorpdfstring{$\cH_\Sigma$}{H\_Sigma}} \label{sec:stringnetdefinition}

Let $\Sigma$ be a compact, orientable surface with boundary, and a single marked point on each boundary component.  We begin by associating to $\Sigma$ a ``ribbon graph'' Hilbert space $\cH_\Sigma$.

A ribbon graph (also called skein/spin network or string net) on $\Sigma$ is the smooth embedding into the interior of $\Sigma$ of a graph in which each vertex has degree either two or three, except that a single vertex of degree one is allowed to be mapped to each marked point on a boundary component.  The ribbon graph Hilbert space $\cH_\Sigma$ is the space of formal linear combinations of  ribbon graphs embedded in~$\Sigma$ modulo the following local relations:
 
\begin{align}
\raisebox{-.8em}{\includegraphics{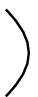}}
&=
\raisebox{-.8em}{\includegraphics{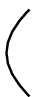}}
\label{eq:firstlevinwen}
\\
\raisebox{-0.7em}{\includegraphics{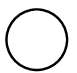}} &= \tau \label{e:looprule} \\
\horizontalt &= \sqrt \tau \verticalv - \frac{1}{\sqrt \tau} \horizontalv \label{eq:lastlevinwen}
\end{align}
Here the first rule is meant to indicate that ribbon graphs related by smooth deformations of the embedding are equivalent.  The second rule allows us to eliminate loops, picking up a factor of 
\begin{equation}
\tau := \frac{1 + \sqrt 5}{2}
 \enspace ,
\end{equation}
provided that the loop can be contracted to a point within the surface.  The third rule allows for eliminating adjacent trivalent vertices from the ribbon graph.  

By applying the above rules, we derive 
\begin{equation} \label{eq:tadpolederive}
\raisebox{-1.8ex}{\includegraphics{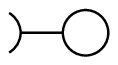}}
=
\sqrt \tau
\raisebox{-1.8ex}{\includegraphics{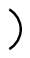}}
- \frac 1 {\sqrt \tau}
\raisebox{-1.8ex}{\includegraphics{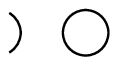}}
= 0
\end{equation}
Therefore we impose as an additional rule
\begin{equation} \label{eq:tadpolerule}
\raisebox{-1.6ex}{\includegraphics{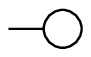}}
= 0
\end{equation}
The only case in which this is not already a consequence of the other rules is when the dangling edge connects to a boundary point.

\subsection{Computational bases for and inner product on \texorpdfstring{$\cH_\Sigma$}{H\_Sigma}} \label{sec:netbasis}

The space $\cH_\Sigma$ is finite dimensional, with a dimension that depends on the topology of $\Sigma$.  We will concentrate on the case where $\Sigma$ is the $n$-punctured sphere, $\Sigma_n = S^2 \smallsetminus (A^1\cup \cdots \cup A^n)$, i.e., the sphere $S^2$ with $n$ discs $A^i$ removed and one boundary point $p_i \in \partial A^i$ for every hole.  Then $\cH_\Sigma$ has some simple bases associated to the dual graphs of certain triangulations.  We call these bases ``computational bases" because they are useful for encoding $\cH_{\Sigma_n}$ using qubits.  

A labeling $\ell$ of $\Sigma$ associates an element $\ell(p) \in \{0,1\}$ to every marked boundary point $p$.  A ribbon graph is consistent with $\ell$ if it has an edge ending at exactly those $p$ with $\ell(p) = 1$.  Then the space $\cH_\Sigma$ can be decomposed as
\begin{equation} \label{eq:boundarylabeldecomp}
\cH_\Sigma = \bigoplus_\ell \cH_\Sigma^\ell
 \enspace ,
\end{equation}
where $\cH_\Sigma^\ell$ is the subspace spanned by ribbon graphs that are consistent with the labeling $\ell$.  

First consider the case where we restrict the ribbon graph to have open boundary conditions, i.e., to have no edges touching the boundaries of $\Sigma$.  In this case, fill in every boundary piece with a disk, and consider a (degenerate) triangulation of $\Sigma$ whose vertices are the centers of these disks.  Then the graph dual to this triangulation is trivalent.  Placing a bit, $0$ or $1$, in the center of each edge, determines a ribbon graph by interpreting a $1$ as the presence of an edge and a $0$ as an absence.  The set of all such $0$/$1$ assignments that have no dead ends---that is, for every vertex of the dual graph, zero, two or three of the incident edges must be set to $1$---is a basis for $\cH_\Sigma$.  Moreover, we can define an inner product on $\cH_\Sigma$ by setting these basis states to be orthonormal.  One can verify that this inner product does not depend on the choice of basis, i.e., on the triangulation.  Several examples are given in \figref{f:basisexamples}.  

\begin{figure*}
\centering
\begin{tabular}{c@{$\qquad$}c@{$\qquad$}c@{$\qquad$}c}
\subfigure[]{\label{f:basisexamples1}\raisebox{.1in}{\includegraphics{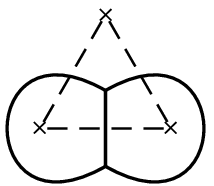}}}
&
\subfigure[]{\label{f:basisexamples2}\raisebox{.1in}{\includegraphics{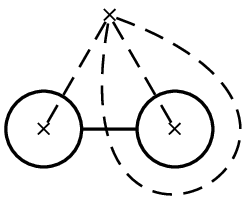}}}
&
\subfigure[]{\label{f:basisexamples3}\raisebox{.1in}{\includegraphics{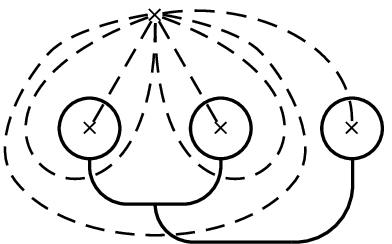}}}
&
\subfigure[]{\label{f:basisexamples4}\raisebox{.1in}{\includegraphics{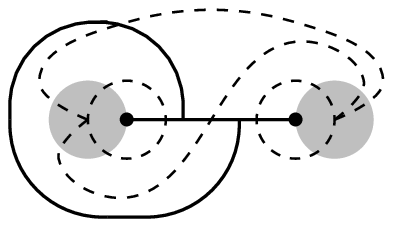}}}
\end{tabular}
\caption{In (a) and (b) are shown two different examples of basis choices for $\cH_\Sigma$ for the case of~$\Sigma$ being the sphere with three punctures (indicated with crosses).  Dashed lines mark the point-set triangulation of the punctures, while the dual graph, which is trivalent, is indicated with solid lines.  Part (c) gives another example of a basis for the sphere with four punctures.  In (a-c), it is assumed that there are no marked points on the boundaries of $\Sigma$.  Part (d) shows the more general situation, for the sphere with two shaded holes, each of which has a marked point on its boundary.} \label{f:basisexamples}
\end{figure*}

The fully general setting, in which ribbon graph edges are allowed to go to a marked point on each boundary component, is slightly more complicated.  In this case, again fill in the holes and place a vertex in the middle of each hole.  For each of these vertices, add a loop going around the marked point on the boundary.  Then complete this graph to a triangulation of the vertices.  The dual graph will be trivalent, except with dead ends at the marked boundary points.  Then the basis and inner product are defined as above on this graph.  

For example, if $n = 0$ or $1$, then $\cH_{\Sigma_n} \cong \C$ is one-dimensional.  (In the case $n = 1$, the rule in Eq.~\eqnref{eq:tadpolerule} ensures that $\cH_{\Sigma_1}^1$, ribbon graphs with an edge attached to the boundary, is zero-dimensional.)  \figref{f:basisexamples4} shows the $n = 2$ case.  Note that there are four edges in the dual graph.  Since each edge can be absent or present ($0$ or $1$), this gives a natural isometry of $\cH_{\Sigma_4}$ into $(\C^2)^{\otimes 4}$.  However, there are only seven edge assignments without dead ends; $\dim(\cH_{\Sigma_4}) = \dim(\cH_{\Sigma_4}^{(0,0)}) + \dim(\cH_{\Sigma_4}^{(1,1)}) + \dim(\cH_{\Sigma_4}^{(0,1)}) + \dim(\cH_{\Sigma_4}^{(1,0)}) = 2 + 3 + 1 + 1 = 7$.  More generally, for $n \geq 2$, a basis derived in this way from a point-set triangulation gives a natural isometry from $\cH_{\Sigma_n}$ into $(\C^2)^{\otimes (5n-6)}$, i.e., a way of implementing $\cH_{\Sigma_n}$ using $5n-6$ qubits.  Therefore, we call such a basis for $\cH_{\Sigma_n}$ a \emph{computational basis}.  It defines an inner product on $\cH_{\Sigma_n}$.  

The examples in \figref{f:basisexamples} suggest a general procedure for defining a basis for $\cH_{\Sigma_n}$, based on a \emph{pants decomposition} of the surface~$\Sigma_n$.  For example, the pants decomposition of $\Sigma_4$ corresponding to~\figref{f:basisexamples3} is
\begin{equation}
\raisebox{-2cm}{
\begin{picture}(1,130)
\scalebox{.85}{
\raisebox{5.25cm}{
\makebox[0pt]{
\rotatebox{-30}{$\begin{matrix} 
&\labeledcylinder{ }   &  & \labeledcylinder{ } &\\
\labeledhcylinder{ }&\labeledhpants{ } & \labeledhcylinder{ } & \labeledhpants{ } & \labeledhcylinder{ }
\end{matrix}$}}}}
\end{picture}
}
\end{equation}
In general, the pants decomposition corresponds to a rooted binary tree with $n-1$ leaves, the root and each leaf associated to one of the holes in $\Sigma_n$.  The degree-three internal vertices correspond to ``pants segments," each isomorphic to $\Sigma_2$, and the edges are ``cylindrical segments," each isomorphic to $\Sigma_3$.  The computational basis is obtained by associating four qubits to every leaf, and one qubit to the root and to every internal pant segment.

\subsection{Actions on \texorpdfstring{$\cH_\Sigma$}{H\_Sigma}} \label{sec:mappingclassgroupaction}

Let $\Sigma_A$ and $\Sigma_B$ be two surfaces, and let $\Sigma$ be the surface obtained by gluing together one or more boundary circles of the surfaces, in such a way that the marked boundary points are matched.  Then ribbon graphs on $\Sigma_A$ and $\Sigma_B$ with matching boundary labels can be glued together to obtain a ribbon graph on $\Sigma$.  In this way, every element $s\in\cH_{\Sigma_B}^{\ell_B}$ defines a linear operator $\hat s : \cH_{\Sigma_A}\rightarrow \cH_\Sigma$.  The result of applying $\hat s$ to $t\in\cH_{\Sigma_A}$ is the linear combination of ribbon graphs obtained by gluing together all ribbon graphs comprising~$s$ and~$t$ with matching labels, up to the equivalences in \secref{sec:stringnetdefinition}.  The map $\h{\ \ }:\cH_{\Sigma_B} \rightarrow \h{\cH_{\Sigma_B}}$ is clearly an isomorphism of vector spaces; we will therefore interpret ribbon graphs interchangeably as states or operators.  For example, if $\Sigma_B = \Sigma_2$ and $\Sigma_A = \Sigma_3$, 
\begin{equation}
\h{\cylinderSdag} \left(\alpha \cdot \pantsttt + \beta \cdot \pantsttv \right) = \alpha \pantstwisted
\end{equation}
In alternative notation, the same equation can be written 
\begin{equation} \label{eq:rootdehntwist}
\h{\Sdag} \left(\alpha \cdot \bttt + \beta \cdot \bttv \right) = \alpha \threecob{\wvvI}{\wvvI}{\ttt}{\threeSdag}
 \enspace .
\end{equation}
This result can be expressed in terms of the basis ribbon graphs by application of the rules~\eqnref{eq:firstlevinwen}--\eqnref{eq:lastlevinwen}.  

Now specialize once more to the case of the $n$-punctured sphere $\Sigma = \Sigma_n$.  The surface $\Sigma$ together with its marked boundary points $p_1,\ldots,p_n$ defines the mapping class group $\cM(\Sigma, \{p_1,\ldots,p_n\})$.  Its elements are isotopy classes of orientation-preserving diffeomorphisms of $\Sigma$ which fix the boundary points.  The mapping class group is generated by operators of two kinds, Dehn-twists and braid-moves.  A Dehn-twist is a $2\pi$-twist along a simple closed curve.  A braid-move affects a pair of pants with two holes in $\Sigma$, and is defined as a $\pi$-twist around a simple closed curved on $\Sigma$ enclosing both holes, followed by (or preceded by) $\pi$-twists on each of the legs, viz.
\begin{equation} \label{eq:braidpicture3d}
\pantsttt \overset{\textrm{braid}}{\longmapsto} \pantsdeformed = \pantsbraided
\end{equation}

The mapping class group $\cM$ acts on the ribbon graph space $\cH_\Sigma$ by linearly extending its action on the basis ribbon graphs to the full space $\cH_\Sigma$.  This action is \emph{unitary} with respect to the inner product defined previously.  

Let us give a more explicit description of the action of Dehn-twists and braid-moves on the ribbon graph Hilbert space.  First consider a non-contractible curve~$\gamma$.  
\begin{equation} \label{eq:dehntwistright}
\raisebox{-.5cm}{
\begin{picture}(25,35)(0,00)
\put(0,20){\cylinderdehntwistcurve}
\put(5,10){$\gamma$}
\end{picture}}
\end{equation}
The surface $\Sigma$ is topologically unchanged if we cut along $\gamma$ and insert into the gap a cylindrical segment $\Sigma_2$.  Then the Dehn-twist about the~$\gamma$ corresponds to the linear operator $D(\gamma): \cH_\Sigma \rightarrow \cH_\Sigma$ given by 
\begin{equation} \label{eq:dehntwistopdef}
D(\gamma) = \h{\cylinder} + \h{\cylinderS} = \h{\vvI} + \h{\S}
\end{equation}

Similarly, the  operator corresponding to braiding can be specified by
\begin{equation}
B = \sum_{i,j \in \{0,1\}} \fourpuncturedspherebraid{i}{j}
\end{equation}
where the binary variables $i$ and $j$ indicate the presence or absence of an edge.  Indeed, stacking this on top of the two legs in the diagram on the left-hand side of Eq.~\eqnref{eq:braidpicture3d} gives the right-hand side thereof. 

We will next identify irreducible subspaces of the mapping class group action.  

%% file: anyonsdefinition.tex
\section{Doubled Fibonacci anyons in \texorpdfstring{$\cH_\Sigma$}{H\_Sigma}} \label{sec:anyonsdefinition}

Let $\Sigma_n$ be the $n$-punctured sphere.  \secref{sec:netbasis} showed how to construct 
a computational basis for $\cH_{\Sigma_n}$ based on one of a family of point-set triangulations.  
It also specialized this construction to the case of a triangulation derived recursively from a pants decomposition of $\Sigma_n$.  These computational bases are useful for encoding $\cH_{\Sigma_n}$ using qubits, and we will study this implementation further in \secref{sec:realization}.

However, these bases are inconvenient operationally, because they do not transform cleanly under the application of Dehn-twists and braid-moves.  In particular, it is complicated to express in the pants decomposition computational basis the action of a Dehn-twist about either the root or an internal cylindrical segment.  From the right-hand side of Eq.~\eqnref{eq:rootdehntwist}, the loop has to be propagated upward toward the leaves, potentially requiring many applications of the ribbon graph equivalence rules.  Additionally, the basis is asymmetrical, since the puncture at the root of the pants decomposition tree is treated differently from the other punctures.  Intuitively, too much information is kept at the leaves.  

In this section, we will define an orthonormal basis that simultaneously diagonalizes the Dehn-twists.  Each basis element will be represented by an ``anyon fusion diagram."  
Although slightly more complicated to define, the advantage of this anyonic basis is that it will reveal more mathematical structure and greatly clarify the steps necessary for obtaining universal quantum computation.  

Let a bold, dashed line indicate a linear combination of the two diagrams with and without that line: 
\begin{equation} \label{eq:vacuumlinesfibonacci}
\raisebox{.1cm}{\includegraphics{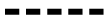}}
=
\frac{1}{\sqrt{1+\tau^2}} ( 
\raisebox{.1cm}{\qquad\quad}
+
\tau \,
\raisebox{.1cm}{\includegraphics{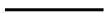}})
\end{equation}
Such lines are called vacuum lines, because a simple calculation shows that other lines can pass over them freely: 
\begin{equation} \label{eq:vacuumquick}
\sqrt{1+\tau^2}
\raisebox{-.6cm}{\includegraphics{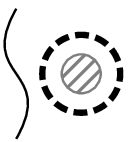}}
= 
\raisebox{-.6cm}{\includegraphics{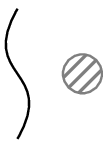}}
+
\tau
\raisebox{-.6cm}{\includegraphics{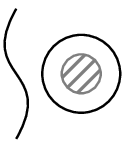}}
=
\raisebox{-.6cm}{\includegraphics{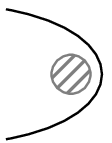}}
+
\tau
\raisebox{-.6cm}{\includegraphics{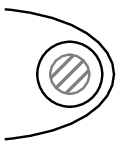}}
= 
\sqrt{1+\tau^2}
\raisebox{-.6cm}{\includegraphics{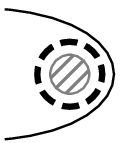}}
\end{equation}
Here, Eq.~\eqnref{eq:vacuumquick} is to be interpreted as a local identity, that is ribbon graphs inside the hatched area or outside the diagram are assumed to be the same for every term.  Its proof is a straightforward application of the rules~\eqnref{eq:firstlevinwen}-\eqnref{eq:lastlevinwen}.  \lemref{lem:vacuum} states further properties of the vacuum lines.

\subsection{Anyonic fusion basis for \texorpdfstring{$\cH_{\Sigma_n}$}{H\_\{Sigma\_n\}}} \label{s:anyonicfusionbasis}

Fix a pants decomposition of $\Sigma_n$, corresponding to a rooted binary tree $T$.  Fix a labeling $\ell$ of the marked boundary points, as in Eq.~\eqnref{eq:boundarylabeldecomp}.  

A labeling of $T$ is an assignment to each edge of $T$ 
either $\oneanyon$ or $\tauanyon$.  A labeling is fusion-consistent if no internal vertex has exactly one incident edge labeled $\tauanyon$.  A pair of labelings of $T$ is boundary-consistent with~$\ell$ if for each boundary point $p$, with labels of the corresponding edges $b_+$ and $b_-$, 
\begin{equation}\begin{split}
\ell(p) = 0 &\Rightarrow b_+ b_- \in \{\oneone, \tautau\} \\
\ell(p) = 1 &\Rightarrow b_+ b_- \in \{\tauone, \onetau, \tautau\}
 \enspace .
\end{split}\end{equation}
A pair of fusion-consistent labelings of $T$  which is boundary-consistent with $\ell$ is called an $\ell$-consistent doubled anyon fusion diagram.  

The \emph{anyonic fusion basis} states for $\cH_{\Sigma_n}^\ell$ are an orthonormal basis, with basis elements indexed by $\ell$-consistent doubled anyon fusion diagrams. 
In the remainder of this subsection, we will define the state $\ket{\ell,d} \in \cH_{\Sigma_n}^\ell$ that corresponds to an $\ell$-consistent doubled anyon fusion diagram~$d$.  
The state $\ket{\ell, d}$ is defined by first constructing a \emph{three-dimensional} ribbon graph living in the thickened surface $\Sigma \times [-1, 1]$, and then reducing this ribbon graph down to two dimensions.  

Think of each of the two labelings in~$d$ as a ribbon graph, by interpreting label~$\oneanyon$ as the absence of an edge and label $\tauanyon$ as the presence of an edge along the edges of the tree $T$.  Place these two ribbon graphs on $\Sigma \times \{1\}$ and $\Sigma \times \{-1\}$ in $\Sigma \times [-1, 1]$.  Then add vacuum lines around each puncture in $\Sigma \times \{0\}$, except not around the puncture at the root of the pants decomposition.  

In the resulting three-dimensional ribbon graph, edges can end at $(p, -1)$ and $(p, 1)$ for each marked boundary point $p$ of $\Sigma$.  However, we would like there only to be a single boundary condition at the point $(p, 0)$.  To fix the boundary conditions, close off the ribbon graphs at each boundary with one of the five following diagrams.  Here, the vertical axis represents the second coordinate, from $-1$ to $1$, with the point $(p, 0)$ marked.  
\begin{equation}
\begin{array}{c @{$\quad\;$} c @{$\quad\;$} c @{$\quad\;$} c @{$\quad\;$} c @{$\quad\;$} c}
&
\includegraphics{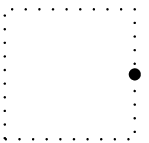}
&
\includegraphics{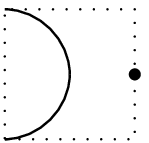}
&
\includegraphics{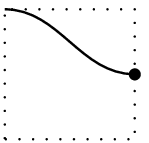}
&
\includegraphics{images/boundaryconditionstauone1}
&
\includegraphics{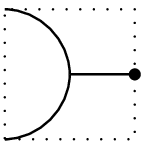}
\\
b_+ b_-: & \oneone & \tautau & \tauone & \onetau & \tautau \\
\ell(p): & 0 & 0 & 1 & 1 & 1
\end{array}
\end{equation}
Boundary-consistency of $d$ implies that one of these diagrams applies.  Note that this closing off is done inside the vacuum line around the hole at $p$.  

Finally, visualize the three-dimensional ribbon graph as a two-dimensional diagram with crossings.  (This requires slightly offsetting the ribbon graph boundaries in $\Sigma \times \{1\}$ from the boundaries in $\Sigma \times \{-1\}$.  Follow any convention for choosing the direction of the offset; different conventions will only change the phase of $\ket{\ell,d}$.)  Use repeatedly the rule
\begin{equation}\begin{split} \label{e:overcrossing}
\overcrossing &= e^{-3\pi i/5} \overcrossingresolvedv + e^{3\pi i/5} \overcrossingresolvedupv \\
\end{split}\end{equation}
to eliminate ribbon crossings, finally obtaining a ribbon graph in $\Sigma \times \{0\} \cong \Sigma$.  This state is~$\ket{\ell,d}$.  

Note that ribbon graphs in three dimensions can be manipulated similarly to ribbon graphs in two dimensions, for example by smoothly deforming the embedding.  However, it is useful to think of the ribbon graphs in three dimensions as actually being thin ribbons, since kinks cannot be freely eliminated.  For example, we derive from Eq.~\eqnref{e:overcrossing} that 
\begin{equation}\begin{split} \label{e:dekink}
\raisebox{-2pt}{\includegraphics[scale=.5]{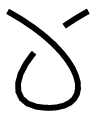}} &= e^{-4\pi i/5} \;
\raisebox{-2pt}{\includegraphics[scale=.5]{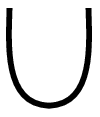}} \\
\raisebox{-4pt}{\includegraphics[scale=.5]{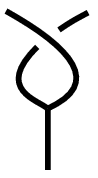}} &= e^{3\pi i/5} \raisebox{-4pt}{\includegraphics[scale=.5]{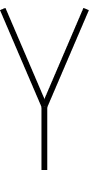}}
\end{split}\end{equation}

The anyonic fusion basis states form an orthonormal basis.  We will prove this and other properties, in a more general setting, in \appref{sec:normalizationgeneral}.

\subsubsection*{Example: Anyonic fusion basis states for $\cH_{\Sigma_2}$}

For $\Sigma_2$, the only pants decomposition is as a single cylindrical segment.  $\cH_{\Sigma_2}$ is seven-dimensional, spanned by the states $\{\vvI, \ttI, \S,\Sdag, \O,\tvD, \vtD\}$, defined by 
\begin{equation} \label{eq:twopuncturedbasis}
\begin{matrix}
\; \cylinder \;&\; \cylinderttI \;&\; \cylinderS \;&\; \cylinderSdag \;&\; \cylinderO \;&\; \cylindertvD \;&\; \cylindervtD \\
\\
\vvI & \ttI & \S & \Sdag & \O & \tvD & \vtD
\end{matrix}
\end{equation}
To explain the notation, notice, for example, that $\vtD$ corresponds to looking downward on the final cylinder above, as the base of the cylinder spreads out.  The basis elements are all of unit length and orthogonal, \emph{except} $\braket \S \Sdag = 1/\tau$.  

An anyonic fusion basis for $\cH_{\Sigma_2}^\ell$ is determined by a single label $\oneone$, $\tautau$, $\onetau$ or $\tauone$ consistent with the boundary conditions $\ell$.  Explicitly, the orthonormal 
bases for $\cH_{\Sigma_2} = \cH_{\Sigma_2}^{(0,0)} \oplus \cH_{\Sigma_2}^{(1,1)} \oplus \cH_{\Sigma_2}^{(1,0)} \oplus \cH_{\Sigma_2}^{(0,1)}$ are given by: 
\begin{eqnarray} \label{e:anyontocomputationalbasisbasecase}
\cH_{\Sigma_2}^{\ell = (0,0)}: & \Bigg\{ \begin{aligned}
\ket{\ell, \oneone} &:= \sqrt{1+\tau^2} \, \oneone \ket{\vvI}, \\
\ket{\ell, \tautau} &:= \sqrt{1+ \bar\tau^2} \, \tautau\coeff\v \ket{\vvI} 
\end{aligned} \Bigg \} \nonumber \\
\cH_{\Sigma_2}^{\ell = (1,1)}: & \left\{ \begin{aligned}
\ket{\ell, \tautau} &:= 5^{1/4} \, \tautau\coeff\t \ket{\ttI}, \\
\ket{\ell, \onetau} &:= \sqrt{1+\tau^2} \, \onetau \ket{\ttI}, \\
\ket{\ell, \tauone} &:= \sqrt{1+\tau^2} \, \tauone \ket{\ttI}
\end{aligned} \right\} \\
\cH_{\Sigma_2}^{\ell = (1,0)}: & \{ \ket{\ell, \tautau} := \ket{\tvD} \} \nonumber \\
\cH_{\Sigma_2}^{\ell = (0,1)}: & \{ \ket{\ell, \tautau} := \ket{\vtD} \} \nonumber
\end{eqnarray}
Here 
\begin{equation}
\taubar := - \frac 1 \tau = \frac{1 - \sqrt 5}{2}
\end{equation}
and we are using the five operators 
\begin{align} \label{e:fiveoperators}
\oneone &= \frac{1}{1+\tau^2} \big( \h\vvI + \tau \h\O \big) \nonumber \\
\tautau\coeff\v &= \frac{1}{1+\bar{\tau}^2} \big( \h\vvI + \bar\tau \h\O \big)
&\tautau\coeff\t &= \frac{1}{1+\tau^2} \big( \tau \h\ttI + \h\S + \h\Sdag \big) \\
\onetau &= \frac{1}{1+\tau^2} \big( \h\ttI + \tau \omega^{-3} \h\S + \tau \omega^3 \h\Sdag \big)
&\tauone &= \frac{1}{1+\tau^2} \big( \h\ttI + \tau \omega^3 \h\S + \tau \omega^{-3} \h\Sdag \big) \nonumber
\end{align}
with $\omega=e^{\pi i/5}$.
These operators are the five minimal central idempotents in the algebra of operators $\cH_{\Sigma_2} \rightarrow \cH_{\Sigma_2}$ of the form $\hat s$, where $s \in \cH_{\Sigma_2}$.

We give more details and additional examples for $\cH_{\Sigma_3}$ in \appref{sec:fibonacciapplication}.

\subsection{Action of Dehn-twists and braid-moves}
The fusion basis states for a particular pants decomposition are simultaneous eigenstates of every Dehn-twist about a cylindrical segment in the decomposition.  This is a main advantage of the fusion basis.  

For concreteness, consider a cylindrical segment with boundary labels $\ell = (1, 1)$ and an anyonic fusion basis state $\ket{\ell, d} = \ket{(1, 1), \tauone}$ such that the segment is labeled $d_+ = \tauanyon$ and $d_- = \oneanyon$  in the two labelings.  Then the ribbon graph, before reducing to two dimensions, looks like 
\begin{equation}
\raisebox{-1.5cm}{\includegraphics{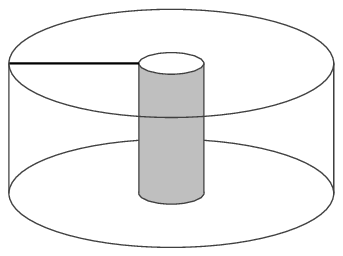}}
\end{equation}
where we show only the $\text{cylindrical segment} \times [-1, 1]$, i.e., $\Sigma_2 \times [-1, 1]$, and not the portion of $\Sigma$ to either side.  By applying Eq.~\eqnref{eq:vacuumquick}, a vacuum loop can be pulled out from a puncture within the shaded region, giving 
\begin{equation}
\raisebox{-.8cm}{\includegraphics{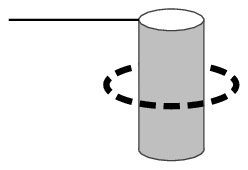}}
\end{equation}
For clarity we have stopped drawing the outer boundary of the $\text{cylindrical segment} \times [-1,1]$.  

The action of a Dehn-twist $D$ on this cylindrical segment is to apply a full rotation.  The Dehn-twist was defined in \secref{sec:mappingclassgroupaction} as acting on $\Sigma$ and ribbon graphs in $\Sigma$, but it is naturally extended to $\Sigma \times [-1, 1]$, and its action on ribbon graphs in $\Sigma \times [-1, 1]$ is easily seen to commute with the reduction to two dimensions $\Sigma \times \{0\}$.  The resulting state simplifies by moving the ribbon down to $\Sigma \times \{0\}$, pulling it over the vacuum loop with Eq.~\eqnref{eq:vacuumquick}, and then removing the kink using Eq.~\eqnref{e:dekink}: 
\begin{equation} \label{e:fibanyondehntwistexample}
\begin{split}
D
\raisebox{-.8cm}{\includegraphics{images/fibanyondehntwistexample1}}
&=
\raisebox{-.8cm}{\includegraphics{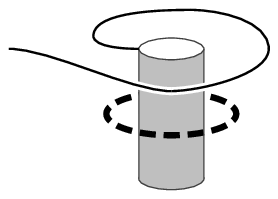}}
=
\raisebox{-.8cm}{\includegraphics{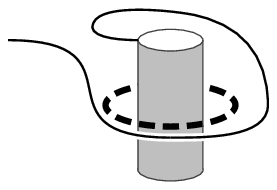}}
\\
&=
\raisebox{-.8cm}{\includegraphics{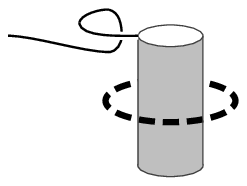}}
=
e^{-4 \pi i / 5}
\raisebox{-.8cm}{\includegraphics{images/fibanyondehntwistexample1}}
\end{split}\end{equation}
Thus $D \ket{(1, 1), \tauone} = e^{-4 \pi i / 5} \ket{(1, 1), \tauone}$.  Had the two edge labels been both $\oneanyon$ or both $\tauanyon$, then similar calculations would have shown $\ket{\ell, d}$ to be a $+1$ eigenstate of the Dehn-twist $D$.  The eigenvalue is~$e^{+4 \pi i / 5}$ when the edge labels are $\onetau$.

Now fix a pants segment of $\Sigma$, and let us compute the action of a braid move $R$ in the anyonic fusion basis.  Let $d$ be a doubled anyonic fusion diagram, and say for example that the labels on the edges incident to the pants segment are given by $\fuse{\text{\scriptsize{$\tautau$}}}{\text{\scriptsize{$\tautau$}}}{\text{\scriptsize{$\tauone$}}}$.  Then the corresponding ribbon graph in three dimensions is given by 
\begin{equation}
\raisebox{-.8cm}{\includegraphics{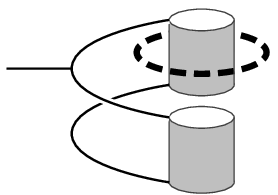}}
\end{equation}
where we have pulled a vacuum loop out from one of the boundaries of $\Sigma$.  Applying $R$ to this state yields 
\begin{equation}\begin{split}
R \;
\raisebox{-.8cm}{\includegraphics{images/fibanyonbraidexample1}}
&=
\raisebox{-1.4cm}{\includegraphics{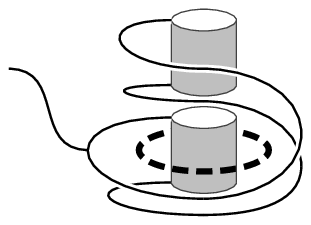}}
=
\raisebox{-1.15cm}{\includegraphics{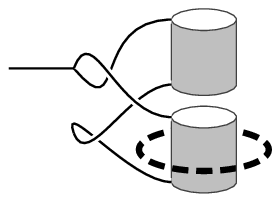}}
=
e^{3 \pi i / 5}
\raisebox{-1.2cm}{\includegraphics{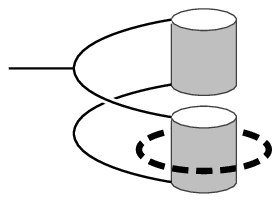}}
\end{split}\end{equation}
Here we have again pulled edges across the vacuum line and applied Eq.~\eqnref{e:dekink} to remove the twists.  

Thus the action of $R$ in this case is to swap the two portions of $d$ beyond two of the boundaries of the pants segment, and add a phase.  In fact, this would have been the action regardless of the doubled anyonic fusion diagram $d$, just with a different phase.  The phase acquired corresponds to the $R$-matrix of the doubled anyon model.  Its general form will be derived in \secref{s:stringnetsdoubledmodel}.

\subsection{Interpretation as a doubled Fibonacci anyon model}

The anyonic fusion basis can be interpreted in terms of the general theory of anyons, in particular of the doubled Fibonacci anyon model~\cite{Fidkowskietal08}.  This interpretation will justify our notation, but readers unfamiliar with anyonic theory may safely skip over it.  

The Fibonacci category $\Fib$ has two anyons, $\oneanyon$ and $\tauanyon$.  Each of the two labelings of $T$ defined in \secref{s:anyonicfusionbasis} is a Fibonacci anyon fusion diagram.  An edge of $T$ labeled $\oneanyon$ represents a vacuum anyon, while an edge labeled $\tauanyon$ represents a $\tauanyon$ anyon.  A vacuum anyon can also be interpreted as the absence of a particle.  The fusion-consistency constraint on a labeling requires that each internal vertex satisfy the fusion rules of $\Fib$: 
\begin{equation}\begin{split}
\oneanyon \anyontimes \oneanyon &\anyoncong \oneanyon \\
\oneanyon \anyontimes \tauanyon &\anyoncong \tauanyon \\
\tauanyon \anyontimes \oneanyon &\anyoncong \tauanyon \\
\tauanyon \anyontimes \tauanyon &\anyoncong \oneanyon \anyonplus \tauanyon
\end{split}\end{equation}
The final equation above should be interpreted to mean that fixing two of the incident labels at a vertex to be $\tauanyon$ and $\tauanyon$, the third label can be either $\oneanyon$ or $\tauanyon$.  However, in the other three cases, fixing two of the incident labels at a vertex also fixes the third label.  

Notice that under both Dehn-twists and braid-moves the two ribbon graphs, in $\Sigma \times \{1\}$ and in $\Sigma \times \{-1\}$, from a doubled anyon fusion diagram transform independently of each other.  Moreover, if these two ribbon graphs are swapped, then the phase of the eigenvalue of any Dehn-twist or braid-move is negated.  

Therefore, a doubled anyon fusion diagram $d$ is the fusion diagram for the \emph{doubled} Fibonacci category $\DFib \cong \Fib \otimes \Fib^*$.  That is, $\DFib$ consists of two copies of the Fibonacci category with opposite chiralities.  If we indicate the four possible pairs of labels for an edge by $\oneone$, $\onetau$, $\tauone$ or $\tautau$, then the allowed fusion rules are
\begin{align}
&\begin{matrix}
\oneone \anyontimes \onetau \anyoncong&\onetau\\
\onetau \anyontimes \oneone \anyoncong&\onetau\\
\onetau \anyontimes \onetau \anyoncong&\oneone \anyonplus \onetau
\end{matrix}
&&\begin{matrix}
\oneone \anyontimes \tauone \anyoncong&\tauone\\
\tauone \anyontimes \oneone \anyoncong&\tauone\\
\tauone \anyontimes \tauone \anyoncong&\oneone \anyonplus \tauone
\end{matrix} \nonumber \\ \\
&\begin{matrix}
\oneone \anyontimes \tautau \anyoncong&\tautau\\
\onetau \anyontimes \tautau \anyoncong&\tautau \anyonplus \onetau\\
\tauone \anyontimes \tautau \anyoncong&\tautau \anyonplus \tauone\\
\tautau \anyontimes \tautau \anyoncong&\oneone \anyonplus \onetau \anyonplus \tauone \anyonplus \tautau\\
\end{matrix}
&&\begin{matrix}
\tautau \anyontimes \oneone \anyoncong&\tautau\\
\tautau \anyontimes \onetau \anyoncong&\tautau \anyonplus \tauone\\
\tautau \anyontimes \tauone \anyoncong&\tautau \anyonplus \onetau\\
\end{matrix} \nonumber \\
&&&\begin{matrix}
\onetau \anyontimes \tauone \anyoncong&\tautau\\
\tauone \anyontimes \onetau \anyoncong&\tautau
\end{matrix} \nonumber
\end{align} 
More compactly, the above fusion rules can be written as
\begin{equation}
(a_1a_2) \anyontimes (b_1b_2) \anyoncong (a_1 \anyontimes b_1)(a_2 \anyontimes b_2)
 \enspace ,
\end{equation}
where $a_i, b_i\in \{\oneanyon, \tauanyon\}$.  They simply require that each coordinate separately be fusion-consistent.  

The fusion rules also allow us to determine the quantum dimensions of these particules, by regarding them as equations for the dimensions.  The result is
\begin{equation}\begin{split}
d_{\oneone} &=1 \\
d_{\onetau} &= d_{\tauone} = \tau \\
d_{\tautau} &= \tau+1
\end{split}\end{equation}
where, recall, $\tau = (1 + \sqrt 5)/2$.

\subsection{Changing the pants decomposition: The  $F$-matrix} \label{sec:changingpants}

To conclude this section, we would like to derive an expression for the basis change between anyonic fusion bases associated to different pants decompositions of the surface $\Sigma$.  A pants decomposition of $\Sigma$ corresponds to a rooted binary tree.  A braid-twist swaps the subtrees at a pants segment, which allows for some basis changes.  However, this only reorders subtrees, and does not change the parent-child structure of the tree.  

Any two pants decompositions having the same root puncture and the same left-to-right ordering of the leaf punctures can be related via a sequence of moves of the following type: 
\begin{equation} \label{e:changingpants}
\fusethree{}{}{}{}{}
\leftrightarrow
\fusethrees{}{}{}{}{}
\end{equation}
This move takes two pants segments connected by a cylindrical segment, and reconnects the middle subtree from one side of the lower pants segment to the other side.  

Let $d$ be a doubled anyon fusion diagram for the pants decomposition on the left-hand side of Eq.~\eqnref{e:changingpants}.  The change in basis to a superposition of doubled anyon fusion diagrams for the right-hand-side pants decomposition is most easily seen in the three-dimensional ribbon graph picture.  
From Eq.~\eqnref{eq:lastlevinwen}, we derive the two equivalent equations 
\begin{equation}\begin{split}
\raisebox{.25cm}{\rotatebox{-45}{\makebox[20pt]{\horizontalv}}} \; &= \;\;\; \frac 1 \tau \; \raisebox{.25cm}{\rotatebox{-45}{\makebox[20pt]{\verticalv}}} \; + \frac 1 {\sqrt \tau} \raisebox{.25cm}{\rotatebox{-45}{\makebox[20pt]{\verticalt}}} \\[5pt]
\raisebox{.25cm}{\rotatebox{-45}{\makebox[20pt]{\horizontalt}}} \; &= \frac 1 {\sqrt \tau} \; \raisebox{.25cm}{\rotatebox{-45}{\makebox[20pt]{\verticalv}}} \; - \frac 1 \tau \;\;\; \raisebox{.25cm}{\rotatebox{-45}{\makebox[20pt]{\verticalt}}}
\end{split}\label{eq:Fmoveeqsimple}\end{equation}
By applying these rules separately to the ribbon graph in $\Sigma \times \{1\}$ and to the ribbon graph in $\Sigma \times \{-1\}$, we obtain the anyon fusion basis change.  As for Dehn-twists and braid-moves, the two ribbon graphs transform independently of each other.  
For example, we find that 
\begin{equation}\begin{split}
\fusethree{\onetau}{\onetau}{\onetau}{\oneone}{\onetau}
&=
\frac{1}{\tau} \fusethrees{\onetau}{\onetau}{\onetau}{\oneone}{\onetau} + \frac{1}{\sqrt{\tau}} \fusethrees{\onetau}{\onetau}{\onetau}{\onetau}{\onetau} \\
\fusethree{\onetau}{\onetau}{\onetau}{\onetau}{\onetau}
&=
\frac{1}{\sqrt{\tau}} \fusethrees{\onetau}{\onetau}{\onetau}{\oneone}{\onetau} - \frac{1}{\tau} \fusethrees{\onetau}{\onetau}{\onetau}{\onetau}{\onetau}
\end{split}\end{equation}
since in the above cases the first labeling is always trivial, i.e., the ribbon graph in $\Sigma \times \{1\}$ is empty.  To give an example of the basis change where both labelings are non-trivial, note that 
\begin{equation}
\fusethree{\tautau}{\tautau}{\tautau}{\tautau}{\tautau}
=
\frac{1}{\tau} \fusethrees{\tautau}{\tautau}{\tautau}{\oneone}{\tautau} + \frac{1}{\tau^2} \fusethrees{\tautau}{\tautau}{\tautau}{\tautau}{\tautau} - \frac{1}{\tau^{3/2}} \left( \fusethrees{\tautau}{\tautau}{\tautau}{\onetau}{\tautau} + \fusethrees{\tautau}{\tautau}{\tautau}{\tauone}{\tautau} \right)
 \enspace .
\end{equation}

In general anyonic theory, the matrix of this basis-change unitary is known as the $F$-matrix; in $\Fib \otimes \Fib^*$, the $F$-matrix is the tensor product of the $F$-matrix for $\Fib$ with the $F$-matrix for $\Fib^*$.  In fact, since the $F$-matrix for $\Fib$ has only real-valued entries, it equals the $F$-matrix for $\Fib^*$.

\subsection{Recursive construction of the anyonic fusion basis states by gluing} \label{sec:recursivegluing}
Let $\Sigma_A = \Sigma_{m+1}$ and $\Sigma_B = \Sigma_{n+1}$ be two spheres
with $n+1$ and $m+1$~punctures, respectively.  By gluing $\Sigma_A$ to $\Sigma_B$ along a boundary component, one obtains the $(m+n)$-punctured sphere $\Sigma = \Sigma_{m+n}$.  Pants decompositions for $\Sigma_A$ and $\Sigma_B$ then combine to give a pants decomposition for $\Sigma$.  Anyonic fusion basis states on $\Sigma$ can thus be understood as the result of combining anyonic fusion basis states for $\Sigma_A$ and $\Sigma_B$.  

Indeed, for a labeling $\ell \in \{0,1\}^{m+n}$ of the boundary components of $\Sigma$, and a fusion-consistent doubled anyon labeling~$d$ of the pants decomposition, we have 
\begin{equation} \label{eq:gluingrule}
\ket{\ell, d}_\Sigma = \sum_{k \in \{0,1\}} \alpha_{k,d} \widehat{\ket{\ell^k_B, d_B}}_{\Sigma_B} \ket{\ell^k_A, d_A}_{\Sigma_A}
 \ .
\end{equation}
Here, $\ell_A^k$ is the restriction of $\ell$ to the boundary components of $\Sigma_A$, additionally with label $k$ assigned to the glued boundary; and $d_A$ is $d$ restricted to the pants decomposition of $\Sigma_A$.  The quantities $\ell_B^k$ and $d_B$ are defined similarly.  The coefficient~$\alpha_{k,d}$ depends on the anyon labels $b_+b_-$ assigned to the connecting edge by~$d$, according to 
\begin{equation} \label{eq:alphakddef}
\alpha_{k,d} = 
\frac{1}{\sqrt{1+\tau^2}}
\begin{cases}
1 & \text{if $b_+ b_- \neq \tautau$} \\
\frac{1}{\tau} & \text{if $b_+ b_- = \tautau$ and $k = 0$} \\
\frac{1}{\sqrt \tau} & \text{if $b_+ b_- = \tautau$ and $k = 1$}
\end{cases}
\end{equation}
This can be derived by using the first identity of Eq.~\eqnref{eq:Fmoveeqsimple} to change the ribbon graph locally so that at most one ribbon crosses from $\Sigma_A$ to $\Sigma_B$.  

%% file: levinwenmodel.tex
\section{Realizing the Fibonacci ribbon-graph Hilbert space \texorpdfstring{$\cH_\Sigma$}{H\_Sigma}} \label{sec:realization}

In this section, we define the Fibonacci surface code.  Starting with a triangulation of a surface~$\Sigma$, certain fixed boundary conditions $\ell$, and qubits placed on each edge of the triangulation, we describe commuting local projection operators such that their joint $+1$ eigenspace is the code space $\cH^\ell_\Sigma$.  
In \secref{sec:levinwen}, we define the local projection operators.  These operators are essentially the same as those in the Levin-Wen Hamiltonian, except defined to take into account the boundary conditions.  In subsequent subsections, we argue that the code space is $\cH^\ell_\Sigma$, and describe various properties of the code.

\subsection{The Fibonacci code} \label{sec:levinwen}
Let $\Sigma = \Sigma_n = S^2 \smallsetminus (A^1 \cup \cdots \cup A^n)$ be the $n$-punctured sphere, and let $\ell$ be a $0$ or $1$ labeling of each hole $A^i$.  

Let $\cT$ be a triangulation of $\Sigma$, and let $\widehat\cT = (V, E)$ be the dual graph to $\cT$.  $\widehat\cT$ is a connected graph embedded in $\Sigma$.  The components of $\Sigma \smallsetminus \widehat\cT$ are simply connected, and we refer to them as plaquettes.  The vertices of $\widehat\cT$ all have degree at most three.  A vertex $v$ in $\cT$ can be univalent (respectively, bivalent) if the corresponding triangle in $\cT$ has two edges (resp., one edge) along a boundary of $\Sigma$.  It will be convenient to imagine boundary edges in the dual graph that correspond to the edges along the boundary in $\cT$.  A boundary edge in the dual graph leaves a vertex $v$ and crosses into some hole $A^i$.  If these boundary edges are included, then each vertex $v$ in $\cT$ is trivalent.  

Now define a labeling on the boundary edges.  For each boundary component, arbitrarily fix one boundary edge, and label that edge by the value of $\ell$ on that component.  Label all other boundary edges $0$.  We will refer to this new labeling also by $\ell$, cf.~\figref{fig:levinwenmodelboundary}.

\begin{figure*}
\centering
\begin{tabular}{c@{$\qquad$}c@{$\qquad$}c}
\subfigure[]{\label{fig:firstgluingdiscr}\raisebox{.1in}{\includegraphics[scale=1]{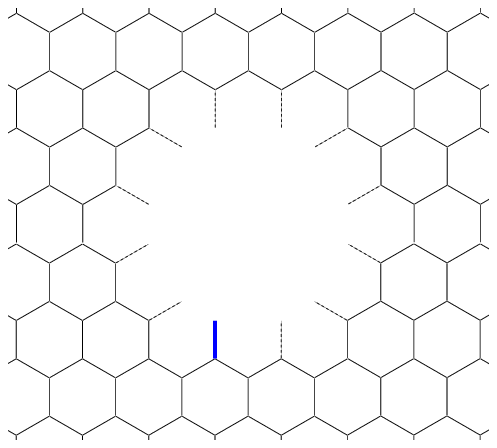}}}
&
\subfigure[]{\label{fig:secondgluingdiscr}\raisebox{.425in}{\includegraphics[scale=1]{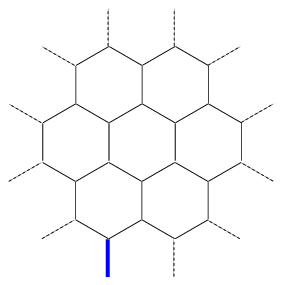}}}
&
\subfigure[]{\label{fig:gluingdiscrete}\raisebox{.1in}{\includegraphics[scale=1]{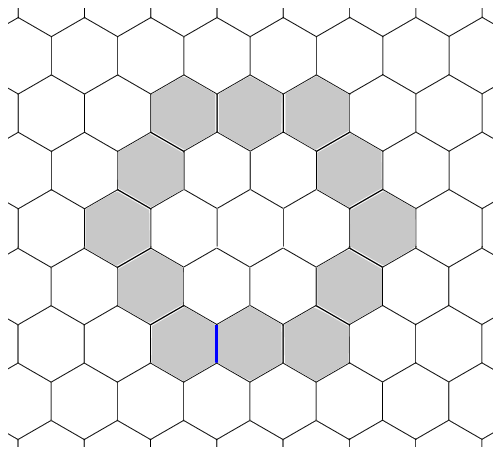}}}
\end{tabular}
\caption{The Levin-Wen model with a boundary, realizing $\cH_\Sigma^\ell$: qubits sit on the solid lattice edges.  In (a) and (b), dashed edges represent ``virtual'' qubits fixed to~$0$.  The thick, blue edge represents a virtual qubit fixed to the label~$\ell(p)$ of the marked boundary point~$p$.  
In (c) is illustrated the gluing of anyon fusion basis states described in \secref{sec:anyonfusionbasislattice}.  The $t$ virtual qubits are replaced by qubits, and prepared in states $\ket k \otimes \ket 0^{\otimes (t-1)}$, for $k \in \{0,1\}$, entangled with states on either side.  Subsequently, all plaquette operators touching the boundary (shaded plaquettes) are applied.  } \label{fig:levinwenmodelboundary}
\end{figure*}
 
The Hilbert space~$\cH_{\widehat \cT} = (\mathbb{C}^2)^{\otimes \abs{E}}$ of our system is obtained by placing a qubit on each edge $e$.  Note that boundary edges are not given qubits.  We use orthonormal bases $\{\ket{0}_e, \ket{1}_e\}$ for the local Hilbert spaces.  As in \secref{sec:netbasis}, we will use the convention that a $1$ indicates the presence of a ribbon along the edge, while a $0$ indicates the absence of a ribbon.  

The Fibonacci code is a subspace of $(\C^2)^{\otimes E}$, defined as the simultaneous~$+1$ eigenspace of a set of commuting projections.  There are two different kinds of projections: 
\begin{itemize}
\item
For every vertex $v$ of the graph, there is a vertex projection $Q_v$.  $Q_v$ depends only on the three edges incident to $v$, and on those three edges is the diagonal operator 
\begin{equation} \label{eq:qvdef}
Q_v = \sum_{\substack{i, j, k \in \{0,1\} : \\ i + j + k \neq 1}} \ketbra{ijk}{ijk}
 \enspace .
\end{equation}
That is, $Q_v$ imposes that the vertex $v$ is a legal ribbon graph vertex without dead-ends.   
If $v$ is connected to one or two boundary edges, then $Q_v$ is defined as in Eq.~\eqnref{eq:qvdef}, except with the boundary edges fixed by the labeling $\ell$.  For example, if there is one boundary edge $e$, then $Q_v$ is given by $\sum_{i, j : \, i + j + \ell(e) \neq 1} \ketbra{ij}{ij}$ on the other two edges.  

Since all of the vertex projections are diagonal in the same basis, they commute with one another.  Let $P$ be the number of plaquettes in $\widehat\cT$ and let $\Delta \cong \Sigma_{n+P}$ be the surface obtained by placing a puncture in the center of each plaquette in $\Sigma$.  Then the simultaneous $+1$ eigenspace of the vertex projections is~$\cH_\Delta^{(\ell,0^P)}$, the space of valid ribbon graphs on $\Delta$ with open boundary conditions on the new punctures.  The computational basis for the qubits corresponds to the computational basis for $\cH_\Delta^{(\ell,0^P)}$ defined in \secref{sec:netbasis}.  We call this space the {\em physical subspace} of $\cH_{\widehat \cT}$.

In more detail, we regard every physical configuration as a ribbon graph on~$\Delta$, using the embedding of $\widehat \cT$ into $\Delta$.  Because of the punctures~$\cP$, no two different computational basis states or physical configurations are identified with the same ribbon graph.  This shows that physical subspace is indeed isomorphic to $\cH_{\Delta}^{(\ell,0^P)}$.
\item
For every plaquette $p$ of the graph, there is a plaquette projection $B_p$.  $B_p$ is supported on~$\cH_\Delta^{(\ell,0^P)}$ and acts on it by adding a vacuum loop, divided by $\sqrt{1+\tau^2}$, around the puncture in the middle of $p$.   Its action on $(\C^2)^{\otimes E}$ can be obtained by reducing the resulting ribbon graph back to the computational basis.  This definition of the plaquette projection~$B_p$ in terms of a vacuum loop is given in the appendix of~\cite{LevinWen} (see \secref{sec:levinwengeneralized} for more details).

$B_p$ depends on all the edges and boundary edges incident to the vertices around the plaquette~$p$.  However, it only possibly changes the qubits on the edges circling the plaquette, i.e., it is a controlled projection, controlled by the edges with exactly one endpoint on the plaquette boundary.  

Since the vacuum loop added by $B_p$ is separated from other vertices or plaquettes, the plaquette operators commute with each other and with the vertex operators.  
\end{itemize}

By Eq.~\eqnref{eq:vacuumquick}, the plaquette operator $B_p$ has the effect of removing the puncture from plaquette~$p$, allowing ribbons to be pulled across it.  Therefore, the simultaneous $+1$ eigenspace of all of the vertex and plaquette operators is isomorphic to~$\cH_\Sigma^\ell$ (see \lemref{lem:mainisomorphism} below).  Thus we have defined a quantum error-correcting code (QECC) based on ribbon graphs.  It encodes $\cH_\Sigma^\ell$ inside of~$(\C^2)^{\otimes \abs{E}}$.  We denote the code space by~$\cH^{\ell, gs}_{\widehat \cT} \subset \cH_{\widehat \cT}$.  Here the superscript~$gs$ stands for \emph{ground space}, because $\cH^{\ell, gs}_{\widehat \cT} \cong \cH^\ell_\Sigma$ is the ground space for the Levin-Wen Hamiltonian
\begin{equation}
H^\ell_{\widehat\cT} = - \sum_{\text{plaquettes $p$}} B_p - \sum_{\text{vertices $v$}} Q_v
 \enspace .\label{eq:levinwenhamiltonian}
\end{equation}

To describe the isomorphism between $\cH^\ell_\Sigma$ and the simultaneous $+1$~eigenspace $\cH_\Delta^{(\ell,0^P)}$ of all plaquette- and vertex-operators on $\cH_{\widehat\cT} = (\C^2)^{\otimes \abs E}$ in more detail, let $B$ be the projection 
\begin{equation} \label{eq:Bdef}
B = \prod_p B_p
 \enspace .
\end{equation}
Then the following lemma, proved in \appref{app:levinwenmodel}, describes how  ribbon graphs are mapped to the qubits on $\widehat \cT$.

\begin{lemma} \label{lem:mainisomorphism}
Given a state $\ket{\Psi_\Sigma} \in \cH^\ell_\Sigma$, deform the ribbon graph (arbitrarily) to avoid all the points in~$\cP$.  The state can then be regarded as an element $\ket{\Phi_\Delta} \in \cH_\Delta^{(\ell,0^P)}$.  Then $B\ket{\Phi_{\Delta}}$ does not depend on the initial deformation (i.e., $\ket{\Phi_\Delta}$).  The map
\begin{align}
\Lambda : \cH^\ell_\Sigma &\rightarrow \cH^{\ell,gs}_{\widehat\cT}\\
 \ket{\Psi_\Sigma} &\mapsto B \ket{\Phi_{\Delta}}
 \enspace ,
\end{align} is an isomorphism (preserving the inner product).
\end{lemma}

\lemref{lem:mainisomorphism} allows us essentially to forget about the graph $\widehat \cT$, and work completely in the off-lattice picture $\cH_\Sigma^\ell$.  For any state in $\cH_\Sigma^\ell$, the corresponding codeword, or ground state, $\ket{\Psi_{\widehat \cT}} \in \cH_{\widehat \cT}^{\ell,gs}$ is obtained by mapping the off-lattice ribbon graph to the graph~$\widehat \cT$, getting a state $\ket{\Phi_\Delta}$), and then applying the projection~$B$.  

It will often be convenient to use the state $\ket{\Phi_\Delta}$ to represent the ground state $\ket{\Psi_{\widehat \cT}} = B \ket{\Phi_\Delta}$.  Of course the projection $B$ is not one-to-one.  However, one can argue that $B \ket{\Phi_\Delta} = B \ket{\Phi'_\Delta}$ if and only if $\ket{\Phi_\Delta}$ and $\ket{\Phi'_\Delta}$ are equivalent under certain discrete lattice versions of the local relations in Eqs.~\eqnref{eq:firstlevinwen}--\eqnref{eq:lastlevinwen} (see~\cite{LevinWen}).

\subsection{Anyonic fusion basis states and gluing} \label{sec:anyonfusionbasislattice}

Here we describe the discrete analog of the gluing property explained in \secref{sec:recursivegluing}.  Let $\cT$ be a triangulation of a surface~$\Sigma$ that is  obtained by gluing $\Sigma_A$ to $\Sigma_B$ along a closed curve~$\gamma$.  Assume that $\gamma$ is part of the triangulation~$\cT$.  Let $\cT_A$ and $\cT_B$ be the restricted triangulations of $\cT$ to $\Sigma_A$ and $\Sigma_B$, respectively, and $\widehat{\cT}_A$ and $\widehat{\cT}_B$ their dual graphs.  \figref{fig:levinwenmodelboundary} shows these dual graphs schematically; edges intersected by~$\gamma$ are shown as ``virtual edges'' attached to the boundaries of~$\Sigma_A$ and $\Sigma_B$, respectively.  Pick an edge~$e$ from among the $t$ virtual edges.  

The images $\ket{\ell, d}_{\widehat \cT} = \Lambda \ket{\ell, d} \in \cH^{\ell, gs}_{\widehat \cT}$ of the anyonic fusion basis vectors~$\ket{\ell, d} \in \cH^\ell_\Sigma$ satisfy essentially the same properties as their off-lattice counterparts.  
Observe that plaquette operators within $\Sigma_A$ and $\Sigma_B$ do not affect qubits on the virtual edges, and can be applied before the plaquette operators $B(\gamma) = \prod_{p: p \cap \gamma \neq \emptyset}B_p$ of plaquettes intersecting~$\gamma$, by commutativity.  If we apply the isomorphism $\Lambda$ to a state decomposed as in~\eqnref{eq:gluingrule}, the result can therefore be written as
\begin{equation} \label{eq:discreteglueing}
\ket{\ell, d}_{\widehat \cT}
=
B(\gamma) \sum_{k \in \{0,1\}} \alpha_{k,d} \ket{\ell^k_A,d_A}_{\widehat{\cT}_A} \ket{\ell^k_B,d_B}_{\widehat{\cT}_B} \ket{k}_e \ket{0}^{\otimes (t-1)} 
 \enspace ,
\end{equation}
illustrated in \figref{fig:gluingdiscrete}.  Here, the states $\ket{\ell^k_A,d_A}_{\widehat{\cT}_A}$, $\ket{\ell^k_B, d_B}_{\widehat{\cT}_B}$ are the  anyonic fusion basis states associated with Hamiltonians $H^{\ell^k_A}_{\widehat{\cT}_A}$, $H^{\ell^k_B}_{\widehat{\cT}_B}$, respectively.  Thus $\ket{\ell, d}_{\widehat \cT}$ results from applying projections along~$\gamma$ to a certain superposition of product states on $\cH_{\widehat{\cT}_A} \otimes \cH_{\widehat{\cT}_A} \otimes (\mathbb{C}^2)^{\otimes t}$.  

Consider now the case $t = 1$. 
Then $B(\gamma) = B_q$ for a single plaquette $q$.  The following lemma states that the qubit on edge~$e$ directly reveals some information about the anyon labels $b_+b_-$ assigned to the connecting edge between $\Sigma_A$ and $\Sigma_B$ by the doubled fusion diagram~$d$.  

\begin{lemma} \label{lem:tadpolegluing}
Let $\widehat\cT$, $\widehat{\cT}_A$, $\widehat{\cT}_B$ be graphs as described above with $t=1$, and let $e$ be the connecting edge.  
Then the anyonic fusion basis state $\ket{\ell,d}_{\widehat \cT}$ is given by 
\begin{equation}
\ket{\ell, d}_{\widehat \cT}
= 
\begin{cases}
\ket{\ell^0_A,d_A}_{\widehat{\cT}_A} \ket{\ell^0_B,d_B}_{\widehat{\cT}_B} \ket{0}_e
&
\text{if $b_+ b_- = \oneone$} \\
\ket{\ell^1_A,d_A}_{\widehat{\cT}_A} \ket{\ell^1_B,d_B}_{\widehat{\cT}_B} \ket{1}_e
&
\text{if $b_+ b_- \in \{\onetau, \tauone\}$} \\
\frac 1 \tau \ket{\ell^0_A,d_A}_{\widehat{\cT}_A} \ket{\ell^0_B,d_B}_{\widehat{\cT}_B} \ket{0}_e + \frac{1}{\sqrt \tau} \ket{\ell^1_A,d_A}_{\widehat{\cT}_A} \ket{\ell^1_B,d_B}_{\widehat{\cT}_B} \ket{1}_e
&
\text{if $b_+ b_- = \tautau$} 
\end{cases}
\end{equation}
\end{lemma}
A proof of this statement is given in \appref{app:levinwenmodel}.  

\subsection{Relating ground spaces for different triangulations} \label{sec:differenttriangulations}

In preparation for presenting a scheme for computing on codewords, in \secref{sec:computation}, let us conclude this section by deriving how the code space changes under local changes to the triangulation~$\cT$.  

Let $e$ be an edge in $\cT$ that does not go along the boundary of $\Sigma$.  Let $\cT'$ be a triangulation obtained from $\cT$ by reconnecting the edge~$e$ as in \figref{fig:fmove}.  In the dual graph $\widehat\cT$, this corresponds to an operation on up to five edges to construct $\widehat\cT'$.  

\begin{figure*}
\begin{center}
\begin{equation*}
\raisebox{-5em}{	
\includegraphics[totalheight=0.20\textheight,viewport=425 213 575 408,clip]{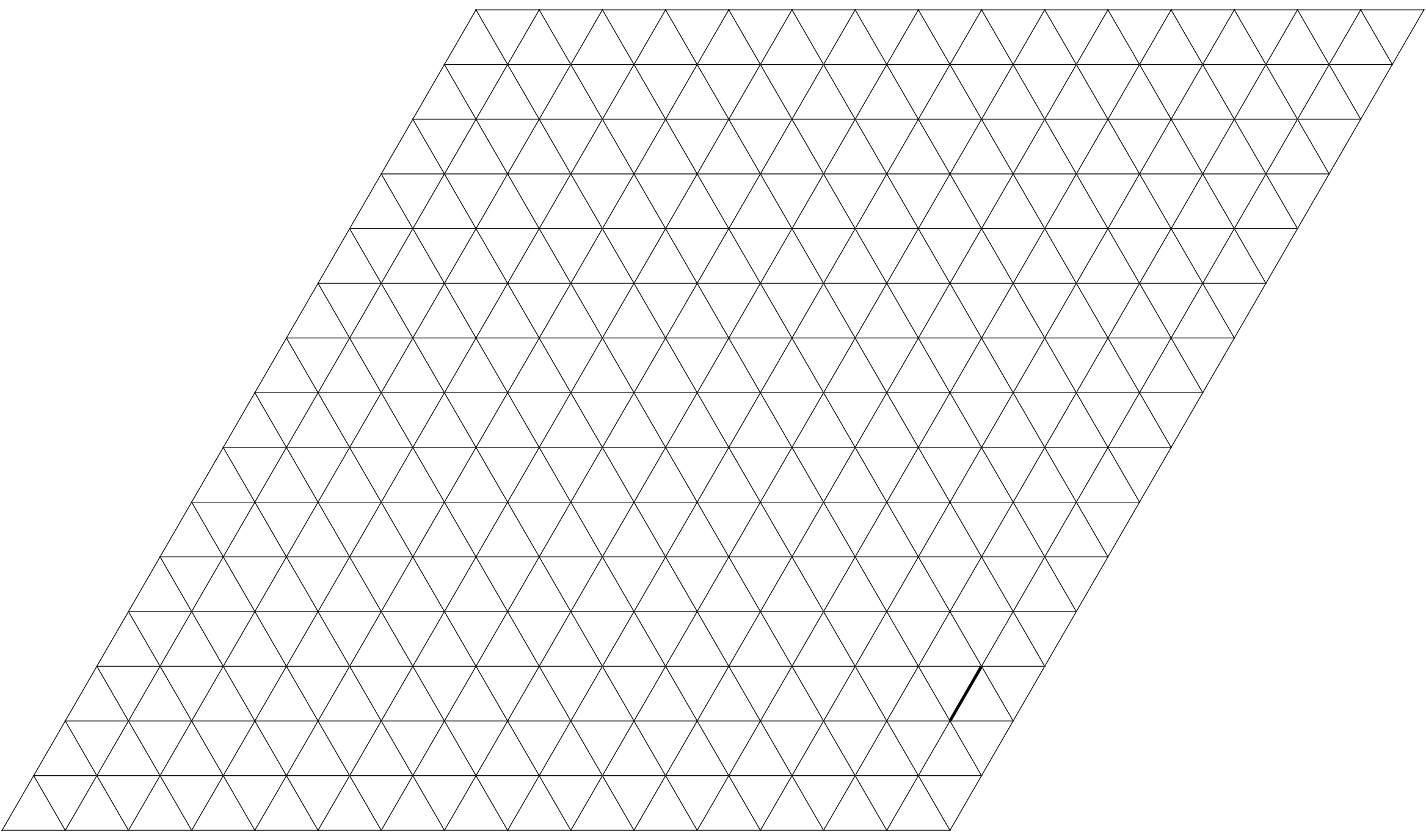}
\put(-51,63){$e$}
}
\mapsto 
\raisebox{-5em}{\includegraphics[totalheight=0.20\textheight,viewport=425 213 575 408,clip]{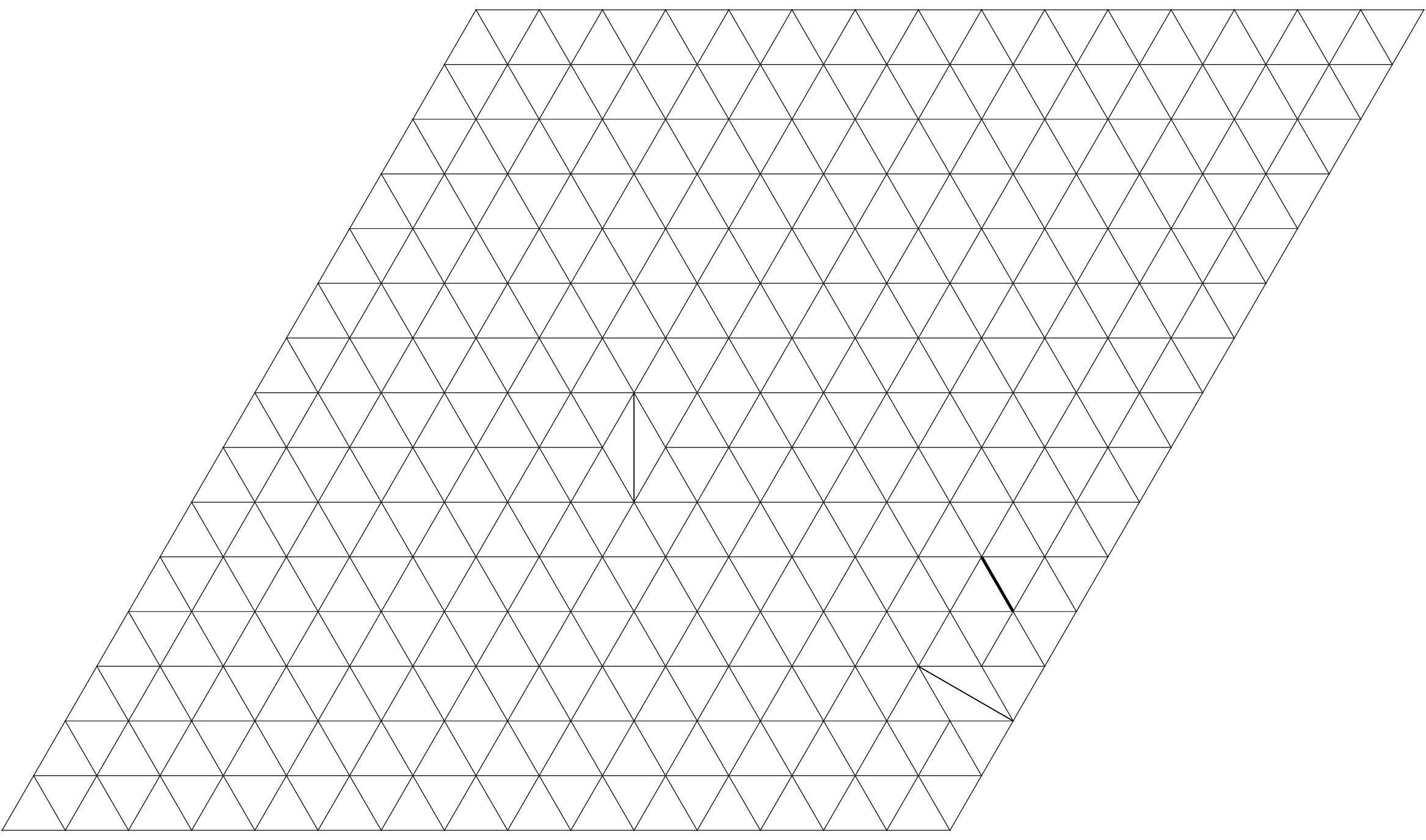}}
\qquad
\raisebox{-5em}{
\includegraphics[totalheight=0.20\textheight,viewport=450 243 600 437,clip]{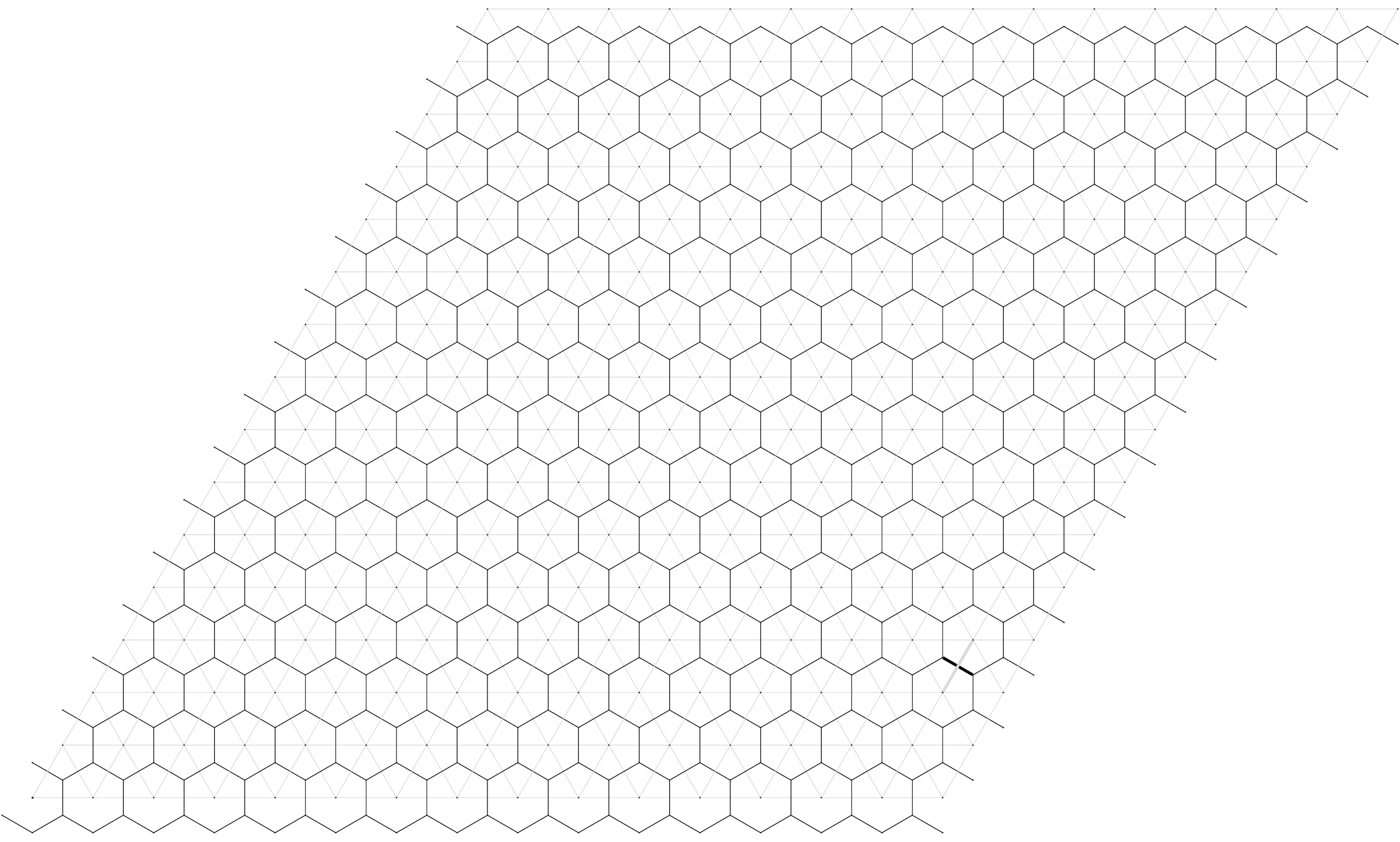}
\put(-58,61){$\widehat{e}$}
}\mapsto 
\raisebox{-5em}{
\includegraphics[totalheight=0.20\textheight,viewport=450 243 600 437,clip]{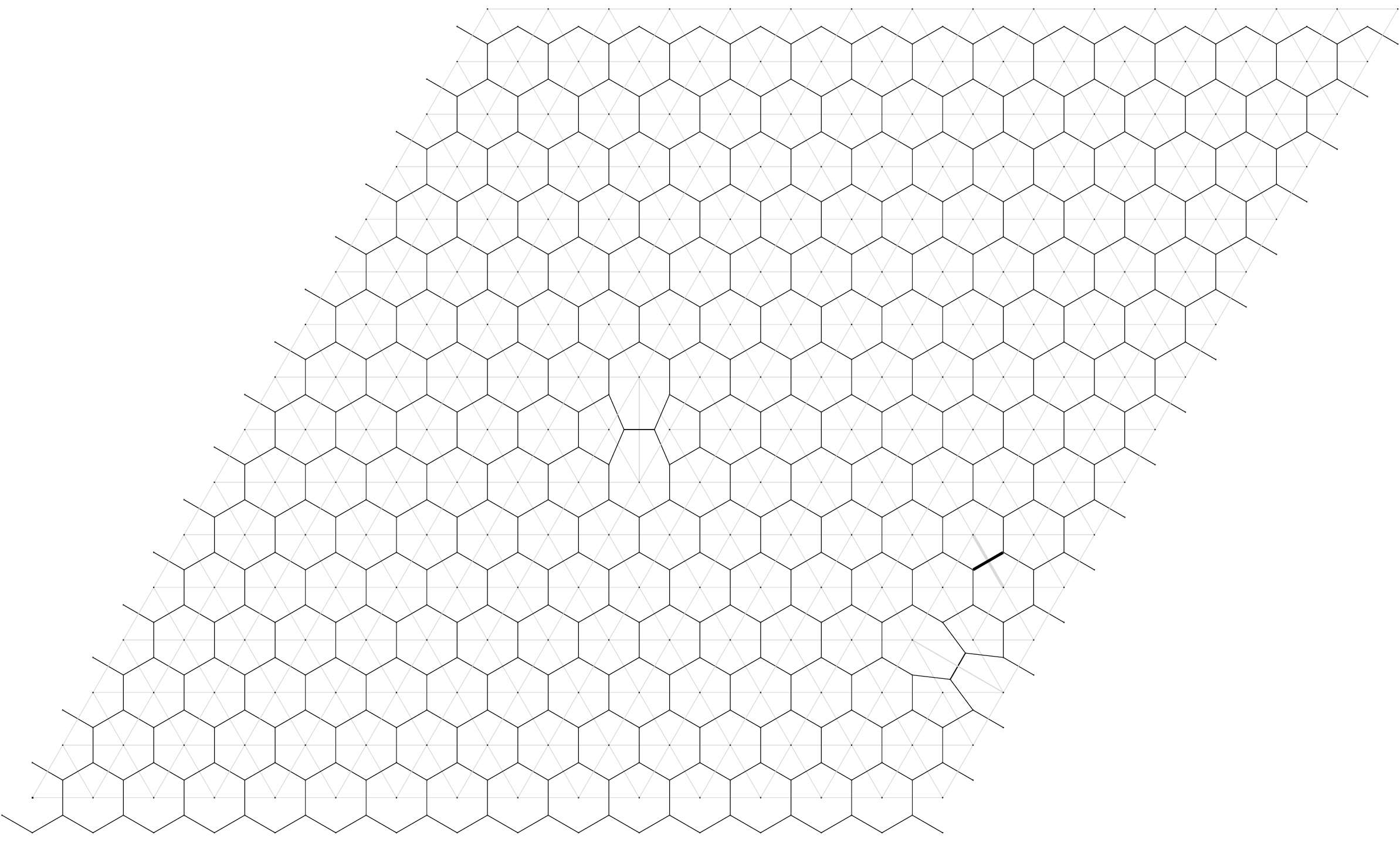}
}
\end{equation*}
\end{center}
\caption{Change of triangulation $\cT$ by flipping edge~$e$ (left), and the effect on the dual lattice~$\widehat\cT$.} \label{fig:fmove}
\end{figure*}

Since $e$ does not go along the boundary of $\Sigma$, we can use the same boundary conditions $\ell$ for $\widehat\cT'$ as for $\widehat\cT$.  Moreover, the dual graph $\widehat{\cT}'$ has the same number of edges~$\abs E$ and punctures~$P$, and can be used to define a punctured surface~$\Delta'$.  Thus as before we can define Hilbert spaces $\cH_{\widehat{\cT'}}$, $\cH^\ell_{\widehat{\cT'}}$, a Hamiltonian~$H^\ell_{\widehat \cT'}$, and the associated subspaces~$\cH^{(\ell,0^P)}_{\Delta'}$ and $\cH^{\ell,gs}_{\widehat\cT'}$.  

To study the relationship between the two ground spaces $\cH^{\ell,gs}_{\widehat\cT}$ and $\cH^{\ell,gs}_{\widehat\cT'}$, we define a linear operator $F_{\widehat e} : \cH^{(\ell,0^P)}_{\Delta} \rightarrow \cH^{(\ell,0^P)}_{\Delta'}$ by 
\newcommand*{\basise}[5]{\raisebox{-0.7em}{
\scalebox{.7}{
\begin{pspicture}(-.8,-.5)(.8,.5)
\psset{unit=.7cm}
\psset{linewidth=.8pt} 
\psset{labelsep=2.5pt} 
\SpecialCoor
\pnode(-.5,0){A}
\pnode(.5,0){B}
\psline[linestyle=#5]{-}(A)(B)
\psline[linestyle=#1]{-}([nodesep=1,angle=120]A)(A)
\psline[linestyle=#2]{-}([nodesep=1,angle=-120]A)(A)
\psline[linestyle=#4]{-}([nodesep=1,angle=60]B)(B)
\psline[linestyle=#3]{-}([nodesep=1,angle=-60]B)(B)
\end{pspicture}}}}

\newcommand*{\basisf}[5]{\raisebox{-0.7em}{\scalebox{.7}{
\begin{pspicture}(-.8,-.5)(.8,.5) 
\psset{unit=.7cm}
\psset{linewidth=.56pt} 
\psset{labelsep=2.5pt} 
\SpecialCoor
\pnode(0,.5){A}
\pnode(0,-.5){B}
\pnode(-.5,0){Aa}
\pnode(.5,0){Ba}
\pnode([nodesep=1,angle=120]Aa){a}
\pnode([nodesep=1,angle=-120]Aa){b}
\pnode([nodesep=1,angle=-60]Ba){c}
\pnode([nodesep=1,angle=60]Ba){d}
\psline[linestyle=#5]{-}(A)(B)
\psline[linestyle=#1]{-}(a)(A)
\psline[linestyle=#2]{-}(b)(B)
\psline[linestyle=#3]{-}(c)(B)
\psline[linestyle=#4]{-}(d)(A)
\end{pspicture}}}}

\begin{align} \label{eq:fmoveexplicit}
\ket{\basise{solid}{solid}{solid}{solid}{dotted}} &\mapsto \frac{1}{\tau}\ket{\basisf{solid}{solid}{solid}{solid}{dotted}} + \frac{1}{\sqrt{\tau}} \ket{\basisf{solid}{solid}{solid}{solid}{solid}} \nonumber \\[.25cm]
\ket{\basise{solid}{solid}{solid}{solid}{solid}} &\mapsto \frac{1}{\sqrt{\tau}}\ket{\basisf{solid}{solid}{solid}{solid}{dotted}}-\frac{1}{\tau}\ket{\basisf{solid}{solid}{solid}{solid}{solid}} \nonumber \\[.25cm]
\ket{\basise{solid}{solid}{dotted}{dotted}{dotted}} &\mapsto \ket{\basisf{solid}{solid}{dotted}{dotted}{solid}} &\ket{\basise{dotted}{dotted}{solid}{solid}{dotted}} &\mapsto \ket{\basisf{dotted}{dotted}{solid}{solid}{solid}} \\[.25cm]
\ket{\basise{solid}{dotted}{solid}{dotted}{solid}} &\mapsto \ket{\basisf{solid}{dotted}{solid}{dotted}{solid}} &\ket{\basise{dotted}{solid}{dotted}{solid}{solid}} &\mapsto \ket{\basisf{dotted}{solid}{dotted}{solid}{solid}} \nonumber \\[.25cm]
\ket{\basise{dotted}{solid}{solid}{dotted}{solid}} &\mapsto \ket{\basisf{dotted}{solid}{solid}{dotted}{dotted}} &\ket{\basise{solid}{dotted}{dotted}{solid}{solid}} &\mapsto \ket{\basisf{solid}{dotted}{dotted}{solid}{dotted}} \nonumber
\end{align}
in the computational bases.  Here, a solid line represents the state $\ket 1$, whereas a dotted line represents the state~$\ket 0$.  We can easily extend this to a unitary $F_{\widehat e} : \cH_{\widehat\cT} \rightarrow \cH_{\widehat\cT'}$.  (To make sense of unitarity, use the obvious isomorphism between the two spaces, each isomorphic to $(\C^2)^{\otimes \abs E}$.)  We will henceforth refer to this as an $F$-move; depending on which is more convenient, we will use either $\widehat e$ or the dual edge~$e$ to label this move.

It can be shown~\cite{KoeReiVid09} that $F_{\widehat e}$ maps the ground space $\cH^{\ell,gs}_{\widehat\cT}$
isomorphically to $\cH^{\ell,gs}_{\widehat\cT'}$.  Moreover, this isomorphism is compatible with the representation of ground states from \lemref{lem:mainisomorphism} in the following sense:

\begin{lemma}[\cite{KoeReiVid09}] \label{lem:gsrepresentationb}
The following diagram commutes, where $B'$ is the product of the plaquette terms for $\widehat\cT'$, as in Eq.~\eqnref{eq:Bdef}: 
\begin{equation}
\begin{diagram}[balance,width=2.75em,height=2.75em,tight]
\cH^{(\ell,0^P)}_{\Delta} & \rA^B & \cH^{\ell,gs}_{\widehat\cT} \\
\dA^{F_{\widehat e}} & & \dA^{F_{\widehat e}} \\
\cH^{(\ell,0^P)}_{\Delta'} & \rA^{B'} & \cH^{\ell,gs}_{\widehat\cT'}
\end{diagram}
\end{equation}
\end{lemma}

%% file: anyoniccomputation.tex
\section{Computation with doubled Fibonacci anyons in \texorpdfstring{$\cH_\Sigma$}{H\_Sigma}} \label{sec:computation}

In this section, we explain how to operate on the code: how to encode qubits, prepare and measure states, and apply unitary operators on codewords.

Universality results for Fibonacci anyons were obtained by Freedman et al.~\cite{Freedmanuniversality00,LarsenWang05}, using an encoding of a logical qubit into three-anyon states: 
\begin{align} \label{eq:nontrivialchargeencoding}
\ket{0} &\mapsto \fusethree{\ \tauanyon}{\ \tauanyon}{\ \tauanyon}{\ \oneanyon}{\ \tauanyon}
&
\ket{1} &\mapsto \fusethree{\ \tauanyon}{\ \tauanyon}{\ \tauanyon}{\ \tauanyon}{\ \tauanyon}
\end{align}
This encoding was also used by Bonesteel et al.~\cite{Bonesteeletal05}, who showed how to efficiently approximate any gate in terms of a sequence of braids,  using the Solovay-Kitaev construction~\cite{Kitaevbook}.  This encoding, while efficient, has the physical disadvantage that the total anyonic charge of the code-state is non-trivial.  Hormozi et al.~\cite{Hormozi07} used an alternative encoding, with four $\tauanyon$ anyons for each qubit: 
\begin{align} \label{eq:trivialchargeencoding}
\ket{0} &\mapsto \fusefour{\tauanyon}{\tauanyon}{\tauanyon}{\tauanyon}{\oneanyon}{\oneanyon}{\oneanyon}
&
\ket{1} &\mapsto \fusefour{\tauanyon}{\tauanyon}{\tauanyon}{\tauanyon}{\tauanyon}{\tauanyon}{\oneanyon}
\end{align}
As explained in~\cite{Freedmanuniversality00,Bonesteeletal05,Hormozi07} (cf.~\cite{Kitaev03,preskill,Nayaketal08}), an arbitrary $n$-qubit circuit can be efficiently approximated if we can
\begin{itemize}
\item
Create the encoded state $\ket{0}^{\otimes n}$.  In both encodings~\eqnref{eq:nontrivialchargeencoding} and~\eqnref{eq:trivialchargeencoding}, this amounts to preparing the state of $n$~pairs of $\tau$ anyons, each pair fusing to the trivial particle~$\oneanyon$.  
\item
Execute arbitrary braids.  
\item
Measure the anyon label $\{\oneanyon,\tauanyon\}$ of a fixed edge in a fixed fusion tree basis.  
\end{itemize}

We will use the ``doubled'' versions $\oneone$ and $\tauone$ instead of $\oneanyon$ and $\tauanyon$.  The other anyons in the doubled theory, $\onetau$ and $\tautau$, will not be used for computation.

\subsection{State preparation and measurement} \label{sec:statepreparation}

A method for preparing the (unique) ground state of the Levin-Wen Hamiltonian on a sphere has been given in~\cite{KoeReiVid09}.  The essential idea is to apply a version of discretized surgery, which easily generalizes.  

Assume that we have a codeword $\ket \Psi$ for the Fibonacci code on a surface $\Sigma$, and a codeword $\ket \Phi$ on the surface $\Upsilon$.  Then~\cite{KoeReiVid09} gives a way of cutting a hole inside a plaquette of $\Sigma$, and a hole inside a plaquette of $\Upsilon$, and gluing these holes together.  Here we discuss the details needed to create specific codewords on the $n$-punctured sphere.

As explained in \secref{sec:anyonfusionbasislattice}, the gluing of codewords associated with two surfaces~$\Sigma_A$ and $\Sigma_B$ takes a particularly simple form for graphs with tadpole-like structures.  To make use of this fact, we change the triangulation to obtain such structures: \figref{fig:tadpolgenerationsummary} shows how a region $\Sigma_A$ bounded by a closed curve $\gamma$ can be isolated from the rest by a procedure~$cut(\gamma)$ consisting of $F$-moves.  The inverse procedure~$glue(\gamma)$ integrates a tadpole-like structure into a regular lattice.  

\begin{figure*}  
\begin{center} 
\raisebox{-6.5em}{\begin{picture}(200,175)(0,0)
\includegraphics[totalheight=0.25\textheight,viewport=165 80 275 170,clip]{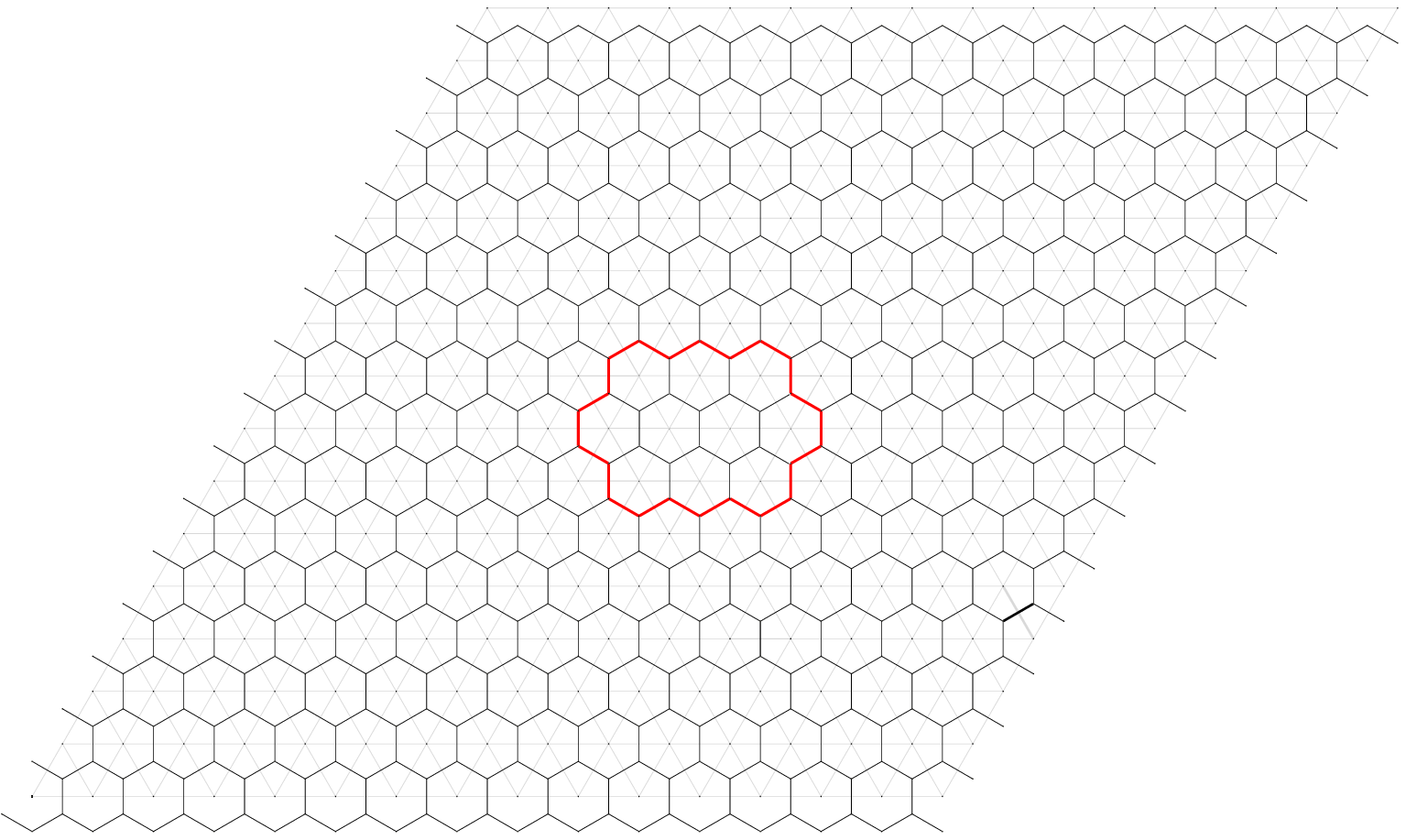}
\put(-120,85){$\Sigma_A$}
\put(-45,90){$e_{i}$} 
\put(-43,107){$e_{i+1}$} 
\put(-90,20){$\Sigma_B$}
\end{picture}}
$\begin{matrix}
cut(\gamma)\\
\longrightarrow\\
\\ \\
glue(\gamma)\\
\longleftarrow
\end{matrix}$ \enspace \enspace
\raisebox{-6.5em}{\begin{picture}(190,175)(0,0)
\includegraphics[totalheight=0.25\textheight,viewport=145 80 255 170,clip]{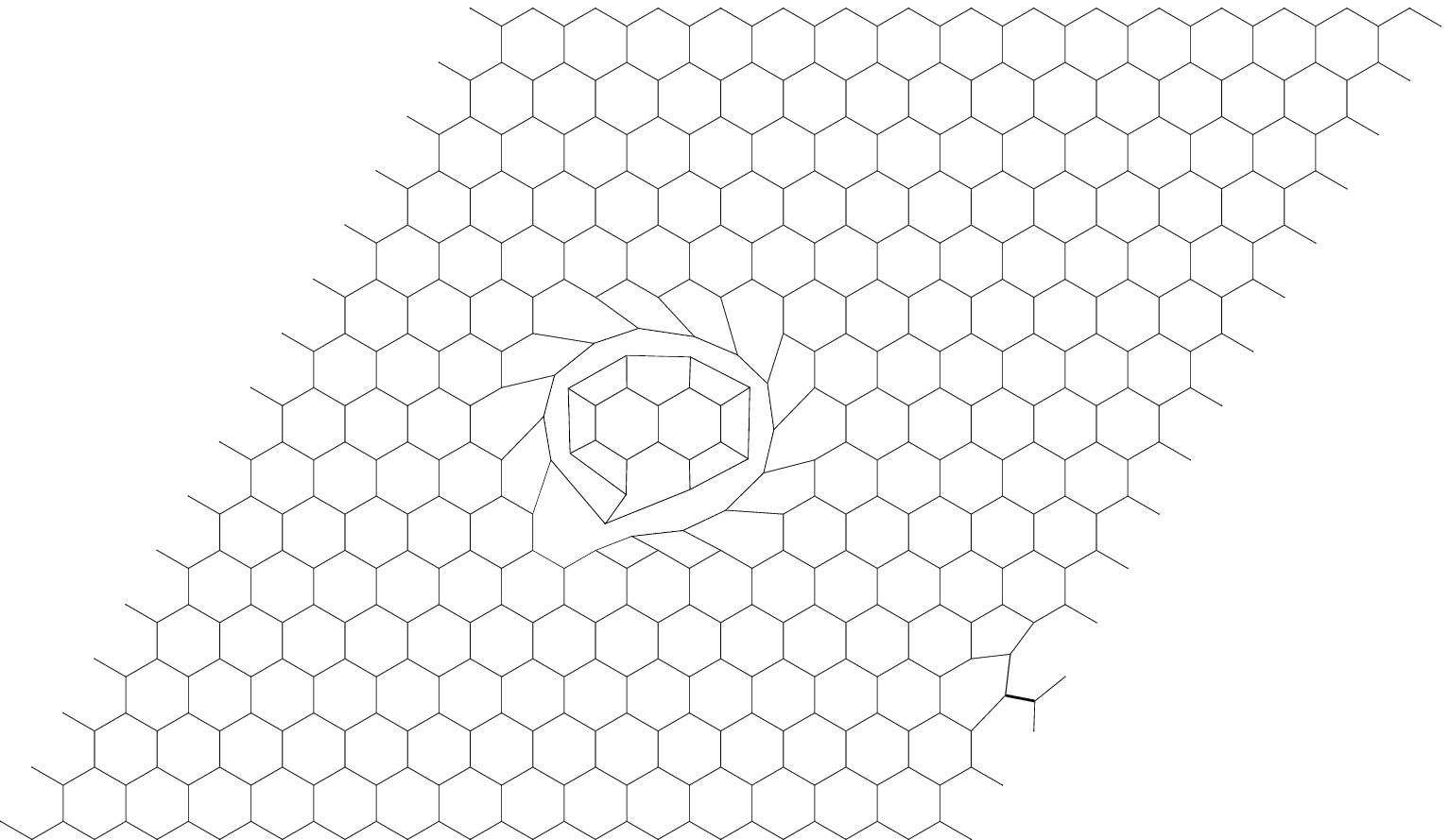}
\put(-120,85){$\Sigma_A$}
\put(-145,40){$e$}
\put(-55,20){$\Sigma_B$}
\end{picture}}
\end{center}
\caption{The procedure $cut(\gamma)$ isolates a region $\Sigma_A$ enclosed by a closed curve~$\gamma$, by making $F$-moves in sequence along the edges $e_1, e_2, \ldots, e_n$, colored red, counterclockwise along $\gamma$ on the dual graph~$\widehat \cT$.  The result is a tadpole-like structure as in \lemref{lem:tadpolegluing}, with a single edge~$e$ and plaquette~$q$ between the two regions.  The inverse operation is $glue(\gamma)$.} \label{fig:tadpolgenerationsummary}
\end{figure*}

\begin{description}
\item[Procedure $cut(\gamma)$:] Let $e_1, e_2, \ldots, e_n \in \widehat \cT$ be the sequence of edges in counterclockwise order constituting a closed path on $\widehat \cT$ along~$\gamma$.  
For $i$ from $1$ up to $n$, apply $F_{e_i}$.  
\item[Procedure $glue(\gamma)$:] Let $e_1', e_2', \ldots, e_n'$ be the edges in the deformed lattice corresponding to the sequence $e_1, e_2, \ldots, e_n$ in the original lattice.  For $i$ from $n$ down to $1$, apply $F_{e_i'}$.  
\end{description}

\noindent \figref{fig:tadpolegenerationexample} shows the intermediate steps in an example of applying $cut(\gamma)$.  Note that the $F$-moves constituting $cut(\gamma)$ do not commute with each other.  Note also that these $F$-moves lead to a degenerate triangulation. 
By adding more $F$-moves, this procedure can be modified to keep the degree of the bounding plaquette small.  

\newcommand*{\tadpolepicture}[2]{
$\begin{matrix}
\fbox{\includegraphics[totalheight=0.17\textheight,viewport=245 135 410 280,clip]{images/tadpolegeneration/out#1}}\\
#2
\end{matrix}$
}

\begin{figure*}
\centering
\begin{tabular}{c@{$\quad$}c@{$\quad$}c}
\subfigure{\tadpolepicture{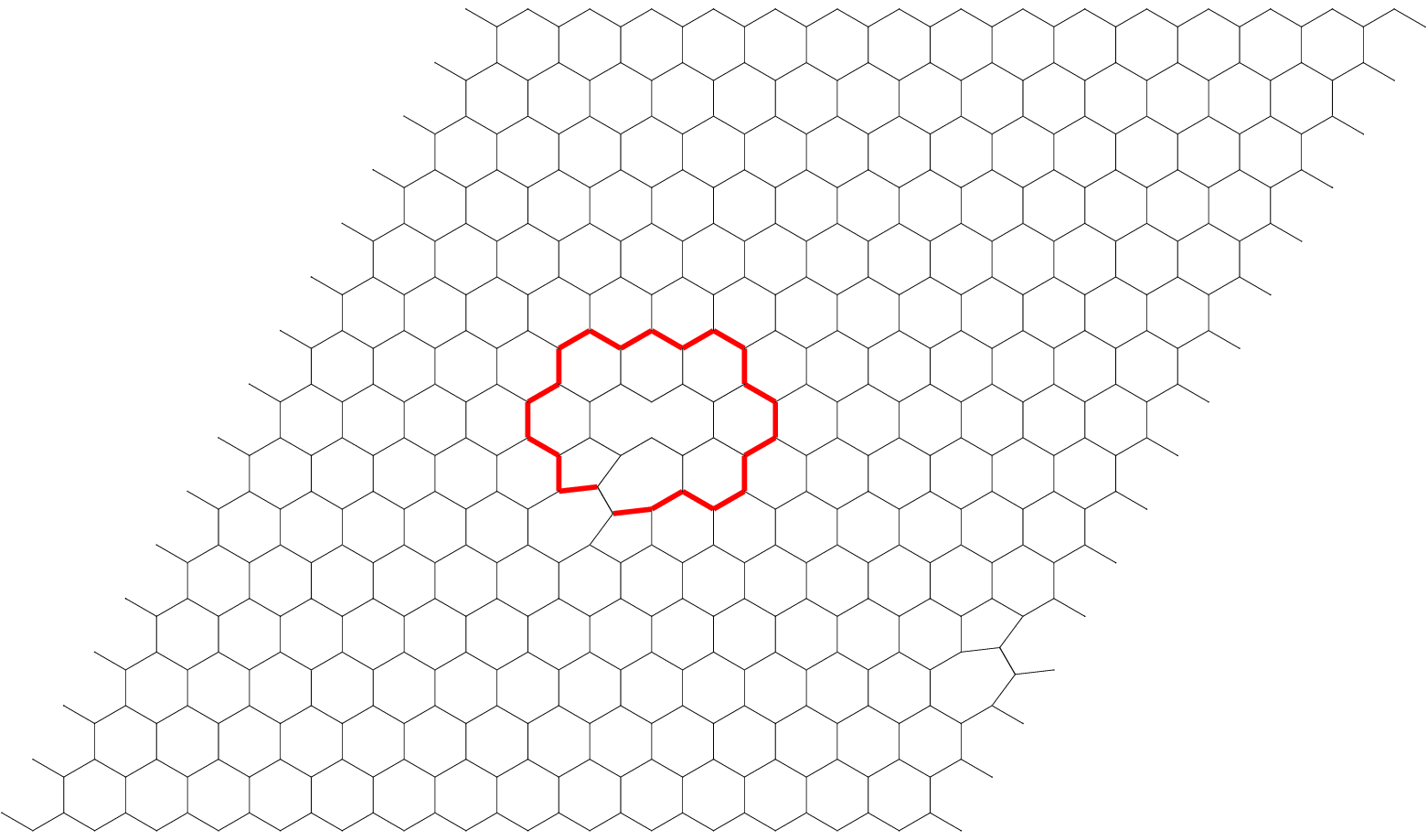}{i=1}}&
\subfigure{\tadpolepicture{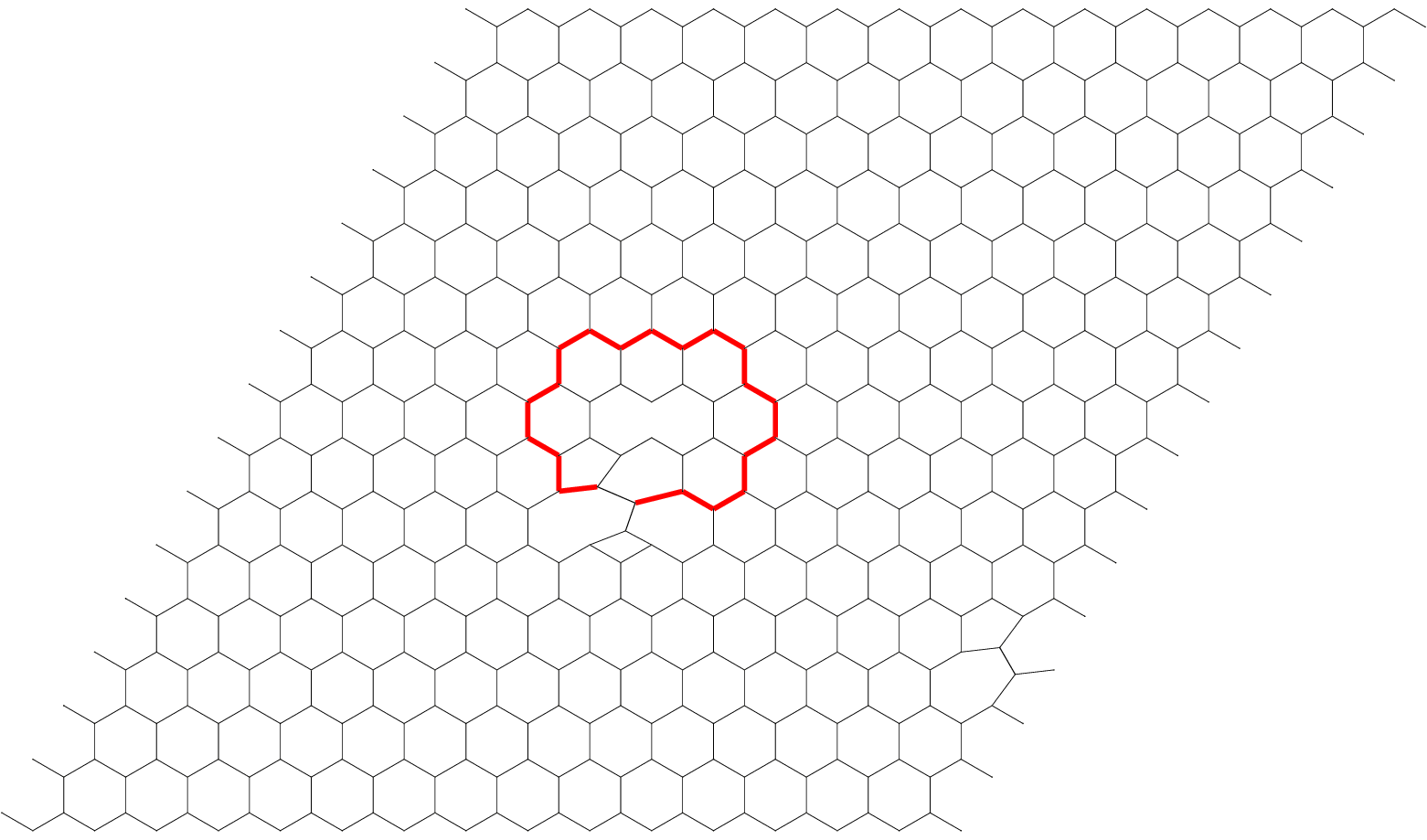}{i=2}}&
\subfigure{\tadpolepicture{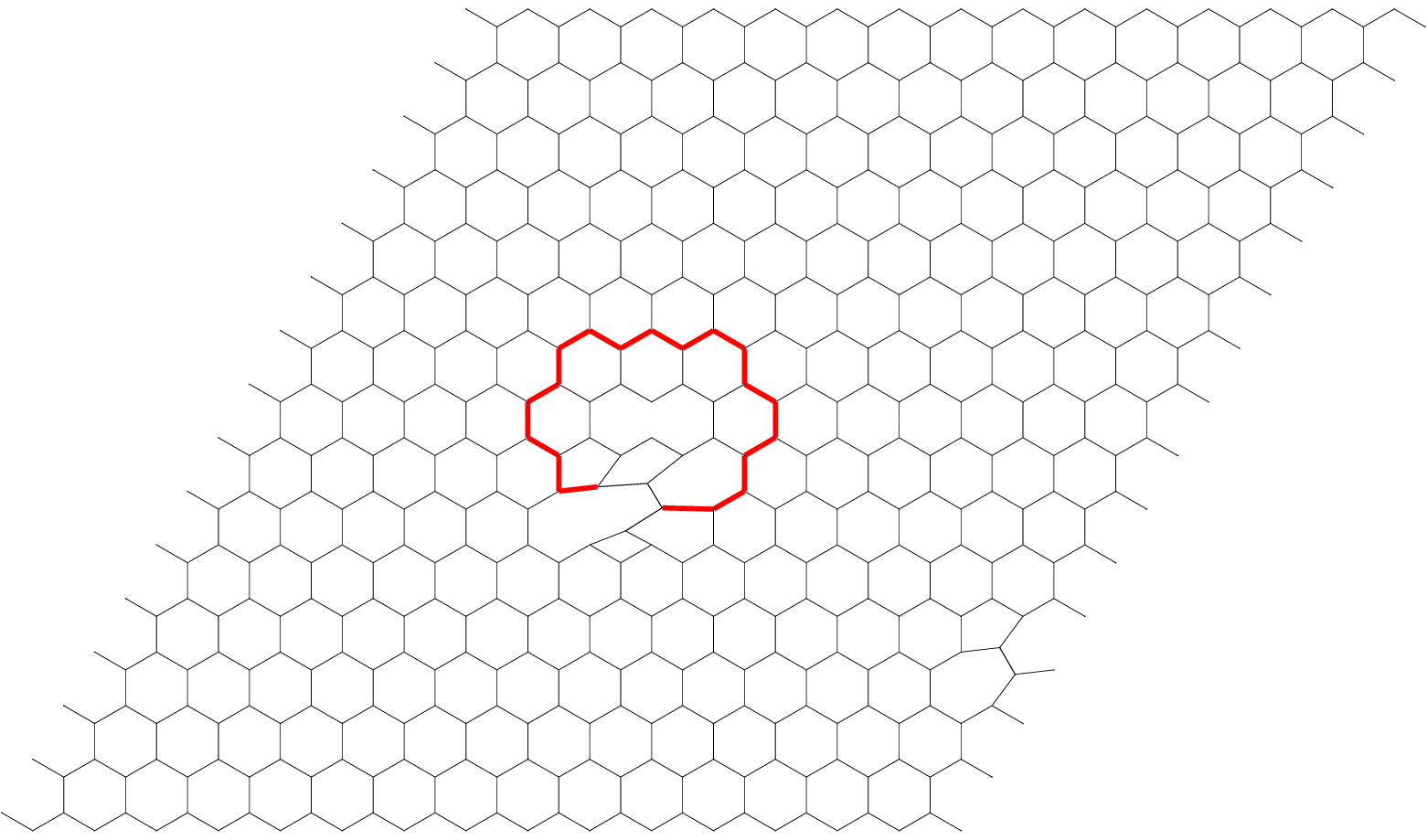}{i=3}}\\
\subfigure{\tadpolepicture{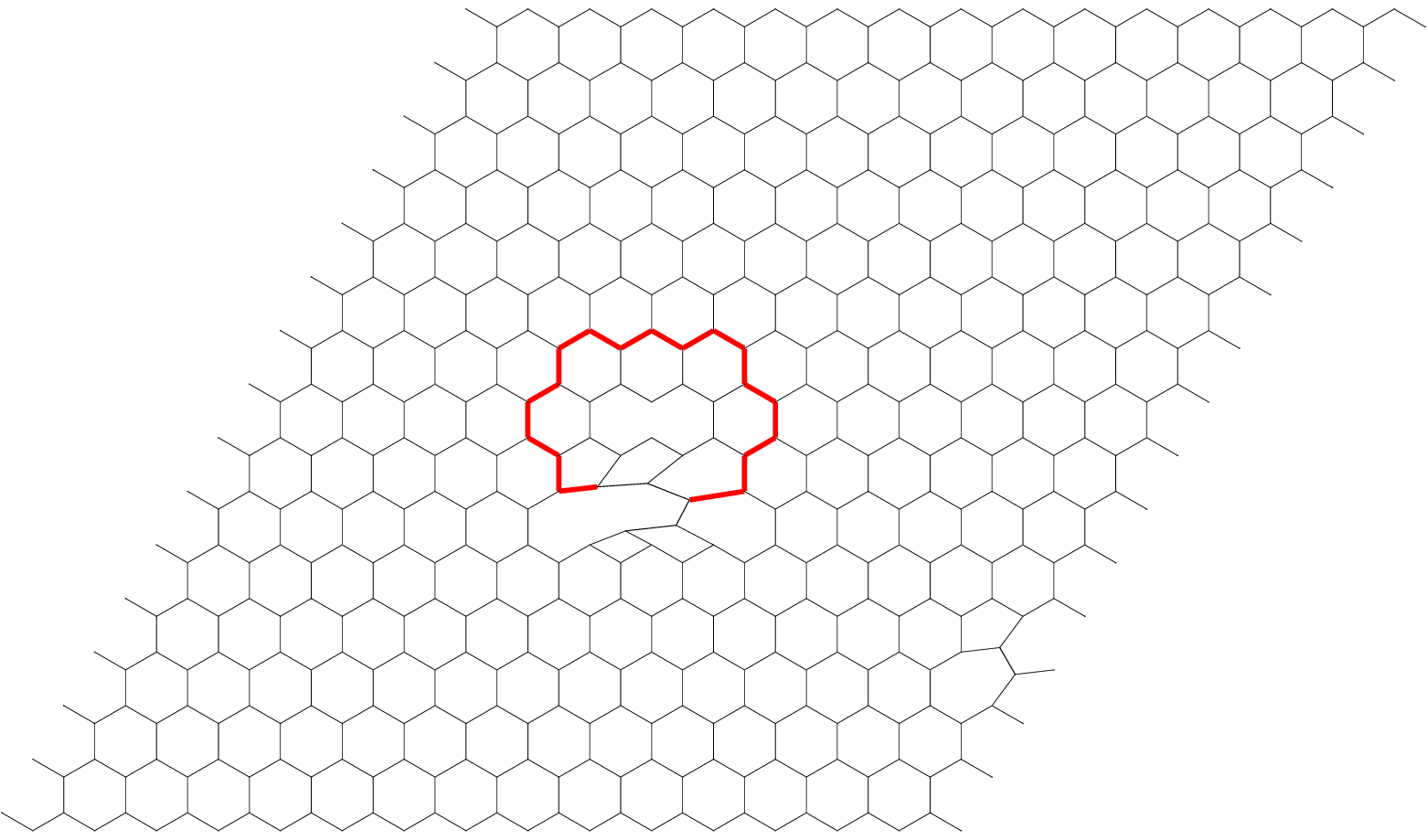}{i=4}}&
\subfigure{\tadpolepicture{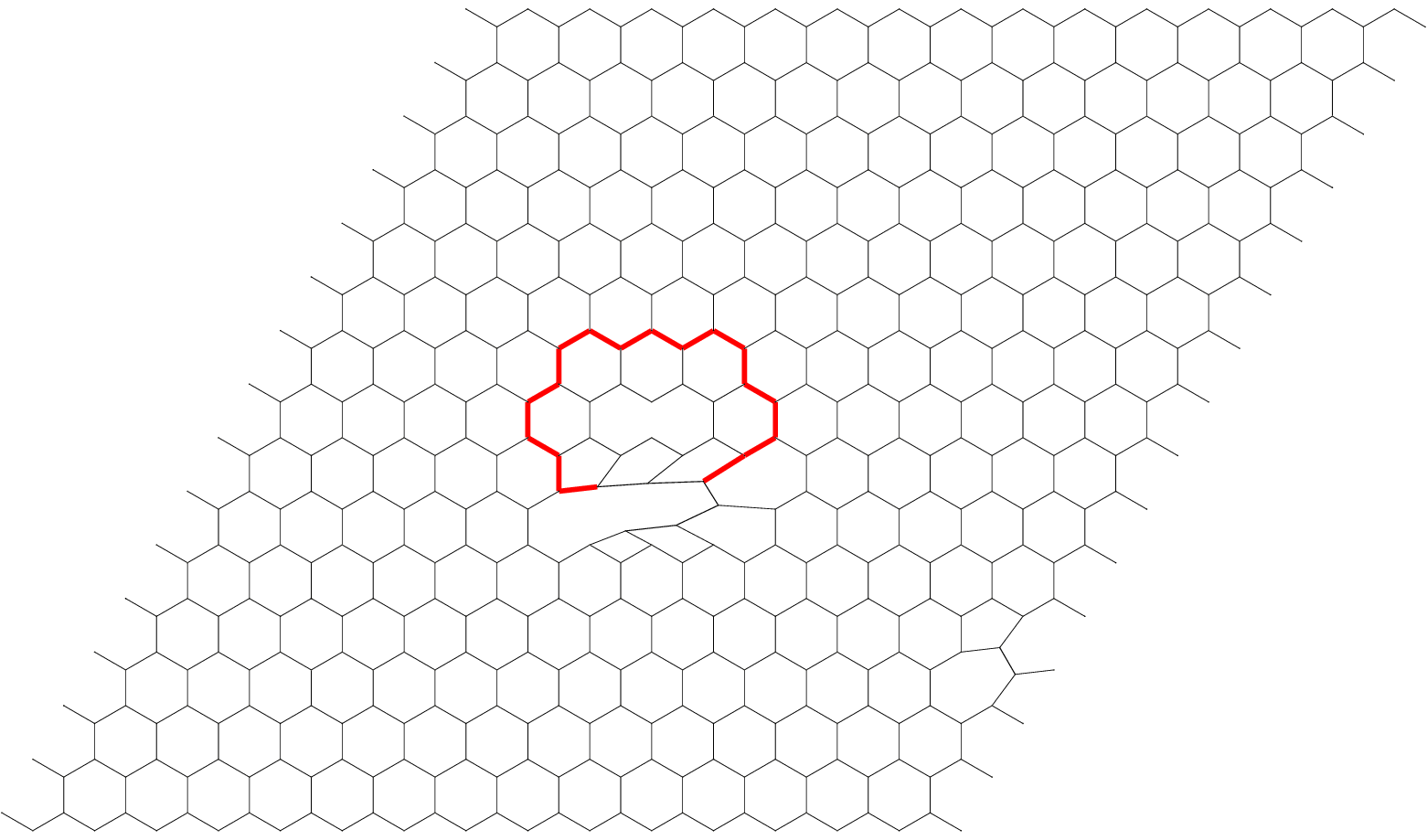}{i=5}}&
\subfigure{\tadpolepicture{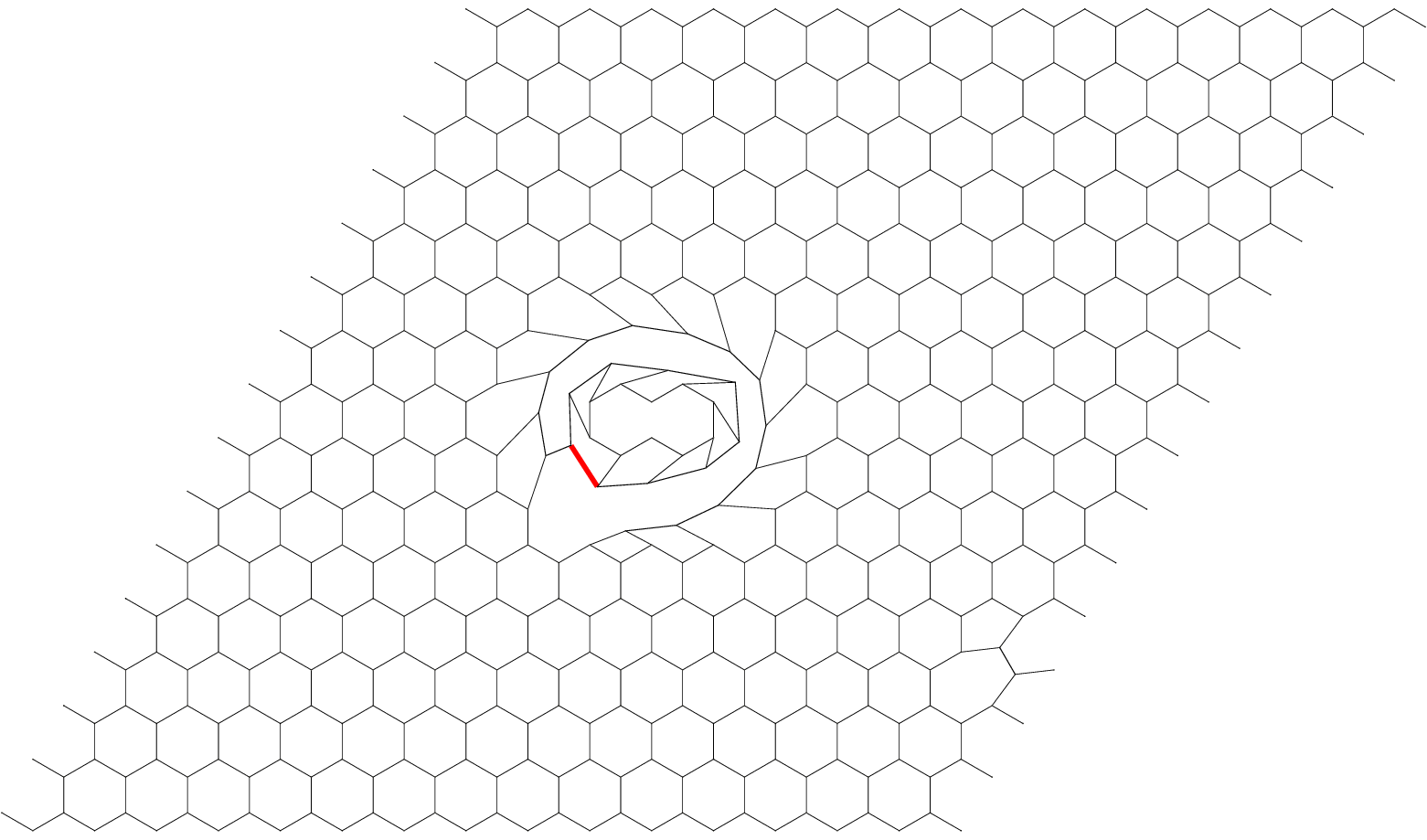}{i=21}}
\end{tabular}
\caption{The successive figures show the intermediate steps of the $cut(\gamma)$ procedure, after applying $F_{e_i}$.  The thick, red edges are $e_1, \ldots, e_n$.  The overall effect of the procedure is shown in \figref{fig:tadpolgenerationsummary}.  } \label{fig:tadpolegenerationexample}
\end{figure*}

Our goal is to create $2n$~anyons of type~$\tauone$, which fuse to~$\oneone$ in pairs.  
This can be achieved by $n$ times cutting out a circular disk and replacing it by a three-punctured sphere with state $\fuse{\tauone}{\tauone}{\oneone}$.  

To describe this in more detail, let $\cT$ be a triangulation of $\Sigma$, a surface obtained by gluing $\Sigma_A$ to $\Sigma_B$.  Let $d_A$ and $d_B$ be fusion diagrams with each root edge carrying the trivial label~$\oneone$, and consider the state $\ket{\ell_A \ell_B, d_A d_B}_{\widehat \cT}$.  

The necessary steps for cutting out $\Sigma_A$ and replacing it by a different surface~$\Sigma_{A'}$ are especially simple when $\widehat \cT$ has the tadpole-like form considered in \lemref{lem:tadpolegluing}; we merely have to trace out the qubits on $\widehat \cT_A$, and replace them by a state $\ket{\ell_{A'}, d_{A'}}_{\Sigma_{A'}}$.  Indeed, by \lemref{lem:tadpolegluing}, we have
\begin{equation}\begin{split}
\ket{\ell_A \ell_B, d_A d_B}_{\widehat \cT}
&= \ket{\ell^0_A, d_A}_{\widehat \cT_A} \ket{\ell^0_B, d_B}_{\widehat \cT_B} \ket{0}_e \\
&\overset{\tr_{\widehat{\cT_A}}}{\longmapsto} \ket{\ell^0_B, d_B}_{\widehat \cT_B} \ket{0}_e \\
&\overset{\textrm{append}}{\longmapsto} \ket{\ell^0_{A'}, d_{A'}}_{\widehat \cT_{A'}} \ket{\ell^0_B, d_B}_{\widehat \cT_B} \ket{0}_e \\
&= \ket{\ell_{A'}\ell_B,d_{A'}d_B}_{\widehat{\cT}'}
 \enspace .
\end{split}\end{equation}
To give $\widehat \cT$ the tadpole form, choose any closed curve $\gamma$ not enclosing any holes in the surface---this ensures that the corresponding fusion diagrams~$d_A$ and~$d_B$ have trivial total charge.  Then use the $cut(\gamma)$ and $glue(\gamma)$ procedures as in the following circuit: 
\smallskip
\begin{equation}\raisebox{1cm}{
\mbox{\Qcircuit @C=1em @R=1.2em{
&                                    &\multigate{2}{cut(\gamma)} &\ustick{\widehat \cT_B}\qw&\qw&\qw             &\qw &\qw &\qw  &\multigate{2}{glue(\gamma)}& \\
\lstick{\ket{\ell_{A}\ell_B,d_{A}d_B}}&\qw &\ghost{cut(\gamma)}            &\ustick{e}\qw&\qw&\qw             &\qw &\qw &\qw  &\ghost{glue(\gamma)}&\qw&\rstick{\ket{\ell_{A'}\ell_B,d_{A'}d_B}}\\
&                                   & \ghost{cut(\gamma)}      &\ustick{\widehat \cT_{A}}\qw&\qw&\measure{\mbox{$\tr$}}& & \qw &\qw &\ghost{glue(\gamma)}&\\
\lstick{\ket{\ell_{A'},d_{A'}}_{\widehat \cT_{A'}}}&\qw &    \qw                    & \qw                       &\qw&\qw              & \qw\qwx[-1] & & &
}}\vspace{0.3cm}
}\end{equation}

In our application, the triangulation $\cT_{A'}$ can be chosen to be a coarse triangulation of $\Sigma_3$, such that the state $\ket{\ell_{A'}, d_{A'}}_{\widehat \cT_{A'}}$ corresponding to a $\tauone$ anyon pairs can be created by a projective measurement on a small number of qubits.

According to \lemref{lem:tadpolegluing}, the qubit on the tail of a tadpole is $0$ or $1$ according to whether the anyon label of the corresponding pant segment is $\oneone$ or $\tauone$, respectively.  This leads to the following circuit for measuring the anyon across a particular pant segment:
\smallskip
\begin{equation}\raisebox{1cm}{
\mbox{\Qcircuit @C=1em @R=1.2em{
                                      &                           &\multigate{2}{cut(\gamma)} &\ustick{\widehat \cT_B}\qw&\qw  \\
\lstick{\ket{\ell,d}}&\qw                        &\ghost{cut(\gamma)}        &\ustick{e}\qw           & \meter\\ 
                                      &                           & \ghost{cut(\gamma)}       &\ustick{\widehat \cT_{A}}\qw&\qw}}\vspace{0.3cm}
}\end{equation}

Implementing these procedures in the presence of noise, with periodic error correction, both measurements and preparations expose the encoded information.  Beneath a threshold, a constant measurement error rate can be reduced by taking the majority of multiple measurement outcomes.  Errors in preparation can be efficiently dealt with using the composite Fibonacci anyon distillation scheme of~\cite{Koenig09distillation}, in order to concentrate the entropy into specific regions using reversible gates.

\subsection{Implementing diffeomorphisms on \texorpdfstring{$\Sigma$}{Sigma} by changes of triangulation} \label{sec:implementingdiffeomorphisms}

Recall that the mapping class group of $\Sigma$ acts on $\cH_\Sigma$ by deforming embedded ribbon graphs.  If $\cT$ is a triangulation of $\Sigma$, such that the dual graph $\widehat \cT$ has $N$ edges, then $\cH_{\widehat \cT} \cong (\C^{2})^{\otimes N}$.  In this section, we show how to implement the mapping class group's action on $\cH_{\widehat \cT}^{\ell, gs} \cong \cH_\Sigma^\ell$---both Dehn-twists and braid-moves---by local unitary operators on~$(\C^{2})^{\otimes N}$.  As before, the implementation is based on changes of triangulation using the operators~$F_e$ from Eq.~\eqnref{eq:fmoveexplicit}.  

Consider first a Dehn-twist, specified by a simple closed curve~$\gamma$, as in Eq.~\eqnref{eq:dehntwistright}.  
Assume that $\gamma$ is supported on edges of the triangulation~$\cT$, and let $\Sigma_L$ and $\Sigma_R$ be the triangulated surfaces on either side of~$\gamma$.  The Dehn-twist operator $D(\gamma)$ is obtained by applying $F$-moves along edges of one of these surfaces, more precisely by the following procedure.

\begin{description}
\item[Dehn-twist $D(\gamma)$:] 
Let $\gamma$ consist of~$n$ edges. 
Repeat the following for rounds $r = 1, \ldots, 2(n-1)$:

In round~$r$, consider the triangulation $\cT^{r-1}$ constructed after $r-1$ rounds.  $\cT^0 = \cT$ is the initial triangulation.  Let $T_0, \ldots, T_{n-1}$ be the triangles bordering $\gamma$ in $\Sigma_R$, in counterclockwise order.  Assume that each triangle has exactly two vertices on $\gamma$; if true in the first round, then this will hold always.  Let $v_i$ be the vertex shared between $T_i$ and $T_{i + 1 \pmod n}$, let $w_i$ be the vertex of $T_i$ not in $\gamma$, and let $e_i$ be the edge $(v_i, w_i)$.  
For $i$ from $1$ to $n$, apply the maps $F_{e_i}$ to obtain the triangulation $\cT^r$.  
\end{description}

\newcommand*{\braidingpicturedual}[1]{
\raisebox{.1in}{\fbox{\includegraphics[totalheight=0.165\textheight,viewport=245 130 600 405,clip]{images/braidingfigures/outdual#1}}}
}

\begin{figure*}
\centering
\begin{tabular}{c@{$\quad$}c@{$\quad$}c}
\subfigure[$\cT^0$]{\braidingpicturedual{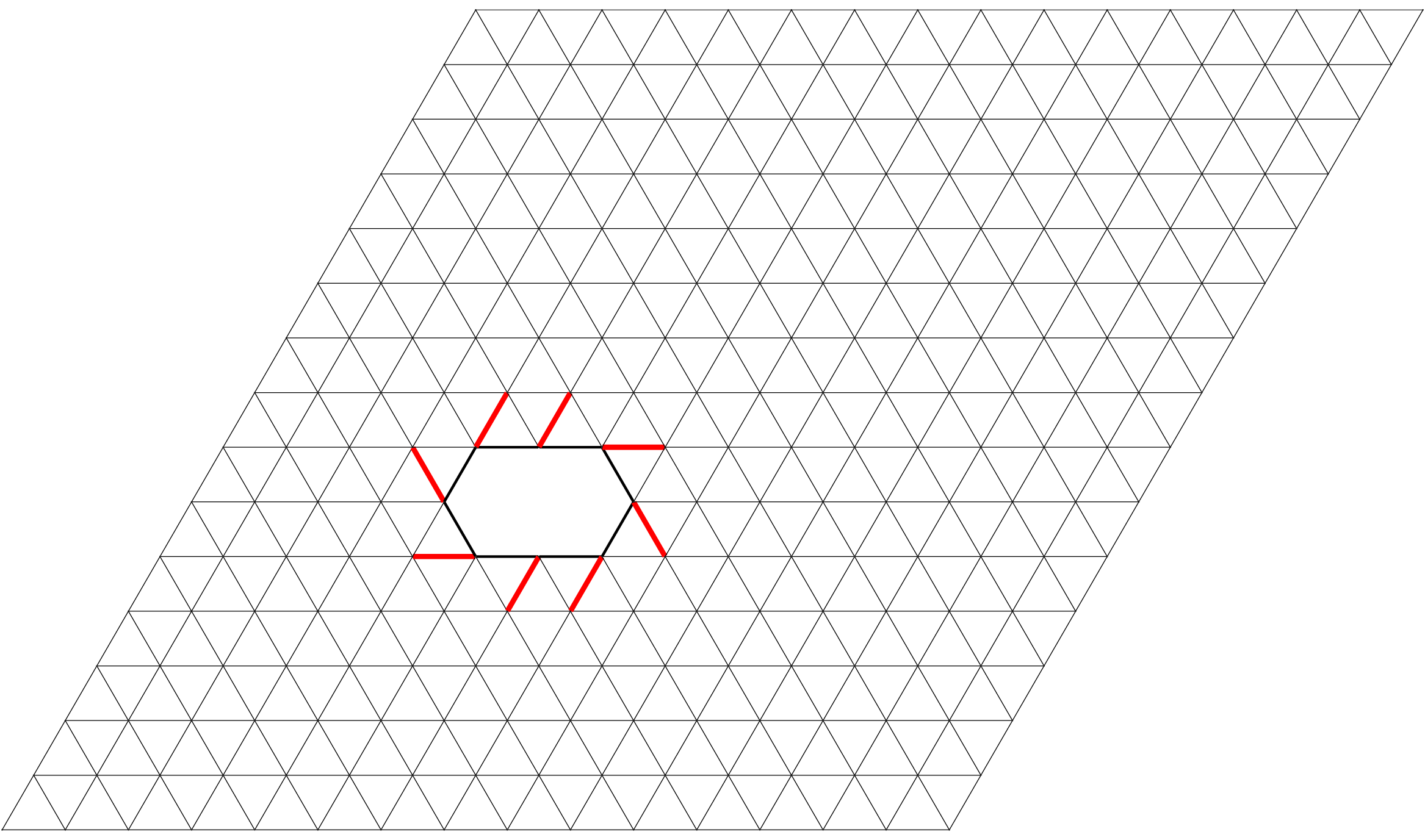}}&
\subfigure[$\cT^1$]{\braidingpicturedual{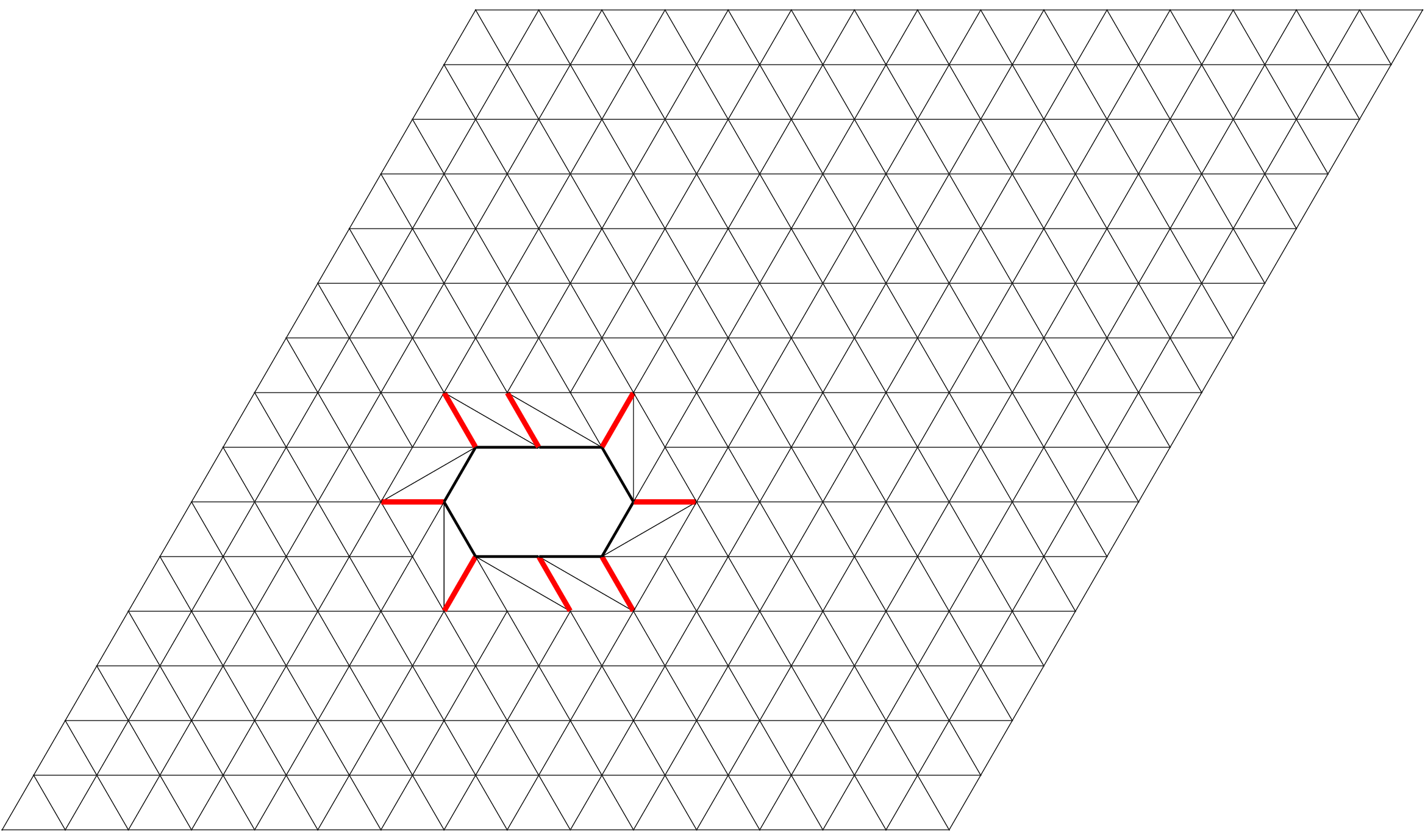}}&
\subfigure[$\cT^2$]{\braidingpicturedual{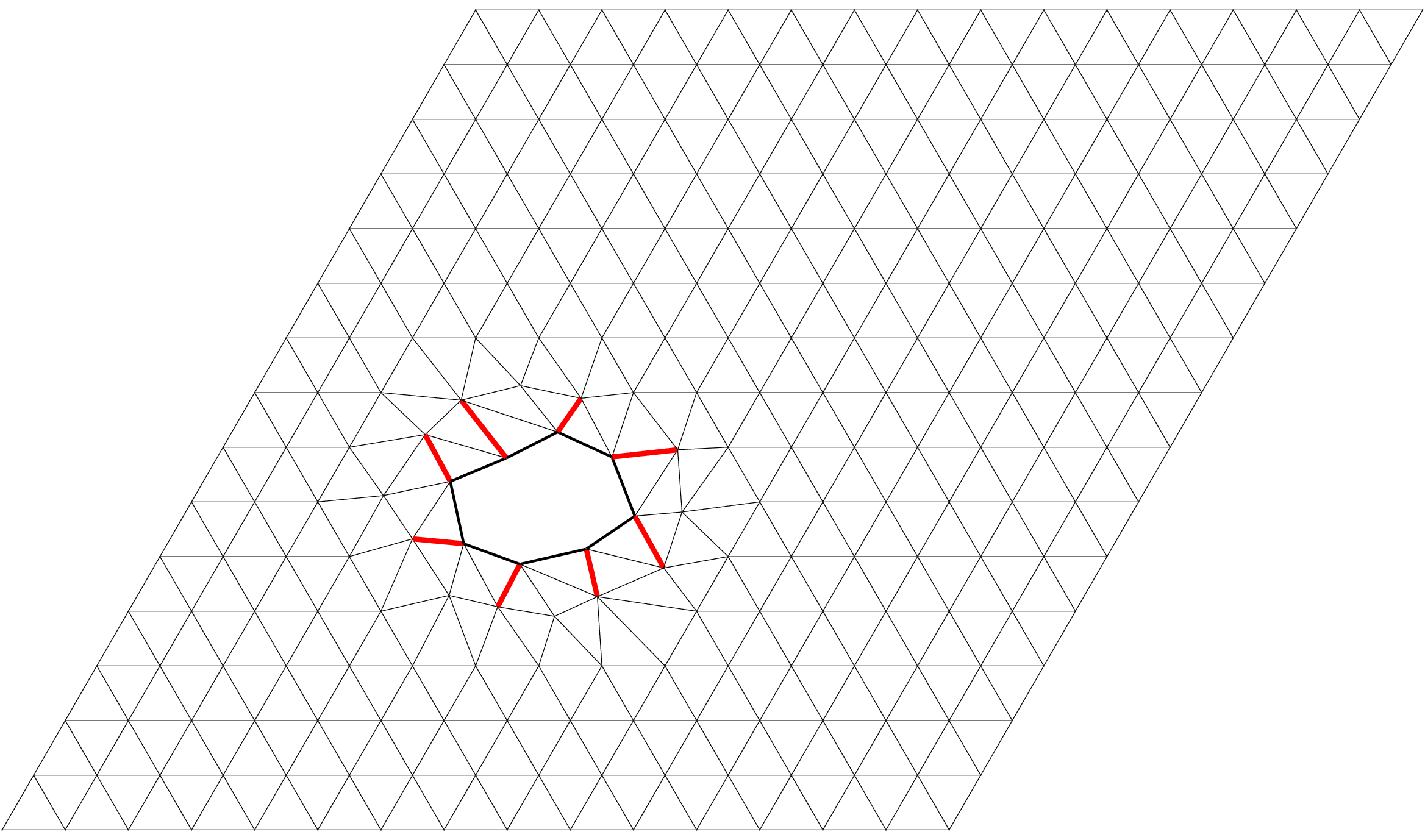}}
\end{tabular}
\caption{
From left to right, the triangulations $\cT^0$, $\cT^1$ and $\cT^2$ obtained while implementing a Dehn-twist about the inner curve~$\gamma$.  The portion $\Sigma_L$ of the surface within $\gamma$ is not shown.  The thick, red edges are the locations of $F$-moves in the next round, i.e., the edges $\{ e_i \}$.  
} \label{fig:duallatticedeformationfigure}
\end{figure*}

Note that the $F$-moves within each round commute, so can be applied in an arbitrary order.  \figref{fig:duallatticedeformationfigure} shows an example of this procedure.  
A clockwise Dehn-twist can be defined similarly.  

With the isomorphism $\Lambda$ described in \lemref{lem:mainisomorphism}, and the explicit form~\eqnref{eq:fmoveexplicit} of the maps $F_e$ and \lemref{lem:gsrepresentationb}, it is straightforward to see that the operator $D(\gamma)$ implements a ribbon-graph Dehn-twist on $\cH_\Sigma \cong \cH^{gs}_{\widehat \cT}$, as studied in \secref{sec:mappingclassgroupaction}.  Figures~\ref{fig:stringnetdeformationaction} and~\ref{fig:stringnetdeformationactionlast} show an example, tracking through $F$-moves states $\ket \Phi \in \cH^{phys}_{\widehat \cT}$ that represent ground states.

A half-twist twists along a curve~$\gamma$ by $\pi$ instead of $2\pi$.  Define an operator $\sqrt{D(\gamma)}$ on $(\C^2)^{\otimes N}$ to execute only rounds $r = 1, \ldots, n-1$ of $D(\gamma)$.  
The braid move of Eq.~\eqnref{eq:braidpicture3d} can then be implemented 
by three half Dehn-twists 
about the curves $\gamma_A$, $\gamma_B$ and $\gamma_C$ given by
\begin{equation}
\raisebox{-3em}{\pantsdehnind}
\end{equation}
Indeed, we have
\begin{equation}\begin{split}
\sqrt{D(\gamma_C)} \sqrt{D(\gamma_B)} \sqrt{D(\gamma_A)^{-1}} \pantsttt
&= \sqrt{D(\gamma_C)} \pantsdehnfirst \\
&= \pantsdehnsecond = \pantsbraided
\end{split}\end{equation}

Therefore both braids and Dehn-twists on the ribbon-graph Hilbert space can be implemented using $F$-moves.  

\newcommand*{\braidingpicture}[1]{
$\begin{matrix}
\fbox{\includegraphics[totalheight=0.19\textheight,viewport=320 160 580 390,clip]{images/braidingfigures/out#1}}\\
i = #1
\end{matrix}$
}
\newcommand*{\braidingpicturespecial}[2]{
$\begin{matrix}
\fbox{\includegraphics[totalheight=0.19\textheight,viewport=320 160 580 390,clip]{images/braidingfigures/out#1#2}}\\
i = #1'
\end{matrix}$
}

\begin{figure*}
\centering
\begin{tabular}{c@{$\quad$}c@{$\quad$}c}
\subfigure{\braidingpicture{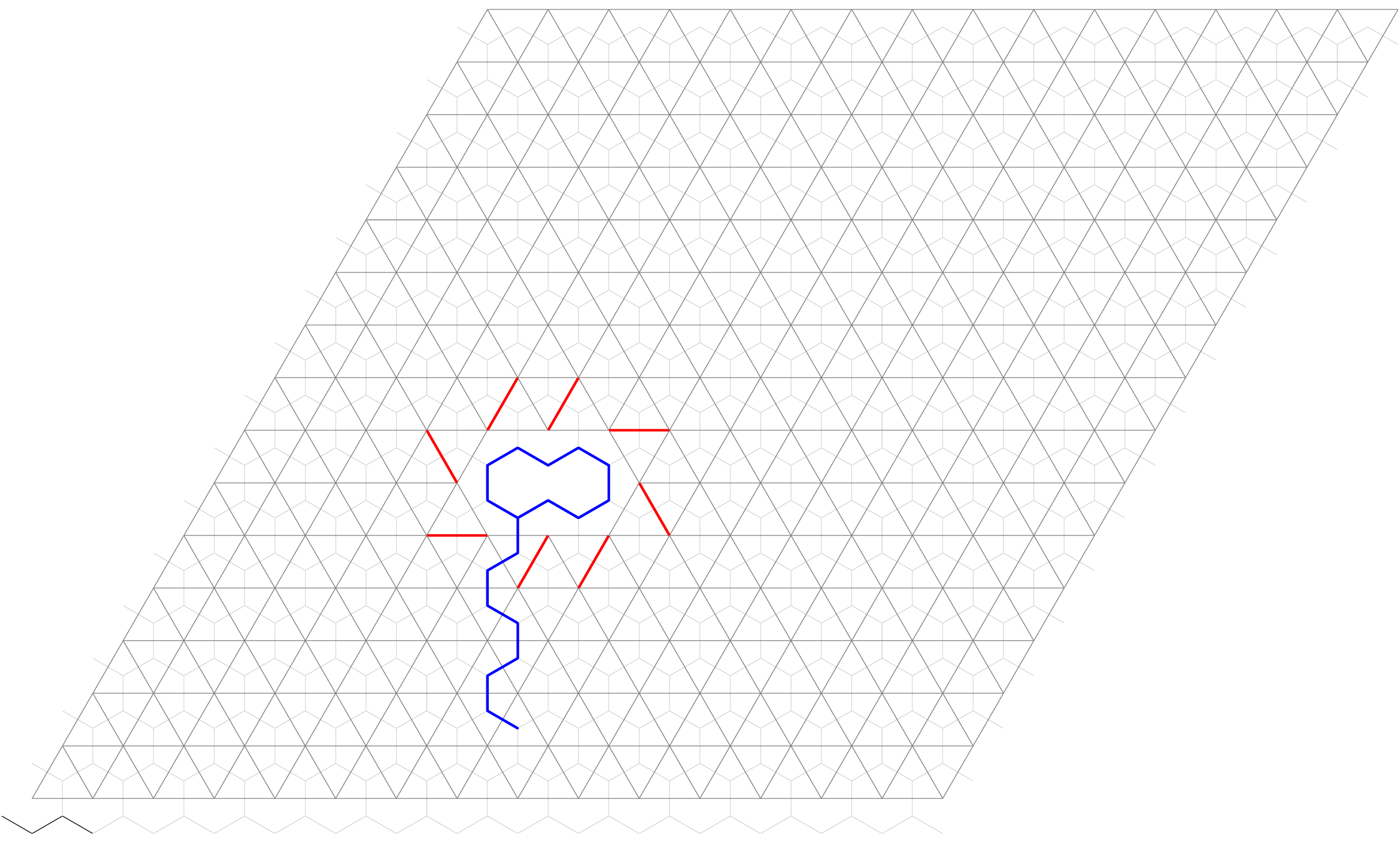}}&
\subfigure{\braidingpicture{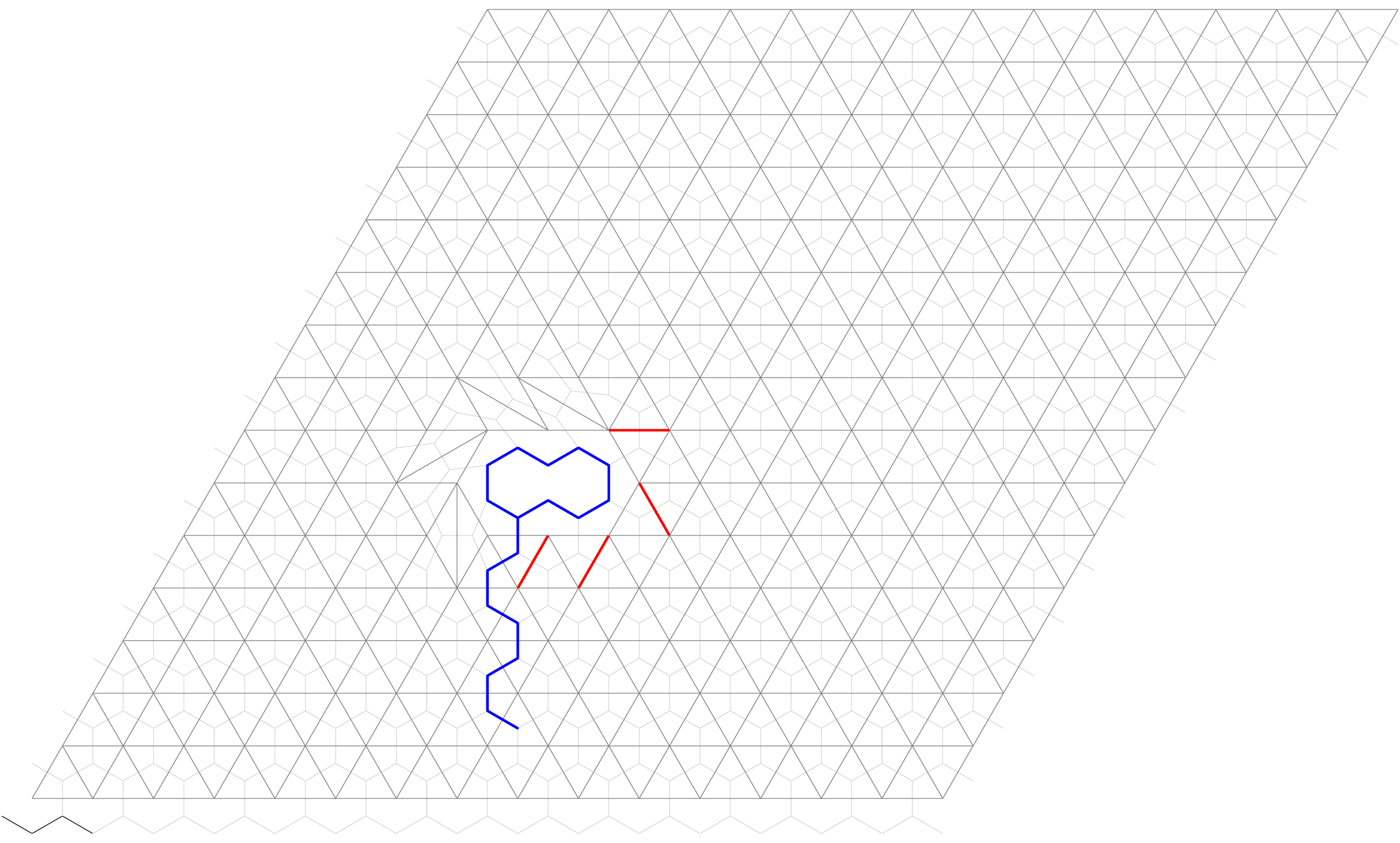}}&
\subfigure{\braidingpicture{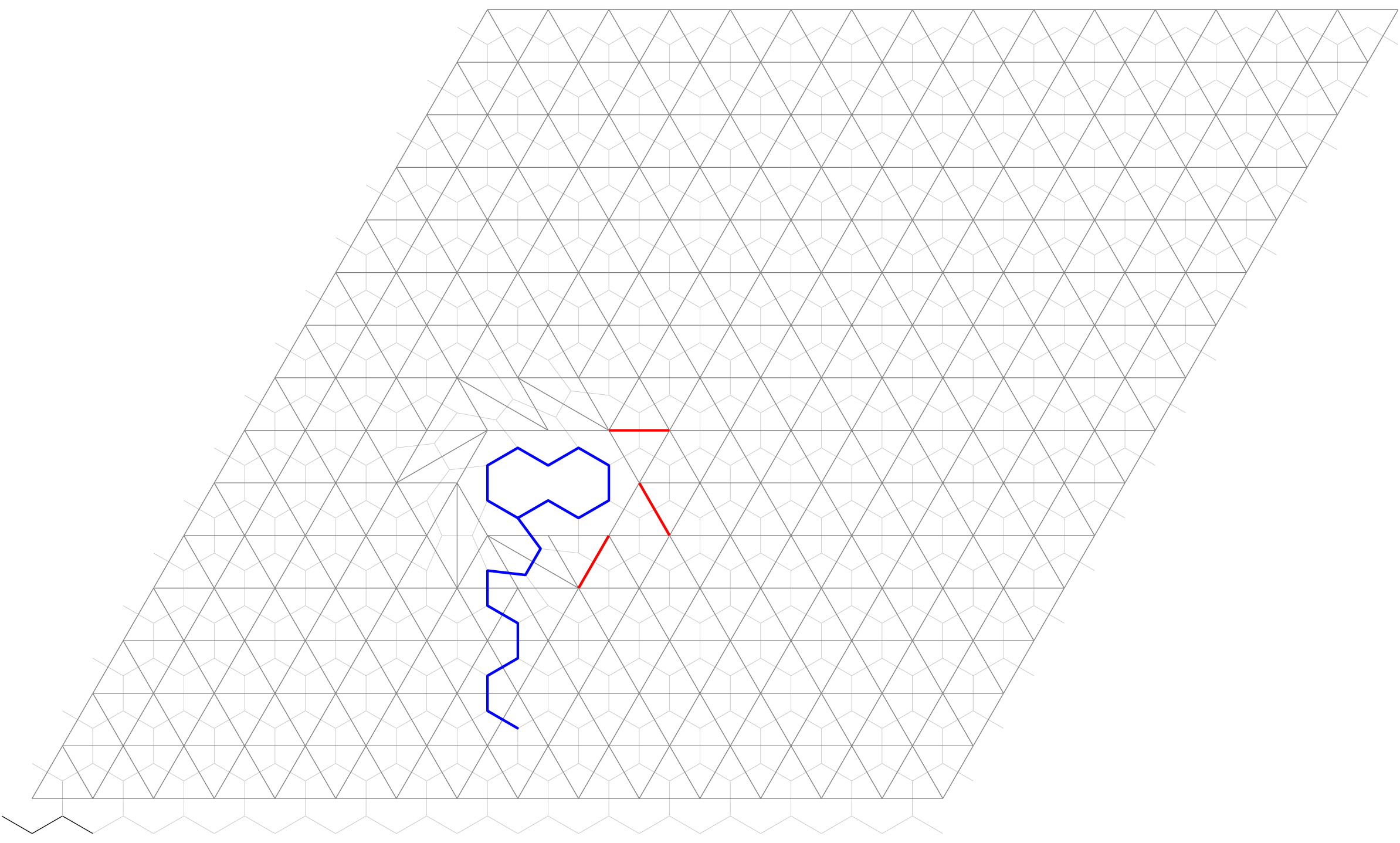}}
\end{tabular}
\caption{This figure shows for various $i$ how $i$ $F$-moves in the first round $(r=1)$ of the Dehn-twist $D(\gamma)$ affect a ribbon graph $\ket{\Phi} \in \cH^{(\ell,0^P)}_\Delta$ embedded in the dual graph $\widehat{\protect\cT}$.  } \label{fig:stringnetdeformationaction}
\end{figure*}

\begin{figure*}
\centering
\begin{tabular}{c@{$\quad$}c@{$\quad$}c}
\subfigure{\braidingpicture{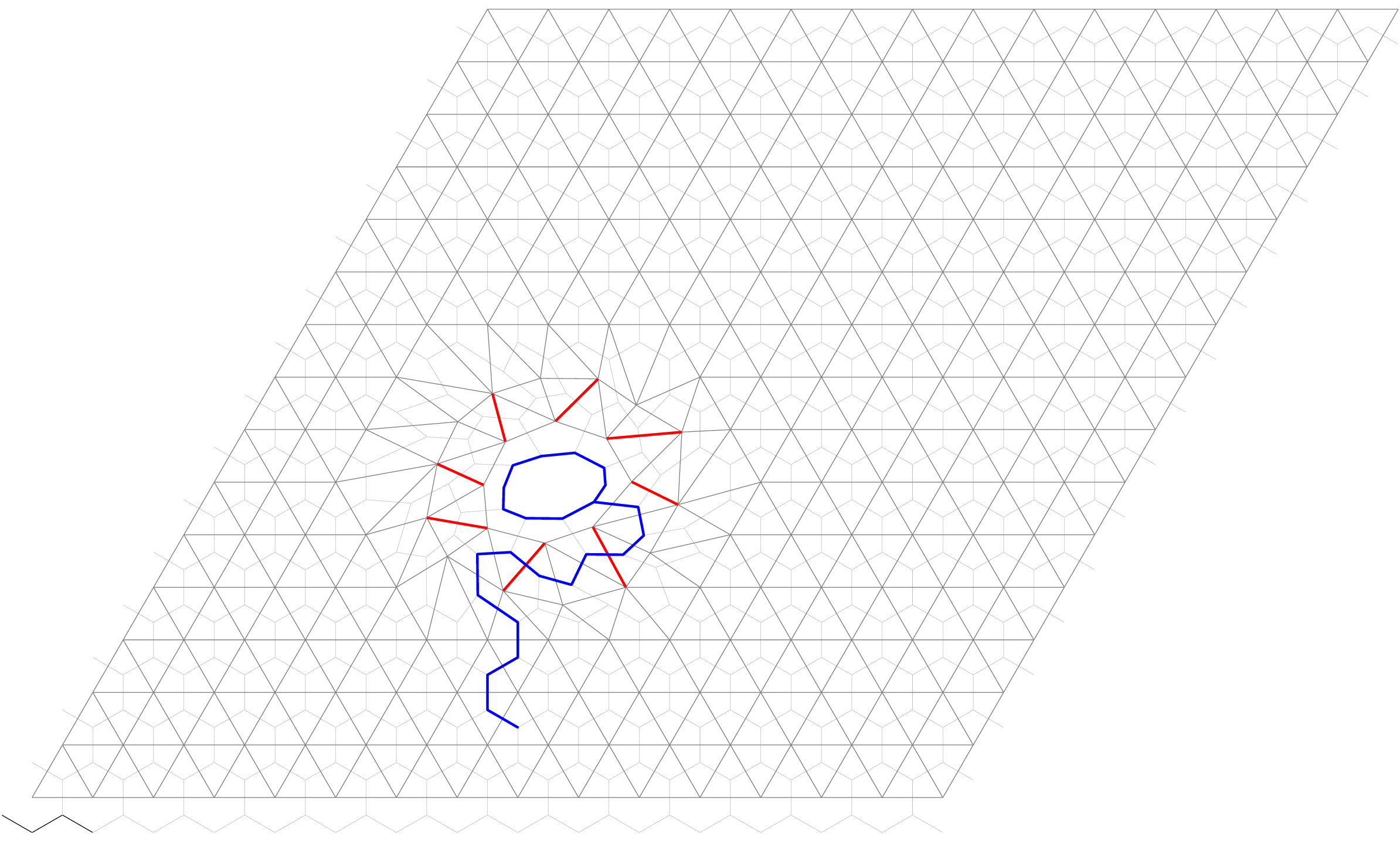}}&
\subfigure{\braidingpicture{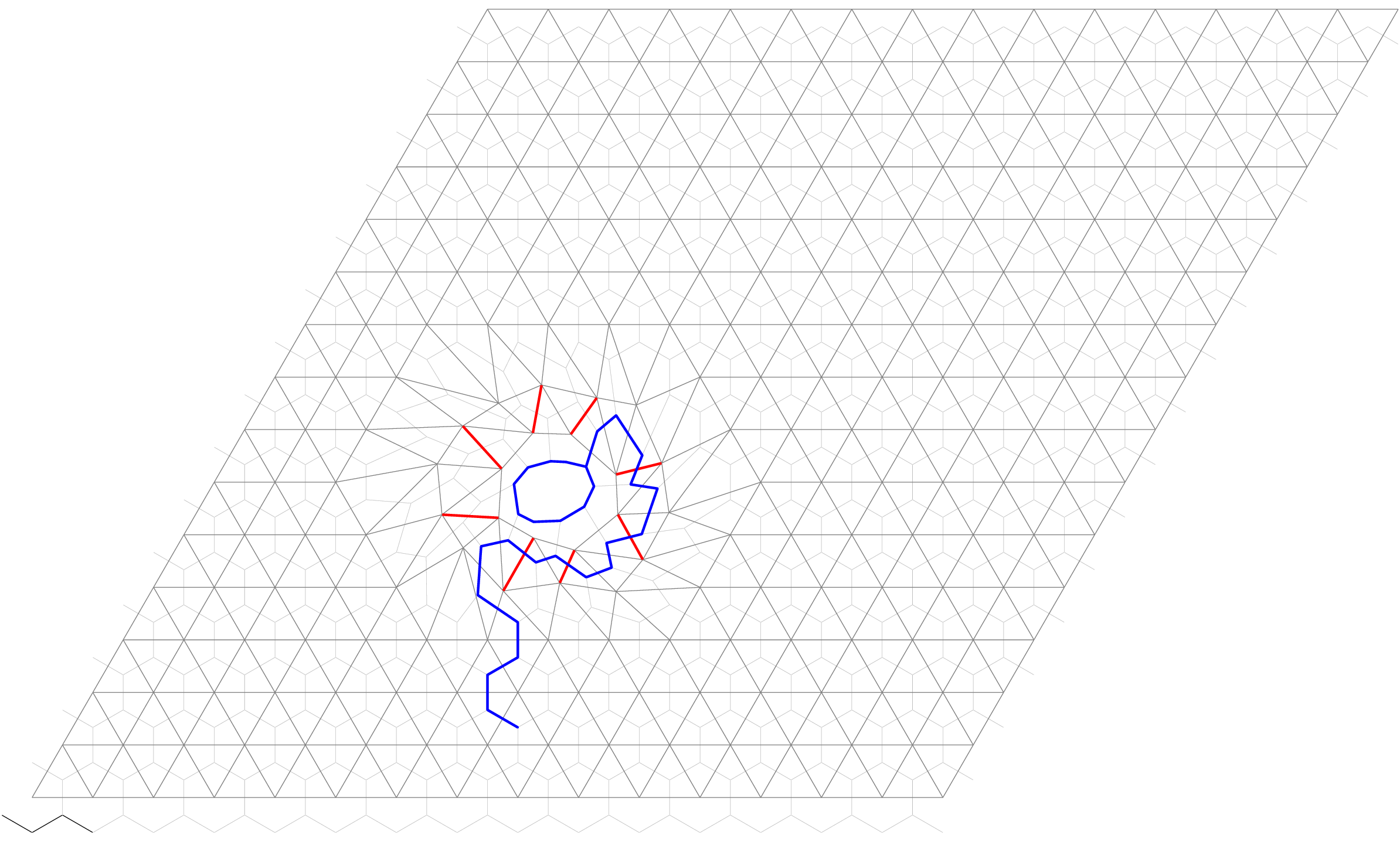}}&
\subfigure{\braidingpicture{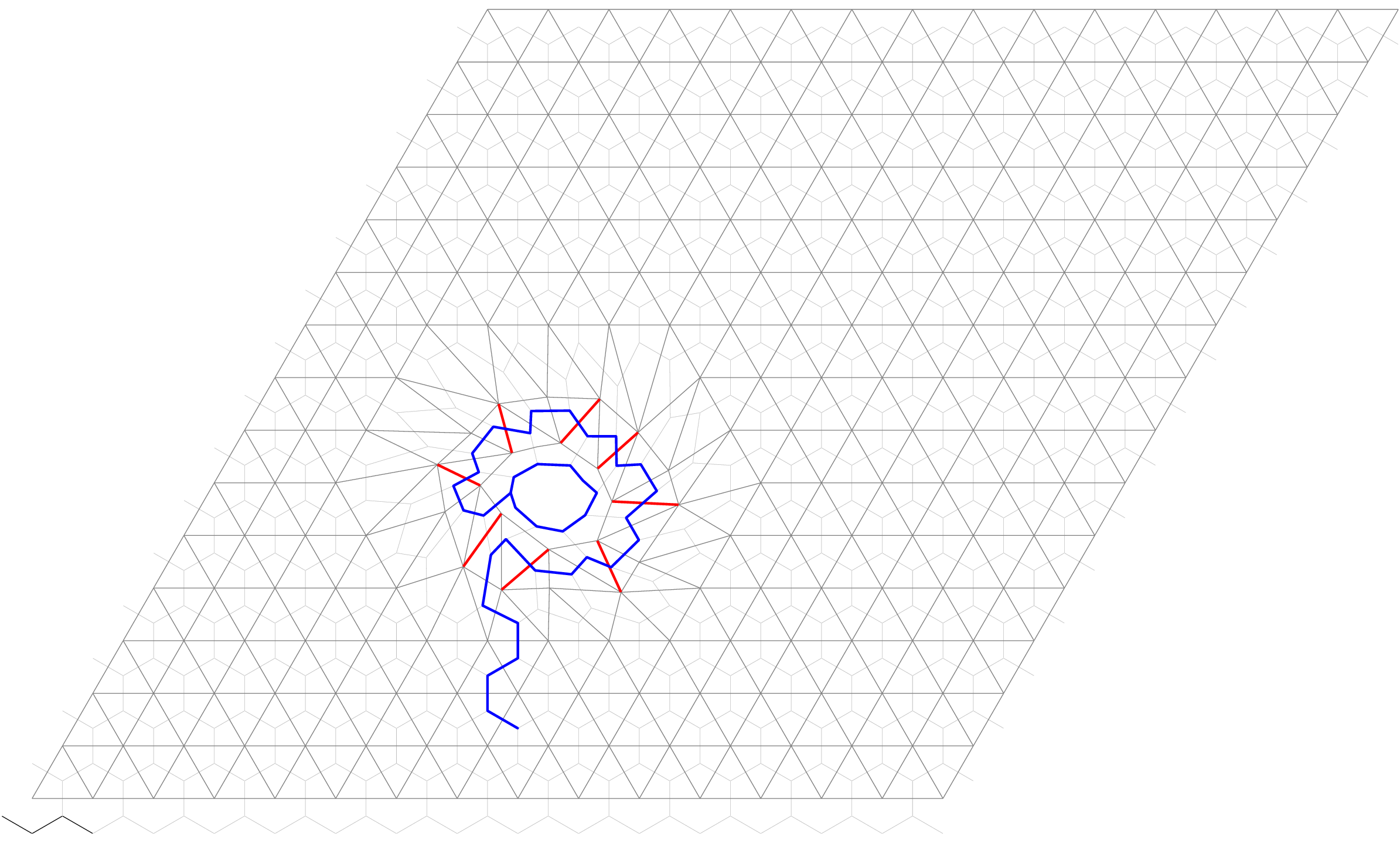}}\\
\subfigure{\braidingpicture{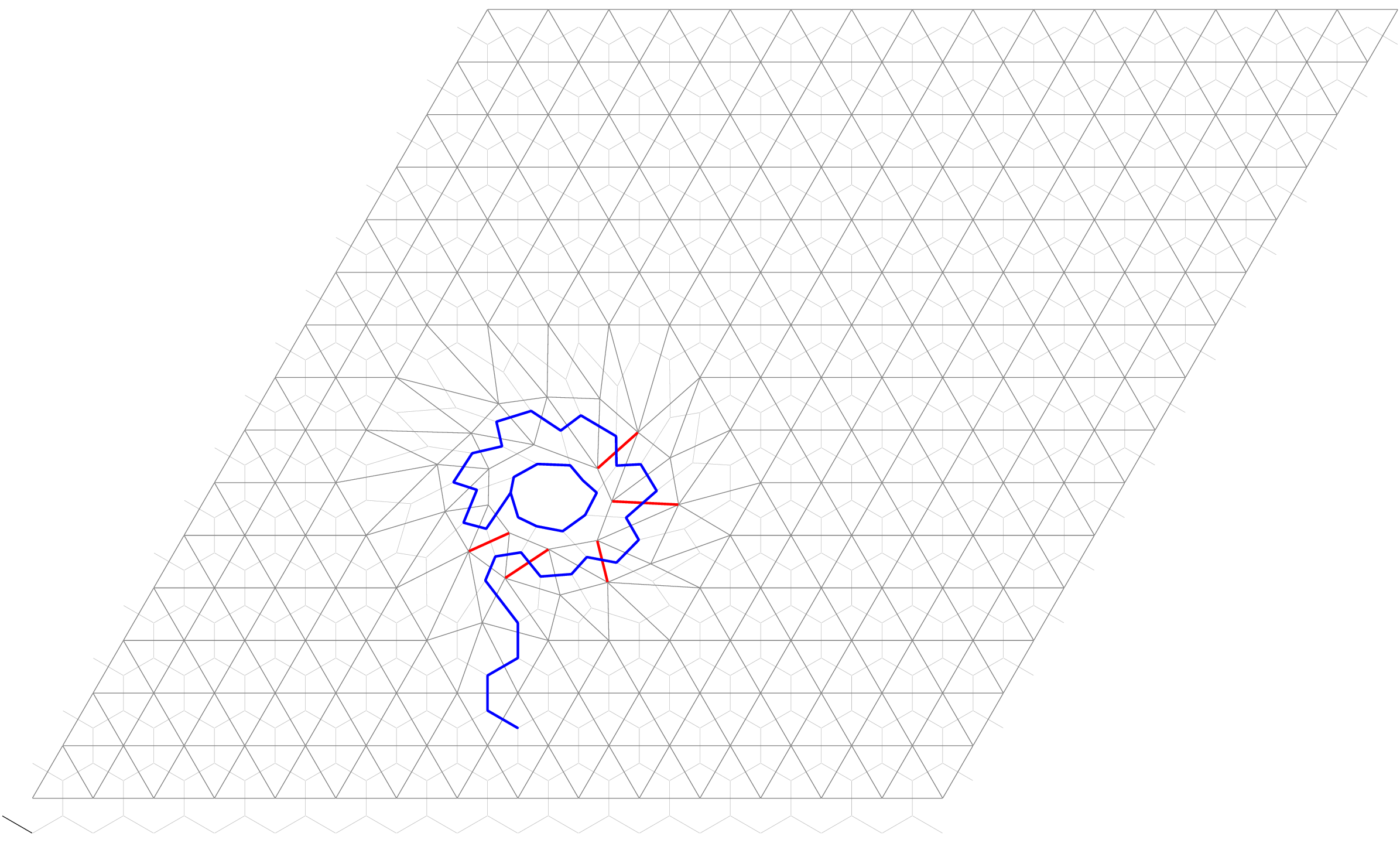}}&
\subfigure{\braidingpicturespecial{107}{b}}&
\subfigure{\braidingpicturespecial{112}{c}}
\end{tabular}
\caption{This figure continues \figref{fig:stringnetdeformationaction} for larger $i$, in order to illustrate that $D(\gamma)$ indeed corresponds to the deformation of an off-lattice ribbon graph by a Dehn-twist.  The half-Dehn-twist $\sqrt{D(\gamma)}$ is implemented by $i = 56$ $F$-moves.  After the $107$-th $F$-move, we switch from $\ket \Phi \in \cH^{(\ell, 0^P)}_\Delta$ to a different representative $\ket{\Phi'} \in \cH^{(\ell, 0^P)}_\Delta$ for convenience (note that $B \ket \Phi = B \ket{\Phi'}$).  } \label{fig:stringnetdeformationactionlast}
\end{figure*}

%% file: generalization.tex
\section{Quantum computation with Turaev-Viro codes from general anyon models} \label{sec:generalanyonmodels} 

Although for concreteness and to minimize notational complexity we have focused on the doubled Fibonacci anyon model, the calculations and methods generalize.  In this section and \appref{app:levinwenmodel}, we will go over the construction of surface codes for other doubled anyon models.  \secref{sec:parametersgeneralanyon} begins by collecting the data needed to specify an anyon model.  \secref{sec:stringnethilbertspacegeneral} introduces the ribbon graph Hilbert space $\cH_\Sigma$ associated with a surface~$\Sigma$ and includes the derivations of a number of useful local identities.  \secref{s:stringnetsdoubledmodel} introduces the doubled fusion basis states by considering ribbon graphs in $\Sigma \times [-1,1]$.  \secref{sec:levinwengeneralized} describes how the ribbon graph Hilbert space $\cH_\Sigma$ is realized as the ground space of the Levin-Wen Hamiltonian.

\subsection{Background on categories}

While we only rely on basic definitions of ``anyon models'' for the generalization discussed below,  some category-theoretic terminology is required to set our work in context.  A thorough discussion of these concepts is beyond the scope of this introductory article.  We restrict ourselves to a summary  and refer the interested reader to the literature for detailed definitions. 

Roughly, a {\em spherical category} as introduced by Barrett and Westbury~\cite{BarrettWestbury93spherical} is a set of data which gives rise to an isotopy invariant of trivalent directed labeled planar graphs embedded in the two-sphere.  A corresponding ``coherence theorem''~\cite{BarrettWestbury93spherical} guarantees consistency (of local rules) under a finite set of conditions on the data.  More specific properties are needed in our context: a {\em unitary braided fusion category} (see, e.g., the appendix of~\cite{KitaevAnyons}) additionally has a notion of braiding, and is equivalently referred to as a {\em unitary ribbon category}.  The reason for this terminology is that the braiding notion allows to make sense of crossing edges (i.e., non-planar graphs) and also allows to introduce a notion of twisting.  This leads to an invariant of ribbon tangles, as studied by Reshetikhin and Turaev~\cite{ReshTur90}.  The coherence theorem in this context is more commonly referred to as Mac Lane's theorem~\cite{MacLane98}.  Finally, a {\em unitary modular tensor category} is a unitary ribbon category with a non-singular $S$-matrix, as defined below; the terminology here stems from the fact that this give rise to a (projective) representation of the modular group.  

There are multiple ways of obtaining new categories from old ones.  The {\em categorical double} or {\em Drinfeld centre} $D\cC$~(see~\cite{Mue03} or~\cite[Section XIII.4]{Kassel95}) of a (not necessarily braided) fusion category~$\cC$ is always a braided category.  Moreover, if~$\cC$ is spherical, then $D\cC$ is modular, as shown by M\"uger~\cite{Mue03}.  The category~$D\cC$ has a particularly simple structure if~$\cC$ is itself already modular: in this case, $D\cC \cong \cC \otimes \cC^*$ is isomorphic to the direct product of~$\cC$ and the conjugate category $\cC^*$ obtained by complex conjugation of all data (see below).  We formulate all our results for this case.

\subsection{Quantum codes from categories}

A Turaev-Viro TQFT can be defined using any spherical category~$\cC$.  In particular, the arguments given below in \secref{sec:turaevvirocode} imply the following: For any finite unitary spherical category~$\cC$, and for any triangulated surface~$\Sigma$, with or without boundary, there is a corresponding quantum error-correcting code.  This code has qudits assigned to the edges, where $d$~is equal to the number of simple objects in~$\cC$.  Unless~$\cC$ is multiplicity-free, the triangles also need qudits.   

For certain categories~$\cC$, this code can be expressed as the ground state of a Levin-Wen local stabilizer Hamiltonian~\cite{LevinWen}.  
Here we focus on the case where $\cC$ is modular, and show that the anyons of the code are from the category~$\cC\otimes\cC^*$. A more general theorem~\cite{Turaevpers10}, proved in the modular case in~\cite{Walker91,Turaev94book},  asserts that if~$\cC$ is any finite semisimple spherical category, then the Turaev-Viro or Barrett-Westbury TQFT is the same as the Reshetikhin-Turaev TQFT for the quantum double~$D\cC$ (see e.g.,~\cite{Mue03a}).

\subsection{Parameters of an anyon model} \label{sec:parametersgeneralanyon}
Consider an anyon model described by a 
tensor category~$\cC$.  Such a model has particle types $\{\vac, i, i^*, j, j^*,\ldots \}$, where $^*$ denotes charge conjugation and $\vac = \vac^*$ is the trivial particle.  Let $d_i = d_{i^*}$ be the quantum dimension of particle~$i$ (note that $d_\vac = 1$), and let $\cD = \sqrt{\sum_i d_i^2}$ be the total quantum dimension.  Let $\delta_{abc^*} = 1$ if the fusion space~$V^{c}_{ab}$ is at least one-dimensional, and $0$~otherwise.  Assume for simplicity that the anyon model has no fusion multiplicities.  (Our results directly generalize to anyon models with fusion multiplicities.)  These fusion rules satisfy
\begin{equation} \label{eq:fusionassociativity}
\sum_m \delta_{i j m^*}\delta_{m k l^*} = \sum_m \delta_{j k m^*} \delta_{i m l^*}
 \enspace ,
\end{equation}
which expresses associativity of fusion. The quantum dimensions satisfy
\begin{equation} \label{eq:qdimensionsidentity}
d_i d_j = \sum_k \delta_{i j k} d_k
 \enspace .
\end{equation}

It is further necessary to specify the operations for fusing and braiding anyons.  The fusion tensor, or $F$-move, for changing fusion bases has entries
\begin{equation}
\fusethreedirected{i}{j}{k}{m}{\ell} = \sum_n F^{i j m^*}_{k \ell^* n} 
\fusethreesdirected{i}{j}{k}{n}{\ell}
 \enspace .
\end{equation}
It satisfies the five properties, for all $i, j, \ldots, s$, 
\begin{align} \label{eq:levinwensymmetries}
\def\customspacing{.2cm}
\begin{tabular}{r r@{$\,=\,$}l}
\small{physicality:} 
&
$F^{ijm}_{k \ell n} \delta_{ijm} \delta_{k \ell m^*}$ & $F^{ijm}_{k \ell n} \delta_{i \ell n} \delta_{jkn^*}$
\\[\customspacing]
\small{pentagon identity:}
&
$\sum_n F^{m \ell q}_{kpn} F^{jip^*}_{mns} F^{jsn}_{\ell kr}$ & $F^{jip^*}_{q^*kr} F^{r^*iq^*}_{m \ell s}$
\\[\customspacing]
\small{unitarity:}
& 
$(F^{ijm}_{k \ell n})^*$ & $F^{i^*j^*m^*}_{k^*\ell^*n^*}$
\\[\customspacing]
\small{tetrahedral symmetry:} 
&
$F^{ijm}_{k \ell n} = F^{jim}_{\ell kn^*}$ & $F^{\ell k m^*}_{jin} = F^{imj}_{k^* n \ell} \sqrt{ \frac{d_m d_n}{d_j d_\ell} }$
\\[\customspacing]
\small{normalization:} 
&
$F^{ii^*\vac}_{j^*jk}$ & $\sqrt{\frac{d_k}{d_i d_j}} \delta_{ijk}$
\end{tabular}
\end{align}

The braiding operations, called $R$-moves, are given by 
\begin{equation} \label{eq:Rmatrixfirstdef}
\arrowbraiddiagram{j}{i}{k} = R^{ij}_k \arrowfusiondiagram{j}{i}{k}
\qquad
\arrowbraidoppositediagram{i}{j}{k} = (R^{i^*j^*}_{k^*})^* \arrowfusiondiagram{i}{j}{k}
 \enspace .
\end{equation}
The entries satisfy $\abs{R^{ij}_k} = 1$ and the hexagon identities: 
\begin{equation}\begin{split}
R^{ki}_m F^{k^* i^* m}_{\ell j^* g} R^{kj}_g &= \sum_n F^{i^* k^* m}_{\ell j^* n} R^{k n}_\ell F^{j^* i^* n}_{\ell k^* g} \\
R^{ik}_m F^{k i m^*}_{\ell^* j g^*} R^{jk}_g &= \sum_n F^{i k m^*}_{\ell^* j n^*} R^{nk}_\ell F^{j i n^*}_{\ell^* k g^*}
 \enspace .
\end{split}\end{equation}

Combined with $F$-moves, $R$-moves allow for the resolution of crossings, via, e.g., 
\begin{equation} \label{eq:overcrossingresolution}
\threeDcrossingunresolved{i}{j} = \sum_k \sqrt{\frac{d_k}{d_i d_j}} \delta_{i j k^*} R^{ij}_k\threeDcrossingresolved{i}{j}{k}
 \enspace .
\end{equation}
They also allow for defining the topological phases and the topological $S$-matrix.  The topological phases are given by $\theta_i = (R^{i^* i}_{\vac})^*$ and allow a ribbon to be untwisted: 
\begin{equation} \label{eq:topphase}
\threeDlinetwistright{i} = \theta_i \threeDline{i}
 \enspace .
\end{equation}
The topological $S$-matrix is defined by
\begin{equation}
S_{ij} = \frac{1}{\cD}\ \Smatrix{i}{j}
 \enspace .
\end{equation}
This matrix is symmetric and unitary; it satisfies
\begin{equation}
S_{ij} = S_{ji} = S_{i^*j^*} = S_{j^*i^*}, \qquad S_{i\vac} = d_i / \cD
 \enspace .
\end{equation}

For an anyon model with particles~$\{ \vac_+, a_+, b_+, \ldots \}$ described by a tensor category~$\cC_+$, we can define a dual category~$\cC_-$ with anyons~$\{ \vac_-, a_-, b_-, \ldots \}$.  Here~$a_-$ is the same as~$a_+$, but with opposite chirality, i.e., the topological phase is given by $\theta_{a_-} = \theta^*_{a_+}$.  Furthermore, the $R$-matrix is replaced by its adjoint, which for the multiplicity-free case considered here boils down to $R_{a_-}^{b_-c_-} = (R_{a_+}^{b_+c_+})^*$.  
The {\em double} of~$\cC = \cC_+$ is the category~$D\cC = \cC_+ \otimes \cC_-$; its particles are $\{a_+ \otimes b_-\ |\ a_+ \in \cC_+, b_- \in \cC_-\}$, and the braiding and fusion matrices are tensor products of those of~$\cC_+$ and~$\cC_-$.  (Note: Here we used the assumption that $\cC$ is modular.  For a general category~$\cC$, the double $D\cC$ need not be isomorphic to~$\cC_+\otimes\cC_-$.)

\subsection{The ribbon graph Hilbert space \texorpdfstring{$\cH_\Sigma$}{H\_Sigma}} \label{sec:stringnethilbertspacegeneral}

Using the fusion rules $\delta_{ijk}$, the quantum dimensions $d_i$ and the $F$-tensor $F^{ijm}_{k \ell n}$,  we can define the ribbon graph Hilbert space $\cH_\Sigma$ associated with a surface $\Sigma$ as in \secref{sec:stringnetdefinition}.  A (colored) ribbon graph is a graph with labeled directed edges embedded into $\Sigma$, with vertices of degree two and three in the interior, and degree-one vertices at the marked boundary points of $\Sigma$.  Switching the direction of an edge corresponds to conjugating the anyon label, and edges assigned the trivial label~$\vac$ can be removed or added from the picture according to
\begin{align}
\begin{pspicture}[shift=-.7](-1.25,-0.75)(1.25,0.75)
\psset{unit=.7cm}
\psset{linewidth=.56pt} 
\psset{labelsep=2.5pt} 
\SpecialCoor
\pnode(-.5,0){A}
\pnode(.5,0){B}
\pnode([nodesep=1,angle=120]A){Aa}
\pnode([nodesep=1,angle=-120]A){Ab}
\pnode([nodesep=1,angle=60]B){Ba}
\pnode([nodesep=1,angle=-60]B){Bb}
\pscurve{->}(Aa)(A)(Ab)
\pscurve{->}(Ba)(B)(Bb)
\rput(A){\uput[30](.5;110){$i$}}
\rput(B){\uput[150](.5;60){$j$}}
\end{pspicture}
&= 
\begin{pspicture}[shift=-.7](-1.25,-0.75)(1.25,0.75)
\psset{unit=.7cm}
\psset{linewidth=.56pt}
\psset{labelsep=2.5pt} 
\SpecialCoor
\pnode(-.5,0){A}
\pnode(.5,0){B}
\psline{-}(A)(B)
\psline{->}([nodesep=1,angle=120]A)(A)
\psline{<-}([nodesep=1,angle=-120]A)(A)
\psline{->}([nodesep=1,angle=60]B)(B)
\psline{<-}([nodesep=1,angle=-60]B)(B)
\uput[90](0,0){$\vac$}
\rput(A){\uput[30](.5;120){$i$}}
\rput(A){\uput[-30](.5;-120){$i$}}
\rput(B){\uput[-150](.5;-60){$j$}}
\rput(B){\uput[150](.5;60){$j$}}
\end{pspicture} 
\end{align}
Valid ribbon graphs are those that satisfy the fusion rules at every vertex, i.e., having $\delta_{ijk} = 1$ for every vertex with incoming edges labeled $i$, $j$ and $k$.  (For a vertex of degree two, this rule is adapted by adding an edge with label~$\vac$.)  The Hilbert space $\cH_\Sigma$ is the space of formal linear combinations of such ribbon graphs, modulo the local relations 
\begin{align}
\begin{pspicture}[shift=-.4](-0.75,-.5)(0.0,.5)
\psset{unit=.7cm}
\psset{linewidth=.56pt}
\psset{labelsep=2.5pt} 
\SpecialCoor
\pnode(-.5,0){A}
\pnode(.5,0){B}
\pnode([nodesep=1,angle=120]A){Aa}
\pnode([nodesep=1,angle=-120]A){Ab}
\pnode([nodesep=1,angle=60]B){Ba}
\pnode([nodesep=1,angle=-60]B){Bb}
\pscurve{->}(Aa)(A)(Ab)
\rput(A){\uput[30](.5;110){$i$}}
\end{pspicture}
&= \begin{pspicture}[shift=-.4](0.0,-.5)(0.75,.5)
\psset{unit=.7cm}
\psset{linewidth=.56pt}
\psset{labelsep=2.5pt} 
\SpecialCoor
\pnode(-.5,0){A}
\pnode(.5,0){B}
\pnode([nodesep=1,angle=120]A){Aa}
\pnode([nodesep=1,angle=-120]A){Ab}
\pnode([nodesep=1,angle=60]B){Ba}
\pnode([nodesep=1,angle=-60]B){Bb}
\pscurve{->}(Ba)(B)(Bb)
\rput(B){\uput[150](.5;60){$i$}}
\end{pspicture}\\
\raisebox{-1.0em}{\scalebox{1.0}{
\begin{pspicture}(-.4,-.4)(.4,.4)
\psset{linewidth=.56pt} 
\psset{labelsep=2.5pt}
\SpecialCoor
\psline{->}(-.7,0)(-.35,0)
\put(0,0){\psarc{->}(0,0){0.35}{-180}{180}}
\uput[30](.35;30){$i$}
\uput[150](.35;150){$j$}
\end{pspicture}
}}&= d_i \delta_{j\vac} \label{e:looprulegeneral} \\
\begin{pspicture}[shift=-.5](-1.25,-.5)(1.0,.5)
\psset{unit=.7cm}
\psset{linewidth=.56pt}
\psset{labelsep=2.5pt} 
\SpecialCoor
\pnode(-.5,0){A}
\pnode(.5,0){B}
\psline{<-}(A)(B)
\psline{->}([nodesep=1,angle=120]A)(A)
\psline{->}([nodesep=1,angle=-120]A)(A)
\psline{->}([nodesep=1,angle=60]B)(B)
\psline{->}([nodesep=1,angle=-60]B)(B)
\uput[90](0,0){$m$}
\rput(A){\uput[30](.5;120){$i$}}
\rput(A){\uput[-30](.5;-120){$j$}}
\rput(B){\uput[-150](.5;-60){$k$}}
\rput(B){\uput[150](.5;60){$\ell$}}
\end{pspicture}
&= 
\sum_n F^{ijm}_{k\ell n}
\begin{pspicture}[shift=-.75](-.5,-.8)(.5,.8)
\psset{unit=.7cm}
\psset{linewidth=.56pt}
\psset{labelsep=2.5pt} 
\SpecialCoor
\pnode(0,.5){A}
\pnode(0,-.5){B}
\psline{<-}(A)(B)
\psline{->}([nodesep=1,angle=30]A)(A)
\psline{->}([nodesep=1,angle=150]A)(A)
\psline{->}([nodesep=1,angle=-30]B)(B)
\psline{->}([nodesep=1,angle=-150]B)(B)
\uput[0](0,0){$n$}
\rput(A){\uput[60](.5;150){$i$}}
\rput(B){\uput[120](.5;-150){$j$}}
\rput(B){\uput[60](.5;-30){$k$}}
\rput(A){\uput[120](.5;30){$\ell$}}
\end{pspicture} \label{e:Frulesmooth}
\end{align}
Compare to Eqs.~\eqnref{eq:firstlevinwen}--\eqnref{eq:lastlevinwen}.  

The Hilbert space $\cH_\Sigma$ decomposes into spaces $\cH^\ell_\Sigma$ indexed by labelings~$\ell$ of the marked boundary points on the boundary components of~$\Sigma$, as in Eq.~\eqnref{eq:boundarylabeldecomp}.  $\cH^\ell_\Sigma$ is the subspace of ribbon graphs with, for each boundary point~$p$, an edge carrying label~$\ell(p)$ leaving~$p$.  An orthonormal basis of~$\cH^\ell_\Sigma$ can be constructed using ribbon graphs living on the dual graph of a triangulation of~$\Sigma$, as explained in \secref{sec:stringnetdefinition}.  For example, two orthonormal bases of the space $\cH_{\Sigma_2}^{(k^*, \ell)}$ on the annulus are given by the ribbon graphs
\begin{equation} \label{eq:annulusorthonormalbasis}
\cB := \Bigg\{ \raisebox{-.75em}{\scalebox{.8}{\Danyonsstandardbasistwo{k}{i}{j}{\ell}}} \,\Bigg\vert\, \text{$\delta_{k i^* j^*}, \delta_{i j \ell^*} \neq 0$} \Bigg\} \qquad \cB' := \Bigg\{ \raisebox{-.75em}{\scalebox{.8}{\Danyonsstandardbasisone{k}{i}{j}{\ell}}} \,\Bigg\vert\, \text{$\delta_{k i^* j^*} \neq 0, \delta_{i j \ell^*} \neq 0$} \Bigg\}
 \enspace .
\end{equation}
The inner product can be defined using any such basis; the result is independent of the chosen basis because different triangulations of $\Sigma$ with the same number of edges are related by edge flips and $F$-moves are unitary.

\subsection{Ribbon graphs on \texorpdfstring{$\Sigma \times [-1,1]$}{Sigma x [-1,1]} and fusion basis states of the doubled model} \label{s:stringnetsdoubledmodel}

The anyonic fusion basis states for the doubled theory are defined, as in \secref{s:anyonicfusionbasis}, by reducing ``doubled" ribbon graphs in the manifold~$\Sigma \times [-1,1]$ down to the surface~$\Sigma$.  

In the following, we use undirected, unlabeled, dashed lines to denote the superposition of all anyon types, weighting~$i$ by~$d_i/\cD$: 
\begin{equation} \label{eq:threeDlinevacuum}
\threeDlinevacuum = \frac{1}{\cD} \sum_i d_i \threeDline{i} = \frac{1}{\cD} \sum_i d_i \threeDlinedown{i}
\end{equation}
These vacuum lines generalize Eq.~\eqnref{eq:vacuumlinesfibonacci}.  \lemref{lem:vacuum} below summarizes their main properties.  

Consider a fusion diagram describing a state of $n$~anyons of a doubled theory $D\cC = \cC_+ \otimes \cC_-$.  This is a tree~$G$ with $n+1$ leaves at the punctures of the $(n+1)$-punctured sphere~$\Sigma = \Sigma_{n+1}$, and with edge labels of the form~$i_+ \otimes j_-$.  
Consider the manifold~$\cM = \Sigma \times [-1,1]$, and take two copies of $G$,
$G_\pm = G \times \{\pm 1\} \subset \Sigma \times \{\pm 1\}$.   
For an edge $e \in G$ labeled $i_+ \otimes j_-$ in the original fusion diagram, label the corresponding edges $e_+$ and $e_-$ in $G_+$ and $G_-$ by $i$ and $j$, respectively.  
In addition to ``doubling'' the edges in this fashion, add vacuum lines, from Eq.~\eqnref{eq:threeDlinevacuum}, embedded in $\Sigma \times \{ 0 \}$ around $n$ of the punctures.  

For concreteness, consider the case of a single anyon $i_+ \otimes j_-$, corresponding to the ribbon graph
\begin{equation} \label{eq:singleanyonthreed}
\raisebox{-.5in}{\includegraphics{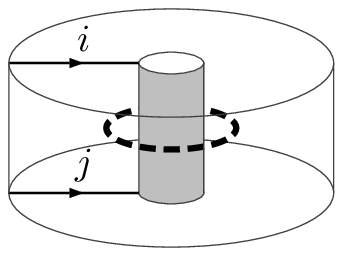}}
\end{equation}
on $\Sigma_2 \times [-1,1]$.  In this figure, the shaded area is the inner component of $(\partial\Sigma_2) \times [-1,1]$.

To see that this is consistent with an $i_+\otimes j_-$ anyon, observe first that it has the required idempotency property, since stacking gives the same diagram up to a constant factor:
\begin{equation}
\raisebox{-.3in}{\includegraphics{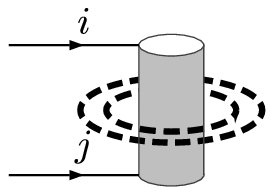}}
= 
\cD \raisebox{-.3in}{\includegraphics{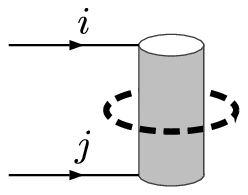}}
\end{equation}
according to the ``doubling'' property~\eqnref{it:firstpropertyvacuum} of vacuum lines stated in~\lemref{lem:vacuum}.  We can also verify that it carries the correct topological phase~$\theta_{i_+} \theta_{j_-} = \theta_i \theta_j^*$, by computing the action of a Dehn-twist~$D(\gamma)$: 
\begin{equation}
D(\gamma) \raisebox{-.3in}{\includegraphics{images/3Dtwiststraight}}
= \raisebox{-.43in}{\includegraphics{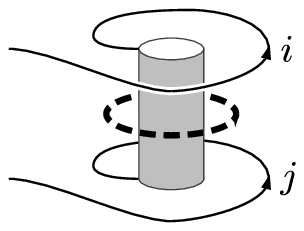}}
= \raisebox{-.32in}{\includegraphics{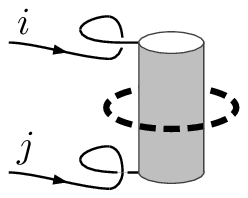}} 
= \theta_i \theta_j^* \raisebox{-.3in}{\includegraphics{images/3Dtwiststraight}}
\end{equation}
Here, as in Eq.~\eqnref{e:fibanyondehntwistexample}, we have used Eq.~\eqnref{it:secondpropertyvacuum} to pull the ribbons over the vacuum line,
and then removed the kinks using the definition of the topological phase, Eq.~\eqnref{eq:topphase}.  

Next, let us check that this realization reproduces the correct $R$-matrix, by considering a braid: 
\begin{equation}
\raisebox{-1.1cm}{\includegraphics{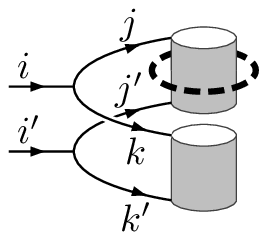}}
\mapsto
\raisebox{-1.16cm}{\includegraphics{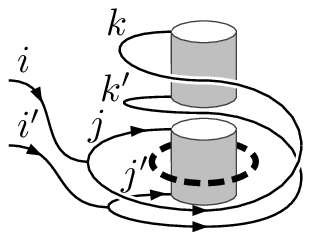}}
=
\raisebox{-1.12cm}{\includegraphics{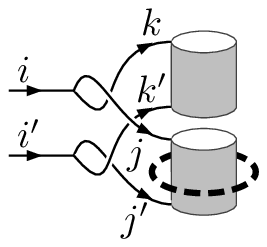}}
\end{equation}
Using Eq.~\eqnref{eq:Rmatrixfirstdef}, we find that, indeed, 
\begin{equation}
R_{i\otimes i'}^{j\otimes j' k\otimes k'} = R_i^{j k} (R_{i'^*}^{k'^* j'^*})^*
 \enspace .
\end{equation}

It is easy to see that this definition of ``doubled'' ribbon graph on $\Sigma \times [-1,1]$ also gives rise to the correct $F$-matrix of the doubled theory~$D\cC$, since the tensor factors behave independently.  

Note that this realization of the doubled theory by ribbon graphs on $\Sigma \times [-1,1]$ is only an intermediate step; we are ultimately interested in ribbon graphs on~$\Sigma$.  Ribbon graphs on $\Sigma \times [-1,1]$ can be reduced to~$\Sigma$.  First, pick a label~$\ell(p)$ for every boundary point~$p \in \partial\Sigma$ fusion-consistent with the anyon labels $i_+$ at~$(p, 1)$ and $i_-$ at~$(p, -1)$, respectively.  Then, connect up these two ends, and attach an edge with label~$\ell(p)$ ending at~$(p, 0)$.  Clearly, the resulting ribbon graphs have only one edge ending at every hole; in particular, we can project the string-net onto~$\Sigma$, and then resolve crossings by use of identity~\eqnref{eq:overcrossingresolution}.  This eventually gives a ribbon graph on~$\Sigma$, an anyonic fusion basis state.

For example, for $\Sigma_2$ we obtain from Eq.~\eqnref{eq:singleanyonthreed} the ribbon graphs 
\begin{equation} \label{eq:kaelldef}
\myvecstand{k}{\ell}{i}{j} := \threeDdoubleringsingle{k}{i}{j}{\ell} = \Danyons{k}{i}{j}{\ell}
\end{equation}
provided $\delta_{k i^*j^*} \neq 0$ and $\delta_{ij \ell^*} \neq 0$.  
Such ribbon graphs are explicitly reduced to elements of $\cH_{\Sigma_2}$ for the Fibonacci theory in \appref{sec:fibonacciapplication}.  
By \lemref{lem:idempotentsfromdouble} below, the ribbon graphs $\myvecstand{k}{\ell}{i}{j}$ are  idempotents under stacking.  To give another example, the described recipe gives  basis vectors for $\cH_{\Sigma_4}$ of the form
\begin{equation} \label{eq:generalizedfourpunctureds}
\myvecstandardpic{k}{\ell}{a}{b}{c}{d}
\end{equation}
\appref{sec:normalizationgeneral} proves that the anyonic fusion basis states form a complete, orthonormal basis.

\subsection{The Turaev-Viro code as the ground space of the Levin-Wen Hamiltonian} \label{sec:levinwengeneralized}

In this section, we present Levin and Wen's lattice Hamiltonians~\cite{LevinWen} realizing the  space $\cH_\Sigma^\ell$ of (colored) ribbon graphs for a general anyon model with no fusion multiplicities, as in \secref{sec:parametersgeneralanyon}.  This generalizes the construction of \secref{sec:realization} for the Fibonacci model.  

As in \secref{sec:realization}, let $\cT$ be a triangulation of $\Sigma$, and let $\widehat \cT = (V, E)$ be the dual graph of~$\cT$.  Let $\cH_{\widehat \cT} \cong (\mathbb{C}^N)^{\otimes \abs E}$ be the Hilbert space obtained by associating a $N$-dimensional qudit with orthonormal basis $\{ \ket i \}_i$ corresponding to the $N$~different anyon labels to every edge.  Computational basis vectors will be represented by directing the edges of~$\widehat \cT$ and labeling them by anyon labels; reversing the direction of an edge corresponds to conjugating the label.  Labelings~$\ell$ of boundary points $p \in \partial \Sigma$ by anyon labels will be extended to virtual boundary edges attached to~$\widehat\cT$ as discussed in \secref{sec:realization}; that is, we associate the trivial label~$\vac$ to all but one edge that carries label~$\ell(p)$.  
 
On $\cH_{\widehat \cT}$, we define vertex and plaquette projections as follows:
\begin{itemize}
\item
For every vertex $v$ of $\widehat\cT$, there is a vertex projection~$Q_v$.  This operator depends only on the edges incident on $v$, and enforces the fusion rules at~$v$.  It is given by 
\begin{equation}
Q_v = \sum_{i,j,k} \delta_{ijk}\proj{ijk}
 \enspace .
\end{equation}
As before, these are commuting projections.  Their simultaneous $+1$-eigenspace is (isomorphic) to the space~$\cH_\Delta^{(\ell, \vac^P)}$ of ribbon graphs on a surface~$\Delta$.  Here $\Delta \cong \Sigma_{n+P}$ is obtained by placing a puncture into each plaquette of $\widehat \cT$ embedded in~$\Sigma_n$.  Let $\cP$ be the set of these new punctures.  The ribbon graphs in $\cH_\Delta^{(\ell, \vac^P)}$ have open boundary conditions on punctures $p \in \cP$, and are in one-to-one correspodence to the computational basis for the qudits.

\item
For each plaquette~$p$ of~$\widehat \cT$, there is a plaquette projection~$B_p$.  $B_p$ is supported on the subspace~$\cH^{(\ell, \vac^P)}_\Delta$, and acts by adding a vacuum loop, divided by total quantum dimension~$\cD$.  By reducing the resulting ribbon graph, its action on computational basis vectors is given by 
\begin{equation} \label{eq:bpidefac}
B_p = \frac{1}{\cD^2} \sum_i d_i O^i_p \quad \textrm{where} \quad O^i_p \, \Bigg |
\raisebox{-0.6cm}{\includegraphics[scale=.7]{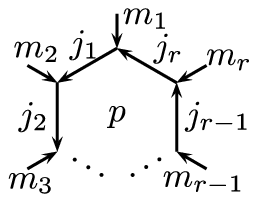}}
\Bigg \rangle
= \!\!
\sum_{k_1,\ldots,k_n} 
\!\!\!\! \bigg( \! \prod_{\nu = 1}^r F^{m_\nu j_\nu j_{\nu-1}}_{i k_{\nu-1} k_\nu} \! \bigg)
\Bigg |
\raisebox{-0.6cm}{\includegraphics[scale=.7]{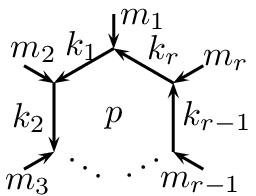}}
\Bigg \rangle
\end{equation}
Here we have assumed that $p$ has $r$~boundary edges, and have identified $j_0 = j_r$ and $k_0 = k_r$.  The plaquette operators commute with each other and with the vertex operators.
\end{itemize}

The Levin-Wen Hamiltonian on $H_{\widehat \cT}$ on $\cH_{\widehat \cT}$ is the sum of these operators, as in Eq.~\eqnref{eq:levinwenhamiltonian}.  According to \lemref{lem:mainisomorphism} (which holds for the more general definitions given in this section), its ground space $\cH^{\ell,gs}_{\widehat \cT} \subset \cH_{\widehat \cT}$, i.e., the simultaneous~$+1$-eigenspace of all plaquette and vertex operators, is isomorphic to the ribbon graph space~$\cH^\ell_\Sigma$.  Quantum computations can therefore be performed in a similar manner as before using $F$-moves, i.e., five-qudit gates.  A generalized version of \lemref{lem:tadpolegluing} allowing for state preparation by gluing and measurement of topological charge is given as \lemref{lem:tadpolegluinggeneral} in the appendix.  

%% file: TVdefinition.tex
\newcommand*{\pTV}{\mathbb{TV}}
\newcommand*{\TV}{\mathsf{TV}}
\newcommand*{\CKK}{\mathsf{CKK}}
\newcommand*{\MCG}{\mathsf{MCG}}
\input{newdefstvinvariant}

\section{The Turaev-Viro invariant and the Turaev-Viro code} \label{sec:tvinvariantcode}

In \secref{sec:levinwengeneralized}, we have discussed how the ribbon graph Hilbert space~$\cH_\Sigma$ can be realized as the ground space of the Levin-Wen qudit lattice Hamiltonian.  Here we  give a more direct formulation that justifies calling this ground space the Turaev-Viro code.  This approach historically preceeds the Hamiltonian construction and explains its origin.  It is the extension of the Turaev-Viro invariant of $3$-manifold to a TQFT~\cite{TuraevViro92topology}.   

We proceed as follows.  In \secref{sec:turaevvirodefinition}, we define the Turaev-Viro invariant~\cite{TuraevViro92topology} in the general category-theoretic formulation due to Barrett and Westbury~\cite{BarrettWestbury96invariants}.  \secref{sec:turaevvirocode} defines the Turaev-Viro code for a surface~$\Sigma$, based on the Turaev-Viro invariant for the thickened surface $\Sigma \times [-1,1]$.  \secref{sec:localstabilizers} shows that this definition corresponds with the Levin-Wen Hamiltonian's ground space, by deriving the Levin-Wen Hamiltonian; applying a plaquette operator corresponds to attaching a ``blister" to the surface.  

\appref{sec:crane} gives an alternative definition of the Turaev-Viro invariant, due to Kontsevich, Crane and Kohno~\cite{Kontsevich88, Crane91, Kohno92}, and sketches how the anyonic fusion basis state decomposition of~$\cH_\Sigma$ can be used to argue that the different definitions are equivalent.  This background is useful for the proof in~\cite{AlagicJordanKoenigReichardt10TuraevViro} that approximating the Turaev-Viro invariant is a BQP-complete problem.

\subsection{Definition of the Turaev-Viro invariant} \label{sec:turaevvirodefinition}

In this section, we present the definition of the Turaev-Viro invariant~$\TV_\cC(M)$ of an oriented $3$-manifold~$M$.  For simplicity, we assume that the $F$ tensor of the category~$\cC$ satisfies tetrahedral symmetry, though this is not necessary (see~\cite{BarrettWestbury96invariants}).  The manifold~$M$ may also have a boundary, in which case the Turaev-Viro invariant~$\TV_\cC(M, \chi)$ depends on~$\chi$, a labeling by anyon types of the oriented edges of a triangulation of the boundary.  

Triangulate~$M$ compatibly with the possible boundary triangulation.  For each tetrahedron~$t$, place a total order $<_t$ on its four vertices, such that the orders are consistent across tetrahedra, i.e., if~$t$ and $t'$ share an edge $\{v, w\}$, $v<_t w$ when $v <_{t'} w$.  For example, such orderings may be obtained from a single ordering of all the vertices in the triangulation.  

For a labeling~$\phi$ of the edges of the triangulation by anyon types from~$\cC$, and for a tetrahedron~$t$, define a scalar~$g_t^\phi$ as follows.  First, we consider the ``standard'' tetrahedron, with~vertices ordered $v_0 < \cdots < v_3$ and labels $i, j, k, \ell, m, n$ of the six edges as shown in \figref{f:stndtetrahedron}.  Here we have oriented the edges according to the vertex orders.  Define the tensor 
\begin{equation} \label{eq:tvdefinitionFmatrix}
\left|
\begin{matrix}
i & j & m\\
k & \ell & n
\end{matrix}
\right|
:= 
\frac{F^{i^* j m}_{k \ell^* n}}{\sqrt{d_m d_n}}
 \enspace .
\end{equation}
Now for a tetrahedron~$t$, the value~$g_t^\phi$ is obtained from this tensor by arbitrarily aligning~$t$ to the standard tetrahedron, and then conjugating a label if the orientation of the edge does not agree with the standard orientation.  For example, if~$t$ is the tetrahedron of \figref{f:nonstandardtetrahedronfirst},  
then 
\begin{equation}
g_t^\phi = \left|\begin{matrix} i & j & m \\ k^* & \ell^* & n^* \end{matrix}\right|
 \enspace .
\end{equation}
The tetrahedral symmetry of~$F$ implies that $g_t^\phi$ is well-defined.  

\begin{figure*}
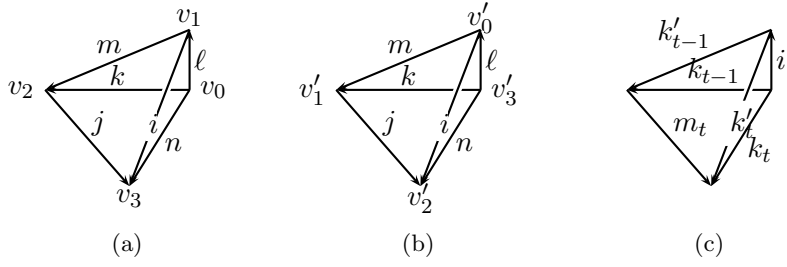

\centering
\begin{tabular}{c@{$\qquad\qquad$}c@{$\qquad\qquad $}c}
\subfigure[]{\label{f:stndtetrahedron}\raisebox{40pt}{\standardtetrahedronfirst}}
&
\subfigure[]{\label{f:nonstandardtetrahedronfirst}\raisebox{40pt}{\nonstandardtetrahedronfirst}}
&
\subfigure[]{\label{f:standardtetrahedronthird}\raisebox{40pt}{\standardtetrahedronthird}}
\end{tabular}
\caption{A standard tetrahedron labeling and edge orientations (a), and two other labelings.} 
\end{figure*}

Then set 
\begin{equation} \label{eq:tvdefinition}
\TV_\cC(M, \chi) = \cD^{-2 \abs{V_M} + \abs{V_{\partial M}}} \sum_{\substack{\textrm{labelings }\phi: \\
\phi|_{\partial M} = \chi}} \prod_{\substack{\text{edges}\\e \in M \backslash \partial M}} d_{\phi(e)} \prod_{\substack{\text{edges}\\e \in \partial M}} \sqrt{d_{\phi(e)}} \prod_{\text{tetrahedra $t$}} g_t^\phi
 \enspace ,
\end{equation}
where $|V_M|$ is the number of vertices in the triangulation, and $|V_{\partial M}|$ is the number of vertices on the boundary.  The sum is over all labelings~$\phi$ that agree with~$\chi$ on the boundary after conjugating labels so edge orientations on the boundary triangulation agree with the orientations from the full triangulation.  It can be shown that the quantity~$\TV_\cC(M, \chi)$ does not depend on the choice of triangulation in the interior of~$M$.

\subsection{The Turaev-Viro code} \label{sec:turaevvirocode}

Let~$\Sigma$ be a triangulated two-dimensional surface.  For simplicity, assume that $\Sigma$ lacks a boundary.  Consider the $3$-manifold~$M = \Sigma \times [-1,1]$.  Triangulate its boundary~$\partial M = \Sigma \times \{-1, 1\}$ with two copies of the triangulation of~$\Sigma$.  Then a labeling~$\chi$ of this triangulation consists of a pair $(\chi_+, \chi_-)$ of labelings of the triangulation of~$\Sigma$.  In particular, we can interpret the quantity~$\TV_\cC(M, \chi)$ as the matrix element of an operator 
\begin{equation} \label{eq:TVcodeprojection}
\pTV_\cC(\Sigma \times [-1,1]) = \sum_{\chi = (\chi_+, \chi_-)} \TV_\cC(M, \chi) \ket{\chi_+}\bra{\chi_-}
 \enspace .
\end{equation}
This operator acts on the Hilbert space~$(\mathbb{C}^{m})^{\otimes |E|}$ of labelings of the $|E|$~edges of the triangulation of~$\Sigma$ by the $m$~anyon labels of~$\cC$.  It only depends on the boundary triangulation, and not on a triangulation of~$M$.  

Define the Turaev-Viro code on~$\Sigma$ as the range of $\pTV_\cC(\Sigma \times [-1,1])$.  Since gluing together two copies of $\Sigma \times [-1,1]$ along a boundary component gives again $\Sigma \times [-1,1]$, and since the gluing of manifolds corresponds to the concatenation of the linear maps from Eq.~\eqnref{eq:TVcodeprojection}, $\pTV_\cC(\Sigma \times [-1,1])$ is a projection.  

We will next show that the Turaev-Viro code is the same as~$\cH_\Sigma$, i.e., the Hilbert space of equivalence classes of colored ribbon graphs embedded in~$\Sigma$.

\subsection{Local stabilizers for the Turaev-Viro code: The Levin-Wen Hamiltonian} \label{sec:localstabilizers}

Consider a vertex~$p$, with degree~$r$ in the triangulation, as in \figref{f:pyramidbase}.  Place a ``blister" put on top of this configuration by introducing a new vertex~$q$, with the new edges on the boundary labeled $\vec{k}' = (k_1', \ldots, k_r')$, as shown in \figref{f:pyramidcomplete} (the edge from $p$ to $q$ is oriented upwards).  The blister is a triangulation of a ball~$B$ with boundary edges labeled $(\vec{k}, \vec{k}', \vec{m})$.  Define 
\begin{equation} \label{eq:bpblisterdefinition}
B_p := \sum_{\vec{k}, \vec{k}', \vec{m}}\frac{\cD^r}{\sqrt{\prod_t d_{m_{t}}}}\TV_\cC(B,(\vec{k}, \vec{k}', \vec{m})) \ketbra{\vec {k'}}{\vec k} \otimes \ketbra{\vec m}{\vec m}
 \enspace ,
\end{equation}
a linear map on the space~$(\mathbb{C}^m)^{\otimes 2r}$ of labelings $(\vec k, \vec m)$.  Similarly as before, since gluing two balls together results in another ball, $B_p$ is a projection.  The additional factors compared to Eq.~\eqnref{eq:TVcodeprojection} are necessary because the outer edges with labels~$\vec{m}$ remain on the boundary when gluing, and the number of vertices on the boundary stays the same.  

\begin{figure}
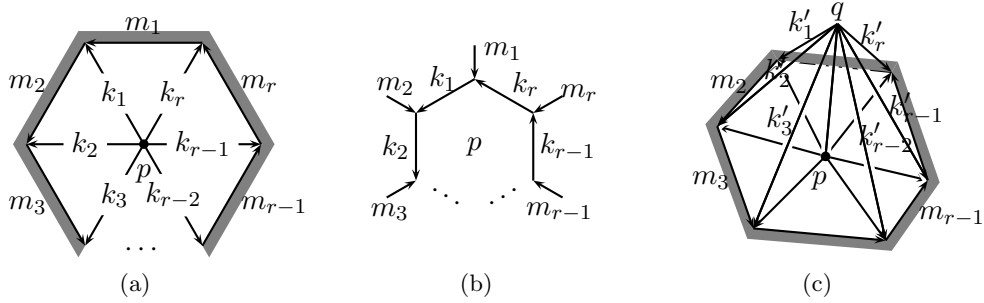

\centering
\begin{tabular}{c@{$\qquad\qquad$}c@{$\qquad\qquad$}c}
\subfigure[]{\label{f:pyramidbase}\raisebox{.2in}{\orientedpneighborhood}}
&
\subfigure[]{\label{f:orientedlabeledplaquette}\raisebox{.65in}{\orientedlabeledplaquette{k}{m}}}
&
\subfigure[]{\label{f:pyramidcomplete}\raisebox{.2in}{\secondorientedpblister}}
\end{tabular}
\caption{The neighborhood of a vertex~$p$ in a triangulated surface can be oriented as in~(a), in turn orienting the trivalent dual graph in~(b).  Attaching a ``blister" above~$p$ gives~(c) when $r = 6$.} 
\end{figure}

For different vertices $p$ and $p'$, the operators $B_p$ and $B_{p'}$, acting on $(\mathbb{C}^m)^{\otimes \abs E}$, commute.  This is clear if $p$ and $p'$ are nonadjacent.  When $p$ is adjacent to $p'$, both products $B_{p'} B_p$ and $B_p B_{p'}$ correspond to attaching the same bipyramid to~$\Sigma$, only with different internal triangulations, as shown in \figref{f:blisteroperatorscommute} in the case that $p$ and $p'$ are connected by a single edge.  Since the Turaev-Viro invariant does not depend on the internal triangulation, we conclude $B_p B_{p'} = B_{p'} B_p$. 

\begin{figure}
\centering
\begin{tabular}{c@{$\qquad\qquad$}c}
\subfigure[]{\raisebox{0.1in}{
\begin{pspicture}[shift=0](0,0)(4.3,2.5)
\put(0,0.15){\scalebox{0.8}{\protect\epsfig{file=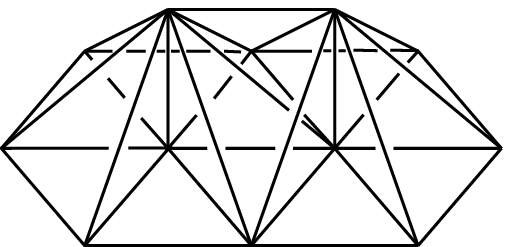}}}
\put(1.25,0.5){$p$}
\put(2.5,0.5){$p'$}
\end{pspicture}}}
&
\subfigure[]{\raisebox{0.1in}{
\begin{pspicture}[shift=0](0,0)(4.3,2.5)
\put(0,0.15){\scalebox{0.8}{\protect\epsfig{file=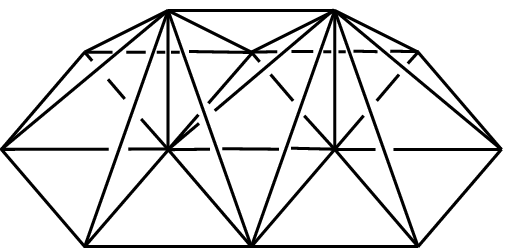}}}
\put(1.25,0.5){$p$}
\put(2.5,0.5){$p'$}
\end{pspicture}}}
\end{tabular}
\caption{For two adjacent vertices $p$ and $p'$, the products $B_{p'} B_p$ and $B_p B_{p'}$ correspond to different internal triangulations of the same bipyramid.  Hence $B_p$ and $B_{p'}$ commute.}  
\label{f:blisteroperatorscommute}
\end{figure}

Similarly, the product~$\prod_p B_p$ of over all vertices of the triangulation corresponds to a ``thickening'' of $\Sigma$, and therefore 
\begin{equation}
\prod_p B_p = \pTV(\Sigma \times [-1,1])
 \enspace .
\end{equation}
Hence the Turaev-Viro code equals the simultaneous $+1$-eigenspace of the $B_p$ projections.  

In fact, these $B_p$ operators are the same as the $B_p$ plaquette operators in the definition of the Levin-Wen Hamiltonian~\eqnref{eq:levinwenhamiltonian}.  To see this, evaluate the invariant~$\TV_\cC(B,(\vec{k}, \vec{k}', \vec{m}))$ associated with \figref{f:pyramidcomplete} explicitly.  With fixed boundary labels $(\vec k, \vec k', \vec m)$, the only free index is the label~$i$ of the internal edge~$\{p,q\}$.  A typical tetrahedron has the form shown in \figref{f:standardtetrahedronthird}, whence 
\begin{equation}
g_t^\phi = \frac{F^{(k'_t)^*m_t k'_{t-1}}_{k_{t-1}i^*k_t}}{\sqrt{d_{k'_{t-1}} d_{k_t}}} = \frac{F^{m_t k_t^*k_{t-1}}_{ik'_{t-1}(k'_t)^*}}{\sqrt{d_{k_{t-1}}d_{k'_t}}}
 \enspace .
\end{equation}
Substituting into Eqs.~\eqnref{eq:tvdefinition} and~\eqnref{eq:bpblisterdefinition}, we obtain 
\begin{equation} \label{eq:derivedplaquette}
B_p = \frac{1}{\cD^2}\sum_{\vec{k}, \vec{k}', \vec{m}} \sum_i d_i\left(\prod_{t=1}^r F^{m_t k_t^*k_{t-1}}_{ik_{t-1}'(k'_t)^*}\right)\ket{\vec{k}', \vec{m}}\bra{\vec{k}, \vec{m}}
 \enspace .
\end{equation}
Considering the dual graph, in \figref{f:orientedlabeledplaquette}, we recognize in Eq.~\eqnref{eq:derivedplaquette} the definition~\eqnref{eq:bpidefac} of the Levin-Wen Hamiltonian's plaquette operators.

%% file: newdefstvinvariant.tex
\newcommand*{\doubletetrahedron}{
\begin{pspicture}[shift=-.875](-1,-.1)(1.2,2.7)
\psset{unit=1.6cm} 
\psset{linewidth=.8pt} 
\psset{labelsep=1pt} 
\SpecialCoor
\pnode(0,0){v3}
\pnode(1.5,1.1){v0}
\pnode(.5,1.8){v1}
\pnode(-.7,.8){v2}
\pnode(1.4,0.4){w}
\pspolygon*[fillstyle=solid,linecolor=lightgray,fillcolor=lightgray](v1)(v2)(v3)\pspolygon*[fillstyle=solid,linecolor=lightgray,fillcolor=lightgray](v0)(v1)(w)
\pcline{->}(v0)(v2)
  \Bput{$\ \ \ k$}
\psline[linestyle=solid,linewidth=3pt, linecolor=white](v1)(v3)
\pcline{->}(v1)(v2) 
  \Bput{$m$}
\pcline{->}(v1)(v3)
\Bput{$i$}
\pcline{->}(v0)(v3)
\Aput{$n$}
\pcline{->}(v2)(v3)
\Aput{$j$}
\psline[linestyle=solid,linewidth=3pt, linecolor=white](v1)(w)
\pcline{->}(v1)(w)
\Aput{$q$}
\pcline{->}(v3)(w)
\Bput{$r$}
\pcline{->}(v0)(w)
\Aput{$s$}
\pcline{->}(v0)(v1) 
  \Bput{$\ell$}
\end{pspicture}}
\newcommand*{\standardtetrahedronfirst}{
\begin{pspicture}[shift=-.875](-1,-.1)(1.2,2.1)
\psset{unit=1.6cm} 
\psset{linewidth=.8pt} 
\psset{labelsep=1pt} 
\SpecialCoor
\pnode(0,0){v3}
  \rput(0,-0.1){$v_3$}
\pnode(.5,.8){v0}
  \rput(.7,.8){$v_0$}
\pnode(.5,1.3){v1}
  \rput(.5,1.4){$v_1$}
\pnode(-.7,.8){v2}
  \rput(-.9,.8){$v_2$}
\pcline{->}(v0)(v2)
  \Bput{$k$}
\psline[linestyle=solid,linewidth=3pt, linecolor=white](v1)(v3)
\pcline{->}(v1)(v2) 
  \Bput{$m$}
\pcline{->}(v1)(v3)
  \lput*(.6){$\! i\!\!$}
\pcline{->}(v0)(v1) 
  \Bput{$\ell$}
\pcline{->}(v0)(v3)
\Aput{$n$}
\pcline{->}(v2)(v3)
\Aput{$j$}
\end{pspicture}}

\newcommand*{\standardtetrahedronsecond}{
\begin{pspicture}[shift=-.875](-1,-.1)(1.2,2.1)
\psset{unit=1.6cm} 
\psset{linewidth=.8pt} 
\psset{labelsep=1pt} 
\SpecialCoor
\pnode(0,0){v3}
\pnode(.5,.8){v0}
\pnode(.5,1.3){v1}
\pnode(-.7,.8){v2}
\pcline{->}(v0)(v2)
  \Bput{$k$}
\psline[linestyle=solid,linewidth=3pt, linecolor=white](v1)(v3)
\pcline{->}(v1)(v2) 
  \Bput{$m$}
\pcline{->}(v1)(v3)
  \lput*(.6){$\! i\!\!$}
\pcline{->}(v0)(v1) 
  \Bput{$\ell$}
\pcline{->}(v0)(v3)
\Aput{$n$}
\pcline{->}(v2)(v3)
\Aput{$j$}
\end{pspicture}}

\newcommand*{\standardtetrahedronthird}{
\begin{pspicture}[shift=-.875](-1,-.1)(1.2,2.1)
\psset{unit=1.6cm} 
\psset{linewidth=.8pt} 
\psset{labelsep=1pt} 
\SpecialCoor
\pnode(0,0){v3}
\pnode(.5,.8){v0}
\pnode(.5,1.3){v1}
\pnode(-.7,.8){v2}
\pcline{->}(v0)(v2)
  \Bput{$\ \  \ k_{t-1}$}
\psline[linestyle=solid,linewidth=3pt, linecolor=white](v1)(v3)
\pcline{->}(v1)(v2) 
  \Bput{$k_{t-1}'$}
\pcline{->}(v1)(v3)
  \lput*(.6){$k_t'\!\!\!$}
\pcline{->}(v0)(v1) 
  \Bput{$i$}
\pcline{->}(v0)(v3)
\Aput{$k_t$}
\pcline{->}(v2)(v3)
\Aput{$m_t$}
\end{pspicture}}

\newcommand*{\nonstandardtetrahedronfirst}{
\begin{pspicture}[shift=-.875](-1,-.1)(1.2,2.1)
\psset{unit=1.6cm} 
\psset{linewidth=.8pt} 
\psset{labelsep=1pt} 
\SpecialCoor
\pnode(0,0){v3}
  \rput(0,-0.1){$v'_2$}
\pnode(.5,.8){v0}
  \rput(.7,.8){$v'_3$}
\pnode(.5,1.3){v1}
  \rput(.5,1.4){$v'_0$}
\pnode(-.7,.8){v2}
  \rput(-.9,.8){$v'_1$}
\pcline{->}(v0)(v2)
  \Bput{$k$}
\psline[linestyle=solid,linewidth=3pt, linecolor=white](v1)(v3)
\pcline{->}(v1)(v2) 
  \Bput{$m$}
\pcline{->}(v1)(v3)
  \lput*(.6){$\! i\!\!$}
\pcline{->}(v0)(v1) 
  \Bput{$\ell$}
\pcline{->}(v0)(v3)
\Aput{$n$}
\pcline{->}(v2)(v3)
\Aput{$j$}
\end{pspicture}}

\newcommand*{\nonstandardtetrahedronsecond}{
\begin{pspicture}[shift=-.875](-1,-.1)(1.2,2.1)
\psset{unit=1.6cm} 
\psset{linewidth=.8pt} 
\psset{labelsep=1pt} 
\SpecialCoor
\pnode(0,0){v3}
\pnode(.5,.8){v0}
\pnode(.5,1.3){v1}
\pnode(-.7,.8){v2}
\pcline{->}(v2)(v0)
  \Aput{$k$}
\psline[linestyle=solid,linewidth=3pt, linecolor=white](v1)(v3)
\pcline{->}(v1)(v2) 
  \Bput{$m$}
\pcline{->}(v1)(v3)
  \lput*(.6){$\! i\!\!$}
\pcline{->}(v1)(v0)
  \Aput{$\ell$}
\pcline{->}(v3)(v0)
\Bput{$n$}
\pcline{->}(v2)(v3)
\Aput{$j$}
\end{pspicture}}

\newcommand*{\standardtetrahedron}{
\begin{pspicture}[shift=-.875](-1,-.1)(1.2,2.1)
\psset{unit=1.6cm} 
\psset{linewidth=.8pt} 
\psset{labelsep=1pt} 
\SpecialCoor
\pnode(0,0){v3}
  \rput(0,-0.1){$v_3$}
\pnode(.5,.8){v0}
  \rput(.7,.8){$v_0$}
\pnode(.5,1.3){v1}
  \rput(.5,1.4){$v_1$}
\pnode(-.7,.8){v2}
  \rput(-.9,.8){$v_2$}
\pcline{->}(v0)(v2)
  \Bput{$k$}
\psline[linestyle=solid,linewidth=3pt, linecolor=white](v1)(v3)
\pcline{->}(v1)(v2) 
  \Bput{$m$}
\pcline{->}(v1)(v3)
  \lput*(.6){$\! i\!\!$}
\pcline{->}(v0)(v1) 
  \Bput{$\ell$}
\pcline{->}(v0)(v3)
\Aput{$n$}
\pcline{->}(v2)(v3)
\Aput{$j$}
\end{pspicture}}

\newcommand*{\tetrahedrexample}{
\begin{pspicture}[shift=-.875](-1,-.1)(1.2,2.1)
\psset{unit=1.6cm} 
\psset{linewidth=.8pt} 
\psset{labelsep=1pt} 
\SpecialCoor
\pnode(0,0){v3}
\pnode(.5,.8){v0}
\pnode(.5,1.3){v1}
\pnode(-.7,.8){v2}
\pcline{->}(v2)(v3)
\Aput{$m_{\nu}$}
\pcline{->}(v0)(v2)
  \Bput{$\!\! j_{\nu-1}\!\!$}
\pcline{->}(v1)(v3)
  \lput*(.6){$k_{\nu}\!\!$}
\pcline{->}(v0)(v1) 
  \Bput{$i$}
\pcline{->}(v1)(v2) 
  \Bput{$k_{\nu-1}$}
\pcline{->}(v0)(v3)
\Aput{$j_{\nu}$}
\end{pspicture}}

\newcommand*{\labeledplaquette}[2]{
\begin{pspicture}[shift=-.875](-1.1,-.9)(1.4,1.2)
\psset{unit=.7cm} 
\psset{linewidth=.8pt} 
\psset{labelsep=2.5pt} 
\SpecialCoor
\pnode(0,0){O}
\pnode(1;150){v1}
\pnode(1;90){v2}
\pnode(1;30){v3}
\pnode(1;-30){v4}
\pnode(1;-90){v5}
\pnode(1;-150){v6}
\psline{->}(v4)(v3)
\psline{->}(v3)(v2)
\psline{->}(v2)(v1)
\psline{->}(v1)(v6)
\def\leglength{.5}
\def\rthree{.866025}
\psline{<-}(v1)([nodesep=\leglength,angle=150]v1)
\psline{<-}(v2)([nodesep=\leglength,angle=90]v2)
\psline{<-}(v3)([nodesep=\leglength,angle=30]v3)
\psline{<-}(v4)([nodesep=\leglength,angle=-30]v4)
\psline{<-}(v6)([nodesep=\leglength,angle=-150]v6)
\rput(O){$p$}
\rput[b]{-30}([nodesep=\rthree,angle=-120]O){$\ldots$}
\rput[b]{30}([nodesep=\rthree,angle=-60]O){$\ldots$}
\rput[b]([nodesep=.5,angle=200]v2){${#1}_1$}
\uput[180]([nodesep=\rthree,angle=180]O){${#1}_2$}
\uput[0]([nodesep=\rthree,angle=0]O){${#1}_{r-1}$}
\uput[110]([nodesep=.15,angle=90]v3){${#1}_r$}
\rput[b]([nodesep=1.35,angle=72]O){${#2}_1$}
\uput[90]([nodesep=1.35,angle=150]O){${#2}_2$}
\rput([nodesep=.6,angle=-130]v6){${#2}_3$}
\rput([nodesep=.6,angle=-50]v4){${#2}_{r-1}$}
\uput[30]([nodesep=1.35,angle=30]O){${#2}_r$}
\end{pspicture}}

\newcommand*{\pneighborhood}{
\begin{pspicture}[shift=-0.2](-1.3,-1.1)(1.3,1.2)
\psset{unit=.7cm} 
\psset{linewidth=.8pt} 
\psset{labelsep=2.5pt} 
\SpecialCoor
\def\sqrt3{1.73205}
\pnode(0,0){O}
\pnode(\sqrt3;-60){p1}
\pnode(\sqrt3;0){p2}
\pnode(\sqrt3;60){p3}
\pnode(\sqrt3;120){p4}
\pnode(\sqrt3;180){p5}
\pnode(\sqrt3;240){p6}

\def\shadingsep{.2}
\pnode([nodesep=\shadingsep,angle=-60]p1){p1p}
\pnode([nodesep=\shadingsep,angle=0]p2){p2p}
\pnode([nodesep=\shadingsep,angle=60]p3){p3p}
\pnode([nodesep=\shadingsep,angle=120]p4){p4p}
\pnode([nodesep=\shadingsep,angle=180]p5){p5p}
\pnode([nodesep=\shadingsep,angle=240]p6){p6p}
\pspolygon[linestyle=none,fillstyle=solid,fillcolor=gray](p1p)(p2p)(p3p)(p4p)(p5p)(p6p)(p6)(p5)(p4)(p3)(p2)(p1)

\psdots(O)
\psline(p1)(p2)(p3)(p4)(p5)(p6)
\psline(p1)(O)(p2) \psline(p3)(O)(p4) \psline(p5)(O)(p6)

\uput{.3}[-90](O){$p$}
\rput([nodesep=1.55,angle=-90]O){$\cdots$}
\end{pspicture}
}

\newcommand*{\pblister}{
\begin{pspicture}[shift=-0.2](-1.2,-1.1)(1.4,1.2)
\psset{unit=1.2cm} 
\psset{linewidth=.8pt} 
\psset{labelsep=2.5pt} 
\SpecialCoor
\pnode(0,-.1){p}
\pnode(.1,1){q}
\pnode(.5,-.8){v1}
\pnode(.8,.2){v2}
\pnode(-.4,.7){v3}
\pnode(-.9,.15){v4}
\pnode(-.6,-.7){v5}

\pnode([nodesep=.1,angle=-60]v1){v1p}
\pnode([nodesep=.1,angle=30]v2){v2p}
\pnode([nodesep=.1,angle=120]v3){v3p}
\pnode([nodesep=.1,angle=170]v4){v4p}
\pnode([nodesep=.1,angle=-120]v5){v5p}
\pspolygon[linestyle=none,fillstyle=solid,fillcolor=gray](v1p)(v2p)(v3p)(v4p)(v4)(v3)(v2)(v1)

\psdots(p)

\psline(p)(v2) \psline(p)(v3) \psline(p)(v4)
\psline(v2)(v3)
\psset{border=2pt}
\psline(v1)(q) \psline(v4)(q) \psline(v5)(q)
\psset{border=0pt}
\pspolygon[linestyle=none,fillstyle=solid,fillcolor=gray](v3p)(v4p)(v5p)(v1p)(v2p)(v2)(v1)(v5)(v4)(v3)
\psline(v1)(q) \psline(v4)(q) \psline(v5)(q)
\psline(v2)(q)(v3)
\psline(v5)(p)(v1)
\psline(v3)(v4)(v5)(v1)(v2)
\psline(p)(q)
\rput(-0.05,-0.4){$p$}
\end{pspicture}
}

\newcommand*{\orientedlabeledplaquette}[2]{
\begin{pspicture}[shift=-.875](-1.1,-.9)(1.4,1.2)
\psset{unit=.9cm} 
\psset{linewidth=.8pt} 
\psset{labelsep=2.5pt} 
\SpecialCoor
\pnode(0,0){O}
\pnode(1;150){v1}
\pnode(1;90){v2}
\pnode(1;30){v3}
\pnode(1;-30){v4}
\pnode(1;-90){v5}
\pnode(1;-150){v6}
\psline{->}(v4)(v3)
\psline{->}(v3)(v2)
\psline{->}(v2)(v1)
\psline{->}(v1)(v6)
\def\leglength{0.5}
\def\rthree{.866025}
\psline{->}([nodesep=\leglength,angle=150]v1)(v1)
\psline{->}([nodesep=\leglength,angle=90]v2)(v2)
\psline{->}([nodesep=\leglength,angle=30]v3)(v3)
\psline{->}([nodesep=\leglength,angle=-30]v4)(v4)
\psline{->}([nodesep=\leglength,angle=-150]v6)(v6)
\rput(O){$p$}
\rput[b]{-30}([nodesep=\rthree,angle=-120]O){$\ldots$}
\rput[b]{30}([nodesep=\rthree,angle=-60]O){$\ldots$}
\rput[b]([nodesep=.5,angle=200]v2){${#1}_1$}
\uput[180]([nodesep=\rthree,angle=180]O){${#1}_2$}
\uput[0]([nodesep=\rthree,angle=0]O){${#1}_{r-1}$}
\uput[110]([nodesep=.15,angle=90]v3){${#1}_r$}
\rput[b]([nodesep=1.35,angle=72]O){${#2}_1$}
\uput[90]([nodesep=1.35,angle=150]O){${#2}_2$}
\rput([nodesep=.6,angle=-130]v6){${#2}_3$}
\rput([nodesep=.6,angle=-50]v4){${#2}_{r-1}$}
\uput[30]([nodesep=1.35,angle=30]O){${#2}_r$}
\end{pspicture}}

\newcommand*{\orientedpneighborhood}{
\begin{pspicture}[shift=-0.2](-1.6,-1.4)(1.5,1.6)
\psset{unit=.9cm} 
\psset{linewidth=.8pt} 
\psset{labelsep=2.5pt} 
\SpecialCoor
\def\sqrt3{1.73205}
\def\outerlabl{2.13205}
\pnode(0,0){O}
\pnode(\sqrt3;-60){p1}
\pnode(\sqrt3;0){p2}
\pnode(\sqrt3;60){p3}
\pnode(\sqrt3;120){p4}
\pnode(\sqrt3;180){p5}
\pnode(\sqrt3;240){p6}
\def\shadingsep{.2}
\pnode([nodesep=\shadingsep,angle=-60]p1){p1p}
\pnode([nodesep=\shadingsep,angle=0]p2){p2p}
\pnode([nodesep=\shadingsep,angle=60]p3){p3p}
\pnode([nodesep=\shadingsep,angle=120]p4){p4p}
\pnode([nodesep=\shadingsep,angle=180]p5){p5p}
\pnode([nodesep=\shadingsep,angle=240]p6){p6p}
\pspolygon[linestyle=none,fillstyle=solid,fillcolor=gray](p1p)(p2p)(p3p)(p4p)(p5p)(p6p)(p6)(p5)(p4)(p3)(p2)(p1)

\psdots(O)
\pcline{->}(p1)(p2)
  \Bput{$m_{r-1}$}
\pcline{->}(p2)(p3)
   \Bput{$m_r$}
\pcline{->}(p3)(p4)
   \Bput{$m_1$}
\pcline{->}(p4)(p5)
   \Bput{$m_2$}
\pcline{->}(p5)(p6)
   \Bput{$m_3$}
\pcline{->}(O)(p1)
   \lput*(.5){$k_{r-2}$}
\pcline{->}(O)(p2)
   \lput*(.5){$k_{r-1}$}
\pcline{->}(O)(p3)
   \lput*(.5){$k_r$}
\pcline{->}(O)(p4)
  \lput*(.5){$k_1$}
\pcline{->}(O)(p5)
   \lput*(.5){$k_2$}
\pcline{->}(O)(p6)
   \lput*(.5){$k_3$}

\uput{.3}[-90](O){$p$}
\rput([nodesep=1.55,angle=-90]O){$\cdots$}
\end{pspicture}
}

\newcommand*{\orientedpblister}{
\begin{pspicture}[shift=-0.2](-1.6,-1.4)(1.5,1.6)
\psset{unit=1.6cm} 
\psset{linewidth=.8pt} 
\psset{labelsep=2.5pt} 
\SpecialCoor

\pnode(0,-.1){p}
\pnode(.1,1){q}
\pnode(.5,-.8){w1}
\pnode(.85,-.3){w2}
\pnode(.55,.6){w3}
\pnode(-0.4,.7){w4}
\pnode(-.9,.15){w5}
\pnode(-.6,-.7){w6}
\pnode([nodesep=.1,angle=-60]w1){w1p}
\pnode([nodesep=.1,angle=-10]w2){w2p}
\pnode([nodesep=.1,angle=60]w3){w3p}
\pnode([nodesep=.1,angle=120]w4){w4p}
\pnode([nodesep=.1,angle=170]w5){w5p}
\pnode([nodesep=.1,angle=-120]w6){w6p}
\psdots(p)
\psline{->}(w3)(w4)
\psline{->}(p)(w3)
\psline{->}(p)(w2) \psline{->}(p)(w4) \psline{->}(p)(w5)

\psset{border=1.5pt}
\psline(w1)(q) \psline(w5)(q) \psline(w6)(q)
\psline(w2)(q)
\psline(p)(q)
\psset{border=0pt}
\pspolygon[linestyle=none,fillstyle=solid,fillcolor=gray](w4p)(w5p)(w6p)(w1p)(w2p)(w3p)(w3)(w2)(w1)(w6)(w5)(w4)
\pspolygon[linestyle=none,fillstyle=solid,fillcolor=gray](w4p)(w4)(w3)(w3p)
\psline{->}(q)(w1) \psline{->}(q)(w5) \psline{->}(q)(w6)
\psline{->}(q)(w2)
\psline{->}(q)(w4)

\psline{->}(p)(w6)
\psline{->}(p)(w1)

\psline{->}(w4)(w5)
\psline{->}(w5)(w6)
\psline{->}(w6)(w1)
\psline{->}(w1)(w2)
\psline{->}(p)(q)
\psline{->}(q)(w3)\psline{->}(w2)(w3)
\psdots(p)
\rput(-0.05,-0.3){$p$}
\rput(.1,1.1){$q$}
\end{pspicture}
}

\newcommand*{\secondorientedpblister}{
\begin{pspicture}[shift=-0.2](-1.6,-1.4)(1.5,1.6)
\psset{unit=1.6cm} 
\psset{linewidth=.8pt} 
\psset{labelsep=2.5pt} 
\SpecialCoor

\pnode(0,-.1){p}
\pnode(.1,1){q}
\pnode(.5,-.8){w1}
\pnode(.85,-.3){w2}
\pnode(.55,.6){w3}
\pnode(-0.4,.7){w4}
\pnode(-.9,.15){w5}
\pnode(-.6,-.7){w6}
\pnode([nodesep=.1,angle=-60]w1){w1p}
\pnode([nodesep=.1,angle=-10]w2){w2p}
\pnode([nodesep=.1,angle=60]w3){w3p}
\pnode([nodesep=.1,angle=120]w4){w4p}
\pnode([nodesep=.1,angle=170]w5){w5p}
\pnode([nodesep=.1,angle=-120]w6){w6p}
\psdots(p)
\psline{->}(w3)(w4)
\psline{->}(p)(w3)
\psline{->}(p)(w2) \psline{->}(p)(w4) \psline{->}(p)(w5)

\psset{border=1.5pt}
\psline(w1)(q) \psline(w5)(q) \psline(w6)(q)
\psline(w2)(q)
\psline{->}(p)(q)
\psset{border=0pt}
\pspolygon[linestyle=none,fillstyle=solid,fillcolor=gray](w4p)(w5p)(w6p)(w1p)(w2p)(w3p)(w3)(w2)(w1)(w6)(w5)(w4)
\pspolygon[linestyle=none,fillstyle=solid,fillcolor=gray](w4p)(w4)(w3)(w3p)
\psline{->}(q)(w1) \psline{->}(q)(w5) \psline{->}(q)(w6)
\psline{->}(q)(w2)
\psline{->}(q)(w4)

\psline{->}(p)(w6)
\psline{->}(p)(w1)

\psline{->}(w4)(w5)
\psline{->}(w5)(w6)
\psline{->}(w6)(w1)
\psline{->}(w1)(w2)
\psline{->}(p)(q)
\psline{->}(q)(w3)\psline{->}(w2)(w3)
\psdots(p)
\rput(-0.05,-0.3){$p$}
\rput(.1,1.1){$q$}
\rput(-0.2,1.0){$k_1'$}
\rput(.4,0.9){$k_r'$}
\rput(.75,0.3){$k_{r-1}'$}
\rput(-0.38,0.2){$k_3'$}
\rput(-0.4,0.6){$k_2'$}
\rput(.5,0.05){$k_{r-2}'$}
\rput(-0.8,0.5){$m_2$}
\rput(-0.95,-0.3){$m_3$}

\rput(1.05,-0.6){$m_{r-1}$}
\end{pspicture}
}

\newcommand*{\orientedtetrahedronlabeled}{
\begin{pspicture}[shift=-.875](-1,-.1)(1.2,2.1)
\psset{unit=1.6cm} 
\psset{linewidth=.8pt} 
\psset{labelsep=1pt} 
\SpecialCoor
\pnode(0,0){v1}
\pnode(.5,.8){v2}
\pnode(.5,1.3){v3}
\pnode(-.7,.8){v4}

\pcline{->}(v1)(v4)
\Aput{$m_\nu$}

\pcline{->}(v4)(v2)
\Bput{$j_{\nu-1}$}

\pcline{<-}(v1)(v3) \lput*(.6){$\!\!k_\nu^*\!\!$}
\pcline(v2)(v3) \Bput{$i^*$}
\pcline{->}(v3)(v4) \Bput{$k_{\nu-1}$}

\pcline{->}(v2)(v1)
\Aput{$j_\nu^*$}
\end{pspicture}
}
\newcommand*{\gluingsurfacepatchesone}{
\begin{pspicture}[shift=0](-2.4,-1.2)(2.1,1.7)
\psset{unit=1.2cm} 
\psset{linewidth=.8pt} 
\psset{labelsep=2.5pt} 
\SpecialCoor
\pnode(-2,.8){nw}
\pnode(-1.5,-1){sw}
\pnode(-.25,-.5){se}
\pnode(-.8,1.2){ne}
\pnode(-.9,.3){center}
\nccurve[angleA=-50,angleB=70,ncurv=.67]{nw}{sw}
\nccurve[angleA=30,angleB=-160,ncurv=.67]{sw}{se}
\nccurve[angleA=90,angleB=-50,ncurv=.67]{se}{ne}
\nccurve[angleA=180,angleB=60,ncurv=.67]{ne}{nw}
\rput{60}{\psellipse[linestyle=solid,linewidth=.4pt,fillstyle=solid,fillcolor=lightgray](center)(.22,.2)}

\pnode(2,.8){nw}
\pnode(1.5,-1){sw}
\pnode(.25,-.5){se}
\pnode(.8,1.2){ne}
\pnode(.9,.3){center}
\nccurve[angleA=-130,angleB=110,ncurv=.67]{nw}{sw}
\nccurve[angleA=150,angleB=-20,ncurv=.67]{sw}{se}
\nccurve[angleA=90,angleB=-130,ncurv=.67]{se}{ne}
\nccurve[angleA=0,angleB=120,ncurv=.67]{ne}{nw}
\rput{120}{\psellipse[linestyle=solid,linewidth=.4pt,fillstyle=solid,fillcolor=lightgray](center)(.22,.2)}
\end{pspicture}
}

\newcommand*{\gluingsurfacepatchestwo}{\begin{pspicture}[shift=0](-2.4,-1.2)(2.1,1.7)
\psset{unit=1.2cm} 
\psset{linewidth=.8pt} 
\psset{labelsep=2.5pt} 
\SpecialCoor
\pnode(-2,.8){nw}
\pnode(-1.5,-1){sw}
\pnode(-.25,-.5){se}
\pnode(-.8,1.2){ne}
\pnode(-.9,.3){center}
\nccurve[angleA=-50,angleB=70,ncurv=.67]{nw}{sw}
\nccurve[angleA=30,angleB=-160,ncurv=.67]{sw}{se}
\nccurve[angleA=90,angleB=-50,ncurv=.67]{se}{ne}
\nccurve[angleA=180,angleB=60,ncurv=.67]{ne}{nw}

\pnode(2,.8){nw}
\pnode(1.5,-1){sw}
\pnode(.25,-.5){se}
\pnode(.8,1.2){ne}
\pnode(.9,.3){center}
\nccurve[angleA=-130,angleB=110,ncurv=.67]{nw}{sw}
\nccurve[angleA=150,angleB=-20,ncurv=.67]{sw}{se}
\nccurve[angleA=90,angleB=-130,ncurv=.67]{se}{ne}
\nccurve[angleA=0,angleB=120,ncurv=.67]{ne}{nw}

\pscircle*[linecolor=white](-.37,.3){.2}
\pscircle*[linecolor=white](.37,.3){.2}
\pnode(-.95,.55){nw}
\pnode(.95,.55){ne}
\nccurve[angleA=-20,angleB=-160,ncurv=.5]{nw}{ne}
\pnode(-1,0){sw}
\pnode(1,0){se}
\nccurve[angleA=25,angleB=155,ncurv=.67]{sw}{se}
\end{pspicture}
}

\newcommand*{\gluingthreetetrahedra}{\def\puncture {\psdots[dotstyle=+,dotangle=45,dotscale=1](0,0)}
\def\dot {\psdots[dotscale=1](0,0)}
\begin{pspicture}[shift=0](-2.6,-1.6)(2.6,1.9)
\psset{unit=1.2cm} 
\psset{linewidth=.8pt} 
\psset{labelsep=2.5pt} 
\SpecialCoor

\pnode(-1,1){l1}
\pnode(-1.5,0){l2}
\pnode(-.5,-.5){l3}
\pnode(1,1){r1}
\pnode(1.5,0){r2}
\pnode(.5,-.5){r3}

\pcline(l2)(r2) \lput{:U}(.4){\puncture}

\pcline[border=.1](l3)(r1) \lput{:U}{\puncture}
\pcline(l3)(r2) \lput{:U}(.4){\puncture}

\pcline[border=.1](r1)(r3) \mput{\dot}
\pcline(r1)(r2) \mput{\dot}
\pcline(r2)(r3) \mput{\dot}

\pcline(l2)(r1) \lput{:U}{\puncture}

\pcline[border=.1](l1)(l3) \mput{\dot}
\pcline(l1)(l2) \mput{\dot}
\pcline(l2)(l3) \mput{\dot}

\pcline(l1)(r1) \mput{\puncture}
\pcline(l3)(r3) \mput{\puncture}

\pcline(l3)(r1)

\pnode(-2,1.5){nw}
\pnode(2,1.5){ne}
\nccurve[angleA=-20,angleB=-160,ncurv=.5]{nw}{ne}
\pnode(-2.2,-1.2){sw}
\pnode(2.2,-1.2){se}
\nccurve[angleA=25,angleB=155,ncurv=.67]{sw}{se}
\end{pspicture}}

\newcommand*{\newdot}{\begin{pspicture}(-0.2,-.2)(0.2,0.2)
\psset{unit=1cm} 
\psset{linewidth=.8pt} 
\psset{labelsep=2.5pt} 
\SpecialCoor
\psdots[dotscale=1](0,0)
\end{pspicture}
}

\newcommand*{\newpuncture}{\def\puncture {\psdots[dotstyle=+,dotangle=45,dotscale=1](0,0)}
\begin{pspicture}(-0.2,-0.2)(0.2,0.2)
\psset{unit=1cm} 
\psset{linewidth=.8pt} 
\psset{labelsep=2.5pt} 
\SpecialCoor
\puncture
\end{pspicture}}

%% file: anyongeneral.tex
\section{Ribbon graphs for general categories}

In this appendix, we collect various facts about ribbon graphs and the fusion basis states discussed in the main text.  In \appref{sec:basicpropertiesappendix}, we summarize a few general identities satisfied by ribbon graphs.  In \appref{sec:normalizationgeneral}, we give alternative definitions of the inner product on the ribbon graph Hilbert space~$\cH_\Sigma$, and use this to show that the fusion basis states are mutually orthogonal.  \appref{sec:fibonacciapplication} provides a few explicit examples of fusion basis states for the Fibonacci model.  We conclude with a short explanation of the generalization to higher-genus surfaces in \appref{sec:highergenuss}.

\subsection{Ribbon graph identities} \label{sec:basicpropertiesappendix}

Here we  state a few useful equivalence relations between ribbon graphs.  First, bubbles can be removed by the rule
\begin{equation} \label{eq:bubblerule}
\threeDdoubleringsinglenoboundary{m}{i}{j}{n} = \delta_{mn}\sqrt{\frac{d_id_j}{d_n}}\threeDdoubleringstraightline{n}
\end{equation}
Next we summarize some of the main properties of the vacuum lines defined in Eq.~\eqnref{eq:threeDlinevacuum}.  
\begin{lemma} \label{lem:vacuum}
The following identities hold, regardless of the contents of the shaded region: 
\begin{gather}
\label{it:secondpropertyvacuum}
\raisebox{-2.4ex}{\begin{picture}(40,35)(2,5)\twoDleftline{i} \twoDvacuum \twoDhole\end{picture}}
= \raisebox{-2.4ex}{\begin{picture}(40,35)(2,5) \twoDrightline{i} \twoDvacuum\twoDhole\end{picture}}
\\
\label{it:firstpropertyvacuum}
\raisebox{-2.4ex}{\begin{picture}(40,35)(2,5)\twoDoutervacuum\twoDvacuum\twoDhole\end{picture}}
=\cD \raisebox{-2.4ex}{\begin{picture}(40,35)(2,5)\twoDvacuum\twoDhole\end{picture}}
\\
\label{it:thirdpropertyvacuum}
\threeDlinewithringvacuum{j} 
= \cD \delta_{j\vac}
 \enspace .
\end{gather}
\end{lemma}

\begin{proof}
To prove Eq.~\eqnref{it:secondpropertyvacuum}, we apply an $F$-move, getting
\begin{equation}\begin{split}
\frac{1}{\cD} \sum_j d_j\raisebox{-2.4ex}{\begin{picture}(40,35)(2,5)\twoDsolid{j}\twoDleftline{i}\twoDhole\end{picture}}
&= \frac{1}{\cD} \sum_{j,k} d_j F^{i^* i\vac}_{j^*jk} \raisebox{-2.4ex}{\begin{picture}(40,35)(2,5)\twoDtopbottom{k}{j}{i}{i}\twoDhole\end{picture}} \\
&= \frac{1}{\cD} \sum_{j,k} \sqrt{\frac{d_i}{d_jd_k}} \delta_{i^*jk} \raisebox{-2.4ex}{\begin{picture}(40,35)(2,5)\twoDtopbottom{k}{j}{i}{i} \twoDhole\end{picture}} \\
\end{split}\end{equation}
By symmetry, this equals the right-hand side of Eq.~\eqnref{it:secondpropertyvacuum}.  
Eq.~\eqnref{it:firstpropertyvacuum} also follows: 
\begin{equation}\begin{split}
\raisebox{-2.4ex}{\begin{picture}(40,35)(2,5)\twoDoutervacuum\twoDvacuum\twoDhole\end{picture}}
= \frac{1}{\cD} \sum_j d_j \raisebox{-2.4ex}{\begin{picture}(40,35)(2,5)\twoDvacuum\twoDhole\twoDoutersolid{$j$}\end{picture}}
= \frac{1}{\cD} \sum_j d_j \!\!\! \raisebox{-2.4ex}{\begin{picture}(70,35)(2,5)\twoDvacuum\twoDhole \!\!\! \twoDoutersolidshifted{$j$}\end{picture}}
= \frac{1}{\cD} \sum_j d_j^2 \!\!\! \raisebox{-2.4ex}{\begin{picture}(40,35)(2,5)\twoDvacuum\twoDhole\end{picture}} 
\end{split}\end{equation}

For the proof of Eq.~\eqnref{it:thirdpropertyvacuum}, observe that necessarily 
\begin{equation}
\threeDlinewithring{j}{i} = \alpha_{ij}\threeDline{j}
\end{equation}
for some scalar~$\alpha_{ij} \in \mathbb{C}$.  Taking the trace, i.e., closing the loop with label~$j$, gives $\alpha_{ij} = \cD S_{ij} / d_j$.  The claim then immediately follows because 
\begin{equation}
\frac{1}{\cD} \sum_i d_i \threeDlinewithring{j}{i} = \frac{\cD}{d_j} \Bigg(\sum_i \underbrace{\frac{d_i}{\cD}}_{S_{i\vac}^*} S_{ij} \Bigg) \threeDline{j}
\end{equation}
where the parenthesized term is equal to $\delta_{j\vac}$ by the unitarity of the $S$-matrix.
\end{proof}

An immediate consequence of Eq.~\eqnref{it:thirdpropertyvacuum} in \lemref{lem:vacuum} is
\begin{equation} \label{eq:Fmovevacuumstring}
\threeDlinetopbottomFmove{i}{i'}{\ell}{j}{j'} = \cD \, \delta_{ii'} \delta_{jj'} \sqrt{\frac{d_\ell}{d_id_j}}\threeDlinetopbottomunconnected{i}{j}
\end{equation}
Indeed, applying an $F$-move to the edge with label~$\ell$ gives
\begin{equation}\begin{split}
\threeDlinetopbottomFmove{i}{i'}{\ell}{j}{j'} &= \sum_k F^{i j \ell^*}_{j^{'*}i^{'*}k} \threeDlinetopbottom{i}{i'}{k}{j}{j'} \\
&=\cD \, \delta_{ii'} \delta_{jj'} F^{i j \ell^*}_{j^*i^*\vac}\threeDlinetopbottomunconnected{i}{j}
\end{split}\end{equation}
and the claim follows by Eq.~\eqnref{eq:levinwensymmetries}.

%% file: normalization_ip.tex
\subsection{The anyonic fusion basis states form an orthonormal basis} \label{sec:normalizationgeneral}

Fix a pants decomposition, corresponding to a rooted binary tree~$T$, of the $n$-punctured sphere $\Sigma = \Sigma_n$.  In this section, we show that the anyonic fusion basis states defined in \secref{s:stringnetsdoubledmodel} form a complete, orthonormal basis.  

The proof is straightforward, but is made inconvenient by the definition of the inner product in \secref{sec:stringnethilbertspacegeneral}.  Recall that an anyonic fusion basis state is specified by two edge-labelings of~$T$ by particles from the category~$\cC$, together with particle labels for each of the $n$ boundary conditions.  The corresponding state is given by placing the two fusion diagrams on $\Sigma \times \{\pm 1\} \subset \Sigma \times [-1,1]$, adding vacuum loops around $n-1$ of the punctures, closing up the boundary conditions, and then reducing to a ribbon graph in~$\Sigma$.  For example, Eq.~\eqnref{eq:generalizedfourpunctureds} shows a typical fusion basis state on~$\Sigma_4$, with the pants decomposition~$\raisebox{-.25em}{\includegraphics[scale=.33]{SVG/fusiondiagramthree_first}}$, before reducing to two dimensions.  

To compute the inner product, a ribbon graph in $\Sigma \times [-1,1]$ must first be reduced to~$\Sigma$, and then to an orthonormal basis of $\cH_\Sigma$, as for example in Eq.~\eqnref{eq:annulusorthonormalbasis}.  To avoid these steps, we will define a new inner product, the trace inner product, that can be computed directly on ribbon graphs without first reducing them to two dimensions.  Although the definition is quite different, the trace inner product $\langle \cdot \vert \cdot \rangle_{tr}$ agrees with the original inner product $\langle \cdot \vert \cdot \rangle$ from \secref{sec:stringnethilbertspacegeneral}.  It is also convenient for extending the definition of~$\cH_\Sigma$ to more general surfaces (see \appref{sec:highergenuss}).  

Slightly more background on ribbon diagrams is necessary.  A line in a two-dimensional ribbon diagram represents a narrow ribbon that is flat in the plane.  A full twist in the ribbon can be drawn as a kink, as in Eq.~\eqnref{eq:topphase}, and is equivalent to adding a phase to the untwisted ribbon.  The same is not true for half twists, though, and the ribbon-graph notation we have developed so far does not allow for specifying half twists.  For all the equivalence relations it should be assumed that the involved ribbons all share the same up/down orientation.  

Let $\Sigma$ be a compact, orientable surface, with a marked point on each boundary component.  Consider two ribbon graphs in~$\Sigma$.  To compute their trace inner product, embed the two diagrams in $\Sigma \times \{ \pm 1 \} \subset \Sigma \times [-1,1]$.  Conjugate all particle labels in the first diagram, and also give them the outward orientation, i.e., turn all ribbons upside down.  Next, connect together the ribbons at the boundary points.  (Since the ribbons on one half have been turned upside down, there is no need to twist the ribbons to connect them together; a twist would have introduced a phase ambiguity.)  Finally, reduce the resulting ribbon graph to two dimensions.  In order to carry out this reduction, it may be useful to introduce pairs of half-twists on some ribbons in order to flip their orientation in local regions, so that the equivalence relations can be applied.  Now, the inner product is the coefficient of the empty or vacuum diagram, divided by the product over all boundary labels $i$ of~$\sqrt{d_i}$.  

Let us first sketch the argument that the trace inner product agrees with the original definition of the inner product.  For computational basis states on the annulus $\Sigma_2$, we have 
\begin{equation}
\Big\langle \Danyonsstandardbasisone{k}{i}{j}{\ell} \Big| \Danyonsstandardbasisone{k'}{i'}{j'}{\ell'} \Big\rangle_{tr}
= 
\frac{\delta_{k k'} \delta_{\ell \ell'}}{\sqrt{d_k d_\ell}}
\raisebox{-8.0ex}{
\begin{picture}(80,95)(-5,25)
\put(17,56.5){\Danyonsstandardbasisonex{\,k}{i'}{j'}{\ell}}
\put(0,0){\includegraphics[scale=.6]{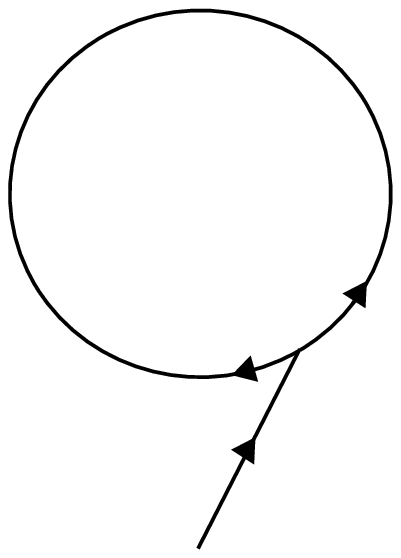}}
\put(47,40){$\ell$}
\put(36,49){$i$}
\put(65,70){$j$}
\end{picture}}
\end{equation}
where on the right-hand side the two ribbons labeled $\ell$ should be matched up and then the diagram reduced to its vacuum coefficient.  Carrying out this calculation---beginning with an $F$-move between the two ribbons $j$ and $j'$---indeed gives $\delta_{k k'} \delta_{\ell \ell'} \delta_{i i'} \delta_{j j'}$, agreeing with the standard inner product.   The general calculation of the trace inner product on computational basis states on the $n$-punctured sphere $\Sigma_n$ is quite similar; for example, 
\begin{equation} \label{eq:traceproductstringnet}
\Big\langle \fusethreedirectedreverse{a_1}{a_2}{a_3}{c_1}{c_2} \vert \fusethreedirectedreverse{a_1'}{a_2'}{a_3'}{c_1'}{c_2'} \Big\rangle_{tr}
= 
\frac{\delta_{\vec{a} \vec{a'}} \delta_{c_2 c_2'}}{\sqrt{d_{a_1} d_{a_2} d_{a_3} d_{c_2} }}
\raisebox{-8.0ex}{
\begin{picture}(110,90)(0,30)
\put(0,0){\Danyonsfourpuncturedflippedsingle}
\put(0,0){\Danyonsfourpuncturedsingle{a}{c}}
\end{picture}}
\end{equation}
which by Eq.~\eqnref{eq:bubblerule} simplifies to $\delta_{\vec{a} \vec{a'}} \delta_{\vec{c}\vec{c'}}$.  

Next we argue that the anyons are orthonormal under the trace inner product.  Again we begin with the case of the annulus $\Sigma_2$.  The anyonic fusion basis states on $\Sigma_2$ are the ribbon graphs $\myvecstand{k}{\ell}{i}{j}$ of Eq.~\eqnref{eq:kaelldef} with $\delta_{i j \ell^*} = \delta_{i' j' \ell'^*} = 1$.  

\begin{lemma} \label{lem:idempotentsfromdouble}
Under stacking, 
\begin{equation}
\myvecstand{\ell'}{k'}{i'}{j'} \; \myvecstand{k}{\ell}{i}{j} = \delta_{i i'} \delta_{j j'} \delta_{k k'} \; \cD \sqrt{\frac{d_k}{d_i d_j}} \; \myvecstand{\ell'}{\ell}{i}{j}
 \enspace .
\end{equation}
\end{lemma}

\begin{proof}
By Eq.~\eqnref{it:secondpropertyvacuum} in \lemref{lem:vacuum} and Eq.~\eqnref{eq:Fmovevacuumstring},
\begin{align}
\threeDdoubleringvv{\ell'}{i'}{j'}{k}{i}{j}{\ell} &= \threeDdoubleringbvv{\ell'}{i'}{j'}{k}{i}{j}{\ell} \\
&= \cD \delta_{i i'} \delta_{j j'} \sqrt{\frac{d_k}{d_i d_j}} \threeDdoubleringsingle{\ell'}{i}{j}{\ell}
\qedhere
\end{align}
\end{proof}

\noindent
In particular, $\frac{1}{\cD} \sqrt{\frac{d_i d_j}{d_k}} \myvecstand{k}{k}{i}{j}$ is an idempotent corresponding to a realization of the fusion space~$V^{i \otimes j}_{i \otimes j}$.  

Now the trace inner product of $\myvecstand{\ell'}{k'}{i'}{j'}$ and $\myvecstand{k}{\ell}{i}{j}$ is given by $1 / \sqrt{d_k d_\ell}$ times the vacuum coefficient of their stack $\myvecstand{\ell'}{k'}{i'}{j'} \; \myvecstand{k}{\ell}{i}{j}$, where additionally the ribbons $\ell$ and $\ell'$ are connected together.  The coefficient of the trivial particle $\vac$ in a vacuum line is $d_\vac / \cD = 1/\cD$, so by \lemref{lem:idempotentsfromdouble},  
\begin{equation} \label{eq:anyonnormalizationsigmatwo}
\braket{\myvecstand{\ell'}{k'}{i'}{j'}}{\myvecstand{k}{\ell}{i}{j}}
= 
\delta_{i i'} \delta_{j j'} \delta_{k k'} \delta_{\ell \ell'} \frac{1}{\sqrt{d_i d_j d_\ell}} \raisebox{-16pt}{\includegraphics{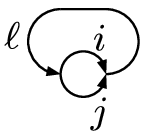}}
= 
\delta_{i i'} \delta_{j j'} \delta_{k k'} \delta_{\ell \ell'}
 \enspace .
\end{equation}

Define anyonic fusion basis states for $\cH_{\Sigma_4}$ 
\begin{equation} \label{eq:generalizedfourpuncturedv}
\myvecstandard{k}{\ell}{a}{b}{c}{d} = 
\myvecstandardpic{k}{\ell}{a}{b}{c}{d}
\end{equation}
Analogous states $\myvecstandard{k}{\ell}{a}{b}{c}{d} \in \cH_{\Sigma_n}$ can be defined for any $n \geq 3$ using a ``standard'' pants decomposition of $\Sigma_n$.  We show that these states, with fusion constraints satisfied, are mutually orthogonal and unit-normalized.  

\begin{lemma} \label{lem:anyonnormalization}
Let $\Psi = \myvecstandard{k}{\ell}{a}{b}{c}{d}$ and $\Psi' = \myvecstandard{k'}{\ell'}{a'}{b'}{c'}{d'}$ be the states in $\cH_{\Sigma_n}$, with $n \geq 3$, as introduced in Eq.~\eqnref{eq:generalizedfourpuncturedv}.  Then 
\begin{equation}
\braket{\Psi'}{\Psi} = \delta_{k k'} \delta_{\vec \ell \vec \ell'} \delta_{\vec a \vec a'} \delta_{\vec b \vec b'} \delta_{\vec c \vec c'} \delta_{\vec d \vec d'}
 \enspace .
\end{equation}
\end{lemma}

\begin{proof}
Although the calculation is fully general, we continue to illustrate the case of $n = 3$.  Applying $F$-moves, we find  
\begin{equation}
\Psi =
\sum_{\vec m} \prod_i F^{a_i a_i^* \vac}_{b_i^* b_i m_i^*}
\raisebox{-9.0ex}{
\begin{picture}(110,103)(0,0)
\put(0,48){\Danyonsfmoveb{m_1}{a_1}{b_1}{\ell_1}}
\put(33,48){\Danyonsfmoveb{m_2}{a_2}{b_2}{\ell_2}}
\put(67,48){\Danyonsfmoveb{m_3}{a_3}{b_3}{\ell_3}}
\put(0,0){\Danyonsfourpunctured{a}{b}{k}{c}{d}}
\end{picture}}
\end{equation}
and similarly for $\Psi'$.  Thus we can separate the calculation of the inner product $\braket{\Psi'}{\Psi}$ into two parts, the inner products of the topmost states in $\cH_{\Sigma_2}$ and the inner product of the lower doubled trees.  From Eq.~\eqnref{eq:anyonnormalizationsigmatwo}, this gives
\begin{equation}\begin{split}
\braket{\Psi'}{\Psi} 
&= 
\sum_{\vec m, \vec m'} \prod_i \Big( \big(F^{a_i' a_i'^* \vac}_{b_i'^* b_i' m_i'^*}\big)^* F^{a_i a_i^* \vac}_{b_i^* b_i m_i^*} 
\braket{\myvecstand{m_i'}{\ell_i'}{a_i'}{b_i'}}{\myvecstand{m_i}{\ell_i}{a_i}{b_i}} \Big) \braket{\Phi'_{\vec m'}}{\Phi_{\vec m}} \\
&= 
\delta_{\vec a' \vec a} \delta_{\vec b' \vec b} \delta_{\vec \ell' \vec \ell}
\sum_{\vec m} \prod_i \Big( \frac{d_{m_i}}{d_{a_i} d_{b_i}} \delta_{a_i b_i m_i^*} \Big) \braket{\Phi'_{\vec m}}{\Phi_{\vec m}}
 \enspace ,
\end{split}\end{equation}
where 
\begin{equation}
\Phi_{\vec m}
= \myvecstandardtwopic{k}{m}{a}{b}{c}{d}
\end{equation}
and $\Phi'_{\vec m'}$ is defined similarly.  

By definition of the trace inner product and using $F$-moves, we have
\begin{align} \label{eq:lasttracebk}
\sum_{\vec m} \prod_i \left( \frac{d_{m_i}}{d_{a_i} d_{b_i}} \delta_{a_i b_i m_i^*} \right) \braket{\Phi'_{\vec m}}{\Phi_{\vec m}}
&= \frac{1}{\sqrt{d_k}} \sum_{\vec m} \prod_i \left( \frac{\sqrt{d_{m_i}} \delta_{a_i b_i m_i^*}}{d_{a_i} d_{b_i}} \right)
\raisebox{-12.0ex}{
\begin{picture}(110,103)(0,0)
\put(0,0){\Danyonsfourpunctured{a}{b}{k}{c}{d}}
\put(-0.25,0){\Danyonsfmovebasetwo{m_1}}
\put(33.1,0){\Danyonsfmovebasetwo{m_2}}
\put(67,0){\Danyonsfmovebasetwo{m_3}}
\put(-0.25,-1){\Danyonsfourpuncturedflippedtwo}
\end{picture}} \nonumber \\
&= \frac{1}{\sqrt{d_k \prod_i d_{a_i} d_{b_i}}}
\raisebox{-9.0ex}{
\begin{picture}(110,110)(4,0)
\put(0,0){\Danyonsfourpuncturedflipped}
\put(0,0){\Danyonsfourpunctured{a}{b}{k}{c}{d}}
\put(0,0){\includegraphics[scale=.6]{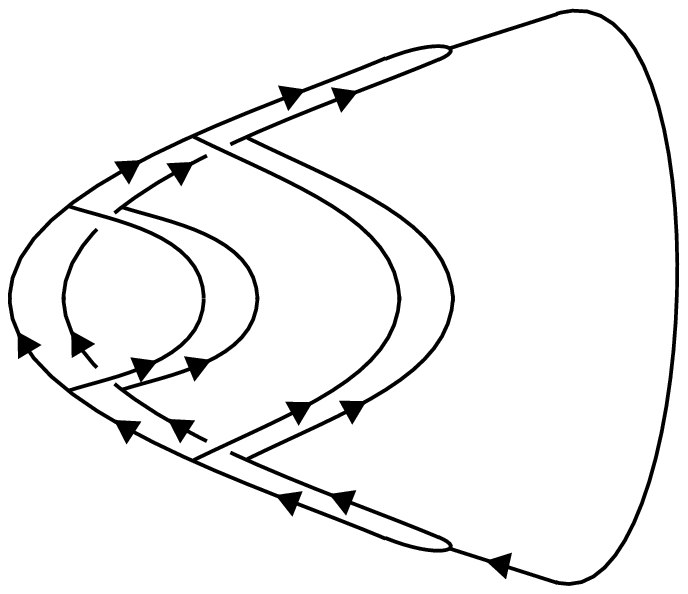}}
\end{picture}} \nonumber \\
&= \delta_{\vec c \vec c'} \delta_{\vec d \vec d'}
 \enspace .
\end{align}
The last step uses Eq.~\eqnref{eq:bubblerule} repeatedly.  
\end{proof}

%% file: completeness.tex
\subsubsection*{Completeness of anyonic fusion basis for $\cH_{\Sigma_n}$}

It remains to show that the anyonic fusion basis spans~$\cH_{\Sigma_n}$.  We do so by verifying that the number of standard computational basis states equals the number of anyonic fusion basis.  

The proof is by induction in $n$.  For $n \in \{0,1\}$, $\cH_{\Sigma_n}$ is one-dimensional.  For $n = 2$, $\Sigma_n$ the annulus, with boundary labels $k^*$ and $\ell$, the numbers of computational basis states and of anyonic fusion basis states are both
\begin{equation}
\sum_{i, j} \delta_{k i^* j^*} \delta_{i j \ell^*}
 \enspace .
\end{equation}
See Eqs.~\eqnref{eq:annulusorthonormalbasis} and~\eqnref{eq:kaelldef}.  

For $\cH_{\Sigma_n}$ with a fixed root puncture, let $s_n(k)$ be the number of standard basis elements with label~$k$ on the root.  This function satisfies the recursion, for any $m$, $2 \leq m \leq n-1$, 
\begin{equation} \label{eq:fnrecursionrelation}
s_n(k) = \sum_{i, j} \delta_{i j k^*} s_m(i) s_{n+1-m}(j)
 \enspace .
\end{equation}

Similarly, let $a_n(k)$ be the number of anyonic fusion states in $\cH_{\Sigma_n}$ with root labeled $k$.  We will show that $a_n(k)$ satisfies the same recursion, Eq.~\eqnref{eq:fnrecursionrelation}, as $s_n(k)$.  For this purpose, let $\tilde a_n(\ell, \ell')$ be the number doubled fusion diagrams on a tree with $n$~leaves, with arbitrary boundary-consistent labels on the $n-1$~leaves and the doubled anyon label $\ell \otimes \ell'$ at the root.  Since picking a boundary-consistent label~$k$ for the root puncture gives a valid anyonic fusion basis state, we have
\begin{equation} \label{eq:gnkdef}
a_n(k) = \sum_{\ell, \ell'} \tilde a_n(\ell, \ell') \delta_{\ell \ell' k^*}
 \enspace .
\end{equation}
Moreover, $\tilde a$ satisfies the recursion, for $2 \leq m \leq n-1$, 
\begin{equation}
\tilde a_{n}(\ell, \ell') = \sum_{i, i', j, j'} \delta_{i j \ell^*} \delta_{i' j' \ell'^*} \tilde a_m(i, i') \tilde a_{n+1-m}(j, j')
 \enspace .
\end{equation}
Reinserting this into Eq.~\eqnref{eq:gnkdef} gives
\begin{align}
a_n(k) &= \sum_{i, i', j, j'} \tilde a_m(i, i') \tilde a_{n+1-m}(j, j') \sum_{\ell, \ell'} \delta_{\ell \ell' k^*} \delta_{i j \ell^*} \delta_{i' j' \ell'^*} \nonumber \\
&= \sum_{\ell, \ell', k} \delta_{\ell \ell' k^*} \Big(\sum_{r,i'} \tilde a_m(i, i') \delta_{i i' \ell^*} \Big) \Big(\sum_{j, j'} \tilde a_{n+1-m}(j, j') \delta_{j j' \ell'^*} \Big) \nonumber \\
&= \sum_{\ell, \ell', k} \delta_{\ell \ell' k^*} a_m(\ell) a_{n+1-m}(\ell')
 \enspace ,
\end{align}
the desired recursion.  
Here in the second step we used 
\begin{equation} \label{eq:fivepuncturedcounting}
\sum_{\ell, \ell'} \delta_{\ell \ell' k^*} \delta_{i j \ell^*} \delta_{i' j' \ell'^*} = \sum_{\ell, \ell'} \delta_{\ell \ell' k^*} \delta_{i i' \ell^*} \delta_{j j' \ell'^*}
 \enspace ,
\end{equation}
which is a consequence of the associativity of fusion, Eq.~\eqnref{eq:fusionassociativity}.  In particular, the left- and right-hand sides respectively equal the number of ribbon graphs  \newcommand*{\fivevalenta}{\raisebox{-3.7ex}{
\begin{picture}(48,58)(-2,-4)
\put(0,0){\includegraphics[scale=.5]{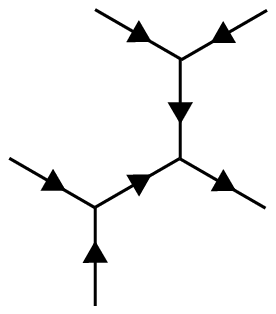}}
\put(15,5){\small $j'$}
\put(30,10){\small $k$}
\put(15,21){\small $\ell'$}
\put(28,26){\small $\ell$}
\put(28,44){\small $j$}
\put(16,44){\small $i$}
\put(7,21){\small $i'$}
\end{picture}}}
\newcommand*{\fivevalentb}{\raisebox{-3.7ex}{
\begin{picture}(48,58)(-2,-4)
\put(0,0){\includegraphics[scale=.5]{SVG/fivevalent}}
\put(15,5){\small $j'$}
\put(30,10){\small $k$}
\put(15,21){\small $\ell'$}
\put(28,26){\small $\ell$}
\put(28,44){\small $i'$}
\put(16,44){\small $i$}
\put(7,21){\small $j$}
\end{picture}}}
\begin{equation} \label{eq:basisfusionrulescnt}
\fivevalenta \qquad\textrm{ and }\qquad \fivevalentb
\end{equation}
with fixed boundary conditions~$(i, i', j, j', k^*)$.  In other words, both sides of Eq.~\eqnref{eq:fivepuncturedcounting} count the dimension of the space of ribbon graphs embedded in a ball, with 
five fixed boundary labels.  

%% file: fibonacciexamples.tex
\subsection{Application to the Fibonacci model} \label{sec:fibonacciapplication}

Let us now specialize these statements to the Fibonacci model, and give a few examples of the anyonic fusion basis states described in Sections~\ref{s:anyonicfusionbasis} and~\ref{s:stringnetsdoubledmodel}.  Here the $F$ and $R$ tensors are specified by Eqs.~\eqnref{eq:Fmoveeqsimple} and~\eqnref{e:dekink}, respectively.  Crossings can therefore be resolved according to Eq.~\eqnref{e:overcrossing}.

\subsubsection*{Computation of idempotents on \texorpdfstring{$\cH_{\Sigma_2}$}{H\_\{Sigma\_2\}}}

For the annulus~$\Sigma_2$, anyonic fusion basis states are equivalent to idempotents, up to normalization, when the boundary labels are identical.  We therefore restrict our attention to the idempotents.  According to \lemref{lem:idempotentsfromdouble}, the idempotents can be written as 
\begin{align}
\oneone &= \frac{1}{\cD} \raisebox{0em}{\doubledanyonsvv} &\tauone &= \frac{1}{\cD} \doubledanyonstv\qquad & \onetau &= \frac{1}{\cD} \doubledanyonsvt  \label{eq:idempotentsdoubled} \\
\tautau\coeff{\v} &= \frac{\tau}{\cD} \doubledanyonsttv &\tautau\coeff{\t} &= \frac{\sqrt \tau}{\cD} \doubledanyonsttt \nonumber
\end{align}
where $\cD = \sqrt{1 + \tau^2}$.  
Writing out this definition and resolving the crossings, it can be verified that these expressions coincide with the idempotents given in Eq.~\eqnref{e:fiveoperators}.  For example, we get
\begin{equation}\begin{split}
\cD^2 \, \onetau 
= \bt + \tau \bttinv 
= \bt+ \tau e^{-3\pi i/5} \bS+\tau e^{3\pi i/5} \bSdag
\end{split}\end{equation}
For later reference, note that the same type of calculation gives
\begin{align} \label{eq:deltaevaluated}
\doubledanyonsdeltaonezero &= e^{-3\pi i/10} \doubledanyonsdeltaonezeroZ 
&
\doubledanyonsdeltazeroone &= e^{3\pi i/10} \doubledanyonsdeltazerooneZ 
\end{align}

\subsubsection*{Computation of anyonic fusion basis states for \texorpdfstring{$\cH_{\Sigma_3}$}{H\_\{Sigma\_3\}}}	

For the $3$-punctured sphere, we have for example the following three fusion basis states, 
\begin{equation} \label{eq:firstexamplefibdoubled}
\fuse{\tautau}{\tauone}{\onetau} =  \doubledexample{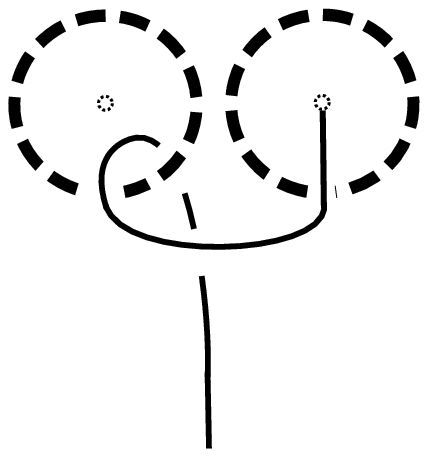} \qquad
\fuse{\tautau}{\tautau}{\tauone} =  \doubledexample{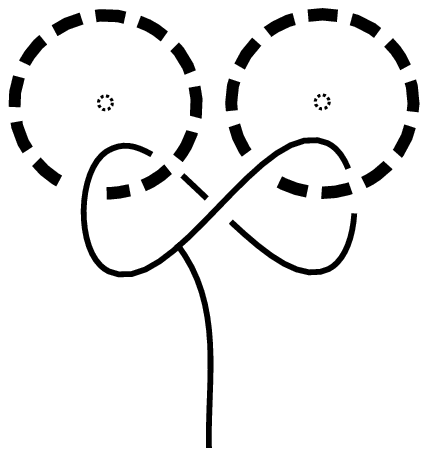} \qquad
\fuse{\tautau\coeff\t}{\tautau}{\tauone} =  \doubledexample{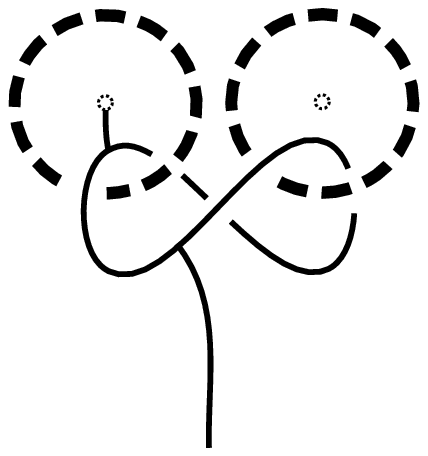}
\end{equation}
where the subscript $\coeff\t$ indicates that we choose to have an edge ending on the boundary.  

Let us reduce to two dimensions the first diagram.  Apply an $F$-move to the left leg to get
\begin{equation}
\fuse{\tautau}{\tauone}{\onetau} = \frac{1}{\tau} \doubledexample{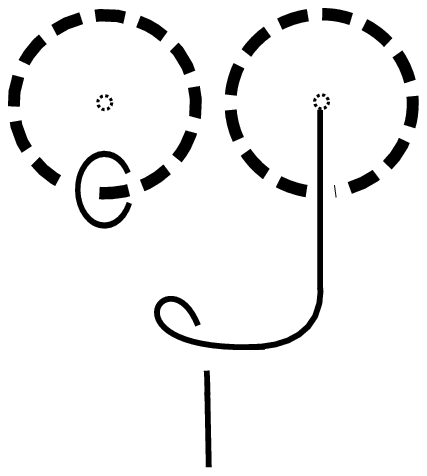} + \frac{1}{\sqrt \tau} \doubledexample{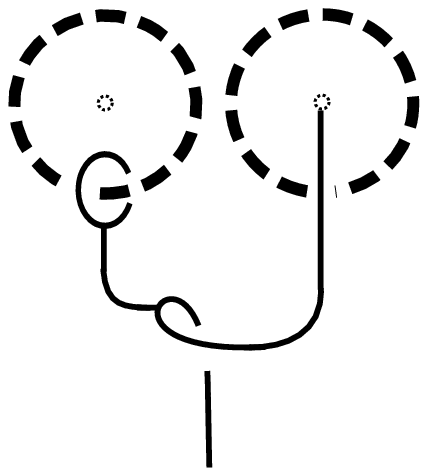}
\end{equation}
Removing the twists with Eq.~\eqnref{e:dekink} and inserting from Eq.~\eqnref{e:fiveoperators} the expressions for the $\tautau\coeff{\v}$ and $\tauone$ idempotents gives 
\begin{equation}\begin{split}
\fuse{\tautau}{\tauone}{\onetau} &=
\frac{e^{-4 \pi i / 5}}{\cD^2} \left( 
\twocob{\zvvI}{\zttI}{\vtt} - \frac{1}{\tau} \twocob{\zO}{\zttI}{\vtt}
 + e^{3\pi i/5} \left( \tau \twocob{\zvvI}{\zS}{\vtt} - \twocob{\zO}{\zS}{\vtt} \right)
 + e^{-3\pi i/5} \left( \tau \twocob{\zvvI}{\zSdag}{\vtt} - \twocob{\zO}{\zSdag}{\vtt} \right)
 \right) \\
&\quad + \frac{e^{9\pi i/ 10}}{\sqrt{\tau} \cD}
\left( \twocob{\zvtD}{\zttI}{\ttt} + \tau e^{3\pi i/5} \twocob{\zvtD}{\zS}{\ttt} + \tau e^{-3\pi i/5} \twocob{\zvtD}{\zSdag}{\ttt} \right)
\end{split}
\end{equation}
Since $\braket \S \Sdag = 1/\tau$, this state indeed has unit norm, as asserted by \lemref{lem:anyonnormalization}.  

%% file: highgenus.tex
\newcommand*{\handlebody}{
\raisebox{-6.5ex}{
\begin{picture}(45,68)(0,0)
\put(0,0){\includegraphics[scale=.3]{SVG/handlebody}}
\end{picture}}}

\newcommand*{\handlebodydehntwists}{
\raisebox{-6.5ex}{
\begin{picture}(45,68)(0,0)
\put(0,0){\includegraphics[scale=.3]{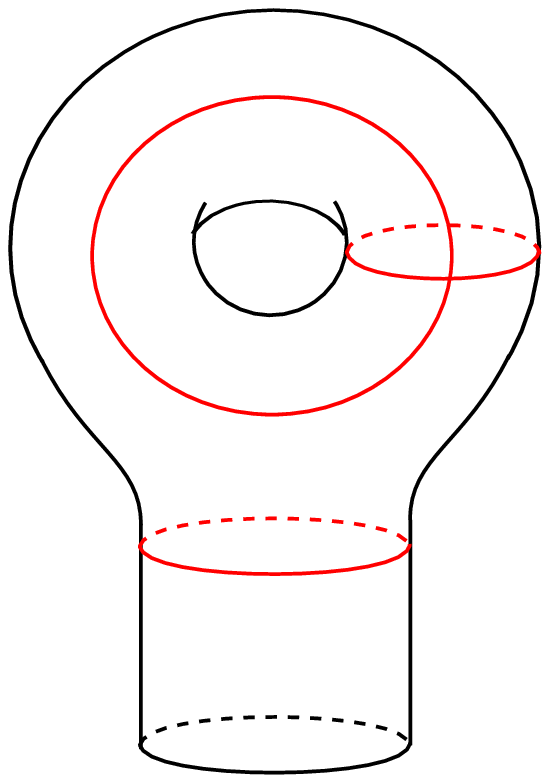}}
\put(20,10){$\alpha$}
\put(20,60){$\beta$}
\put(35,48){$\gamma$}
\end{picture}}}

\newcommand*{\genusbasisone}[2]{
\raisebox{-3.5ex}{
\begin{picture}(65,40)(0,0)
\put(0,0){\includegraphics[scale=.5]{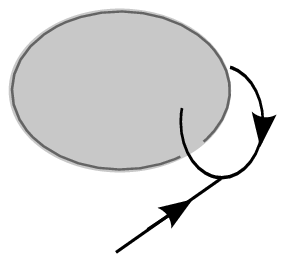}}
\put(31,2){$#1$}
\put(44,16){$#2$}
\end{picture}}}
 
\newcommand*{\genusbasisonethreeD}[4]{
\raisebox{-9.5ex}{
\begin{picture}(85,120)(0,0)
\put(0,0){\includegraphics[scale=.5]{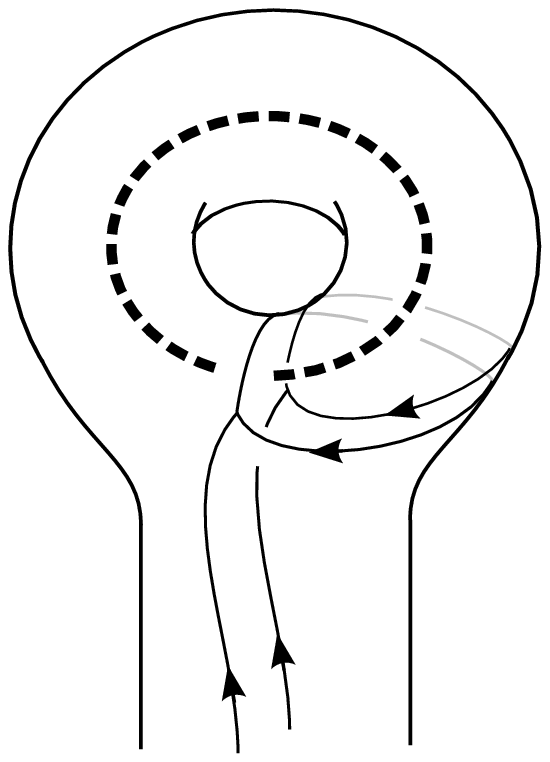}}
\put(20,10){$#1$}
\put(44,16){$#2$}
\put(45,37){$#3$}
\put(52,57){$#4$}
\end{picture}}}

\newcommand*{\genusbasistwothreeD}[4]{
\raisebox{-9.5ex}{
\begin{picture}(85,120)(0,0)
\put(0,0){\includegraphics[scale=.5]{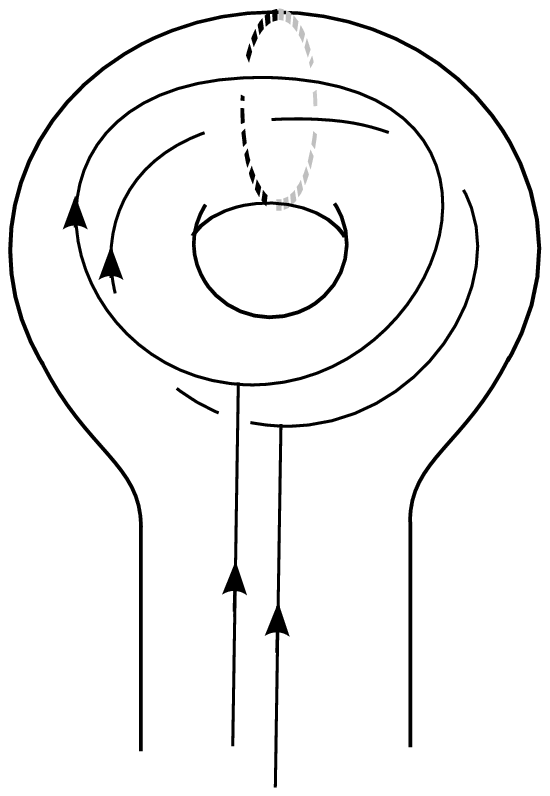}}
\put(20,35){$#1$}
\put(44,20){$#2$}
\put(1,78){$#3$}
\put(18,68){$#4$}
\end{picture}}}

\newcommand*{\genusbasisoneprojected}[4]{
\raisebox{-9.5ex}{
\fbox{\begin{picture}(85,100)(0,0)
\put(0,0){\includegraphics[scale=.5]{SVG/handlebody_doubledprojected}}
\put(20,10){$#1$}
\put(44,16){$#2$}
\put(45,40){$#3$}
\put(68,60){$#4$}
\end{picture}}}}

\newcommand*{\genusbasistwo}[2]{
\raisebox{-3.5ex}{
\begin{picture}(65,40)(0,0)
\put(0,0){\includegraphics[scale=.5]{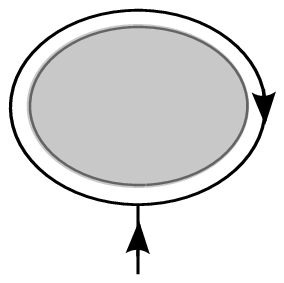}}
\put(22,2){$#1$}
\put(42,22){$#2$}
\end{picture}}}

\newcommand*{\sbmatrixfirst}[3]{\raisebox{-3.7ex}{
\begin{picture}(57,45)(-2,-10)
\put(0,0){\includegraphics[scale=.5]{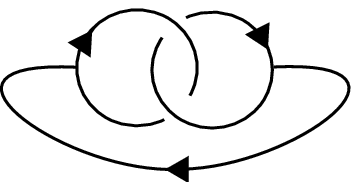}}
\put(11,27){\small $#1$}
\put(30,27){\small $#2$}
\put(23,-8){\small $#3$}
\end{picture}}}

\newcommand*{\sbmatrixsecond}[3]{\raisebox{-3.7ex}{
\begin{picture}(57,45)(-2,-10)
\put(0,0){\includegraphics[scale=.5]{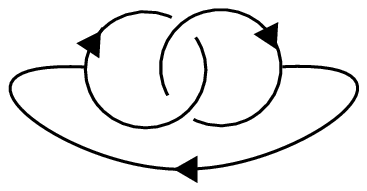}}
\put(11,27){\small $#1$}
\put(30,27){\small $#2$}
\put(23,-8){\small $#3$}
\end{picture}}}

\subsection{Generalization to higher-genus surfaces} \label{sec:highergenuss}
Our results extend to a surface~$\Sigma$ of higher genus.  To describe a fusion diagram basis of $\Sigma$, fix a handle decomposition of $\Sigma$ into an $n$-punctured sphere with handles attached to some of the holes.  A handle is a punctured torus, shown in \figref{f:handle}.  
An anyonic fusion basis for $\cH_\Sigma$ is defined by extending a doubled anyon fusion diagram in $\Sigma_n \times [-1,1]$ by attaching at each handle one of the diagrams in \figref{f:handleanyons}.  Here $a_+a_-$ is the label on the leaf, and $(a_+, b_+, b_+^*)$ and $(a_-, b_-, b_-^*)$ are fusion-consistent triples from $\cC_+$ and $\cC_-$, respectively.  These two choices result in bases that are diagonal with respect to Dehn twists along the loops $\{\alpha, \beta\}$ and $\{\alpha, \gamma\}$, respectively, in \figref{f:handle}.  

\begin{figure*}
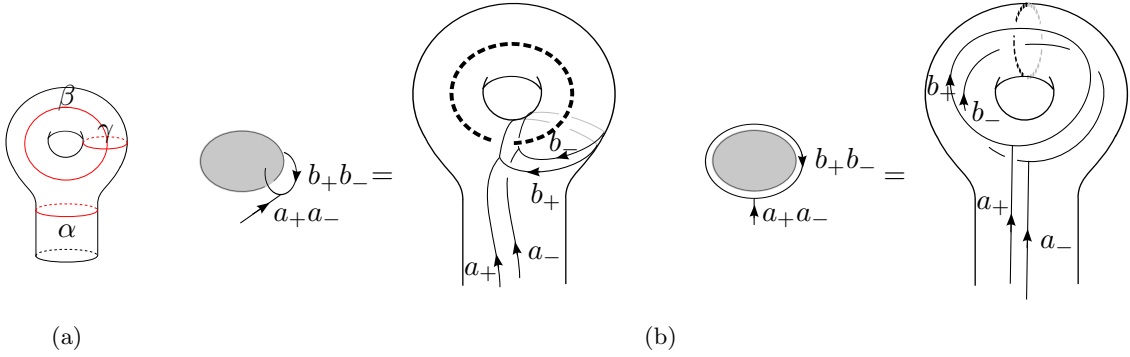

\centering
\begin{tabular}{c@{$\qquad$}c}
\subfigure[]{\label{f:handle}$\!\!\!\!\!\!\!$\raisebox{.75in}{\handlebodydehntwists}}
&
\subfigure[]{\label{f:handleanyons}\raisebox{.75in}{$\genusbasisone{a_+a_-}{b_+b_-} = \genusbasisonethreeD{a_+}{a_-}{b_+}{b_-}\qquad 
\genusbasistwo{a_+a_-}{b_+b_-} = \genusbasistwothreeD{a_+}{a_-}{b_+}{b_-}$}}
\end{tabular}
\caption{An anyonic fusion basis for a surface with $n$ punctures or handles is specified by using for each handle one of the two bases in (b), 
combined with a doubled anyon fusion tree in~$\Sigma_n \times [-1,1]$.} 
\end{figure*}

Together with $F$-moves, the change of basis matrix $S^{a_+a_-} = (S^{a_+a_-}_{b'_+b'_-, b_+b_-})$ defined by
\begin{equation}
\genusbasistwo{a_+a_-}{b'_+b'_-} = \sum_{b_+, b_-} S^{a_+a_-}_{b'_+b'_-, b_+b_-}\genusbasisone{a_+a_-}{b_+b_-}
\end{equation}
fully specifies the action of the mapping class group of $\Sigma$.  (For the torus, there are two inequivalent bases, and a similar change of basis matrix~$S$.)  By similar arguments as those used in
\secref{s:stringnetsdoubledmodel} for the $R$-matrix, one can derive that
$S^{a_+a_-}_{b'_+b'_-,b_+b_-} = A^{a_+}_{b_+',b_+}B^{a_-}_{b_-',b_-}$, where
\begin{equation}
A^a_{b',b} = \frac{1}{\mathcal{D} \sqrt{d_a}} \sbmatrixfirst{b}{b'}{a} \qquad\textrm{ and }\qquad B^a_{b',b} = \frac{1}{\mathcal{D} \sqrt{d_a}} \sbmatrixsecond{b}{b'}{a}
 \enspace . 
\end{equation}
We refer to~\cite[Appendix E]{KitaevAnyons} for a proof of the unitarity of these matrices.

%% file: levinwenappendix.tex
\section{Discretizing and Gluing: Proofs} \label{app:levinwenmodel}

In this appendix, we prove Lemmas~\ref{lem:mainisomorphism} and~\ref{lem:tadpolegluing}. 
\lemref{lem:mainisomorphism} states that the simultaneous $+1$-eigenspace $\cH_{\widehat \cT}^{\ell,gs}$ of all plaquette- and vertex-operators of the Levin-Wen Hamiltonian~\eqnref{eq:levinwenhamiltonian} is isomorphic to the ribbon graph space~$\cH_\Sigma^\ell$.  In preparation for the proof of this fact, we derive the following auxiliary statement.

\begin{lemma} \label{lem:punctureaddition}
Consider two surfaces $\Sigma$ and $\Sigma'$, where $\Sigma'$ is the same as $\Sigma$ but with a puncture inserted at $p \in \Sigma$.  Let $\ell$ be a labeling of the boundary points of $\Sigma$, and let $\cH^{(\ell, \vac)}_{\Sigma'}$ be the space of ribbon graphs on $\Sigma'$ with open boundary condition on~$p$.  Consider a state $\ket{\Psi_\Sigma} \in \cH^\ell_\Sigma$.  We deform the ribbon graph locally around~$p$, and regard the resulting ribbon graph as an element $\ket{\Phi_{\Sigma'}}\in\cH^{(\ell, \vac)}_{\Sigma'}$.  Let $B: \cH^{(\ell, \vac)}_\Sigma \rightarrow \cH^{(\ell, \vac)}_{\Sigma'}$ be the map which adds a vacuum loop around~$p$, divided by~$\cD$.  Then $B \ket{\Phi_{\Sigma'}}$ is independent of the initial deformation (i.e., $\ket{\Phi_{\Sigma'}}$).  Furthermore, the map
\begin{equation}\begin{split}
\Lambda: \cH^\ell_\Sigma &\rightarrow \cH^{(\ell, \vac)}_{\Sigma'}\\
\ket{\Psi_\Sigma} &\mapsto \cD \, B \ket{\Phi_{\Sigma'}}
\end{split}\end{equation}
is a norm-preserving isomorphism.  
\end{lemma}

\begin{proof}
Observe that the state $B \ket{\Phi_{\Sigma'}}$ has a vacuum loop inserted around the puncture~$p$.  Since locally deformed configurations of ribbon graphs near a puncture enclosed by a vacuum loop are equivalent according to \lemref{lem:vacuum}~\eqnref{it:secondpropertyvacuum}, the state $B \ket{\Phi_{\Sigma'}}$ indeed does not depend on the initially chosen deformation.  This implies that the map $\Lambda: \cH^\ell_\Sigma \rightarrow \cH^{(\ell, \vac)}_{\Sigma'}$ is well-defined.  It is easy to see that the map is surjective; a preimage of an element $\ket{\Psi_{\widehat{\cT}}} \in \cH^{(\ell, \vac)}_{\Sigma'}$ is simply the same ribbon graph embedded in $\Sigma$ (this can be checked with \lemref{lem:vacuum}~\eqnref{it:firstpropertyvacuum}). 

Now consider the inner product $\spr{\tilde{\Psi}_\Sigma}{\Psi_\Sigma}$ of two states $\ket{\Psi_\Sigma}, \ket{\tilde{\Psi}_\Sigma} \in \cH^\ell_\Sigma$.  Since local deformations do not change the equivalence class, we may assume that these ribbon graphs avoid~$p$.  To evaluate the inner product $\spr{\tilde{\Psi}_{\Sigma'}}{\Psi_{\Sigma'}}$  of the images $\ket{\tilde{\Psi}_{\Sigma'}} = \Lambda \ket{\tilde{\Psi}_\Sigma}$, $\ket{\Psi_{\Sigma'}} = \Lambda\ket{\Psi_\Sigma}$, consider a ribbon graph basis $\cB'$ of $\cH^{(\ell, \vac)}_{\Sigma'}$ whose ribbon graphs are obtained by attaching a tadpole with head surrounding~$p$ (and fusion-consistent labels) to the elements of a ribbon graph basis~$\cB$ of $\cH^\ell_\Sigma$.  Evaluating the inner product~$\spr{\tilde{\Psi}_{\Sigma'}}{\Psi_{\Sigma'}}$ with respect to this basis reduces to the evaluation of $\spr{\tilde{\Psi}_\Sigma}{\Psi_\Sigma}$ and the norm of a vacuum line around~$p$.  This gives
\begin{equation}
\spr{\tilde{\Psi}_{\Sigma'}}{\Psi_{\Sigma'}} = \spr{\tilde{\Psi}_\Sigma}{\Psi_\Sigma} \cdot \frac{1}{\cD^2}
 \enspace ,
\end{equation}
which implies the claim.
\end{proof}

The proof of \lemref{lem:mainisomorphism} is now immediate.  

\begin{proof}[Proof of \lemref{lem:mainisomorphism}]
Let $\ket{\Psi_\Sigma} \in \cH^\ell_\Sigma$, and let $\ket{\Phi_\Delta} \in \cH^{(\ell, \vac^P)}_\Delta$ be a state obtained by locally deforming $\ket{\Psi_\Sigma}$ to avoid the points in~$\cP$.  Observe that the state $B \ket{\Phi_\Delta}$, where $B$ is the product of all plaquette operators (cf.~\eqnref{eq:Bdef}), has vacuum loops inserted around each puncture~$p \in \cP$.  This corresponds to the situation described in \lemref{lem:punctureaddition}.  By inductively applying this lemma to all $p \in \cP$, we obtain the claim.
\end{proof}

\newcommand*{\tadpoleproof}[5]{
\raisebox{-8.5ex}{
\begin{picture}(225,85)(0,0)
\put(12,52){$#1$}
\put(12,25){$#2$}
\put(172,52){$#4$}
\put(172,25){$#3$}
\put(98,52){$#5$}
\put(0,0){\scalebox{0.35}{\protect\epsfig{file=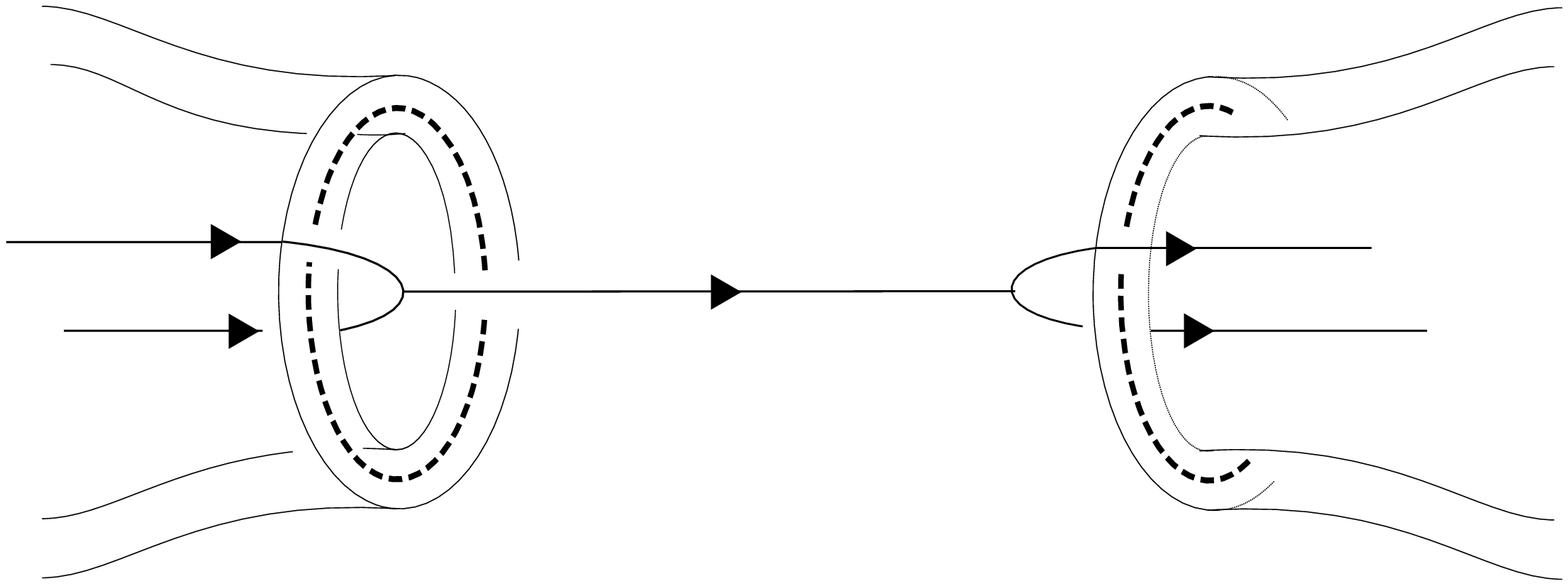}}}
\end{picture}}}

\newcommand*{\tadpoleprooffirst}[5]{
\raisebox{-8.5ex}{
\begin{picture}(225,85)(0,0)
\put(0,0){\scalebox{0.35}{\protect\epsfig{file=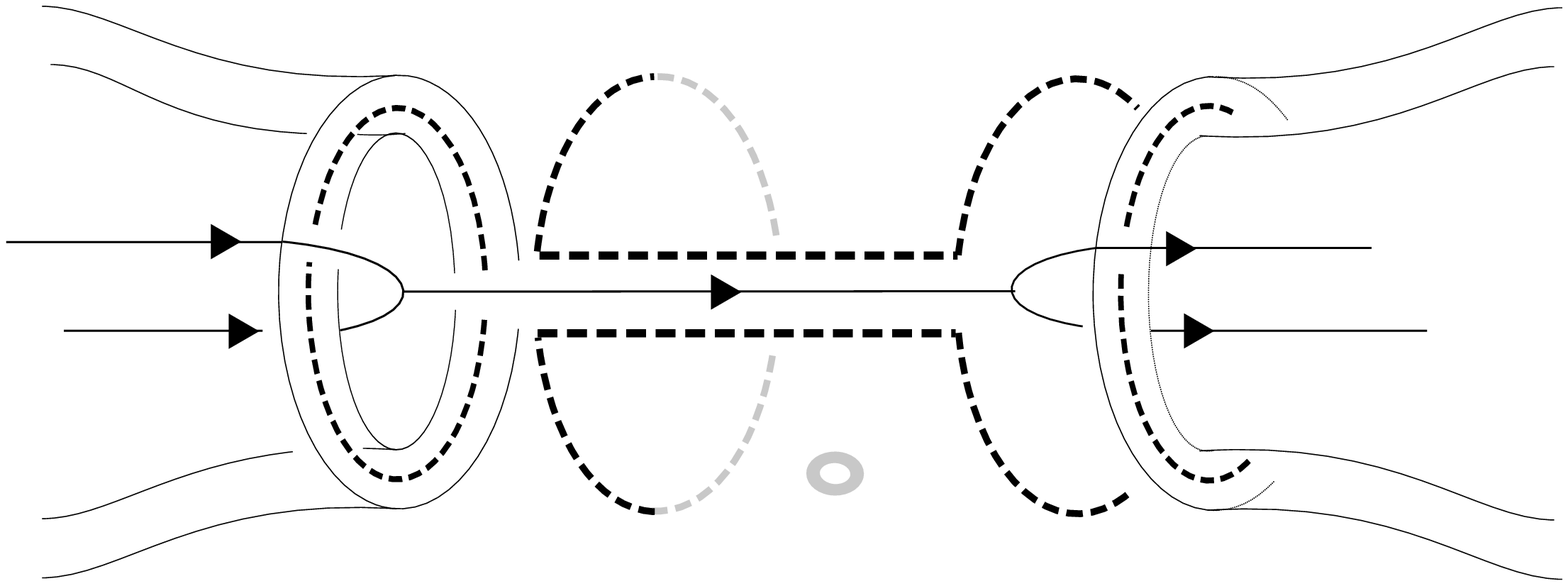}}}
\put(12,52){$#1$}
\put(12,25){$#2$}
\put(172,52){$#4$}
\put(172,25){$#3$}
\put(98,52){$#5$}
\end{picture}}}
\newcommand*{\tadpoleproofsecond}[5]{
\raisebox{-8.5ex}{
\begin{picture}(225,85)(0,0)
\put(0,0){\scalebox{0.35}{\protect\epsfig{file=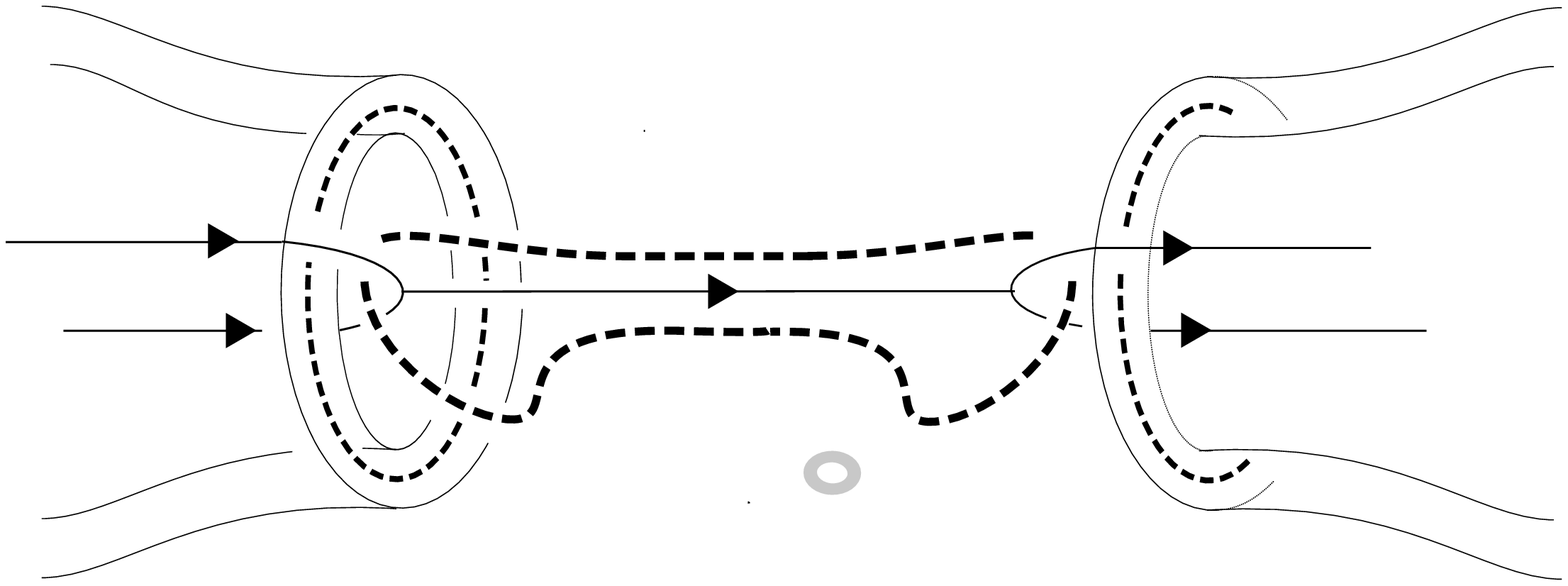}}}
\put(12,52){$#1$}
\put(12,25){$#2$}
\put(172,52){$#4$}
\put(172,25){$#3$}
\put(98,52){$#5$}
\end{picture}}}

We now prove the following more general version of \lemref{lem:tadpolegluing}, which describes the effect of gluing.

\begin{lemma} \label{lem:tadpolegluinggeneral}
Let $\widehat \cT$, $\widehat{\cT}_A$, $\widehat{\cT}_B$, $b_+$, $b_-$ and $e$ be as in \lemref{lem:tadpolegluing}.  Then the anyonic fusion basis state $\ket{\ell, d}_{\widehat \cT}$ is given by 
\begin{equation} \label{eq:generalizedtadpoletoprove}
\ket{\ell,d}_{\widehat \cT}
= \sum_n \sqrt{\frac{d_n}{d_{b_+}d_{b_-}}} \delta_{b_+b_-n^*} \ket{\ell^k_A,d_A}_{\widehat{\cT}_A} \ket{\ell^k_B,d_B}_{\widehat{\cT}_B} \ket{n}_e
\end{equation}
\end{lemma}

\begin{proof}
Consider the off-lattice ribbon graph $\ket{\ell,d}_\Sigma \in \cH_\Sigma$ described in terms of a ribbon graph on $\Sigma \times [-1,1]$.  By doubling a vacuum loop around a boundary component twice, and pulling out the loops, we can obtain two vacuum loops running parallel to the gluing curve~$\gamma$ if we introduce an additional factor of $\frac{1}{\cD^2}$.  Applying an $F$-move then results in
\begin{equation}
\frac{1}{\cD^2} \sum_m \sqrt{\frac{d_m}{d_id_j}} \delta_{ijm^*} \tadpoleproof{i}{j}{j}{i}{m}
\end{equation}
We can analyze the effect of applying the plaquette operator $B_q$ in the off-lattice picture: it adds a puncture, a vacuum loop around it and a factor of $\frac{1}{\cD}$.  It transforms this state to 
\begin{multline}
\frac{1}{\cD} \sum_m \sqrt{\frac{d_m}{d_id_j}} \delta_{ijm^*} \tadpoleprooffirst{i}{j}{j}{i}{m} \\
= \frac{1}{\cD} \sum_m \sqrt{\frac{d_m}{d_id_j}} \delta_{ijm^*} \tadpoleproofsecond{i}{j}{j}{i}{m}
\end{multline}
where we used property~\eqnref{it:secondpropertyvacuum} of vacuum lines.  With Eq.~\eqnref{eq:Fmovevacuumstring}, we conclude that this is equal to
\begin{equation}
\underbrace{\left(\frac{1}{d_i d_j} \sum_m \delta_{ijm^*} d_m \right)}_{=1\textrm{ by~\eqnref{eq:qdimensionsidentity}}}\cdot \sum_n \sqrt{\frac{d_n}{d_id_j}} \delta_{ijn^*} \tadpoleproof{i}{j}{j}{i}{n}
\end{equation} 
It is easy to see that applying the product $\prod_{p \neq q} B_p$ of the remaining plaquette operators results in the state on the right-hand side of Eq.~\eqnref{eq:generalizedtadpoletoprove}.  
\end{proof}

%% file: partitionfunction.tex
\newcommand*{\WRT}{\mathsf{WRT}}
\newcommand*{\partf}[1]{Z(#1)}

\section{The Witten-Reshetikhin-Turaev invariant} \label{sec:crane}

In \secref{sec:highergenuss}, we sketched how, based on a modular category~$\cC$, the Turaev-Viro construction  yields a representation $\rho_{D\cC, \Sigma}$ of the mapping class group $\MCG_\Sigma$ of a surface $\Sigma$.  Furthermore, this representation is  described by the doubled category $DC\cong \cC\otimes\cC^*$.  More generally, any modular category~$\cC$ gives rise to a (projective) representation
\begin{equation} \label{eq:generalprojectiverep}
\rho_{\cC,\Sigma}:\MCG_\Sigma\rightarrow \mathsf{GL}(\cH_{\cC,\Sigma})/<e^{2\pi ic/24}>
 \enspace ,
\end{equation}
where $\cH_{\cC,\Sigma}$ is a Hilbert space of anyonic fusion diagrams.  (In this expression, $<e^{i2\pi i c/24}>$ denotes the cyclic group generated by $e^{2\pi i c/24\varphi}\id$, where $c$~is a scalar called central charge.)  A certain matrix element of the representation~\eqnref{eq:generalprojectiverep} for the boundary~$\Sigma = \partial M_g$  of the g-handlebody $M_g$ (cf.~\figref{fig:ribbongraphclaim}) defines the Witten-Reshetikhin-Turaev (WRT) invariant~\cite{Witten89, ReshetikhinTuraev91}, as we now explain.  

A closed, oriented $3$-manifold $M$ can be represented by a Heegaard splitting $(g, x)$, i.e., a genus~$g \in \mathbb{N}$ and an element $x$ of the mapping class group $\MCG(\partial M_g)$ of the surface~$\partial M_g$.  Roughly, the element~$x$ specifies how to glue two copies of $M_g$ together to obtain~$M$.  The pair $(g, x)$ uniquely specifies the equivalence class of $M$ under homeomorphisms, up to~(i) a certain stabilization move~$(g, x)\rightarrow (g+1, \tilde x)$ which corresponds to attaching a $3$-sphere with a standard genus-$1$ Heegaard splitting, and~(ii) multiplication of~$x$ by the subgroup $\MCG^0_{\partial M_g}$ of mapping class group elements which extend to homeomorphisms of~$M_g$.  

\begin{figure}
\centering
\includegraphics[width=0.5\textwidth]{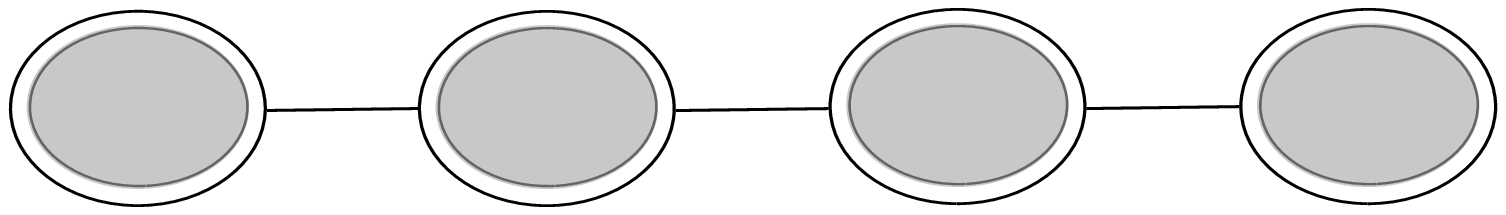}
\caption{The anyonic basis state $\ket{v_{\cC,g}}\in\cH_{\cC,\partial M_g}$, here for $g = 4$.  All edges carry the trivial label~$1$ of~$\cC$.} \label{fig:wrtstate}
\end{figure}

For a modular category~$\cC$, the Witten-Reshetikhin-Turaev invariant of a manifold~$M$ with Heegaard splitting~$(g, x)$ is given by 
\begin{equation} \label{eq:wrtinvariant}
\WRT_\cC(M) := \cD^{g-1} \bra{v_{\cC,g}} \rho_{\cC,g}(x) \ket{v_{\cC,g}}
 \enspace ,
\end{equation}
where $\ket{v_{\cC,g}}$ is the unit-normalized vector corresponding to the anyon diagram shown in~\figref{fig:wrtstate}.  Invariance follows from the fact that the vector $\ket{v_{\cC, g}}$ is invariant under the action of~$\MCG^0_{\partial M_g}$, and the factor $\cD^{g-1}$~compensates for stabilization moves.  Eq.~\eqnref{eq:wrtinvariant} is the Crane-Kohno-Kontsevich~\cite{Kontsevich88, Crane91, Kohno92} version of the WRT invariant.  Its original~\cite{ReshetikhinTuraev91} and more commonly known definition is based on the Dehn surgery presentation of manifolds.  Piunikhin~\cite{Piunikhin93} (see also~\cite[Section~2.4]{Kohno02}) showed that these definitions are equivalent (for $\cC = SU(2)_k$).  

For a closed $3$-manifold~$M$ and a modular category~$\cC$, definition~\eqnref{eq:tvdefinition} of the Turaev-Viro invariant reduces to 
\begin{equation} \label{eq:tvinvariantclosedmanifold}
\TV_\cC(M) = \cD^{-2|V_M|}\sum_{\textrm{labelings }\phi}\prod_{e} d_{\phi(e)} \prod_{\textrm{tetrahedra t}} g_t^\phi
 \enspace .
\end{equation}
The invariants~\eqnref{eq:wrtinvariant} and~\eqnref{eq:tvinvariantclosedmanifold} are related by the categorical double, that is,
\begin{equation} \label{eq:TVWRTmainstatement}
\TV_\cC(M) = \WRT_{D\cC}(M)
\enspace .
\end{equation}
Eq.~\eqnref{eq:TVWRTmainstatement} was shown by Walker and Turaev~\cite{Walker91, Turaev94book} (see also~\cite{Roberts95}) and holds more generally if~$\cC$ is a spherical category~\cite{Turaevpers10}.
As explained in \secref{sec:generalanyonmodels}, the double $D\cC$ takes the form $\cC\otimes\cC^*$ for a modular category~$\cC$ and~\eqnref{eq:TVWRTmainstatement} becomes
\begin{equation} \label{eq:TVWRTmodular}
\TV_\cC(M) = |\WRT_\cC(M)|^2
\enspace .
\end{equation}
This is because $\rho_{\cC \otimes \cC^*} \cong \rho_\cC \otimes \rho_{\cC^*}$ is a tensor-product representation
of~$\cC$ and its conjugate, the state of~\figref{fig:wrtstate} factorizes as $\ket{v_{\cC\otimes\cC^*,g}} \cong \ket{v_{\cC,g}} \otimes \ket{v_{\cC^*,g}}$, and the total quantum dimension of $\cC \otimes \cC^*$ is equal to~$\cD^2$. 
Eq.~\eqref{eq:TVWRTmodular} has direct application to quantum computing: it facilitates the proof in~\cite{AlagicJordanKoenigReichardt10TuraevViro} that approximating the Turaev-Viro invariant is a BQP-complete problem.

In the remainder of this section, we sketch a proof of~\eqnref{eq:TVWRTmainstatement} (for modular~$\cC$) based on our description of doubled anyonic fusion basis states arising in the Turaev-Viro code.  

A first step is the observation that, by definition of a Heegaard splitting, a manifold described by~$(g, x)$ is the result of gluing together three manifolds: two copies~$M_g$, $M_g'$ of the $g$-handlebody, and a ``mapping cylinder'' or cobordism defined by the mapping class group element~$x$.  Contracting the Turaev-Viro-tensor networks corresponding to the handlebodies gives two states $\ket{\partf{M_g}}$ and $\ket{\partf{M'_g}}$ on~$\cH_{\partial M_g}$, and~$\cH_{\partial M'_g}$, respectively, while the mapping cylinder gives rise to a linear map $\rho_{\cC \otimes \cC^*}(x): \cH_{\partial M_g} \rightarrow \cH_{\partial M'_g}$.  
The state $\ket{\partf{M_g}}$ is given by 
\begin{equation} \label{eq:partitionfunction}
\ket{\partf{M_g}} := \sum_\chi \TV_\cC(M_g, \chi) \ket \chi
\end{equation}
where the sum is over labelings~$\chi$ of the edges of the triangulated surface~$\partial M_g$.  (In the context of TQFTs, such states associated to 3-manifolds are sometimes called partition functions~\cite{Walker91}.)  Because of its definition in terms of the Turaev-Viro invariant, this vector is in the Turaev-Viro code subspace on $\partial M_g$, defined by the projection~\eqnref{eq:TVcodeprojection}.  A similar expression holds for~$\ket{\partf{M'_g}}$.  
The linear map $\rho_{\cC\otimes \cC^*}(x)$ can be interpreted as being composed of $F$-moves (i.e., local changes of the triangulation) because of its definition in terms of contractions of the $F$-tensor, and $\rho_{\cC \otimes \cC^*}$ is a representation of $\MCG(\partial M_g)$.  This is the representation~\eqnref{eq:generalprojectiverep} arising from the doubled category~$\cC \otimes \cC^*$ (cf.~\appref{sec:highergenuss}).  

Combining these three components, we conclude that the Turaev-Viro invariant has the form 
\begin{equation} \label{eq:TVvmgintermediate}
\TV_\cC(M) = \bra{\partf{M_g'}} \rho_{\cC \otimes \cC^*}(x) \ket{\partf{M_g}}
 \enspace ,
\end{equation}
which already bears some resemblance with the WRT invariant~\eqnref{eq:wrtinvariant}.  Indeed, Eq.~\eqnref{eq:TVWRTmainstatement} follows from~\eqnref{eq:TVvmgintermediate} and the following lemma, which determines the anyonic fusion basis state corresponding to $\ket{\partf{M_g}}$.  

\begin{lemma} \label{lem:partfunctionevaluation}
Consider a genus-$g$ handlebody $M_g$ with triangulated boundary~$\partial M_g$, and let $\ket{\partf{M_g}}$ be the Turaev-Viro codeword given by Eq.~\eqnref{eq:partitionfunction}.  
\begin{enumerate}[(i)]
\item\label{it:ribbongraphclaim}
 The off-lattice ribbon graph underlying $\ket{\partf{M_g}}$ (cf.~\lemref{lem:mainisomorphism}) consists of a vacuum loop around each handle, as in \figref{fig:ribbongraphclaim}.  
\item\label{it:secondnormclaim}
The state $\ket{\partf{M_g}}$ has squared norm~$\|\ket{\partf{M_g}}\|^2 = \cD^{2(g-1)}$.  
\end{enumerate}
In particular, up to a phase, $\ket{\partf{M_g}}$ is equal to $\cD^{g-1} \ket{v_{\cC \otimes \cC^*, g}}$, where $\ket{v_{\cC \otimes \cC^*, g}}$ is the normalized anyonic fusion basis state given in~\figref{fig:wrtstate}.  
\end{lemma}

\begin{figure}
\centering
\includegraphics[width=0.5\textwidth]{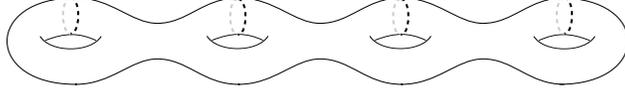}
\caption{The state $\ket{\partf{M_g}}$, here for $g = 4$, is the projection of this ribbon graph, with a vacuum loop around each handle.} \label{fig:ribbongraphclaim}
\end{figure}

The proof of this lemma proceeds by induction in the genus~$g$, using the fact that a genus-$g$ handlebody can be built from a genus-$(g-1)$~handlebody by gluing together two discs from the boundary, as illustrated in \figref{f:loops}.  To derive a gluing formula describing the transition from $\ket{\partf{M_{g-1}}}$ to $\ket{\partf{M_g}}$, we may locally choose a convenient triangulation, and consider the result of contracting the corresponding tensor indices against each other. For example, we may choose each disk to be a single triangle, and   moreover may take the triangle to be degenerate, with two of its sides the same edge.  In this case, the dual graph is locally a tadpole.  From the characterization of the Levin-Wen plaquette operators in \secref{sec:levinwengeneralized} each as adding a vacuum loop around a puncture within the plaquette, and Eqs.~\eqnref{it:secondpropertyvacuum} and~\eqnref{e:looprulegeneral}, the qudit along the tail of a tadpole is fixed to $\vac$.  Hence the state $\ket{\partf{M_{g-1}}}$ can be factored as 
\begin{equation}
\ket{\partf{M_{g-1}}} = \ket{v'} \otimes \ket{\vac\vac} \otimes \Big( \frac{1}{\cD} \sum_i d_i \ket i \Big)^{\otimes 2}
 \enspace ,
\end{equation}
where the second term corresponds to the two tadpole tails, the third term is the state of the tadpole heads, and $\ket{v'}$ is a state on the remaining qudits.  This expression then implies the desired result for $\ket{\partf{M_{g}}}$ when using the definition of $\TV_\cC$.  Below we will give an alternative proof that avoids using degenerate triangulations.  

\begin{figure}
\centering
\begin{equation*}
\raisebox{-36pt}{\includegraphics[scale=1]{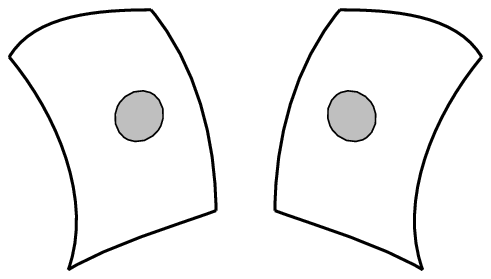}}
\;\longrightarrow\;
\raisebox{-36pt}{\includegraphics[scale=1]{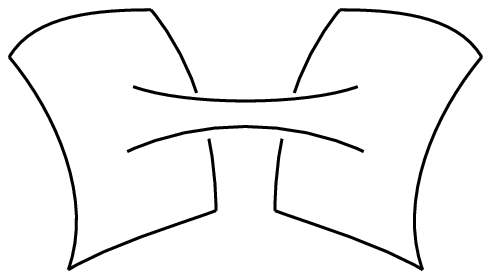}}
\end{equation*}
\caption{Two surfaces can be joined by puncturing each and identifying the boundaries of the holes, as in \secref{sec:statepreparation}.  In the case that the surfaces are the same, the boundary of a $3$-manifold, gluing together two disks from the boundary adds a handle to the manifold.  The Turaev-Viro state for the resulting manifold has a vacuum loop going around the handle.} \label{f:loops}
\end{figure}

\def\dot {\psdots[dotscale=1](0,0)}
\def\puncture {\psdots[dotstyle=+,dotangle=45,dotscale=1](0,0)}

\begin{proof}
We first prove~\eqnref{it:ribbongraphclaim} by induction in the genus~$g$.  For the genus-zero case, the Turaev-Viro codespace on the sphere is one-dimensional, and given by the projection of the ``empty'' ribbon graph; this is the claim~\eqnref{it:ribbongraphclaim}.  For the induction step, we use the gluing procedure  illustrated in \figref{f:loops}.  This is equivalent to taking two patches of surface with the free tensor indices all~$1$, adding a layer of tetrahedra to project to the Turaev-Viro code, and then identifying the indices for patches on either side.  Therefore, the new tensor that we are contracting on corresponds to a sphere with two separated disks of indices all fixed to $1$.  By retriangulating, we may take this sphere to consist of just three tetrahedra, as in \figref{f:gluing_threetetrahedra}.  We want to understand this tensor with the six indices marked $\newdot$ fixed to $1$.  This leaves six free indices, marked $\newpuncture$.  By the fusion constraints on the associated ribbon graph, these indices must all match.  It remains only to determine the coefficient of this shared index; call it $i$.  

\begin{figure}
\centering
\begin{tabular}{c@{$\qquad\qquad\qquad\qquad$}c}
\subfigure[]{\label{f:gluing_threetetrahedra}
\raisebox{-30pt}{\includegraphics[scale=1]{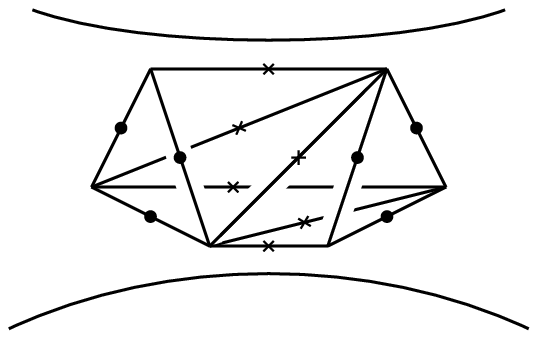}}
}
&
\raisebox{15pt}{\subfigure[]{\label{f:locallabeling}
\raisebox{-20pt}{
\!\!\!\!\!\!\!\!\!\!\!\!\!\!\!\!\!\!\!\!\!\!\!\!\!\!\!\!\!\!
\doubletetrahedron \vspace{40pt}
}}}
\end{tabular}
\caption{Two triangulations of a handle, used in the proof of \lemref{lem:partfunctionevaluation}.}
\end{figure}

Up to a power of $\cD$, we can read off from Eqs.~\eqnref{eq:tvdefinition} and~\eqnref{eq:tvdefinitionFmatrix} the coefficient as 
\begin{equation}
\sqrt{d_i}^6 \frac{F^{111}_{i i^* i}}{\sqrt{d_i}} \frac{F^{111}_{i i^* i}}{\sqrt{d_i}} F^{i i^* 1}_{i i* 1}
 \enspace ,
\end{equation}
where the first term, $\sqrt{d_i}^6$, comes from the term $\prod_{e \in \partial M} \sqrt{d_{\varphi(e)}}$ in the Turaev-Viro expansion, and the other terms are the tensors for the three tetrahedra.  Now substitute in $F^{i i^* 1}_{j j^* k} = \sqrt{\frac{d_k}{d_i d_j}} \delta_{i j k}$, so $F^{111}_{i i^* i} = 1$ and $F^{i i^* 1}_{i i^* 1} = 1/d_i$.  We find that the coefficient is proportional to $d_i$, as claimed.  

For the proof of~\eqnref{it:secondnormclaim}, note that the squared norm of the vector $\ket{\partf{M}}$ associated with a 3-manifold~$M$ is equal to the Turaev-Viro invariant of the closed $3$-manifold obtained by gluing~$M$ to itself.   In particular, the squared norm is independent of the triangulation of~$M$.  We can therefore again choose convenient triangulations.  

As the ball is triangulated by a single tetrahedron, as in Eq.~\eqnref{eq:tvdefinitionFmatrix}, we immediately obtain
\begin{equation}
\norm{\ket{\partf{M_0}}}^2 = \TV_\cC(M_0) = \cD^{-8} \sum_{ijk\ell m n} |F^{i^*jm}_{k\ell^*n}|^2 d_i d_j d_k d_\ell
 \enspace .
\end{equation}
By unitarity, $\sum_n |F^{i^*jm}_{k\ell^*n}|^2 = \delta_{i^*jm^*} \delta_{km\ell^*}$, so applying Eq.~\eqnref{eq:qdimensionsidentity} repeatedly gives
\begin{equation} \label{eq:ballinvariantvec}
\norm{\ket{\partf{M_0}}}^2 = \frac{1}{\cD^2}
 \enspace .
\end{equation}
for the genus-$0$ case.  

To compute the norm of $\ket{\partf{M_g}}$ for $g \geq 1$, fix a triangulation of $M_{g-1}$ and consider the triangulation of $M_g$ obtained by connecting two triangular faces on $\partial M_{g-1}$ using two tetrahedra as shown in \figref{f:locallabeling}.  In this figure, the face with edge labels $(m,j,i)$ is glued to a triangle on $\partial M_{g-1}$, while the face with edge labels $(\ell, q, s)$ is glued to another such triangle.  We claim that 
\begin{equation} \label{eq:inductproofnorm}
\norm{\ket{\partf{M_g}}}^2 = \cD^2 \norm{\ket{\partf{M_{g-1}}}}^2
 \enspace .
\end{equation}
Using~\eqnref{eq:ballinvariantvec}, the claim~\eqnref{it:secondnormclaim} immediately follows from~\eqnref{eq:inductproofnorm}.  

To show~\eqnref{eq:inductproofnorm}, observe that the contribution to $\TV(M, \chi)$ of a labeling~$\chi$ locally matching \figref{f:locallabeling} is given by 
\begin{equation}
\frac{F^{i^*jm}_{k\ell^* n}}{\sqrt{d_m d_n}} \frac{F^{q^*ri}_{n\ell^* s}}{\sqrt{d_i d_s}} \sqrt{d_i d_j d_k d_\ell d_m d_n d_q d_r d_s}
 \enspace ,
\end{equation}
since all edges are on~$\partial M$.  On the other hand, the restriction $\chi' = \chi|_{M_{g-1}}$ gives a contribution 
\begin{equation}
\sqrt{d_m d_i d_j} \cdot \sqrt{d_q d_\ell d_s}
\end{equation}
to $\TV(\chi', M_{g-1})$.  We can therefore write 
\begin{equation} \label{eq:newnorm}
\ket{\partf{M_g}} = \cD \, W \ket{\partf{M_{g-1}}} \ \textrm{ where }\ W = \sum_{ijk\ell mnqrs} \frac{\sqrt{d_rd_k}}{\sqrt{d_m d_i d_s}} F^{q^*ri}_{n\ell^* s}F^{i^*jm}_{k\ell^*n} \ket{ijk\ell mnqrs}(\bra{ijm} \otimes \bra{q \ell s})
\end{equation}
In~\eqnref{eq:newnorm}, the additional factor of $\cD$ accounts for the fact that the number of vertices on~$\partial M$ is equal to $|V_{\partial M}| = |V_{\partial M'}| + |V_{\partial M''}| - 1$ since the two faces share a vertex after gluing.  Using the tetrahedral symmetry  of the $F$-tensor, we get 
\begin{equation}
W^\dagger W = \sum_{ijmq\ell s} \underbrace{\frac{1}{d_m d_\ell} \left(\sum_{knr} d_k |F^{ni^*\ell^*}_{qs^*r}|^2 \, |F^{i^*jm}_{k\ell n}|^2\right)}_{\substack{=:X(ijmq\ell s)}}(\ket{ijm}\bra{ijm} \otimes \ket{q\ell s} \bra{q\ell s})
 \enspace .
\end{equation}
We claim that 
\begin{equation} \label{eq:xijmidentity}
X(ijmq\ell s) = \delta_{\ell^* s^* q} \delta_{i^*jm^*} \qquad\textrm{ for all }\qquad (ijmq\ell s)
 \enspace ,
\end{equation}
that is, $W$ is an isometry on the span of valid labelings of $\partial M_{g-1}$.  This proves the claim~\eqnref{eq:inductproofnorm} and concludes the proof of~\eqnref{it:secondnormclaim}.  Indeed, we have, again by unitarity of the basis change, 
\begin{equation}
\sum_r |F^{ni^*\ell^*}_{qs^*r}|^2 = \delta_{ni^*\ell}\delta_{\ell^* s^* q}
 \enspace ,
\end{equation}
and therefore 
\begin{equation}
\sum_{nr} |F^{ni^*\ell^*}_{qs^*r}|^2 \, |F^{i^*jm}_{k\ell n}|^2 = \delta_{\ell^* s^* q}\delta_{i^*jm^*}\delta_{mk\ell}
 \enspace .
\end{equation}
Multiplying by $\frac{d_k}{d_m d_\ell}$ and summing over $k$ using~\eqnref{eq:qdimensionsidentity} gives~\eqnref{eq:xijmidentity}.  

Having shown~\eqnref{it:ribbongraphclaim} and~\eqnref{it:secondnormclaim}, we can reexpress~$\ket{\partf{M_g}}$ in terms of the anyonic fusion basis as follows.  First, we use properties~\eqnref{it:secondpropertyvacuum} and~\eqnref{it:firstpropertyvacuum} of \lemref{lem:vacuum} to add vacuum lines around the $g - 1$ constrictions  between the handles in~\figref{fig:ribbongraphclaim}.  Since the resulting vacuum lines correspond to a pants decomposition of $\partial M_g$, we can use the correspondence between doubled fusion basis states and ribbon graphs to conclude that the state $\ket{\partf{M_g}}$ is proportional to the anyonic fusion basis state~$\ket{v_{\cC\otimes \cC^*,g}}$.  The proportionality constant is determined by the normalization~\eqnref{it:secondnormclaim}.  
\end{proof}

%% file: main.bbl
\newcommand{\etalchar}[1]{$^{#1}$}
\begin{thebibliography}{FKLW03}
\expandafter\ifx\csname urlprefix\endcsname\relax\def\urlprefix{URL }\fi
\providecommand{\arxiv}[2][]{\href{http://arxiv.org/pdf/#2}{\texttt{arXiv:#2}}}
\providecommand{\doi}[2][]{\href{http://dx.doi.org/#2}{\texttt{doi:#2}}}

\bibitem[AJKR10]{AlagicJordanKoenigReichardt10TuraevViro}
Gorjan Alagic, Stephen~P. Jordan, Robert K\"onig, and Ben~W. Reichardt.
\newblock Approximating the {T}uraev-{V}iro invariant is {BQP}-complete.
\newblock in preparation, 2010.

\bibitem[BHZS05]{Bonesteeletal05}
N.E. Bonesteel, L.~Hormozi, G.~Zikos, and S.~H. Simon.
\newblock Braid topologies for quantum computation.
\newblock {\em Phys. Rev. Lett.}, 95:140503, 2005,
  \href{http://www.arxiv.org/abs/quant-ph/0505065}{{arXiv:quant-ph/0505065}}.

\bibitem[BW96]{BarrettWestbury96invariants}
John~W. Barrett and Bruce~W. Westbury.
\newblock Invariants of piecewise-linear 3-manifolds.
\newblock {\em Trans. Amer. Math. Soc.}, 348:3997--4022, 1996,
  \href{http://www.arxiv.org/abs/hep-th/9311155}{{arXiv:hep-th/9311155}}.

\bibitem[BW99]{BarrettWestbury93spherical}
John~W. Barrett and Bruce~W. Westbury.
\newblock Spherical categories.
\newblock {\em Adv. Math.}, 143:357--375, 1999,
  \href{http://www.arxiv.org/abs/hep-th/9310164}{{arXiv:hep-th/9310164}}.

\bibitem[Cra91]{Crane91}
L.~Crane.
\newblock 2-d physics and 3-d topology.
\newblock {\em Comm. Math. Phys.}, 135(3):615--640, 1991.

\bibitem[DKLP02]{Dennisetal02}
Eric Dennis, Alexei Kitaev, Andrew Landahl, and John Preskill.
\newblock Topological quantum memory.
\newblock {\em J. Math. Phys.}, 43:4452--4505, 2002,
  \href{http://www.arxiv.org/abs/quant-ph/0110143}{{arXiv:quant-ph/0110143}}.

\bibitem[FFN{\etalchar{+}}09]{Fidkowskietal08}
L.~Fidkowski, M.~Freedman, C.~Nayak, K.~Walker, and Z.~Wang.
\newblock From string nets to nonabelions.
\newblock {\em Comm. Math. Phys.}, 287(3), 2009,
  \href{http://www.arxiv.org/abs/cond-mat/0610583}{{arXiv:cond-mat/0610583}}.

\bibitem[FKLW03]{Freedmanetal03}
Michael~H. Freedman, Alexei Kitaev, Michael~J. Larsen, and Zhenghan Wang.
\newblock Topological quantum computation.
\newblock {\em Bull. Amer. Math. Soc.}, 40(1):31--38, 2003,
  \href{http://www.arxiv.org/abs/quant-ph/0101025}{{arXiv:quant-ph/0101025}}.

\bibitem[FKW02]{Freedmanetal02}
Michael~H. Freedman, Alexei Kitaev, and Zhenghan Wang.
\newblock Simulation of topological field theories by quantum computers.
\newblock {\em Comm. Math. Phys.}, 227(3):587--603, 2002.

\bibitem[FLW02]{Freedmanuniversality00}
Michael~H. Freedman, Michael~J. Larsen, and Zhenghan Wang.
\newblock A modular functor which is universal for quantum computation.
\newblock {\em Comm. Math. Phys.}, 227(3):605--622, June 2002,
  \href{http://www.arxiv.org/abs/quant-ph/0001108}{{arXiv:quant-ph/0001108}}.

\bibitem[Fre00]{Freedman00}
Michael~H. Freedman.
\newblock Quantum computation and the localization of modular functors, 2000,
  \href{http://www.arxiv.org/abs/quant-ph/0003128}{{arXiv:quant-ph/0003128}}.

\bibitem[FSG09]{Fowleretal08}
Austin~G. Fowler, Ashley~M. Stephens, and Peter Groszkowski.
\newblock High threshold universal quantum computation on the surface code.
\newblock {\em Phys. Rev. A}, 80:052312, 2009,
  \href{http://www.arxiv.org/abs/0803.0272}{{arXiv:0803.0272 [quant-ph]}}.

\bibitem[HZBS07]{Hormozi07}
L.~Hormozi, G.~Zikos, N.E. Bonesteel, and S.H. Simon.
\newblock Topological quantum compiling.
\newblock {\em Phys. Rev. B}, 75:165310, 2007,
  \href{http://www.arxiv.org/abs/quant-ph/0610111}{{arXiv:quant-ph/0610111}}.

\bibitem[Kas95]{Kassel95}
Christian Kassel.
\newblock {\em Quantum Groups}.
\newblock Springer, New York, 1995.

\bibitem[Kit03]{Kitaev03}
Alexei Kitaev.
\newblock Fault-tolerant quantum computation by anyons.
\newblock {\em Ann. Phys.}, 303(1):2--30, 2003,
  \href{http://www.arxiv.org/abs/quant-ph/9707021}{{arXiv:quant-ph/9707021}}.

\bibitem[Kit06]{KitaevAnyons}
Alexei Kitaev.
\newblock Anyons in an exactly solved model and beyond.
\newblock {\em Ann. Phys.}, 321(1):2--111, 2006,
  \href{http://www.arxiv.org/abs/cond-mat/0506438v3
  [cond-mat.mes-hall]}{{arXiv:cond-mat/0506438v3 [cond-mat.mes-hall]}}.

\bibitem[Koh92]{Kohno92}
Toshitake Kohno.
\newblock Topological invariants of 3-manifolds using representations of
  mapping class groups {I}.
\newblock {\em Topology}, 31(2):203--230, 1992.

\bibitem[Koh02]{Kohno02}
Toshitake Kohno.
\newblock {\em Conformal Field Theory and Topology}, volume 210 of {\em
  Translations of Mathematical Monographs}, chapter 2.4.
\newblock American Mathematical Society, 2002.

\bibitem[Kon88]{Kontsevich88}
M.~Kontsevich.
\newblock Rational conformal field theory and invariants of 3-manifolds.
\newblock preprint of the Centre de Physique Theorique Marseille, CPT-88/p2189,
  1988.

\bibitem[K{\"{o}}n09]{Koenig09distillation}
Robert K{\"{o}}nig.
\newblock Composite anyon coding and the initialization of a topological
  quantum computer.
\newblock 2009, \href{http://www.arxiv.org/abs/0910.2427}{{arXiv:0910.2427
  [quant-ph]}}.

\bibitem[KRV09]{KoeReiVid09}
Robert K\"onig, Ben~W. Reichardt, and Guifr{\'e} Vidal.
\newblock Exact entanglement renormalization for string-net models.
\newblock {\em Phys. Rev. B}, 79:195123, 2009,
  \href{http://www.arxiv.org/abs/0806.4583}{{arXiv:0806.4583 [cond-mat]}}.

\bibitem[KSV99]{Kitaevbook}
A.~Yu. Kitaev, A.H. Shen, and M.N. Vyalyi.
\newblock {\em Classical and Quantum Computation}.
\newblock American Mathematical Society, Providence, 1999.

\bibitem[LW05a]{LarsenWang05}
Michael Larsen and Zhenghan Wang.
\newblock Density of the {$SO(3)$} {TQFT} representation of mapping class
  groups.
\newblock {\em Comm. Math. Phys.}, 260(3):641--658, 2005,
  \href{http://www.arxiv.org/abs/math/0408161}{{arXiv:math/0408161}}.

\bibitem[LW05b]{LevinWen}
Michael~A. Levin and Xiao-Gang Wen.
\newblock String-net condensation: A physical mechanism for topological phases.
\newblock {\em Phys. Rev. B}, 71:045110, 2005,
  \href{http://www.arxiv.org/abs/cond-mat/0404617
  [cond-mat.str-el]}{{arXiv:cond-mat/0404617 [cond-mat.str-el]}}.

\bibitem[{Mac}98]{MacLane98}
Saunders {Mac Lane}.
\newblock {\em Categories for the working mathematician}.
\newblock Springer, 1998.

\bibitem[M{\"u}g03a]{Mue03a}
Michael M{\"u}ger.
\newblock From subfactors to categories and topology {I}: {F}robenius algebras
  in and {M}orita equivalence of tensor categories.
\newblock {\em Journal of Pure and Applied Algebra}, 180(1-2):81--157, 2003.

\bibitem[M{\"u}g03b]{Mue03}
Michael M{\"u}ger.
\newblock From subfactors to categories and topology {II}: {T}he quantum double
  of tensor categories and subfactors.
\newblock {\em Journal of Pure and Applied Algebra}, 180(1-2):159--219, 2003.

\bibitem[NSS{\etalchar{+}}08]{Nayaketal08}
Chetan Nayak, Steven~H. Simon, Ady Stern, Michael~H. Freedman, and Sankar~Das
  Sarma.
\newblock Non-abelian anyons and topological quantum computation.
\newblock {\em Rev. Mod. Phys.}, 80(1083), 2008,
  \href{http://www.arxiv.org/abs/0707.1889}{{arXiv:0707.1889
  [cond-mat.str-el]}}.

\bibitem[Piu93]{Piunikhin93}
Sergey Piunikhin.
\newblock Reshetikhin-{T}uraev and {C}rane-{K}ohno-{K}ontsevich 3-manifold
  invariants coincide.
\newblock {\em J. Knot Theor. Ramif.}, 2(1):65--95, 1993.

\bibitem[Pre04]{preskill}
John Preskill.
\newblock Topological quantum computation.
\newblock Lecture Notes (Chapter 9), 2004.
\newblock \href{http://theory.caltech.edu/people/preskill/ph229/}{http://theory.caltech.edu/people/preskill/ph229/}.

\bibitem[RH07]{Raussendorfharrington06}
Robert Raussendorf and Jim Harrington.
\newblock Fault-tolerant quantum computation with high threshold in two
  dimensions.
\newblock {\em Phys. Rev. Lett.}, 98:190504, 2007,
  \href{http://www.arxiv.org/abs/quant-ph/0610082}{{arXiv:quant-ph/0610082}}.

\bibitem[RHG06]{Raussendorfetal05}
Robert Raussendorf, Jim Harrington, and Kovid Goyal.
\newblock A fault-tolerant one-way quantum computer.
\newblock {\em Ann. Phys.}, 321(2242), 2006,
  \href{http://www.arxiv.org/abs/quant-ph/0510135}{{arXiv:quant-ph/0510135}}.

\bibitem[Rob95]{Roberts95}
Justin~D. Roberts.
\newblock Skein theory and {T}uraev-{V}iro invariants.
\newblock {\em Topology}, 34(4):771--787, 1995.

\bibitem[RSW09]{RowStonZhen09}
E.~Rowell, R.~Stong, and Z.~Wang.
\newblock On classification of modular tensor categories.
\newblock {\em Comm. Math. Phys.}, 292(2):303--605, 2009,
  \href{http://www.arxiv.org/abs/0712.1377}{{arXiv:0712.1377v3 [quant-ph]}}.

\bibitem[RT90]{ReshTur90}
N.Y. Reshetikhin and V.G. Turaev.
\newblock Ribbon graphs and their invaraints derived from quantum groups.
\newblock {\em Comm. Math. Phys.}, 127(1):1--26, 1990.

\bibitem[RT91]{ReshetikhinTuraev91}
N.Y. Reshetikhin and V.G. Turaev.
\newblock Invariants of 3-manifolds via link polynomials and quantum groups.
\newblock {\em Invent. Math.}, 103(3):547--597, 1991.

\bibitem[Tur94]{Turaev94book}
V.G. Turaev.
\newblock {\em Quantum invariants of knots and 3-manifolds}, volume~18 of {\em
  de Gruyter studies in mathematics}.
\newblock de Gruyter, New York, 1994.

\bibitem[Tur10]{Turaevpers10}
V.G. Turaev, 2010.
\newblock personal communication.

\bibitem[TV92]{TuraevViro92topology}
V.G. Turaev and O.Y. Viro.
\newblock State sum invariants of $3$-manifold and quantum $6j$-symbols.
\newblock {\em Topology}, 31(4):865--902, 1992.

\bibitem[Wal91]{Walker91}
Kevin Walker.
\newblock On {W}itten's 3-manifold invariants, 1991.
\newblock \href{http://canyon23.net/math/}{http://canyon23.net/math/}.

\bibitem[Wit89]{Witten89}
Edward Witten.
\newblock Quantum field theory and the {J}ones polynomial.
\newblock {\em Comm. Math. Phys.}, 121(3):351--399, 1989.

\end{thebibliography}
